\documentclass{sig-alternate-mod}

\newcommand{\T}{\mbox{${\cal T}$}}

\renewcommand{\P}{\mbox{${\cal P}$}}

\newtheorem{ex}{EXAMPLE}[section]
\newenvironment{example}{\begin{ex} \nopagebreak
  \begin{rm}}{{\hfill$\Box$}\end{rm}\end{ex}} 

\newtheorem{defin}{Definition}[section]
\newenvironment{definition}[1]{\begin{defin}\begin{rm}({\bf #1})}{{\hfill$\Box$}\end{rm}\end{defin}}

\newtheorem{lemm}{Lemma}[section]
\newenvironment{lemma}{\begin{lemm}}{{\hfill$\Box$}\end{lemm}}

\newtheorem{thm}{Theorem}[section]
\newenvironment{theorem}{\begin{thm} \nopagebreak}{{\hfill$\Box$}\end{thm}}

\newtheorem{prop}{Proposition}[section]
\newenvironment{proposition}{\begin{prop}}{{\hfill$\Box$}\end{prop}}

\newtheorem{corol}{Corollary}[section]
\newenvironment{corollary}{\begin{corol} \nopagebreak}{{\hfill$\Box$}\end{corol}}

\newtheorem{conjec}{Conjecture}[section]
\newenvironment{conjecture}{\begin{conjec} \nopagebreak}{{\hfill$\Box$}\end{conjec}}

\newsavebox{\savepar}

\newcommand{\x}[1]{\mbox{\ensuremath{#1}}}

\newcommand{\squishlist}{
  \begin{list}{$\bullet$}
   {
     \setlength{\itemsep}{0pt}
     \setlength{\parsep}{0pt}
     \setlength{\topsep}{0pt}
     \setlength{\partopsep}{0pt}
     \setlength{\leftmargin}{1.5em}
     \setlength{\labelwidth}{1em}
     \setlength{\labelsep}{0.5em} } }
\newcommand{\squishend}{
   \end{list}  }

\newcommand{\nop}[1]{}                       

\newcommand{\reminder}[2]{}

\begin{document}

\title
{Combined-Semantics Equivalence Is Decidable \\ For a Practical Class of Conjunctive Queries
}

\numberofauthors{1} 

\author{
\alignauthor Rada Chirkova
\\
	\affaddr{Department of Computer Science} \\
	\affaddr{NC State University, Raleigh, NC 27695, USA} \\
	\email{chirkova@csc.ncsu.edu}
}

\maketitle

\begin{abstract} 


The problems of query containment and equivalence are  fundamental problems in the context of query processing and optimization. In their classic work \cite{ChandraM77} published in 1977, Chandra and Merlin solved the two problems for the language of conjunctive queries {\em (CQ queries)} on relational data, under the ``set-semantics'' assumption for query evaluation. 
Alternative semantics, called {\em bag} and {\em bag-set semantics}, 
have been studied since 1993; 
Chaudhuri and Vardi in \cite{ChaudhuriV93} outlined necessary and sufficient conditions for equivalence of CQ queries under these semantics. (The problems of containment of CQ bag and bag-set queries remain open to this day.)  More recently, Cohen \cite{Cohen06-pods,Cohen06} introduced a formalism for treating (generalizations of) CQ queries evaluated under each of set, 
bag, and bag-set 
semantics uniformly as special cases of the more general {\em combined semantics.} 
This formalism provides tools for studying broader classes of practical SQL queries, specifically important types of queries that arise in on-line analytical processing (OLAP). 
Cohen in  \cite{Cohen06} provides a sufficient condition for equivalence of (generalizations of) combined-semantics CQ queries, as well as sufficient and necessary equivalence conditions for several proper sublanguages of the query language of  \cite{Cohen06}. 

Our goal in this paper is to continue the study of equivalence of CQ queries. 
We focus on the problem of determining whether two CQ queries are combined-semantics equivalent. We continue the tradition of \cite{ChandraM77,ChaudhuriV93,Cohen06} of studying this problem using the tool of containment between queries. 
This paper introduces a syntactic necessary and sufficient condition for equivalence of  queries belonging to a large natural language of ``explicit-wave'' combined-semantics CQ queries;  this language encompasses (but is not limited to) all set, bag, and bag-set queries, and appears to cover all combined-semantics CQ queries that are expressible in SQL. Our result solves in the positive the decidability problem of determining  combined-semantics equivalence for pairs of explicit-wave CQ queries. That is, for an arbitrary pair of combined-semantics CQ queries, it is decidable (i) to determine whether each of the queries is explicit wave, and (ii) to determine, in case both queries are explicit wave, whether or not they are combined-semantics equivalent, by using our syntactic criterion. (The problem of determining equivalence for general combined-semantics CQ queries remains open. Even so, our syntactic sufficient containment condition could still be used to determine that two general CQ queries are combined-semantics equivalent.)  
Our equivalence test, as well as our general sufficient condition for containment of combined-semantics CQ queries, reduce correctly to 
the special cases reported in \cite{ChandraM77,ChaudhuriV93} for set, bag, and bag-set semantics.    
Our containment and equivalence conditions also properly generalize the results of \cite{Cohen06}, provided that the latter are restricted to the language of (combined-semantics) CQ queries. 
\end{abstract}

\vspace{-0.2cm} 

\section{Introduction} 
\label{intro-sec} 





Query containment and equivalence are recognized as fundamental problems in  evaluation and optimization of databa- se queries. The reason is, for  conjunctive queries {\em (CQ queries)} --- a broad class of frequently used queries, whose expressive power is sufficient to express select-project-join queries in relational algebra --- query equivalence can be used as a tool in query optimization. 
Specifically, to find a more efficient {\it and} answer-preserving formulation of a given CQ query, it is enough to ``try all ways'' of arriving at a ``shorter'' query formulation, by removing query subgoals, in a process  called query minimization~\cite{ChandraM77}. A subgoal-removal step succeeds only if equivalence (via containment) of the ``original'' and ``shorter'' query formulations can be ensured. The equivalence test of~\cite{ChandraM77} for CQ queries is NP complete, whereas equivalence of general relational queries is undecidable. 

The query-minimization algorithm of \cite{ChandraM77} works under the assumption of {\em set semantics} for query evaluation, where both the database (stored) relations and query answers are treated as sets. Query answering and reformulation in the set-semantics setting have been studied extensively in the database-theory literature. As a basis, these studies have all used the necessary and sufficient containment condition of \cite{ChandraM77} for CQ queries. At the same time, the set semantics is not the default query-evaluation semantics in database systems in practice. For instance, in the standard relational query language SQL, duplicates are removed from the answer to a SQL query {\em only} if the query uses the  {\tt DISTINCT} keyword in its {\tt SELECT}  clause. 
This and other discrepancies between the set semantics for query evaluation and the standard of 
SQL have prompted researchers \cite{ChaudhuriV93,IoannidisR95} to consider ``bag semantics'' and ``bag-set semantics'' for query evaluation. Under {\em bag semantics,} both query answers and stored relations are treated as {\em bags} (that is, {\em multisets).} Under {\em bag-set semantics,} query answers are treated as bags, whereas the database relations are assumed to be sets.  

In an extended abstract \cite{ChaudhuriV93} published in PODS in 1993, Chaudhuri and Vardi focused on the hard problem of bag containment for CQ queries. The paper \cite{ChaudhuriV93} formulates containment and equivalence results, including equivalence tests, for bag and bag-set queries. However, the full version of the paper \cite{ChaudhuriV93} has never appeared, and the problems of bag and bag-set containment for CQ queries remain open to this day. 

The seminal work by Cohen \cite{Cohen06-pods,Cohen06} has provided a Datalog-based formalism for treating queries evaluated under each of set, 
bag, and bag-set 
semantics uniformly as special cases of the more general {\em combined semantics}. 
The focus of the papers \cite{Cohen06-pods,Cohen06} was on formulating conditions for {\em combined-semantics equivalence} of queries expressed using that syntax. Intuitively, two queries are combined-semantics equivalent, denoted $\equiv_C$, if on each database, the two queries return the same answers, with the same multiplicity of each answer tuple.  
To show the practical value of the combined-semantics formalism, Cohen exhibited in \cite{Cohen06} a number of real-life SQL queries that can be expressed as combined-semantics CQ queries, but cannot be expressed using set, bag, or bag-set semantics. In the following example, we show three realistic combined-semantics queries, which we then use to illustrate the contributions of this paper.

\begin{example} 
\label{nex-ex} 
The  application domain that we use here is based on a data-warehousing example from \cite{MumickBook}. 
Consider a retailer that has multiple stores. 
The retailer carries many items and has an elaborate relational database/warehouse for analysis, marketing, and promotion purposes. One of the tables in the database has the schema {\tt Pos(transactionID, storeID, itemID, date, amount)}; the table has one million rows. 
This table represents point-of-sale transactions, with one tuple for every item sold in a transaction. Each tuple has the transaction ID, the ID of the item sold, the ID of the store selling it, the date, and the amount of the sale. 

The business-development division of this store chain would typically analyze general trends and high-level aggregates over the {\tt Pos} relation. We assume that for this reason, the data analysts in that division access data through a {\em view,} with the schema {\tt Sales(storeID, itemID, date, amount)},  rather\linebreak than working with the ``raw'' data in the {\tt Pos} relation. To ensure correctness of the data-analyses results, the {\tt Sales} relation is a {\em bag-valued} relation; that is, for each tuple with specific values of the attributes {\tt storeID}, {\tt itemID}, {\tt date}, and {\tt amount}, the {\tt Sales} relation makes available as many copies of this tuple as there are rows with this information in the base relation {\tt Pos}. 

Suppose that this business-development division of the store chain would like to study the impact, on the total sales, of those transactions in the stores where the item prices are the same as on a fixed date,\footnote{We assume that the transaction amount can be used to determine the price of the item. This is true, for instance, for sales of big-ticket items, where each transaction typically records the sale of one such item. We also assume that item prices do not change in the middle of a business day.} January 1, 2012. Consider a SQL query  {\tt Qc} that could be used for the purpose of this analysis. 

{\small 
\begin{verbatim}  
(Qc) SELECT storeID, amount FROM Sales S 
     WHERE EXISTS 
        (SELECT * FROM Sales 
         WHERE storeID = S.storeID AND itemID = S.itemID 
         AND amount = S.amount AND date = `01/01/12')
\end{verbatim}
} 

For each store ID, query {\tt Qc} returns separately the amount for each transaction that took place in that store on the date {\tt 01/01/12}. Moreover, for each item that was sold in the store on the date {\tt 01/01/12}, the query {\tt Qc} returns all the {\em same} purchase amounts for the same item in the same store, as many times as the purchases have happened,  regardless of the date. 
If the analysts want to calculate correctly the total per-store returns for all the transactions that have the same item prices as the same-store transactions for the date {\tt 01/01/12}, then all they have to do to write such a query is to (i) add to the query {\tt Qc} the clause {\tt GROUP BY storeID}, and to (ii)  replace {\tt amount} by {\tt sum(amount)} in the {\tt SELECT} clause of the resulting query. Note that to evaluate the resulting analysis query, the query processor would first evaluate the query {\tt Qc}, and would then apply to the answer to {\tt Qc} the above grouping and aggregation.

Due to the large size of the relation {\tt Sales} (one million rows), 
the self-join of {\tt Sales} in the query {\tt Qc}, 
via a correlated subquery, should be avoided if at all possible. For some SQL queries, such direct removal of a subquery is indeed possible. 
Consider, for instance, a query {\tt Qc2}, whose definition is almost the same as that of the query {\tt Qc}: 

{\small 
\begin{verbatim}  
(Qc2) SELECT storeID, amount FROM Sales S 
      WHERE date = `01/01/12' 
      AND EXISTS 
         (SELECT * FROM Sales 
          WHERE storeID = S.storeID AND itemID = S.itemID 
          AND amount = S.amount)
\end{verbatim}
} 

\noindent 
The only difference between the queries {\tt Qc} and {\tt Qc2} is that the equality comparison with the date {\tt `01/01/12'}, which can be found in the subquery of the query {\tt Qc}, is a condition in the main body of the query {\tt Qc2}. 

It is easy to argue informally, based on the ``intended meaning'' of the query {\tt Qc2}, that {\tt Qc2} always returns the same number of the same answers as the query {\tt Qc2min}, which is obtained by removing the correlated subquery from {\tt Qc2}: 

{\small 
\begin{verbatim}  
(Qc2min) SELECT storeID, amount FROM Sales S 
         WHERE date = `01/01/12' 
\end{verbatim}} 

\noindent 
Indeed, by the results of this paper, the queries {\tt Qc2} and {\tt Qc2min} are combined-semantics equivalent. 

Thus, one question we could pose is whether the queries {\tt Qc} and {\tt Qc2} (or alternatively the queries {\tt Qc} and {\tt Qc2min}) are combined-semantics equivalent. If we can answer this question in the positive, then the query {\tt Qc2min}, which is very efficient to execute, can be evaluated on the relation {\tt Sales} to obtain the correct answer to the query of interest {\tt Qc}. Another question we could of course pose is whether the query {\tt Qc} is combined-semantics equivalent to the SQL query that results from removing the correlated subquery of the query {\tt Qc}. Intuitively, this is unlikely, as the latter query would lack the equality comparison with the date {\tt `01/01/12'}.  
\end{example} 

Combined-semantics CQ queries such as the query {\tt Qc} of Example~\ref{nex-ex}, with grouping and aggregation added, arise naturally in on-line analytical processing {\em (OLAP)} applications \cite{Lehner09,OlapReport}. Such queries occur whenever a data-analysis task calls for a query structure with nested subqueries. Such queries also arise due to joins that go beyond ``star-schema joins'' \cite{ChaudhuriD97,ChaudhuriDN11}, which are the only well-understood joins in the  literature on OLAP query optimization. 
 See \cite{Lehner09,OlapReport,ShuklaDN00,ZaharioudakisCLPU00} for more detailed discussions of why queries with nested subqueries and with ``non-star joins'' are natural and frequent in OLAP.  (For additional extended illustrations of such queries, see the online version \cite{fullversion} of this paper.)

It turns out that the equivalence tests reported in \cite{ChandraM77,ChaudhuriV93,Cohen06-pods,Cohen06} do not apply to the SQL queries {\tt Qc} or {\tt Qc2} of Example~\ref{nex-ex}, even though both queries are expressible in the syntax of \cite{Cohen06-pods,Cohen06} for combined-semantics conjunctive queries. At the same time, we can use the results of this paper and of \cite{MyLibkinIcdt12journalSubmission} (the latter results are also available online in \cite{fullversion}) to establish the following: 
\begin{enumerate} 
	\item $Qc$ $\equiv_C$ $Qc2$ does not hold; 
	
	
	\item $Qc2$ $\equiv_C$ $Qc2min$ does hold (observe that $Qc2min$ is a conjunctive {\em bag-semantics} query); 
	
	\item $Qc2min$ is a minimized version of $Qc2$; and 
	
	\item $Qc$ is the only minimized version of  itself.  
\end{enumerate} 

\noindent 
To the best of our knowledge, none of the above inferences 1--4 can be made using formal results reported in prior work. We use the results of \cite{MyLibkinIcdt12journalSubmission} in making inferences 3--4 only, and use the results of this current paper in making all the four inferences. 

\paragraph{Our contributions} 

\noindent 
We study equivalence of unaggregated SQL queries with equality comparisons (including comparisons with constants) and possibly with subqueries. We follow the approach of \cite{Cohen06}, where the study concentrates on Datalog translations of such queries, that is on combined-semantics CQ queries. 
The requisite translations from SQL to Datalog are straightforward (``as expected'').\footnote{Section 1 in \cite{Cohen06} provides some details of the 
translations.} In the remainder of this paper, all queries are expressed using the Datalog-based formalism of \cite{Cohen06}. In the remainder of this paper, we refer to combined-semantics CQ queries as {\em CCQ queries.} 

We focus on the problem of determining whether two CCQ queries are combined-semantics equivalent. We continue the tradition of \cite{ChandraM77,ChaudhuriV93,Cohen06}  of studying these problems using the tool of containment between queries. Our first specific contribution is to introduce in this work ``covering mappings'' {\em (CVMs)} 
between CCQ queries, and to show that CVMs furnish a sufficient condition for combined-semantics containment, for all CCQ queries. 

The second specific contribution of this paper is in providing a necessary condition for combined-semantics equivalence of CCQ queries. To formulate this condition, we isolate a large class of CCQ queries, which we call ``explicit-wave queries''.\footnote{The term ``explicit-wave query'' is due to the structures generated by the proof of the necessary equivalence condition reported in this paper.} We show that this class of queries encompasses, but is not limited to, (i) all CQ set, bag, and bag-set queries, and (ii) all CQ queries for which \cite{Cohen06} provides its sufficient and necessary equivalence tests. Further, it appears that all combined-semantics CQ queries that are expressible in SQL are explicit-wave queries. (Please see Section~\ref{necess-equiv-condition} for the details.) 
Our necessary condition for query equivalence is asymmetric -- it states that if for CCQ queries $Q_1$ and $Q_2$ we have the combined-semantics equivalence $Q_1 \equiv_C Q_2$, {\em and} $Q_1$ is  an explicit-wave query, then there exists a CVM from $Q_2$ to $Q_1$. We also show that this necessary condition is tight, in that we exhibit a specific CCQ query $Q$ that is {\em not} explicit wave, such that our necessary equivalence condition does not hold when applied to $Q$ and to another CCQ query $Q'$, even though the two queries are provably combined-semantics equivalent. 
The problem of combined-semantics equivalence for general  CCQ queries remains open. 

The importance of the contributions of this paper is in solving in the positive the decidability problem of determining  combined-semantics equivalence, for pairs of explicit-wave CCQ queries. That is, the results reported in this paper imply immediately that for an arbitrary pair of CCQ queries, it is decidable (i) to determine whether each of the queries is explicit wave, and (ii) to determine, in case both queries are explicit wave, whether or not they are combined-semantics equivalent, by using our syntactic CVM-based criterion. Note that even in case where one or both of the input queries are not explicit wave, our syntactic (CVM-based) sufficient condition for combined-semantics equivalence could still be used to determine that the two queries are equivalent. 

The results of this paper combined with those of \cite{MyLibkinIcdt12journalSubmission} can be used directly in query optimizers for database-management systems, as well as for developing optimization methods for queries in more expressive languages than CQ queries and in presence of integrity constraints. Our results can also be used for developing algorithms for rewriting queries using views, and for view selection under combined semantics. 

The remainder of this paper is organized as follows. In Section~\ref{new-rel-work-sec} we review related work. Section~\ref{prelim-section} formulates the background notions and results. In Section~\ref{containment-mappings-sec}, we introduce covering mappings (CVMs) and present our sufficient CVM-based condition for combined-semantics equivalence, which is applicable to all CCQ queries. Section~\ref{necess-equiv-condition} formalizes the notion of explicit-wave queries, presents in Theorem~\ref{magic-mapping-prop} the CVM-based necessary condition for equivalence of explicit-wave CCQ queries, and shows that this condition does not necessarily apply to general CCQ queries. Section~\ref{necess-equiv-condition} also contains the decidability test, which is the main result of this paper. Finally, Section~\ref{magic-mapping-proof-sec} contains the proof of Theorem~\ref{magic-mapping-prop}.

\reminder{Say that \cite{Cohen06} does {\em not} provide necessary conditions for $Q \equiv_C Q'$, such that those conditions would subsume ours. Cohen's necessary conditions are for the queries: 

\begin{itemize} 
	\item queries without set variables 
	\item queries without multiset variables 
	\item no same-name predicate twice or more in positive (i.e., nonnegated) subgoals 
	\item query is a join of set (i.e., no multiset variables) subquery with multiset (i.e., no set variables) subquery ; the formal definition is that a query must not have a subgoal that would have both a multiset variable and a set variable  
	\item no copy variables
\end{itemize} 
}


\reminder{Move this paragraph to Section~\ref{sccm-sec}: after (i) definition of regularized query, (ii) definition of CVM, and (iii) Theorem~\ref{not-same-num-multiset-vars-thm}. 
The main requirement of containment mappings, that of $\varphi(L) \subseteq L'$ (see Definition~\ref{multiset-homom-def} (3)), appears fundamental as far as our expectations from mappings go: We expect the subgoals of one query to ``map into'' the subgoals of the other query. Observe a departure from this requirement in the definition of CVM, in Definition~\ref{magic-mapping-def} (4). Do I need to go on, by explaining that we use the intuition of {\em individual} satisfying assignments from queries to databases, where the copy numbers do not matter? Hence the intuition for CVM.}

\subsection{Related Work}
\label{new-rel-work-sec} 

In their classic paper \cite{ChandraM77}, Chandra and Merlin presented an NP-complete containment test for CQ queries under set semantics. This sound and complete test has been used in optimization, via minimization, of CQ set-semantics queries, as well as in developing algorithms for rewriting queries (both equivalently and nonequivalently) using views. 
In this current paper we extend the containment and equivalence results of \cite{ChandraM77} to general CQ combined-semantics queries, and show the limitations of each extension. 

Equivalence tests for CQ bag and bag-set queries were formulated by Chaudhuri and Vardi in~\cite{ChaudhuriV93}; correctness of the tests follows from the results of \cite{Cohen06}. Our equivalence results for CQ combined-semantics queries reduce correctly to the special cases of CQ bag and bag-set queries, as given in \cite{ChaudhuriV93}. 
Further, this current paper provides a nontrivial generalization and the first known proof of the well-known sufficient containment condition for CQ bag queries, as outlined in \cite{ChaudhuriV93}. 
 
Definitive results on containment between CQ queries under bag and bag-set semantics have not been  obtained so far. Please see Jayram, Kolaitis, and Vee \cite{KolaitisPods06} for original undecidability results on containment of CQ queries with inequalities under bag semantics. The authors point out that it is not  known whether the problem of bag containment for {\em CQ} queries is even decidable. 
For the case of {\em bag-set} semantics, sufficient conditions for containment of two CQ queries can be expressed via containment of (the suitable) aggregate queries with aggregate function {\tt count(*)}. 
The latter containment problem can be solved using the methods proposed in~\cite{CohenNS03}. 
Please see~\cite{ChaudhuriV93,AfratiDG10} for other results on bag and bag-set containment of CQ queries. 
The general problems of containment for CQ bag and bag-set queries remain open. 

In her papers  \cite{Cohen06-pods,Cohen06}, Cohen  provided an elegant  and powerful  formalism for treating queries evaluated under each of set, 
bag, and bag-set 
semantics uniformly as special cases of the more general combined semantics. The papers also contain a general sufficient condition for combined-semantics equivalence of CQ queries with disjunction, negation, and arithmetic comparisons, as well as necessary and sufficient equivalence conditions for special cases. (Interestingly, when we restrict the language of the queries in question to the language of CQ queries, it turns out that all the necessary and sufficient query-equivalence conditions of \cite{Cohen06} hold for queries belonging collectively to a proper subclass of the class of explicit-wave CQ queries, which (class) we introduce in this current paper.) 
The proof in \cite{Cohen06} of its general sufficient condition for equivalence of queries 
is in terms of containment between the queries under combined semantics. 
That (implicit) sufficient query-containment condition 
 is proved in \cite{Cohen06}  for the case where the two queries have the same number of multiset variables. 
In this current paper we provide proper generalizations of all the results of \cite{Cohen06}, including of its implicit sufficient condition for query containment, provided that the results of \cite{Cohen06} are applied to CQ queries only. 

A discussion of query equivalence and containment for query languages 
that properly contain the language of CQ queries 
is beyond the scope of this paper. The interested reader is referred to \cite{Cohen06}, which contains an excellent overview of the literature in this direction.

\section{Preliminaries}
\label{prelim-section} 

\reminder{Plan for the subsections of this Preliminaries section:
\begin{itemize}

	\item Basics for combined semantics; special cases of bag/bag-set/set semantics (to define each in appendix)
	
	\item Combined-semantics equivalence tests for CQ queries
	
	\item Dependencies and chase:
	
	\begin{itemize}
	
		\item basics from PODS-09 preliminaries -- put into appendix?
		
		\item classes of dependency sets from \cite{JohnsonK84}
		
		\item my PODS-09 results for bag and bag-set semantics??? -- put into appendix?
	
	\end{itemize}

\end{itemize}
} 



\reminder{The contents of Sections~\ref{cutup-scohen-framework-sec} through~\ref{cutup-sound-complete-tests-sec} are all from the paper \cite{Cohen06} by Sara Cohen. Please see Appendix~\ref{cohen-set-bag-sem-sec} for the details and clarifying examples.} 

\subsection{Combined semantics: The framework [9]}
\label{cutup-scohen-framework-sec} 


\reminder{Into the title of this subsection, insert manually the citation ID for \cite{Cohen06}}

\reminder{Must have a separate (sub)section in the Preliminaries, where I would explicitly state all the relevant results of \cite{ChandraM77}.}


\subsubsection{Syntax of queries}
\label{query-syntax-sec}

Predicate symbols are denoted as $p$, $q$, $r$. Databases contain ground atoms for a given set of predicate symbols; we consider finite-size databases only. A database may have several copies of the same atom. To denote this fact, each atom in the database is associated with a {\em copy number} $N$. Formally, if $p$ is an $n$-ary predicate, for an $n \in {\mathbb N}_+$ (with ${\mathbb N}_+$ the set of natural numbers), we write $p(c_1,\ldots,c_n; N)$, with $N \in {\mathbb N}_+$, to denote that there are precisely $N$ copies of $p(c_1,\ldots,c_n)$ in the database. As a shorthand, if $N=1$, we often omit the copy number $N$. The {\em active domain of database} $D$, denoted $adom(D)$, is the set of all constants mentioned in the ground atoms of $D$.  We adopt a convention by which, for each atom of the form $p(c_1,\ldots,c_n)$ such that database $D$ has $N \geq 1$ copies of that atom, $N$ is an element of $adom(D)$ if and only if there exists in $D$  an atom $r(c'_1,\ldots,c'_m; N')$ (for some copy number $N'$ $\geq$ $1$ and where $r$ and $p$ may or may not be the same predicate) such that $N$ is one of $c'_1$, $\ldots$, $c'_m$. 
\reminder{A database $D$ is called {\em set valued} if all copy numbers for the atoms in $D$ equal unity (1). }

For query syntax, we denote variables using $X$, $Y$, $Z$, possibly with subscripts, and $i$, $j$, $k$. The former range over constants in the database (i.e., over $adom(D)$), whereas the latter range over copy numbers.  For this reason, we call the former {\em regular variables} (or simply {\em variables} for short), and we call the latter {\em copy variables.}  
We use $c$, $d$ to denote constants. A {\em term,} denoted as $S$, $T$, is a variable or a constant. 

A {\em relational atom} has the form $p(S_1,\ldots,S_n)$, where $p$ is a predicate of arity $n$. We also use the notation $p(\bar{S})$, where $\bar{S}$ stands for a sequence of terms $S_1,\ldots,S_n$. 
A {\em copy-sensitive atom} has the form $p(\bar{S}; i)$, and is simply a relational atom with copy variable $i$. 
We call relational atom $p(\bar{S})$ {\em the relational template of copy-sensitive atom} $p(\bar{S}; i)$. For each relational atom, its relational template is the atom itself.  
%
A {\em condition,} denoted as $L$, is a conjunction of relational and copy-sensitive atoms, with duplicate atoms allowed, such that all copy variables in $L$ are unique (i.e., appear in a single copy-sensitive atom, and do not appear in other atoms). 
Sometimes it will be convenient for us to view condition $L$ as a bag of all and only the elements in the conjunction $L$. 

We distinguish between the variables that appear in the head of a query, and those that only appear in the body. The former are {\em distinguished (head) variables,} and the latter are {\em nondistinguished (nonhead) variables.} Nondistinguished variables come in two flavors: {\em set variables} and {\em multiset variables.} The intuition for the difference between these two types of variables is as follows. When evaluating a query, different assignments for set variables do not contribute to the multiplicity in which a particular answer is returned by the query. On the other hand, different assignments for multiset variables do contribute to the multiplicity of the returned answers. Technically, in order to differentiate between set variables and multiset variables, we always specify the set of multiset variables in each condition immediately to the right of the condition. As a syntactic requirement, all copy variables must be in the set of multiset variables. 


\begin{definition}{Query syntax: CCQ query} 
\label{ccq-def}
A  {\em copy-sensitive conjunctive query (CCQ query)} is a nonrecursive expression of the form
$$Q(\bar{X}) \leftarrow L, M, $$
where $\bar{X}$ is a (possibly empty) vector, $L$ is a nonempty condition, and $M$ is a set of variables, such that:  
\begin{itemize}
	\item $L$ contains all the variables in $\bar{X}$; that is, $Q$ is {\em safe}; 
	\item $M$ is a subset of the set of nondistinguished variables of $L$ and contains all copy variables of $L$. We denote all the copy variables of $Q$ collectively as $M_{copy} \subseteq M$, and all the remaining (``multiset noncopy'') variables in $M$ as $M_{noncopy} := M - M_{copy}$.  
\end{itemize}
\vspace{-0.5cm}
\end{definition} 


We call each element of the condition $L$ a {\em subgoal} of $Q$. The variables in $M$ are the {\em multiset variables} of $Q$. The variables in $L$ that are not in $\bar{X}$ or in $M$ are the {\em set variables} of $Q$. Consider an illustration. 

\begin{example} 
\label{syntax-ex}
Let CCQ query $Qc$ be as follows. 
\begin{tabbing}
$Qc(X, Y) \leftarrow sales(X, Z, U, Y; i), sales(X, Z, c, Y), \{ Z, U, i \} .$ 
\end{tabbing}

Suppose that we interpret  $c$ as the constant value `01/01/12.' Then this query is the Datalog version of the SQL query {\tt Qc} of Example~\ref{nex-ex} (in Section~\ref{intro-sec}). 

The condition $L$ of the query $Qc$ is the conjunction of the two subgoals of $Q$ with the predicate {\tt sales}. The variables $X$ and $Y$ are the head variables of the query $Qc$. The set $\{ Z, U, i \}$ is the set $M$ of multiset variables of $Qc$; the set $M_{copy}$ of $Qc$ comprises the (only) copy variable $i$ of $Qc$, and the  set $M_{noncopy}$ of $Qc$ has the multiset noncopy variables $Z$ and $U$ of $Qc$. By definition, 
$M$ 
$=$ $M_{copy}$ $\cup$ $M_{noncopy}$ . This query does not have set variables. 
\end{example}


We use $S(Q)$ to denote an arbitrary vector, without repetitions, of the set variables of $Q$, and $\bar{S}(Q)$ to denote an arbitrary vector, without repetitions, of the remaining variables of $Q$ (i.e., the distinguished and multiset variables of $Q$). 
By abuse of notation, we will often refer to a query by its head $Q(\bar{X})$ or simply by its head predicate  $Q$. For a vector of terms $\bar X$ with $k$ $\geq$ $0$ elements, we say that a CCQ query with head $Q({\bar X})$ is a {\em CCQ} $k${\em -ary query.}  In the special case where $k$ $=$ $0$, we say that $Q$ is a {\em CCQ Boolean query,} and denote its head by $Q()$.

We will sometimes be interested in special types of queries. A CCQ query $Q$ is a {\em set query} 
if it has no multiset variables, that is, if $M = \emptyset$. Query $Q$ is a {\em multiset query} if $Q$ has no set variables. Further, a multiset query $Q$ is (i) a {\em bag query} if $Q$ has only copy-sensitive subgoals, and is (ii) a {\em bag-set query} if $Q$ has only relational subgoals. 

\reminder{Remove this paragraph from final paper if needed: In this paper, we also consider CCQ queries that may have copy variables in the head. 
We prohibit this syntactic relaxation of Definition~\ref{ccq-def} for any queries except those that we explicitly construct (to make certain points) from given CCQ queries that satisfy the Definition. Specifically, for a CCQ query $Q(\bar{X}) \leftarrow L, M$ and for a vector $\bar{M}$ that uses (perhaps with repetitions)  all and only multiset variables of $Q$, the {\em fully-relaxed derivative of $Q$ using $\bar{M}$} is the query  $Q(\bar{X},\bar{M}) \leftarrow L, \emptyset$. By definition, each fully-relaxed derivative query is a set query. Further, whenever $M = \emptyset$ in $Q$, all fully-relaxed derivatives of $Q$ coincide with $Q$.  }

\reminder{Put here a clarifying example: It should use the queries of Example~\ref{motiv-ex}, but make the points of Example~\ref{cohen06-exampleTwoTwo}.}

\subsubsection{Combined semantics for queries}
\label{comb-sem-sec}

We define how CCQ query $Q(\bar{X}) \leftarrow L, M$ yields a {\em multiset} of tuples on database $D$. Intuitively, we start by considering satisfying assignments of the condition $L$. We then restrict these assignments to the {\em nonset variables} of $L$, that is to $\bar{S}(Q)$. Each of these restricted assignments yields a tuple in the result. A formal description of the semantics follows.

Let $\gamma$ be a mapping of the terms in condition $L$ to values. We will also apply $\gamma$ to a sequence of terms to derive a sequence of values, in the obvious way. We say that $\gamma$ is a {\em satisfying assignment} of $L$ with respect to database $D$ if all of the following conditions hold: 

\begin{itemize}

	\item $\gamma$ is the identity mapping on constants;
\vspace{-0.1cm} 	
	\item for all relational atoms $p(\bar{T}) \in L$, there exists an $N \in \mathbb{N}_+$ such that we have $p(\gamma\bar{T}; N) \in D$; and 
\vspace{-0.1cm} 	

	\item for all copy-sensitive atoms $p(\bar{T}; i) \in L$, the following two conditions hold:

\vspace{-0.1cm} 	
	\begin{itemize}

		\item $\gamma i \in \mathbb{N}_+$ (i.e, $\gamma i$ is a positive natural number); 
		
\vspace{-0.1cm} 	
		\item there is an $N \geq \gamma i$ such that $p(\gamma\bar{T}; N) \in D$. 
		
	\end{itemize}

\end{itemize}


\mbox{}

Let $\Gamma(Q,D)$ denote the set of satisfying assignments of $L$ with respect to database $D$. Let $\gamma$ be an assignment of the variables in $\bar{S}(Q)$ to constants. We say that $\gamma$ is {\em satisfiably extendible} if there is an assignment $\gamma' \in \Gamma(Q,D)$ such that $\gamma$ and $\gamma'$ coincide on all terms for which $\gamma$ is defined, that is, $\gamma'(X) = \gamma(X)$ for all $X \in \bar{S}(Q)$. Intuitively, this means that it is possible to extend $\gamma$ to derive a satisfying assignment of $L$. We use $\Gamma_{\bar{S}}(Q,D)$ to denote the set of satisfiably extendible assignments of $\bar{S}(Q)$ with respect to $D$. For the $\gamma \in \Gamma_{\bar{S}}(Q,D)$ and for the $\gamma' \in \Gamma(Q,D)$ as specified in this paragraph, we say that $\gamma'$ {\em contributes $\gamma$ to} $\Gamma_{\bar{S}}(Q,D)$. 



\reminder{Do I need the ``$\rho$-permutation'' terminology? Finally, let $\rho(\bar{S}(Q))$ be a permutation of the terms in the vector $\bar{S}(Q)$; then we call the projection $\pi_{\rho(\bar{S}(Q))}(\Gamma_{\bar{S}}(Q,D))$ the {\em $\rho$-permutation of the set} $\Gamma_{\bar{S}}(Q,D)$. By definition of projection, the $\rho$-permutation of $\Gamma_{\bar{S}}(Q,D)$ is obtained, for any $\rho$ of $\bar{S}(Q)$, by simply ``reshuffling the columns'' of the set-valued relation $\Gamma_{\bar{S}}(Q,D)$. } 

We now define the result of applying a query $Q$ to a database $D$. (We use 
 $\{ \hspace{-0.1cm} \{ \ldots \} \hspace{-0.1cm} \}$ to denote a bag of values.) 

\begin{definition}{Combined semantics}
\label{query-semantics-def}
Let $Q(\bar{T})$ $\leftarrow L, M$ be a CCQ query and let $D$ be a database. The {\em result of applying $Q$ to $D$ under combined semantics,} denoted ${\sc Res}_C(Q,D)$, is defined as 

\mbox{}

\noindent 
${\sc Res}_C(Q,D) \ := \ \{ \hspace{-0.1cm} \{ \ \gamma(\bar{T}) \ | \ \gamma \in \Gamma_{\bar{S}}(Q,D) \ \} \hspace{-0.1cm} \} \ .$
\end{definition}

(In the special case where $Q$ is a Boolean query, each $\gamma(\bar{T})$ as above is interpreted as a separate copy of the empty tuple.) Note that ${\sc Res}_C(Q,D)$ is a bag of tuples, that is, 
${\sc Res}_C(Q,$ $D)$ 
may contain multiple occurrences of the same tuple. 
Consider an illustration. 

\begin{example} 
\label{second-syntax-ex}
Let CCQ query $Qc$ be as in Example~\ref{syntax-ex}, and consider a database $D$ whose $Sales$ relation has two distinct tuples: $(85, 433, 01/01/12, 264; 2)$ and $(85, 433,$\linebreak $03/15/12, 264)$. (Note that the first tuple of the $Sales$ relation is present in two copies in the database $D$.) Then ${\sc Res}_C(Qc,D)$ is a bag of exactly three copies of the tuple $(85, 264)$. 
\end{example}

\reminder{ 
\mbox{}

{\em Remark (2.8 in \cite{Cohen06})} Note that the semantics of \cite{Cohen06} allows us to choose whether multiplicities in the database should affect query results. If this is desired, then a copy variable should be used for the particular predicate of interest. Otherwise, omitting copy variables essentially gives the same effect as removing duplications in the database. \reminder{\footnote{See discussion of Example~\ref{cohen06-exampleTwoOne} in Appendix~\ref{scohen-framework-sec}.}} 

\mbox{}

For technical reasons, we will sometimes be interested in evaluating a set of queries over a database. Combined semantics extends naturally to this case, i.e., given a set $\cal Q$ of queries and a database $D$, we define

\begin{tabbing}
${\sc Res}_C({\cal Q},D) := \biguplus \ \{ \hspace{-0.1cm} \{ \ {\sc Res}_C(Q,D)  \ | \ Q \in {\cal Q} \ \} \hspace{-0.1cm} \} .$
\end{tabbing}
where $\biguplus$ is the multiset-union operator. In a similar fashion, upcoming definitions of query containment and equivalence can naturally be extended to sets of queries. (See Examples~\ref{cohen-two-nine-ex} and~\ref{cohen-two-ten-example} in Appendix~\ref{scohen-framework-sec} for illustrations.) 
} 

Under certain circumstances,  combined semantics coincides with set, bag, or bag-set semantics. Please see \cite{Cohen06} 
for the details on the three traditional query semantics, specifically on how these semantics can be formulated as special cases of combined semantics. 


\subsubsection{Query containment and equivalence}

Query containment under combined, set, bag, and bag-set semantics is defined in the standard manner. Formally, $Q$ {\em is contained in} $Q'$ under a given semantics if, for all databases, the bag of values returned by $Q$ is a subbag of the bag of values returned by $Q'$. We write $Q \sqsubseteq_C Q'$, $Q \sqsubseteq_S Q'$, $Q \sqsubseteq_B Q'$, and $Q \sqsubseteq_{BS} Q'$ if $Q$ is contained in $Q'$ under combined, set, bag, and bag-set semantics, respectively. Similarly, we use $Q \equiv_C Q'$, $Q \equiv_S Q'$, $Q \equiv_B Q'$, and $Q \equiv_{BS} Q'$ to denote the fact that $Q$ is equivalent to $Q'$ under each semantics. $Q \equiv_C Q'$ holds if and only if $Q \sqsubseteq_C Q'$ and $Q' \sqsubseteq_C Q$ both hold. The definitions of $Q \equiv_S Q'$, $Q \equiv_B Q'$, and $Q \equiv_{BS} Q'$ parallel that of $Q \equiv_C Q'$ in the obvious manner. 

For CCQ queries $Q$ and $Q'$, we have that (1) $Q \sqsubseteq_S Q'$ iff $Q \sqsubseteq_C Q'$, in case $Q$ and $Q'$ are set queries;  (2) $Q \sqsubseteq_B Q'$ iff $Q \sqsubseteq_C Q'$, in case $Q$ and $Q'$ are bag queries; and (3) $Q \sqsubseteq_{BS} Q'$ iff $Q \sqsubseteq_C Q'$, in case $Q$ and $Q'$ are bag-set queries. 

For a class $\cal Q$ of queries: The {\em $\cal Q$-containment problem for combined semantics} is: {\em Given queries $Q$ and $Q'$ in $\cal Q$, determine whether $Q \sqsubseteq_C Q'$.} The {\em $\cal Q$-equivalence problem} is defined similarly using $\equiv_C$ instead of $\sqsubseteq_C$. The two problems
 can be defined similarly for other semantics. 

\reminder{For both set and bag-set semantics and for important classes of queries, the $\cal Q$-equivalence problem has the {\em small counterexample property.} Formally this means that if $Q, Q' \in {\cal Q}$ and $Q \equiv_S \hspace{-0.5cm} / \hspace{0.4cm} Q'$ (resp.  $Q \equiv_{BS} \hspace{-0.8cm} / \hspace{0.6cm} Q'$), then there is a database $D$ that contains as many atoms as there are subgoals in $Q$ or $Q'$, such that ${\sc Res}_S(Q,D) \neq  {\sc Res}_S(Q',D)$ (resp. ${\sc Res}_{BS}(Q,D) \neq  {\sc Res}_{BS}(Q',D)$ \reminder{Need to really check this small counterexample property for *bag*-*set* semantics}). For example, the class of positive queries with constants, comparisons and disjunctions has the small counterexample property both under set semantics \cite{ChandraM77} and under bag-set semantics \cite{CohenNS99,NuttSS98}. The following theorem statems that this is also the case for bag semantics. 

\begin{theorem}{(Theorem 2.14 of \cite{Cohen06}: Bag semantics: small counterexamples)}
Let $Q$ and $Q'$ be positive queries, possibly with constants, comparisons, and disjunctions. If $Q$ and $Q'$ are not equivalent then there exists a database $D$ that contains as many {\em unique} atoms as there are subgoals in $Q$ or $Q'$, such that ${\sc Res}_{B}(Q,D) \neq {\sc Res}_{B}(Q',D)$. 
\end{theorem}

A more general version of this result is proved in Section 4 in \cite{Cohen06}. \reminder{Need to double check that a *correct* proof is there.} That theorem will  follow from Corollary 4.9 and Proposition 2.12 of \cite{Cohen06}. 

\reminder{Do I need this paragraph?} The small counterexample property is useful, since it shifts the focus to small databases when trying to determine equivalence. In particular, it is often the case that a counterexample to equivalence of $Q$ and $Q'$ can be created out of the body of $Q$ or $Q'$. Unfortunately, as the following example demonstrates, the small counterexample property does not hold even for simple queries evaluated under combined semantics. This makes characterizing equivalence under combined semantics intuitively more difficult. 

\begin{example}{(Example 2.15 from \cite{Cohen06})} 
\label{cutup-cohen-two-six-ex}
\reminder{Do I need this example?} 
Consider the queries
\begin{tabbing}
$Q_7(X) \leftarrow A(X,Y), \emptyset .$ \\
$Q_8(X) \leftarrow A(X,Y), \{ Y \} .$ 
\end{tabbing}

Over any database that is the size of $Q_7$ or $Q_8$ (i.e., that contains a single atom), $Q_7$ and $Q_8$ return the same values. However, $Q_7 \equiv_C \hspace{-0.6cm} / \hspace{0.4cm} Q_8$, as demonstrated by the database
\begin{tabbing}
Tabbb\= ing \= here \= today \kill
$D = \{ A(1,1), A(1,2) \} ,$  
\end{tabbing}
since ${\sc Res}_C(Q_7,D) = \{ \hspace{-0.1cm}\{ (1) \} \hspace{-0.1cm}\}$ and ${\sc Res}_C(Q_8,D) = \{ \hspace{-0.1cm}\{ (1), (1) \} \hspace{-0.1cm}\}$. 
\end{example}
} 


\subsection{Equivalence and minimization results} 
\label{prev-results-sec} 


{\bf Homomorphisms and set queries.} Given two conditions $\phi(\bar{U})$ and $\psi(\bar{V})$, 
a {\it homomorphism} from $\phi(\bar{U})$ to $\psi(\bar{V})$ is a mapping $h$ from the 
set of terms in $\bar{U}$ to the set of terms in $\bar{V}$ such that (1) $h(c) = c$ for each constant $c$, (2) for each  relational atom $r(U_1,\ldots,U_n)$ of $\phi$ we have that  $r(h(U_1),$ $\ldots,$ $h(U_n))$ is in $\psi$, 
and (3) for each copy-sensitive atom $p(W_1,\ldots,W_n; i)$ of $\phi$ we have that  $p(h(W_1),$ $\ldots,$ $h(W_n); h(i))$ is in $\psi$.  
Given two CCQ $k$-ary queries  
$Q_1(\bar{X}) \leftarrow \phi(\bar{X}, \bar{Y}), M_1$ and $Q_2(\bar{X}') \leftarrow \psi(\bar{X}', \bar{Y}'), M_2,$ a {\it containment mapping} from $Q_1$ to $Q_2$ is a homomorphism $h$ from $\phi(\bar{X}, \bar{Y})$ to $\psi(\bar{X}', \bar{Y}')$ such that $h(\bar{X}) = \bar{X}'$. 


\begin{theorem}{\cite{ChandraM77}} 
\label{cm-thm}
Given two CCQ set queries $Q_1$ and $Q_2$ of the same arity,  $Q_1 \sqsubseteq_S Q_2$ holds if and only if there is a containment mapping from $Q_2$ to $Q_1.$ 
\end{theorem} 
%
This classic result of~\cite{ChandraM77}  forms the basis for a sound and complete test for set-equivalence of CCQ set queries $Q$ and $Q'$, by definition of set-equivalence $Q \equiv_S Q'$.

{\bf Bag and bag-set queries.} For bag and bag-set semantics, the following conditions are known for CCQ query equivalence. (Query $Q_c$ is a {\em canonical 
representation} of query $Q$ if $Q_c$ is the result of removing all 
duplicate atoms from the condition of $Q$.) 


\begin{theorem}{~\cite{ChaudhuriV93}}
\label{prelim-cv-theorem} Let $Q$ and $Q'$ be CCQ queries. Then (1) When $Q$ and $Q'$ are bag queries, 
$Q \equiv_B Q'$ iff 
 $Q$ and $Q'$ are isomorphic. (2) When $Q$ and $Q'$ are bag-set queries, $Q \equiv_{BS} Q'$ iff 
  $Q_c$ and $Q'_c$ are isomorphic. 
\end{theorem}

 (Two CCQ queries $Q_1$ and $Q_2$, with respective sets of multiset variables $M_1$ and $M_2$, are {\em isomorphic} if there exists a one-to-one containment mapping from $Q_1$ {\em onto} $Q_2$ such that the mapping induces a bijection from $M_1$ to $M_2$, and there exists another containment mapping from $Q_2$ onto $Q_1$, with (symmetrically) the same properties.) 



{\bf Combined-semantics queries.} The next result is a sufficient condition of~\cite{Cohen06} for equivalence of two queries under combined semantics. In \cite{Cohen06},  Cohen formulates each of Definition~\ref{multiset-homom-def} and Theorem~\ref{cutup-cohen-thm-three-three} for CCQ queries that may also contain negation and inequality comparisons.  
(In condition (3) of Definition~\ref{multiset-homom-def} we treat the query conditions, which are conjunctions of atoms, as bags of the same atoms. Given a bag $B$, we call a set $S$ the {\em core-set} of $B$ if $S$ is the result of dropping all duplicates of all elements of $B$.) 


\begin{definition}{Multiset-homomorphism \cite{Cohen06}} 
\label{multiset-homom-def}
Let\linebreak $Q(\bar{X}) \leftarrow L,  M$ and $Q'(\bar{X}') \leftarrow L',  M'$ be two $k$-ary CCQ queries, for $k \geq 0$. 
Let $\varphi$ be a mapping from the terms of $Q'$ to the terms of $Q$.\footnote{We also apply $\varphi$ to atoms and conjunctions of atoms, in the obvious way, e.g., $\varphi(p(\bar{S})) = p(\varphi(\bar{S}))$.} We say that $\varphi$ is a {\em multiset-homomorphism from $Q'$ to $Q$} if  $\varphi$ satisfies all of the following conditions:
\begin{enumerate}
	\item $\varphi \bar{X}' = \bar{X}$ ;
	\item $\varphi$ is the identity mapping on constants;
	\item the core-set of $\varphi L'$ is a subset of the core-set of $L$ ;
	\item $\varphi M' \subseteq M$ ; and
	\item $\varphi Y \neq \varphi Y'$ for every two variables $Y \neq Y' \in M'$. 
\end{enumerate}
\vspace{-0.6cm}
\end{definition} 

For every mapping  $\varphi$ that satisfies conditions 1--3 of Definition~\ref{multiset-homom-def}, 
we call  $\varphi$ a {\em generalized containment mapping (GCM).} 

\nop{
Conditions 1 and 3 state, respectively, that $\varphi$ maps the head to the head and the body into the body. 
Conditions 4 and 5 state that $\varphi$ is an injective mapping of multiset variables to multiset variables. 
Note that the only difference between a regular homomorphism \cite{ChandraM77} and a multiset-homomorphism is the addition of Conditions 4 and 5. 
} 

We say that two CCQ queries $Q$ and $Q'$ are {\em multiset homomorphic} whenever there is a multiset-homomorphism from $Q$ to $Q'$ and another from  $Q'$ to  $Q$. 


\nop{
\begin{definition}{Definition 3.2 from \cite{Cohen06}: Multiset-homomorphic} 
Consider the queries $Q(\bar{X}) \leftarrow \vee_{i=1}^n L_i, M_i$ and 
$Q'(\bar{X}') \leftarrow \vee_{j=1}^m L'_j, M'_j$ . We say that $Q$ and $Q'$ are {\em multiset homomorphic} if
\begin{itemize}
	\item $n = m$, that is $Q$ and $Q'$ have the same number of disjuncts; and
	\item there is a permutation $\pi$ of $\{ \ 1,\ldots,n \ \}$ such that for all $i$, there is a multiset-homomorphism from $Q_{| i}$ to $Q'_{| \pi(i)}$ and from  $Q'_{| \pi(i)}$ to  $Q_{| i}$ . 
\end{itemize}
\end{definition}
} 


\begin{theorem}{\cite{Cohen06}} 
\label{cutup-cohen-thm-three-three} 
Given CCQ queries $Q$ and $Q'$. If $Q$ and $Q'$ are multiset homomorphic then $Q \equiv_C Q'$. 
\end{theorem} 

\reminder{
\begin{proof}
Consider the queries $Q(\bar{X}) \leftarrow L, M$ and $Q'(\bar{X}') \leftarrow L', M'$ . Let $\varphi$ be a multiset-homomorphism from $Q$ to $Q'$, and let $\varphi'$ be a multiset-homomorphism from $Q'$ to $Q$. Since $\varphi$ and $\varphi'$ are both injective on the multiset variables of $Q'$ and $Q$, respectively, it follows that $Q$ and $Q'$ have the same number of multiset variables. Therefore, $\varphi$ and $\varphi'$ are bijective on the multiset variables of $Q'$ and $Q$, respectively. 

Let $D$ be a database with active domain $adom(D)$, and let $\bar{d}$ be a tuple of constants in $adom(D)$. Suppose that ${\sc Res}_C(Q',D)$ contains $k > 0$ occurrences of $\bar{d}$. We show that ${\sc Res}_C(Q,D)$ contains at least $k$ occurrences of $\bar{d}$. This is sufficient in order to prove equivalence of $Q$ and $Q'$, since by symmetric arguments one can show that every tuple returned by $Q$ is returned at least as many times by $Q'$. 

Let $\Gamma$ be the set of satisfying assignments of $Q'$ into $D$ that map $\bar{X}'$ to $\bar{d}$. Let $\Gamma_{\bar{S}}$ be the restriction of assignments in $\Gamma$ to the variables in $\bar{S}(Q')$. In other words, $\Gamma_{\bar{S}}$ contains all satisfiably extendible assignments of the nonset variables of $Q'$ into $D$. There are exactly $k$ assignments $\gamma_1,\ldots,\gamma_k \in \Gamma_{\bar{S}}$. We associate each assignment $\gamma_i \in \Gamma_{\bar{S}}$ with an assignment  $\gamma^*_i \in \Gamma$ such that $\gamma_i$ is the restriction of  $\gamma_i^*$ to the variables in  ${\bar{S}}(Q')$. If there are several candidates for  $\gamma_i^*$, we choose one arbitrarily. 

Since $\varphi$ is a multiset-homomorphism from $Q$ to $Q'$, we have that $\varphi l \in L'$ for each 
subgoal $l$ of $Q$. Since  $\gamma_i^*$ is a satisfying assignment of $Q'$ into $D$, we have that  $\gamma_i^*(L')$ is satisfied by the database. (In other words, all the ground atoms in $\gamma_i^*(L')$ appear in $D$.) 
By composing the assignments we derive that $\gamma_i^* \circ \varphi(L)$ is satisfied by $D$.  Hence, $\gamma_i^* \circ \varphi$ is a satisfying assignment of $Q$ into $D$. In addition, $\gamma_i^* \circ \varphi(\bar{X}) = \bar{d}$, since $\varphi(\bar{X}) = \bar{X}'$. 

Finally, we show that no two assignments $\gamma_i^* \circ \varphi$ and $\gamma_j^* \circ \varphi$ ($i \neq j$) agree on all the multiset variables of $Q$. By the definition of $\Gamma_{\bar{S}}$, it holds that $\gamma_i$ and $\gamma_j$ differ on at least one multiset variable of $Q'$. Hence $\gamma_i^*$ and $\gamma_j^*$ also differ on at least one multiset variable of $Q'$. Since $\varphi$ maps multiset variables to multiset variables and is bijective on the multiset variables of $Q$, we derive that $\gamma_i^* \circ \varphi$ and $\gamma_j^* \circ \varphi$ differ on at least one multiset variable {\em of} $Q$. Therefore, the restrictions of the assignments $\gamma_j^* \circ \varphi$ (for all $j \leq k$) to $\bar{S}(Q)$ are all different satisfiably extendible assignments of the nonset variables of $Q$ into the database. We conclude that ${\sc Res}_C(Q,D)$ contains at least $k$ occurrences of $\bar{d}$. 

Our arguments apply for all $D$ and for all $\bar{d}$. Therefore, $Q \equiv_C Q'$, as required. 
\end{proof}

} 

{\em Note 1.} Theorem~\ref{cutup-cohen-thm-three-three} is proved in \cite{Cohen06} via showing that for two (generalized) CCQ queries $Q$ and $Q'$, the existence of a multiset-homomorphism from $Q'$ to $Q$ implies $Q$ $\sqsubseteq_C$ $Q'$. 

It is shown in \cite{Cohen06} that the sufficient equivalence condition of Theorem~\ref{cutup-cohen-thm-three-three} is not necessary for the query classes considered in \cite{Cohen06}. \reminder{How about ***CCQ*** queries, that is for my case???} 


\nop{
We demonstrate the notion of multiset-homomorphic queries in the following example. 

\begin{example}{(Example 3.4 from \cite{Cohen06})} 
Consider the following queries. 
\begin{tabbing}
$Q_6(X) \leftarrow P(X,Y), P(X,Z), P(W,U), Y < Z,  \{ Y \} .$ \\
$Q_7(X) \leftarrow P(X,Y), P(X,Z), Y < Z, \{ Y \} .$ \\
$Q_8(X) \leftarrow P(X,Y), P(X,Z), P(W,U), U < Z,  \{ Y \} .$ 
\end{tabbing}

It is not difficult to see that the queries $Q_6$ and $Q_7$ are multiset-homomorphic. Therefore, by Theorem~\ref{cutup-cohen-thm-three-three}, $Q_6 \equiv_C Q_7$. Now, consider queries $Q_6$ and $Q_8$. The mapping $\varphi$ defined as
\begin{tabbing}
$\{ X / X, Y / Y, Z / Z, W / X, U / Y \}$ 
\end{tabbing}
is a multiset-homomorphism from $Q_8$ to $Q_6$. Observe that $\varphi$ \reminder{Did she mean $\varphi^{-1}$?} is not a multiset-homomorphism from $Q_6$ to $Q_8$. In fact, no such multiset-homomorphism exists, since  $Q_6 \equiv_C \hspace{-0.6cm} / \hspace{0.4cm} Q_8$.  
\end{example}
} 


\vspace{-0.2cm} 

\section{Containment and Mappings} 
\label{containment-mappings-sec} 

In this section, for combined-semantics containment of CCQ queries, we outline two necessary conditions, Theorems~\ref{not-same-num-multiset-vars-thm} and~\ref{gen-ce-cm-thm}, which are proved in \cite{MyLibkinIcdt12journalSubmission}, and introduce a sufficient condition, Theorem~\ref{cvm-containm-thm}. The latter result properly generalizes both (i) the sufficient condition outlined in \cite{ChaudhuriV93} for bag containment of CQ queries, and  (ii) the general sufficient containment condition for CCQ queries that can be obtained from \cite{Cohen06}. 
To formulate Theorem~\ref{cvm-containm-thm}, we introduce covering mappings (CVMs) between CCQ queries. We use CVMs in our results throughout the remainder of this paper. 

Throughout this paper, we use the notation $Q(\bar{X}) \leftarrow L,M$ and $Q'(\bar{X}') \leftarrow L', M'$ for the definitions of CCQ queries $Q$ and $Q'$. The conditions of $Q$ and $Q'$ may have constants. 

\reminder{Must explain in this section by which results of this section the equivalence/nonequivalence claims of my intro examples hold: Example~\ref{weird-two-ex} and Example~\ref{mh-vs-sscm-ex}.} 

\subsection{Necessary conditions for containment} 
\label{containm-necess-cond-sec} 

We outline here two necessary conditions for a CCQ query $Q$ being combined-semantics contained in CCQ query $Q'$, Theorems~\ref{not-same-num-multiset-vars-thm} and~\ref{gen-ce-cm-thm}. Both results are proved in \cite{MyLibkinIcdt12journalSubmission}; we will need them in this current paper in our discussion of Theorem~\ref{magic-mapping-prop}.

\begin{theorem}
\label{not-same-num-multiset-vars-thm}
Let $Q$ and $Q'$ be two $k$-ary CCQ queries, for a $k$ $\geq$ $0$. 
Then $Q \sqsubseteq_C Q'$ implies both 
$|M_{copy}|$ $\leq$ $|M'_{copy}|$ and $|M_{noncopy}|$ $\leq$ $|M'_{noncopy}|$. 
\end{theorem}

\noindent 
(For a set $S$ we denote by $|S|$ the cardinality of $S$.)  

We call a pair $(Q,$ $Q')$ of CCQ queries a {\em containment-compatible CCQ pair} if (i) The (nonnegative) head arities of $Q$ and $Q'$ are the same; (ii) $|M_{copy}|$ $\leq$ $|M'_{copy}|$; and (iii) $|M_{noncopy}|$ $\leq$ $|M'_{noncopy}|$. 
(Note the asymmetry in the notation for the pair.) 
Further, we call a pair $(Q,$ $Q')$ of CCQ queries an {\em equivalence-compatible CCQ pair} if each of $(Q,$ $Q')$ and $(Q',$ $Q)$ is a containment-compatible CCQ pair. 
By Theorem~\ref{not-same-num-multiset-vars-thm} we have that, whenever $Q \sqsubseteq_C Q'$ ($Q \equiv_C Q'$, respectively) holds, then  $(Q,$ $Q')$ is a containment-compatible (an equivalence-compatible, respectively) CCQ pair.  

The next result, Theorem~\ref{gen-ce-cm-thm} (also proved in \cite{MyLibkinIcdt12journalSubmission}), generalizes the ``only-if'' part of the classic result of \cite{ChandraM77}, see Theorem~\ref{cm-thm} in Section~\ref{prev-results-sec}, to CCQ queries. For the definition of generalized containment mapping, {\em GCM,} see Section~\ref{prev-results-sec}.  We begin by introducing another definition that we need to formulate Theorem~\ref{gen-ce-cm-thm}. 

For a CCQ query $Q$, we say that 
CCQ query $Q_{ce}$ is a {\em copy-enhanced version of} $Q$ if $Q_{ce}$ is the result of adding a distinct copy variable to each relational subgoal of $Q$. We can show that for a query $Q$, all copy-enhanced versions of $Q$ are identical up to renaming of the copy variables introduced in the construction of  $Q_{ce}$. A formalization is given in \cite{MyLibkinIcdt12journalSubmission}. 

We are now ready to formulate Theorem~\ref{gen-ce-cm-thm}. 

\begin{theorem} 
\label{gen-ce-cm-thm} 
Given CCQ queries $Q$ and $Q'$ such that $Q \sqsubseteq_C Q'$. Then there exists a GCM from $Q'_{ce}$ to $Q_{ce}$.  
\end{theorem} 


Neither Theorem~\ref{not-same-num-multiset-vars-thm} nor Theorem~\ref{gen-ce-cm-thm} provides a sufficient condition for combined-semantics containment of two CCQ queries: The following Example~\ref{weird-two-ex} is a counterexample in both cases.  

\begin{example} 
\label{weird-two-ex}
Consider CCQ queries $Q$ and $Q'$:  
\begin{tabbing}
$Q(X) \leftarrow p(X,Y), p(Y,Z), p(Z,X; i), \{ Y,i \} .$ \\
$Q'(X) \leftarrow p(X,Y), p(Y,Z), p(Z,X; i), \{ Z,i \} .$
\end{tabbing} 

%
Apart from the choice of multiset variables, 
$Q$ and $Q'$ are clearly isomorphic. 
However, $Q \equiv_C Q'$ does not hold, as witnessed by database $D = \{ p(1,2),$ $p(2,3),$ $p(3,1),$ $p(1,4),$ $p(4,3) \}$. 
Our results in this paper permit us to determine $Q \equiv_C \hspace{-0.55cm} / \hspace{0.4cm} Q'$ syntactically,  see Section~\ref{necess-equiv-condition}. (To the  best of our knowledge, no previous work provides a formal procedure to determine $Q \equiv_C \hspace{-0.55cm} / \hspace{0.4cm} Q'$ for queries such as in this example.)  
\end{example}

Each of Theorem~\ref{not-same-num-multiset-vars-thm} and Theorem~\ref{gen-ce-cm-thm} yields a necessary condition for combined-semantics {\em equivalence} of CCQ queries in a natural way. For instance: 

\begin{corollary}
\label{not-same-num-multiset-vars-prop} 
Let $Q$ and $Q'$ be two $k$-ary CCQ queries such that 
$Q \equiv_C Q'$. Then we have $|M_{copy}|$ $=$ $|M'_{copy}|$ and $|M_{noncopy}|$ $=$ $|M'_{noncopy}|$. 
\end{corollary} 

\subsection{Covering mappings for CCQ queries} 
\label{sccm-sec} 

In this subsection, we introduce covering mappings {\em (CVMs)} between CCQ queries, and study properties of CVMs. 


\begin{definition}{Covering mapping (CVM)}
\label{magic-mapping-def} 
Given CCQ queries $Q$ and $Q'$, 
a mapping, call it $\mu$, from the terms of $Q'$ to the terms of $Q$ is called a {\em covering mapping (CVM) from $Q'$ to $Q$} whenever $\mu$ satisfies all of the following conditions:
\begin{itemize} 
	\item[(1)] $\mu$ maps each constant (if any) in $Q'$ to itself; 
\vspace{-0.15cm}
	\item[(2)] applying $\mu$ to the vector $\bar{X}'$ yields the vector $\bar{X}$; 
\vspace{-0.15cm}

	\item[(3)] the set of terms 
	in $\mu M'_{copy}$ is exactly $M_{copy}$, and the set of terms 
	in $\mu M'_{noncopy}$ includes all of $M_{noncopy}$; 
\vspace{-0.15cm}
	
	\item[(4)] for each relational subgoal of $Q'$, of the form $s(\bar{Y})$, there exists in $Q$ either a relational subgoal $s(\mu(\bar{Y}))$, or a copy-sensitive subgoal $s(\mu(\bar{Y}); i)$, with $i \in M_{copy}$; and 	
\vspace{-0.15cm}

	\item[(5)] for each copy-sensitive subgoal of $Q'$ of the form $s(\bar{Y}; i)$, 
	there exists in $Q$ a subgoal $s(\mu(\bar{Y}); \mu(i))$. 

\end{itemize} 
\vspace{-0.6cm}
\end{definition} 

\nop{ 

Given two equivalence-compatible CCQ queries $Q(\bar{X}) \leftarrow L, M$ and $Q'(\bar{X}') \leftarrow L', M'$. 
Then a mapping, call it $\mu$, from the terms of $Q'$ to the terms of $Q$ is called a {\em same-scale covering mapping (SSCM) from $Q'$ to $Q$} whenever $\mu$ satisfies all of the following conditions:
\begin{itemize}
	\item[(1)] $\mu$ maps each constant (if any) in $Q'$ to itself; 
	\item[(2)] $\mu[\bar{X}']$ is the vector $\bar{X}$; 

	\item[(3)] the set of terms 
	in $\mu[M'_{copy}]$ is exactly $M_{copy}$, and the set of terms 
	in $\mu[M'_{noncopy}]$ is exactly $M_{noncopy}$; 
	
	\item[(4)] for each relational subgoal of $Q'$, of the form $s(\bar{Y})$, there exists in $Q$ either a relational subgoal $s(\mu(\bar{Y}))$, or a copy-sensitive subgoal $s(\mu(\bar{Y}); i)$ where $i \in M_{copy}$;	

	\item[(5)] for each subgoal of $Q'$ of the form $s(\bar{Y}; i)$, where $i \in M'_{copy}$, there exists in $Q$ a subgoal $s(\mu(\bar{Y}); \mu(i))$, where $\mu(i) \in M_{copy}$. 

\end{itemize}

} 

By Definition~\ref{magic-mapping-def}, if there exists a CVM from CCQ query $Q'$ to CCQ query $Q$, then $(Q,$ $Q')$ is a containment-compatible CCQ pair. 
It is immediate from Definition~\ref{magic-mapping-def} that 
if a mapping $\mu$ is a CVM from  $Q'$ to $Q$, then $\mu$ induces 
		a surjection 
from the set of copy-sensitive subgoals of $Q'$ to the set of copy-sensitive subgoals of $Q$. Observe also that in case both $Q$ and $Q'$ are set queries, Definition~\ref{magic-mapping-def} becomes the definition of containment mapping \cite{ChandraM77} from $Q'$ to $Q$.  

For the special case where $(Q,$ $Q')$ is an {\em equivalence-}compat- ible CCQ pair, we call each CVM from $Q'$ to $Q$ a {\em same-scale covering mapping (SCVM) from $Q'$ to $Q$.} By definition, each SCVM from $Q'$ to $Q$ is a bijection from the set $M'$ to the set $M$ when restricted to the domain $M'$. 

The intuition for Definition~\ref{magic-mapping-def} comes from our use of CVMs in \cite{MyLibkinIcdt12journalSubmission} 
as a tool for minimizing CCQ queries. Consider the following illustration. 

\begin{example} 
\label{mh-vs-sscm-ex}
Let queries $Q$ and $Q'$ be as follows. 
\begin{tabbing}
$Q(X) \leftarrow p(X,X,Y; i), p(X,Z,Y), \{ Y,i \} .$ \\
$Q'(X) \leftarrow p(X,X,Y; i), \{ Y,i \} .$
\end{tabbing}
By Definition~\ref{multiset-homom-def}, there does not exist a multiset homomorphism \cite{Cohen06}, or even a GCM, from $Q$ to $Q'$.  At the same time, by the results of~\cite{MyLibkinIcdt12journalSubmission}, $Q'$ is a minimized version of $Q$. 
We can ascertain this fact by using a CVM, $\mu$, from $Q$ to $Q'$: $\mu$ $=$ $\{$ $X \rightarrow X$, $Y \rightarrow Y$, $i \rightarrow i$, $Z \rightarrow X$  $\}$ .  
\end{example} 

As illustrated by  Example~\ref{mh-vs-sscm-ex}, CVMs are not GCMs. Indeed, the definition of CVMs gives up explicitly on condition (3) for GCMs (see Definition~\ref{multiset-homom-def}); by this condition, for each subgoal $s$ of $Q$ in Example~\ref{mh-vs-sscm-ex}, we must have that $\mu(s)$ is a subgoal of $Q'$. While CVMs are not GCMs, a nice relationship exists between CVMs and GCMs, see Proposition~\ref{cvm-is-gcm-prop}. To formulate Proposition~\ref{cvm-is-gcm-prop}, we use the following definition, in which we treat query conditions as bags of atoms.  

Given CCQ query $Q$, let ${\cal T}(Q)$ be the set of relational templates of all (if any) copy-sensitive subgoals of $Q$. We recall that CCQ query $Q_c$ is a canonical representation of CCQ query $Q$ if $Q_c$ is the result of removing all duplicate atoms from the condition of $Q$. 

\begin{definition}{(Un)regularizing CCQ query} 
\label{regulariz-def} 
Given CCQ query $Q$, with canonical representation $Q_c$. Then (1) A {\em regularized version of} $Q$ is a CCQ query $Q_r$ obtained by dropping from the condition of $Q_c$ all elements of the set  ${\cal T}(Q)$; (2) A {\em deregularized version of} $Q$ is a CCQ query $Q_d$ obtained by adding to the condition of $Q_r$ all elements of the set  ${\cal T}(Q)$; (3) An {\em unregularized version of} $Q$ is a CCQ query $Q_u$ obtained by adding to the condition of $Q_r$ one or more duplicates of the existing relational subgoals, and/or one or more elements (possibly with duplicates) of the set ${\cal T}(Q)$.  
\end{definition}


The following result is straightforward. 

\begin{proposition} 
\label{regulariz-prop} 
Given a CCQ query $Q$. Then (1) Each of $Q_r$ and $Q_d$ is a well-defined, unique and polynomial-time computable CCQ query; (2) $Q_r \equiv_C Q$ and $Q_d \equiv_C Q$ both hold; and (3) For each unregularized version $Q_u$ of $Q$, we have that $Q_u \equiv_C Q$ holds. 
\end{proposition} 


We are now ready to formulate Proposition~\ref{cvm-is-gcm-prop}.


\begin{proposition} 
\label{cvm-is-gcm-prop} 
Given CCQ queries $Q$ and $Q'$. Then for each CVM, $\mu$, from $Q'$ to $Q$, we have that (1) $\mu$ is a {\em GCM} 
from $Q'$ to the deregularized version of $Q$, and (2) $\mu$ is a CVM from $Q'$ to the regularized version of $Q$. 
\end{proposition} 

In Example~\ref{mh-vs-sscm-ex}, we are given the regularized version $Q'_r$ of the query $Q'$. The deregularized version of $Q'$ is $Q'_d(X) \leftarrow p(X,X,Y; i), p(X,X,Y), \{ Y,i \}$. The mapping $\mu$ of  Example~\ref{mh-vs-sscm-ex} (i) is a GCM from $Q$ to $Q'_d$, (ii) is a CVM from $Q$ to $Q'_r$, and (iii) is not a GCM  from $Q$ to $Q'_r$.

\begin{proof}{(Proposition~\ref{cvm-is-gcm-prop})} 
Proof of (1): Let $\mu$ be a CVM from CCQ query $Q'$ to CCQ query $Q$. 
We apply the mapping $\mu$ to the head and individually to each subgoal of the query $Q'$; we will refer to the result as (query) $\mu(Q')$. In addition, we impose a natural requirement that a variable $Y$ in the query $\mu(Q')$ be a multiset variable of $\mu(Q')$ if and only if $Y$ is a multiset variable of the query $Q$. It is immediate from this requirement and from item (3) of Definition~\ref{magic-mapping-def} that the multiset variables of $\mu(Q')$ are exactly the multiset variables of the query $Q$. 

Now applying $\mu$ to the head vector of the query $Q'$ results in the head vector of $Q$, by  item (2) of Definition~\ref{magic-mapping-def}. Further, for each copy-sensitive subgoal of $Q'$, its image in $\mu(Q')$ is a copy-sensitive subgoal of the query $Q$, by  item (5) of Definition~\ref{magic-mapping-def}. We get a similar desired behavior, by  item (4) of Definition~\ref{magic-mapping-def}, for all relational subgoals of $Q'$ whose images under $\mu$ are {\em relational} subgoals of the query $Q$. 

The only problem in the application of the mapping $\mu$ to the query $Q'$ would arise when, for some relational subgoal of $Q'$ of the form $s(\bar{Y})$, the {\em relational} atom $s(\mu(\bar{Y}))$ is not a subgoal of the query $Q$. For an illustration, consider queries $Q$ and $Q'$ of Example~\ref{mh-vs-sscm-ex}: The mapping $\mu$ $=$ $\{ X \rightarrow X, Y \rightarrow Y, Z \rightarrow X, i \rightarrow i \}$ is a CVM from query $Q$ to query $Q'$ of the Example. Applying this mapping to subgoal $p(X,Z,Y)$ of the query $Q$ results in a relational subgoal $p(X,X,Y)$ that is not present in the query $Q'$.


However,  items (3) through (5) of Definition~\ref{magic-mapping-def} together ensure that for each occurrence of the above problem, query $\mu(Q')$ has ``the copy-sensitive version'', $s(\mu(\bar{Y}); i)$ (for some copy variable $i$ $\in$ $M_{copy}$), of the atom $s(\mu(\bar{Y}))$ of the previous paragraph. That subgoal  $s(\mu(\bar{Y}); i)$ would be added to $\mu(Q')$ as the result of applying $\mu$ to some copy-sensitive subgoal of the query $Q'$. Thus, adding the {\em relational} atom $s(\mu(\bar{Y}))$ to $\mu(Q')$, as the result of applying $\mu$ to the {\em relational} subgoal $s(\bar{Y})$  of the query $Q'$, does not ``take us outside'' the set of subgoals of the {\em deregularized version} of the query $Q$. (Recall the definition of the set ${\cal T}(Q)$.) 
As a result, the condition of the query $\mu(Q')$ is a subset of the condition of the deregularized version of the query $Q$. Q.E.D. 

Proof of (2): Immediate from (1) and from definition of CVM. 
\end{proof}

It turns out that CVMs furnish a rather general sufficient condition for CCQ combined-semantics containment: 


\begin{theorem} 
\label{cvm-containm-thm} 
Given CCQ queries $Q$ and $Q'$, such that there exists a CVM from $Q'$ to $Q$. Then $Q \sqsubseteq_C Q'$ holds. 
\end{theorem} 

Theorem~\ref{cvm-containm-thm}  generalizes properly both (i) the sufficient condition of \cite{ChandraM77} for containment between CCQ set queries, see Theorem~\ref{cm-thm}, and (ii) the well-known result of \cite{ChaudhuriV93} stating that a containment mapping\footnote{The ``containment mapping'' terminology of \cite{ChaudhuriV93} results from the use in that paper of a syntax for bag queries that does not coincide with the syntax of \cite{Cohen06} used in this current paper. See Appendix~\ref{e-three-sec} for a detailed discussion.} from CCQ bag query $Q'$ {\em onto} CCQ bag query $Q$ ensures containment $Q \sqsubseteq_B Q'$. In fact, to the best of our knowledge, the proof of our Theorem~\ref{cvm-containm-thm} is the first formal proof of the latter result from \cite{ChaudhuriV93}. 


The condition of Theorem~\ref{cvm-containm-thm} does not appear to be a necessary condition for containment of CCQ queries. Indeed, a well-known example of \cite{ChaudhuriV93} (see Appendix~\ref{chaudhuri-v-ex-sec}), claims containment $Q \sqsubseteq_C Q'$, but no CVM exists from $Q'$ to $Q$. 

We now prove Theorem~\ref{cvm-containm-thm}. Remarkably, the proof is almost verbatim the proof, in \cite{Cohen06}, of the result that is given as Theorem~\ref{cutup-cohen-thm-three-three} in this current paper. %

\begin{proof}{(Theorem~\ref{cvm-containm-thm})} 
Consider the queries $Q(\bar{X}) \leftarrow L, M$ and $Q'(\bar{X}') \leftarrow L', M'$ . Let $\varphi$ be a CVM from $Q'$ to $Q$. Recall that by item (3) of Definition~\ref{magic-mapping-def}, when $\varphi$ is restricted to the domain $M'$ then we have that the range of $\varphi$ includes all of $M$. 

By Proposition~\ref{cvm-is-gcm-prop}, $\varphi$ is a generalized containment mapping from the query $Q'$ to the deregularized version $Q_d$ of $Q$, $Q_d(\bar{X}) \leftarrow L_d, M_d$.  By Proposition~\ref{regulariz-prop}, $Q_d \equiv_C Q$. (In particular, this means that the set $M_d$ is identical to $M$, and  that the set ${\bar S}(Q_d)$ is identical to ${\bar S}(Q)$.) 

Let $D$ be a database with active domain $adom(D)$, and let $\bar{t}$ be a tuple of constants in $adom(D)$. Suppose that ${\sc Res}_C(Q_d,D)$ contains $k > 0$ occurrences of $\bar{t}$. (The case where the head vectors of the two queries are empty, and hence $\bar{t}$ is the empty tuple, is proved in the same way as the case where $\bar{t}$ is a tuple with at least one constant. We omit the proof of the former case.) We show that ${\sc Res}_C(Q',D)$ contains at least $k$ occurrences of $\bar{t}$. This is sufficient in order to prove combined-semantics containment of $Q_d$ in $Q'$. From this result, by $Q_d \equiv_C Q$ we obtain that $Q \sqsubseteq_C Q'$ also holds. 

Let $\Gamma$ be the set of satisfying assignments of $Q_d$ into $D$ that map $\bar{X}$ to $\bar{t}$. Let $\Gamma_{\bar{S}}$ be the restriction of assignments in $\Gamma$ to the variables in $\bar{S}(Q_d)$. 
There are exactly $k$ assignments $\gamma_1,\ldots,\gamma_k \in \Gamma_{\bar{S}}$. We associate each assignment $\gamma_i \in \Gamma_{\bar{S}}$ with an assignment  $\gamma^*_i \in \Gamma$ such that $\gamma_i$ is the restriction of  $\gamma_i^*$ to the variables in  ${\bar{S}}(Q_d)$. If there are several candidates for  $\gamma_i^*$, we choose one arbitrarily. 

\nop{
Consider the deregularized version $Q_d(\bar{X}) \leftarrow L_d, M$ of the query $Q$. Let $\Gamma(Q_d)$ be the set of satisfying assignments of $Q_d$ into $D$ that map $\bar{X}$ to $\bar{d}$. Let $\Gamma_{\bar{S}}(Q_d)$ be the restriction of assignments in $\Gamma(Q_d)$ to the variables in $\bar{S}(Q_d)$. By definition of deregularized versions of CCQ queries, we have that 
(i) the set $\Gamma(Q_d)$ is identical to the set $\Gamma(Q)$, (ii) the set $\Gamma_{\bar{S}}(Q_d)$ is identical to the set $\Gamma_{\bar{S}}(Q)$, and (iii) the bag ${\sc Res}_C(Q_d,D)$ contains exactly $k$ occurrences of $\bar{d}$. 
} 

Recall that $\varphi$ is a generalized containment mapping from $Q'$ to $Q_d$. Thus, we have that  for each 
subgoal $l'$ of $Q'$, $\varphi l' \in L_d$ holds. 
Since  $\gamma_i^*$ is a satisfying assignment of $Q_d$ into $D$, we have that  $\gamma_i^*(L_d)$ is satisfied by the database. (In other words, all the ground atoms in $\gamma_i^*(L_d)$ appear in $D$.) 
By composing the assignments we derive that $\gamma_i^* \circ \varphi(L')$ is satisfied by $D$.  
Hence, $\gamma_i^* \circ \varphi$ is a satisfying assignment of $Q'$ into $D$. In addition, $\gamma_i^* \circ \varphi(\bar{X}') = \bar{t}$, since $\varphi(\bar{X}') = \bar{X}$. 

Finally, we show that no two assignments $\gamma_i^* \circ \varphi$ and $\gamma_j^* \circ \varphi$ ($i \neq j$) agree on all the multiset variables of $Q'$. By the definition of $\Gamma_{\bar{S}}$, it holds that $\gamma_i$ and $\gamma_j$ differ on at least one multiset variable of $Q_d$. Hence $\gamma_i^*$ and $\gamma_j^*$ also differ on at least one multiset variable of $Q_d$. Since the image of $M'$ under $\varphi$ 
includes all of $M$, 
we derive that $\gamma_i^* \circ \varphi$ and $\gamma_j^* \circ \varphi$ differ on at least one multiset variable {\em of} $Q'$. Therefore, the restrictions of the assignments $\gamma_j^* \circ \varphi$ (for all $j \leq k$) to $\bar{S}(Q')$ are all different satisfiably extendible assignments of the nonset variables of $Q'$ into the database. We conclude that ${\sc Res}_C(Q',D)$ contains at least $k$ occurrences of $\bar{t}$. 

Our arguments apply for all $D$ and for all $\bar{t}$. Therefore, $Q_d \sqsubseteq_C Q'$ and (by $Q_d \equiv_C Q$) we have $Q \sqsubseteq_C Q'$, as required. 
\end{proof} 

\subsection{CVMs and multiset homomorphisms of [9]}

In this subsection we compare CVMs with multiset homomorphisms \cite{Cohen06}, see Definition~\ref{multiset-homom-def}. For a fixed pair of CCQ queries $Q$ and $Q'$,  with respective sets of multiset variables $M$ and $M'$, each CVM from $Q'$ to $Q$ 
has range {\em at least} $M$ when restricted to the domain $M'$, and each multiset homomorphism from $Q'$ to $Q$ has range {\em at most} $M$
when restricted to the domain $M'$. Therefore, general CVMs and multiset homomorphisms are incomparable when applied to pairs of CCQ queries. (See Example~\ref{scvm-vs-mhom-ex} in Appendix~\ref{e-four-sec}.) At the same time, we have the following result for {\em SCVMs} and multiset-homomorphisms. (The proof, which is immediate from Proposition~\ref{cvm-is-gcm-prop}, can be found in Appendix~\ref{e-four-sec}.) 

\begin{proposition} 
\label{sccm-vs-mhomom-prop} 
Given an equivalence-compatible CCQ pair $(Q,$ $Q')$. Then each SCVM from $Q'$ to $Q$ is a multiset-homomorphism from $Q'$ to {\em the deregularized version} of $Q$, and vice versa. 
\end{proposition} 

For instance, consider the mapping $\mu$ of Example~\ref{mh-vs-sscm-ex} from the terms of the query $Q$ to the terms of the query $Q'$ of the example. This mapping is a CVM from $Q$ to $Q'$ and is also a multiset-homomorphism from $Q$ to the deregularized version $Q'_d$ of $Q'$, $Q'_d(X) \leftarrow p(X,X,Y; i), p(X,X,Y), \{ Y,i \} .$ (Observe that there is no GCM from query $Q'_d$ to query $Q'$.)  

As an immediate corollary of Propositions~\ref{cvm-is-gcm-prop} and~\ref{sccm-vs-mhomom-prop}, we have that for each equivalence-compatible CCQ pair $(Q, Q')$, the existence of a multiset-homomorphism from $Q'$ to $Q$ implies the existence of a CVM from $Q'$ to $Q$. From this result and from Example~\ref{mh-vs-sscm-ex}, we obtain that 
the restriction of Theorem~\ref{cutup-cohen-thm-three-three} (due to \cite{Cohen06}) to CCQ queries does not have quite the same power as the sufficient condition for equivalence of CCQ queries that is immediate from Theorem~\ref{cvm-containm-thm}. (By Theorem~\ref{cvm-containm-thm}, we have $Q$ $\equiv_C$ $Q'$ for the queries of Example~\ref{mh-vs-sscm-ex}.) In fact, by Example~\ref{mh-vs-sscm-ex} we have that our Theorem~\ref{cvm-containm-thm} is a proper generalization of the (implicit) query-containment condition of \cite{Cohen06}, provided that the latter is applied to CCQ queries only; see Note 1 in Section~\ref{prev-results-sec}. (By Definition~\ref{multiset-homom-def} and by Theorem~\ref{not-same-num-multiset-vars-thm}, the existence of a multiset-homomorphism from CCQ query $Q'$ to CCQ query $Q$ implies $Q$ $\sqsubseteq_C$ $Q'$ only when $(Q, Q')$ is an equivalence-compatible CCQ pair.)  


\reminder{Old content starts below here in the remainder of this subsection}

\reminder{ 

\begin{corollary}
\label{well-behaved-mu-one-corol}
Given CCQ queries $Q$ and $Q'$: 
Then each same-scale covering mapping from $Q'$ to $Q$ is an MSCM from $Q'$ to the deregularized version of $Q$. Conversely, each MSCM from $Q'$ to the deregularized version of $Q$ is a same-scale covering mapping from $Q'$ to $Q$. 
\end{corollary}

\reminder{Define isomorphism SCVM (for Theorem~\ref{unique-minimiz-vers-thm}).}

} 

\vspace{-0.2cm} 

\section{Equivalence: Asymmetric \\ Necessary Condition} 
\label{necess-equiv-condition} 

In this section we present a necessary condition for  equivalence of CCQ queries, Theorem~\ref{magic-mapping-prop}. To formulate Theorem~\ref{magic-mapping-prop}, we isolate a large well-behaved class of combined-semantics CQ queries, which we call ``explicit-wave queries.''  Theorem~\ref{magic-mapping-prop} is asymmetric: It states that if for CCQ queries $Q$ and $Q'$ the combined-semantics equivalence $Q \equiv_C Q'$ holds, {\em and} we have that $Q$ is  an explicit-wave query, then there exists a CVM from $Q'$ to $Q$.  As we will see in Section~\ref{magic-mapping-proof-sec}, establishing this result  is not trivial.  
We also formulate and prove the main result of this paper, Theorem~\ref{necess-suff-equiv-cond-thm}, and show how it gives rise to an algorithm for determining whether two explicit-wave CCQ queries are or are not combined-semantics equivalent. 

We begin by introducing Definition~\ref{expl-wave-def}. This technical definition is required for the proof of Theorem~\ref{magic-mapping-prop} to go through. Given a CCQ query $Q$, with set $M_{noncopy}$ $\neq$ $\emptyset$ of multiset noncopy variables, we say that a GCM $\mu$ from $Q$ to itself is a {\em noncopy-permuting GCM} if the mapping resulting from restricting the domain of $\mu$ to $M_{noncopy}$ is a bijection from $M_{noncopy}$ to itself. For two noncopy-permuting GCMs, $\mu_1$ and $\mu_2$, from $Q$ to itself, we say that $\mu_1$ and $\mu_2$ {\em agree on} $M_{noncopy}$ if $\mu_1$ and $\mu_2$ induce the same mapping from $M_{noncopy}$ to itself. If for  CCQ query $Q$ we have $M_{noncopy}$ $=$ $\emptyset$, we say that all GCMs from $Q$ to itself are noncopy-permuting GCMs, and that all pairs of such GCMs agree on $M_{noncopy}$. 

In Definition~\ref{expl-wave-def}, for a CCQ query $Q$ and for its copy-enhanced version $Q_{ce}$, we will call ``the original copy-sensitive subgoals of $Q$'' those copy-sensitive atoms that are present in the conditions of both $Q$ and $Q_{ce}$. 

\begin{definition}{Explicit-wave CCQ query} 
\label{expl-wave-def} 
A CCQ query $Q$  is an {\em explicit-wave (CCQ) query} if one of the following conditions holds: 
\begin{itemize} 
	\item[(1)] $Q$ has at most one copy-sensitive subgoal; or 
	\item[(2)] For the set $M_{noncopy}$ of multiset noncopy variables of $Q$, and for each pair $(\mu_1,\mu_2)$ of noncopy-permuting GCMs from $Q_{ce}$ to itself, such that $\mu_1$ and $\mu_2$ {\em agree on} $M_{noncopy}$, for each original copy-sensitive subgoal, $s$, of $Q$ we have that $\mu_1(s)$ and $\mu_2(s)$ have the same relational template. 
\end{itemize} 
\vspace{-0.5cm} 
\end{definition}  

The problem of determining whether a given CCQ query is an explicit-wave query can easily be seen to be in co-NP. It is open whether this upper complexity bound is tight. 

As an example, any CCQ query $Q$ that has a distinct predicate name for each subgoal (i.e., is a query without self-joins) can  be shown to be an explicit-wave query. Further, a polynomial-time checkable sufficient condition for a CCQ query to be explicit wave is that each subgoal of the query not contain both a set variable and a copy variable: 

\begin{proposition} 
\label{suffic-for-expl-wave-prop}
Given a CCQ query $Q$ such that each copy-sensitive subgoal of $Q$ has no set variables. Then $Q$ is an explicit-wave query. 
\end{proposition} 

(This result is straightforward. For completeness, a proof can be found in Appendix~\ref{suffic-for-expl-wave-sec}.)   

By Proposition~\ref{suffic-for-expl-wave-prop}, each CCQ set query is an explicit-wave query, and so is each CCQ bag query and each CCQ bag-set query. In addition, Example~\ref{nex-ex} in Section~\ref{intro-sec} exhibits SQL queries that are  explicit-wave CCQ queries, and such that none of the previously known tests for combined-semantics equivalence are applicable to them. 

Further, again by Proposition~\ref{suffic-for-expl-wave-prop}, we posit that explicit-wave queries include all those CCQ queries that are expressible in SQL. (That is, it appears that all ``practical'' CCQ queries are explicit wave.) Indeed, we are not aware of a way in SQL to enforce, for some relation in the main {\tt FROM} clause of a query without the {\tt DISTINCT} keyword, that the multiplicity of one individual argument of the relation not contribute to the multiplicity of the query answers. 

For each CCQ query $Q$ that is not explicit-wave, we call $Q$ an {\em implicit-wave query.} Consider an illustration. 


\begin{example} 
\label{intro-weird-ex} 
Consider CCQ queries $Q$ and $Q'$. 
\begin{tabbing} 
Hehetab b \= hehe \kill
$Q(X_1) \leftarrow r(X_1,Y_1,Y_2,X_2; i), r(X_1,Y_1,Y_2,X_3; j), \{ Y_1,Y_2, i, j \} .$ \\
$Q'(X_1) \leftarrow r(X_1,Y_1,Y_2,X_2; i), r(X_1,Y_1,Y_2,X_2; j), \{ Y_1,Y_2, i, j \} .$ 
\end{tabbing} 

The only difference between the queries 
is that the two subgoals of the query $Q$ have different {\em set} 
variables, $X_2$ and $X_3$, whereas the two subgoals of $Q'$ have the same set variable $X_2$. 
We can show (see Appendix~\ref{abc-app}) that the query $Q$  is an implicit-wave query.  

There exist both a multiset homomorphism and a CVM from the query $Q$ to the query $Q'$. (Recall that each of the two mappings provides a sufficient condition for $Q'$ $\sqsubseteq_C$ $Q$.) 
%
%
Observe  that there is no isomorphism between $Q$ and $Q'$. The remarkable part is that no multiset homomorphism or CVM exists in the opposite direction, that is from $Q'$ to $Q$. 
Yet, $Q \equiv_C Q'$ does hold (see Appendix~\ref{weird-ex-sec}). It does not help much that there exists a GCM from $Q'$ to $Q$.  
By Theorem~\ref{gen-ce-cm-thm}, the existence of a GCM is a necessary, rather than sufficient, condition for the containment $Q$ $\sqsubseteq_C$ $Q'$. (To  apply Theorem~\ref{gen-ce-cm-thm}, observe that $Q$ and $Q_{ce}$ are identical, as are $Q'$ and $Q'_{ce}$.) 
\end{example} 

Queries such as the query $Q$ of Example~\ref{intro-weird-ex} are of the kind that does not seem to have been studied before. 
For instance, as we have just argued, implicit-wave CCQ queries cannot occur under set, bag, or bag-set semantics. 
By Theorem~\ref{magic-mapping-prop} in this section,  under these three traditional semantics, as well as in other cases of combined semantics,  there exist {\em symmetric CVM} mappings between equivalent CCQ queries. That is, for each pair ($Q$, $Q'$) of combined-semantics CCQ queries such that each of $Q$ and $Q'$ is an explicit-wave query, $Q \equiv_C Q'$ implies that a CVM exists from $Q$ to $Q'$. What is important is that in all such cases, a mapping of the {\em same type} (i.e., also a CVM) always exists also from $Q'$ to $Q$. Example~\ref{intro-weird-ex} illustrates that such symmetry does not hold  for unrestricted pairs of CQ queries under combined semantics.


We now state Theorem~\ref{magic-mapping-prop}. 

\begin{theorem} 
\label{magic-mapping-prop} 
Given CCQ queries $Q$ and $Q'$, such that (i) $Q$ is an explicit-wave query, and (ii) $Q \equiv_C Q'$. Then  there exists a SCVM from $Q'$ to $Q$. 
\end{theorem} 

Theorems~\ref{cvm-containm-thm} and \ref{magic-mapping-prop} yield immediately a necessary and sufficient equivalence condition for CCQ explicit-wave queries, see Theorem~\ref{necess-suff-equiv-cond-thm}. 

Due to the well-known example of \cite{ChaudhuriV93} (Appendix~\ref{chaudhuri-v-ex-sec}), it appears that condition (ii) of  Theorem~\ref{magic-mapping-prop} cannot be replaced by condition $Q \sqsubseteq_C Q'$ (while also replacing SCVMs by CVMs), even when $Q$ is an explicit-wave query. Alternatively, we cannot remove condition (i) of Theorem~\ref{magic-mapping-prop}. Indeed, in Example~\ref{intro-weird-ex} there is a SCVM from query $Q$ to explicit-wave  query $Q'$, but there is no SCVM from $Q'$ to $Q$, even though $Q \equiv_C Q'$ holds. Thus, Theorem~\ref{magic-mapping-prop} provides an {\em asymmetric} necessary condition for CCQ-query equivalence. This asymmetry does not appear to have been explored in previous work. One reason for this is that, as we have shown, under the three traditional semantics all CCQ queries are explicit-wave queries. In \cite{Cohen06}, Cohen explores query classes that properly subsume the class of CCQ queries. When restricted to CCQ queries, all the necessary and sufficient conditions of \cite{Cohen06} for combined-semantics query equivalence require the queries to be explicit-wave queries. 
(We note that none of the necessary and sufficient conditions of \cite{Cohen06} applies to our Examples~\ref{weird-two-ex} or~\ref{mh-vs-sscm-ex}, even though all the queries in the two examples are explicit-wave queries. Yet, by an equivalence test that is immediate from our Theorems~\ref{cvm-containm-thm} and \ref{magic-mapping-prop}, $Q \equiv_C \hspace{-0.55cm} / \hspace{0.4cm} Q'$ for the queries of Example~\ref{weird-two-ex}, and $Q \equiv_C Q'$ for the queries of Example~\ref{mh-vs-sscm-ex}. See Appendix~\ref{cohen-expl-wave-app} for all the details.)  

We now formulate the main result of this paper, our decidability result for combined-semantics equivalence of CCQ explicit-wave queries. 

\begin{theorem} 
\label{necess-suff-equiv-cond-thm}
Given explicit-wave CCQ queries $Q_1$ and $Q_2$. 
Then $Q_1 \equiv_C Q_2$ if and only if there exists a CVM from $Q_1$ to $Q_2$, and another from $Q_2$ to $Q_1$. 
\end{theorem} 

The result of Theorem~\ref{necess-suff-equiv-cond-thm} is immediate from Theorems~\ref{magic-mapping-prop} and~\ref{cvm-containm-thm}. 

We use Theorem~\ref{necess-suff-equiv-cond-thm} to develop the following algorithm for determining whether two explicit-wave CCQ queries $Q$ and $Q'$ are or are not combined-semantics equivalent. (The correctness of the algorithm is by Theorems~\ref{not-same-num-multiset-vars-thm} and~\ref{necess-suff-equiv-cond-thm}.)   

{\bf Input:} Pair $(Q,$ $Q')$ of CCQ queries. 

{\bf Output:} Determination whether  $Q$ $\equiv_C$ $Q'$ does or does not hold, in case both $Q$ and $Q'$ are explicit-wave queries. 

{\bf Procedure:} 
\begin{enumerate} 
	\item If $(Q,$ $Q')$ is not an equivalence-compatible CCQ pair, then stop and return ``not combined-semantics equivalent.'' 
	\item Use Definition~\ref{expl-wave-def} to ascertain that both $Q$ and $Q'$ are explicit-wave queries. 
	\item Whenever both $Q$ and $Q'$ are explicit-wave queries, use the test of Theorem~\ref{necess-suff-equiv-cond-thm} to report whether $Q$ $\equiv_C$ $Q'$ does or does not hold. 
\end{enumerate} 

We note that instead of using in step 2 of the algorithm the procedure of Definition~\ref{expl-wave-def}, which is in co-NP, we can alternatively apply the polynomial-time sufficient condition of Proposition~\ref{suffic-for-expl-wave-prop} for a CCQ query to be explicit wave. As discussed earlier in this section, this sufficient condition appears to cover all practical cases of SQL queries. Further, step 3 of the algorithm is NP-complete; this is easy to see when we recall that containment mappings of \cite{ChandraM77} are a special case of CVMs. (It is straightforward to argue that finding, for an equivalence-compatible CCQ pair $(Q,$ $Q')$, a CVM from $Q$ to $Q'$ when both queries are {\em set} queries is at least as hard as finding a CVM from $Q$ to $Q'$ in the general case. This claim follows from the fact that, in case where $(Q,$ $Q')$ is an equivalence-compatible CCQ pair, each CVM from $Q$ to $Q'$ is by definition a bijection from the set $M$ of multiset variables of $Q$ to the set $M'$ of multiset variables of $Q'$. Intuitively, as CCQ set queries do not have multiset variables, when looking for the existence of a CVM from CCQ set query $Q$ to CCQ set query $Q'$ we ``have to consider the maximal possible number of options'' when enumerating candidate CVMs from the terms of $Q$ to the terms of $Q'$.) 

Finally, we observe that even if one or more of the queries $Q$ and $Q'$ is an implicit-wave CCQ query, then still Theorem~\ref{cvm-containm-thm} could be used in at least some cases, to ascertain that $Q$ $\equiv_C$ $Q'$. (Recall that Theorem~\ref{cvm-containm-thm} applies to all CCQ queries, rather than just explicit wave.)  Similarly, step 1 of the algorithm would determine correctly nonequivalence for pairs of CCQ queries where neither query in the pair has to be explicit wave. 

\vspace{-0.2cm} 

\section{Proof of Theorem 4.1} 
\label{magic-mapping-proof-sec} 

In this section we provide a proof of Theorem~\ref{magic-mapping-prop}. 


\reminder{Into high-level explanation of the proof (in the main text of the paper): Explain how the complexity of the proof is compounded by the need to treat in fundamentally different ways (a) the multiset noncopy variables, and (b) the copy variables} 

\reminder{Say here about the examples that provide illustrations of the notation and constructs of the proof of Theorem~\ref{magic-mapping-prop}: Example~\ref{main-proof-ex}, in Section~\ref{main-proof-ex-sec}, illustrates the constructions and notation introduced in Sections~\ref{basics-sec} through~\ref{monomial-class-mappings-sec} of this proof. Example~\ref{main-proof-cont-ex}, of Section~\ref{main-proof-cont-ex-sec}, illustrates  the constructions and notation introduced in Section~\ref{putting-together-f-sec}, specifically in Subsection~\ref{easy-funct-case-sec}, of this proof. 
Finally, Example~\ref{writeup-weird-ex} in Section~\ref{writeup-weird-ex-sec} illustrates the general case of the constructions introduced in Section~\ref{putting-together-f-sec}. In addition, Example~\ref{writeup-weird-ex} exhibits a CCQ query $Q$ that is {\em not} an explicit-wave query, and a CCQ query $Q'$ such that $Q' \equiv_C Q$. [[[Where do I prove $Q' \equiv_C Q$??? -- I need the proof!!!]]] Hence, Theorem~\ref{magic-mapping-prop} does not guarantee the existence of a SCVM from $Q'$ to $Q$. Indeed, we show that for the $Q$ and $Q'$ of Example~\ref{writeup-weird-ex}, a SCVM from $Q'$ to $Q$ does not exist, even though $Q' \equiv_C Q$ holds.}

\subsection{Intuition for the Proof and Extended \\ Example} 
\label{ext-wave-ex-sec} 

\subsubsection{Intuition for the Proof} 
\label{proof-intuition-sec} 

In this subsection we outline the idea of the proof of Theorem~\ref{magic-mapping-prop}. Intuitively, we generalize the proof, via canonical databases, of the existence of a containment mapping \cite{ChandraM77} from CCQ set query $Q'$ to CCQ set query $Q$ whenever $Q \equiv_S Q'$. The challenge in the generalization is that we are looking for a SCVM from $Q'$ to $Q$, that is, the desired mapping must map each multiset variable of $Q'$ into a distinct multiset variable of $Q$. Showing that we have constructed a mapping with this property is thus an essential part of the proof. (Observe that based on the conditions of Theorem~\ref{magic-mapping-prop},   we have no information about the structural relationship between the two queries.) 

 For a given CCQ query $Q$, the proof of Theorem~\ref{magic-mapping-prop} constructs an infinite number of databases, where each database $D_{{\bar N}^{(i)}}(Q)$, $i \geq 1$, can be thought of as a union  of (suitable modifications of) ``canonical databases'' for $Q$. (See Section \ref{new-ccq-canon-db-sec} for the definition.) Similarly to canonical databases for CCQ set queries, each ground atom in each database $D_{{\bar N}^{(i)}}(Q)$ can be associated, via a mapping that we denote $\nu^{(i)}_Q$, with a unique subgoal of the query $Q$. 
See Section~\ref{db-constr-sec} for the details. 

The role of each database $D_{{\bar N}^{(i)}}(Q)$ in the proof is that the database represents a particular combination of multiplicities of the values of (some of) the multiset variables $Y_1$, $Y_2$, $\ldots,$ $Y_n$, for some $n \geq 1$, of the query $Q$. (We have that $n \geq 1$ for all CCQ queries $Q$ and $Q'$ such that $Q \equiv_C Q'$ and at least one of $Q$ and $Q'$ is not a set query.) For each database $D_{{\bar N}^{(i)}}(Q)$, 
we represent the $n$ respective multiplicities as natural numbers $N_1^{(i)}$ through $N_n^{(i)}$, or equivalently via the $n$-ary vector ${\bar N}^{(i)}$. 

By construction of the databases, we have that some fixed tuple, $t^*_Q$, is an element of the bag ${\sc Res}_C(Q,D_{{\bar N}^{(i)}}(Q))$ for each $i \geq 1$. Moreover, {\em for all queries $Q''$ such that $(Q,$ $Q'')$ is an equivalence-compatible CCQ pair,}  we have that the multiplicity of the tuple $t^*_Q$ in each bag ${\sc Res}_C(Q'',D_{{\bar N}^{(i)}}(Q))$ (that is, for each $i \geq 1$) can be expressed using the symbolic representations, $N_1$ through $N_n$, of the respective elements $N_1^{(i)},$ $\ldots,$ $N_n^{(i)}$ of the vector ${\bar N}^{(i)}$. That is, for each such query $Q''$, we can obtain explicitly a function, ${\cal F}_{(Q)}^{(Q'')}$, in terms of the $n$ variables  $N_1,$ $\ldots,$ $N_n$, such that whenever we substitute $N_j^{(i)}$ for $N_j$, for each $j \in \{ 1,\ldots,n \}$, the resulting expression in terms of $N_1^{(i)},$ $\ldots,$ $N_n^{(i)}$ evaluates to the multiplicity of the tuple $t^*_Q$ in the bag ${\sc Res}_C(Q'',D_{{\bar N}^{(i)}}(Q))$. See Sections~\ref{valid-map-sec} through~\ref{monomial-classes-sec} and Section~\ref{putting-together-f-sec} for the construction of the function. Sections~\ref{main-proof-ex-sec} and~\ref{putting-together-f-sec} contain extended illustrations of the constructions. 

Observe that for each CCQ query $Q'$ such that $Q' \equiv_C Q$, it must be that the functions  ${\cal F}_{(Q)}^{(Q')}$ and  ${\cal F}_{(Q)}^{(Q)}$ output the same value on each database $D_{{\bar N}^{(i)}}(Q)$, $i \geq 1$. 

Consider the simplest case, where our query $Q$ 
has no self-joins  and has $|M|$ $=$ $n$ $\geq$ $1$.  
In this case, by construction of the databases, we have that the function ${\cal F}_{(Q)}^{(Q)}$ {\em for the query $Q$} is the monomial $\Pi_{j=1}^n N_j$. Consider an arbitrary assignment, $\gamma$, from $Q$ to a $D_{{\bar N}^{(i)}}(Q)$. Each such $\gamma$ has contributed to the construction of the database; we call $\gamma$ a {\em generative assignment from $Q$ to} $D_{{\bar N}^{(i)}}(Q)$.  
We can show that the composition $\nu^{(i)}_Q$ $\circ$ $\gamma$ is a SCVM from $Q$ to itself. (Note the presence in the product $\Pi_{j=1}^n N_j$ of the variables for all the $n$ multiset variables of $Q$.) Moreover, for each query $Q'$ such that $Q' \equiv_C Q$, the function ${\cal F}_{(Q)}^{(Q')}$ is forced (by $Q' \equiv_C Q$ and by ${\cal F}_{(Q)}^{(Q)}$ being a multivariate polynomial) to be exactly $\Pi_{j=1}^n N_j$, regardless of the relationship between the structures of $Q$ and $Q'$. We show that whenever ${\cal F}_{(Q)}^{(Q')}$ $=$ $\Pi_{j=1}^n N_j$, an assignment from $Q'$ to a database $D_{{\bar N}^{(i)}}(Q)$ can be composed with the mapping $\nu^{(i)}_Q$ to yield  a SCVM from $Q'$ to $Q$, precisely due to the presence in the function ${\cal F}_{(Q)}^{(Q')}$ of the ``representative'' $N_j$ of each multiset variable $Y_j$ of the query $Q$, for each $j \leq n$. 

The above exposition conveys the general intuition of the proof of Theorem~\ref{magic-mapping-prop}:  For all CCQ queries $Q$, there is a monomial, in terms of {\em all} of $N_1$, $\ldots$, $N_n$, that contributes to the construction of the function ${\cal F}_{(Q)}^{(Q)}$ and that reflects the multiplicity, in the set $\Gamma^{(t^*_Q)}_{\bar{S}}(Q,D_{{\bar N}^{(i)}}(Q))$, of all generative assignments from $Q$ to databases $D_{{\bar N}^{(i)}}(Q)$. (The set $\Gamma_{\bar S}^{(t^*_Q)}(Q,D_{{\bar N}^{(i)}}(Q))$ is the set projection, on the vector ${\bar S}(Q)$, of the set, denoted $\Gamma^{(t^*_Q)}(Q,D_{{\bar N}^{(i)}}(Q))$, of all satisfiable assignments for $Q$ and $D_{{\bar N}^{(i)}}(Q)$ that contribute the tuple $t^*_Q$ to the answer to $Q$ on $D_{{\bar N}^{(i)}}(Q)$.)  We call this monomial, ${\cal P}^{(Q)}_*$, {\em the wave of the query} $Q$ w.r.t. $\{ D_{\bar{N}^{(i)}}(Q) \}$. 
(See Section~\ref{monomial-class-mappings-sec} for the definition.) 
Suppose that, for a query $Q'$ such that $Q' \equiv_C Q$, we can show that the function   ${\cal F}_{(Q)}^{(Q')}$ has, as a term, the wave of $Q$ w.r.t. $\{ D_{\bar{N}^{(i)}}(Q) \}$,  {\em backed up by assignments from $Q'$ to the databases}  $D_{{\bar N}^{(i)}}(Q)$. Then we can use these assignments and the mapping $\nu^{(i)}_Q$ to construct a SCVM from $Q'$ to $Q$. 

There are two challenges in implementing this idea for general CCQ queries. First, the term ${\cal P}^{(Q)}_*$ may not be ``visible'' in the expression for ${\cal F}_{(Q)}^{(Q)}$. As a result, the term ${\cal P}^{(Q)}_*$ does not necessarily contribute to the construction of the function ${\cal F}_{(Q)}^{(Q')}$,  even in case $Q \equiv_C Q'$. (This is exactly the case of queries $Q$ and $Q'$ of Example~\ref{intro-weird-ex}, see Example~\ref{wave-weird-ex} in Section~\ref{illustr-sec} for the details.) 
Second, in general, function  ${\cal F}_{(Q)}^{(Q')}$ 
may have terms that are {\em not} backed  up by assignments from $Q'$ to databases $D_{{\bar N}^{(i)}}(Q)$. Both challenges arise from the fact that the function ${\cal F}_{(Q)}^{(Q'')}$, in terms of $N_1$, $\ldots,$ $N_n$, is, in general, {\em not} a multivariate polynomial on its entire domain. 

To overcome the first challenge, we introduce the restriction that $Q$ be an {\em explicit-wave query.} (Hence Definition~\ref{expl-wave-def} is necessarily technical.) Even under this restriction, overcoming the second challenge requires a nontrivial proof. (See Section~\ref{q-prime-has-wave-sec}, with its main result Proposition~\ref{qprime-goldfish-prop}.) 

\subsubsection{The main propositions that imply the result of Theorem~\ref{magic-mapping-prop}} 
\label{main-results-summary-sec} 

In a little more detail, the proof of Theorem~\ref{magic-mapping-prop} is immediate from the following three results. ({\em Monomial classes} and their {\em multiplicity monomials} are constructs that we introduce to build functions ${\cal F}_{(Q)}^{(Q)}$ and ${\cal F}_{(Q)}^{(Q')}$. All the details on these constructs can be found in Section~\ref{monomial-classes-sec}; Section~\ref{main-proof-ex-sec} contains an extended illustration of these constructs. Intuitively, a monomial class for a query and database is a set of certain-format assignments from the query to the database; the multiplicity monomial of a monomial class is a monomial that expresses how many tuples of a certain form the assignments for that monomial class contribute to the bag of answers to the query on the database.) A high-level map of the entire proof of Theorem~\ref{magic-mapping-prop} can be found in Section~\ref{map-sec}. 
\begin{itemize} 
	\item Proposition~\ref{q-has-wave-prop} of Section~\ref{monomial-class-mappings-sec} states the following: Given a CCQ query $Q$, there exists a nonempty monomial class, call it ${\cal C}_*^{(Q)}$, for the query $Q$ w.r.t. the family of databases  $\{ D_{\bar{N}^{(i)}}(Q) \}$, such that the multiplicity monomial of ${\cal C}_*^{(Q)}$ is the wave of the query $Q$ w.r.t. $\{ D_{\bar{N}^{(i)}}(Q) \}$. 
	\item Proposition~\ref{q-same-scale-mpng-prop} of Section~\ref{monomial-class-mappings-sec} states the following: 
Given  CCQ queries $Q(\bar{X}) \leftarrow L,M$ and $Q'(\bar{X}') \leftarrow L',M'$, such that (i) $Q$ and $Q'$ have the same (positive-integer) head arities, (ii) $|M_{copy}|$ $=$ $|M'_{copy}|$, and (iii) $|M_{noncopy}|$ $=$ $|M'_{noncopy}|$. Suppose that there exists a nonempty monomial class ${\cal C}_*^{(Q')}$ for the query $Q'$  w.r.t. the family of databases $\{ D_{\bar{N}^{(i)}}(Q) \}$, such that the multiplicity monomial of ${\cal C}_*^{(Q')}$ is the wave of the query $Q$  w.r.t. $\{ D_{\bar{N}^{(i)}}(Q) \}$. Then there exists a SCVM from  the query $Q'$ to the query $Q$.    
	\item Proposition~\ref{qprime-goldfish-prop} of Section~\ref{q-prime-has-wave-sec} states the following: Whenever 
	\begin{itemize} 
		\item[(a)] $Q \equiv_C Q'$ for CCQ queries $Q$ and $Q'$, and 
		\item[(b)] $Q$ is an explicit-wave CCQ query (as specified by Definition~\ref{expl-wave-def}), 
	\end{itemize} 
then there exists a (nonempty) monomial class ${\cal C}_*^{(Q')}$ for the query $Q'$ and for the family of databases $\{ D_{\bar{N}^{(i)}}(Q) \}$, such that the multiplicity monomial of ${\cal C}_*^{(Q')}$ is the wave of the query $Q$ w.r.t. $\{ D_{\bar{N}^{(i)}}(Q) \}$. 
\end{itemize}

\subsubsection{An Illustration} 
\label{illustr-sec} 

We now provide an extended illustration of how the term ${\cal P}^{(Q)}_*$ may not be ``visible'' in the expression for ${\cal F}_{(Q)}^{(Q)}$, and of how, in general, the function ${\cal F}_{(Q)}^{(Q)}$ is not a multivariate polynomial on its entire domain. 
Example~\ref{again-writeup-weird-ex} in Section~\ref{second-beyond-easy-case-sec} is an extended variant of this Example~\ref{wave-weird-ex}. In addition,  Example~\ref{again-writeup-weird-ex} illustrates how  function  ${\cal F}_{(Q)}^{(Q)}$ 
may have terms that are {\em not} backed  up by assignments from $Q$ to databases $D_{{\bar N}^{(i)}}(Q)$.   

\begin{example} 
\label{wave-weird-ex} 
For the query $Q$ of Example~\ref{intro-weird-ex}, we use the following notation for its variables: 
\begin{tabbing} 
Hehetab b \= hehe \kill
$Q(X_1) \leftarrow r(X_1,Y_1,Y_2,X_2; Y_3), r(X_1,Y_1,Y_2,X_3; Y_4),$ \\
\> $\{ Y_1,Y_2,Y_3,Y_4 \} .$ 
\end{tabbing} 

\noindent 
This notation makes it clear that the variables $Y_1$ through $Y_4$ are all multiset (copy or noncopy) variables of the query. 

We associate a variable $N_j$ with each of the $n$ $=$ $4$ multiset variables $Y_j$ of the query, $1 \leq j \leq 4$. Consider assignments $N_1$ $:=$ $1$, $N_2$ $:=$ $1$, $N_3$ $:=$ $2$, and $N_4$ $:=$ $3$. For some fixed $i$ and for $1 \leq j \leq 4$, each of these assignments associates the ``value'' $N^{(i)}_j$ with the variable $N_j$. For the resulting vector ${\bar N}^{(i)}$ $=$ $[ \ 1 \ 1 \ 2 \ 3 \ ]$, we construct the database  $D_{{\bar N}^{(i)}}(Q)$ $=$ $\{$ $r(a,b,c,d; 2),$ $r(a,b,c,e; 3)$ $\}$, as specified in the proof of Theorem~\ref{magic-mapping-prop}. We will refer to the ground atom $r(a,b,c,d; 2)$ as $d_1$, and to $r(a,b,c,e; 3)$ as $d_2$. 

Each generative assignment from $Q$ to $D_{{\bar N}^{(i)}}(Q)$ maps $X_1$ $\rightarrow$ $a$, $Y_1$ $\rightarrow$ $b$, $Y_2$ $\rightarrow$ $c$,  $X_2$ $\rightarrow$ $d$, and $X_3$ $\rightarrow$ $e$.  By definition of combined-semantics query evaluation,  these generative assignments together contribute to the set $\Gamma_{\bar S}^{(t^*_Q)}(Q,D_{{\bar N}^{(i)}}(Q))$, for $t^*_Q$ $=$ $(a)$, exactly $\Pi_{j=1}^4$ $N_j^{(i)}$ $=$ $6$ distinct tuples. (The set $\Gamma_{\bar S}^{(t^*_Q)}(Q,D_{{\bar N}^{(i)}}(Q))$ is the set projection, on the vector ${\bar S}(Q)$, of the set, denoted $\Gamma^{(t^*_Q)}(Q,D_{{\bar N}^{(i)}}(Q))$, of all satisfiable assignments for $Q$ and $D_{{\bar N}^{(i)}}(Q)$ that contribute the tuple $t^*_Q$ to the answer to $Q$ on $D_{{\bar N}^{(i)}}(Q)$.) As shown in the proof of Theorem~\ref{magic-mapping-prop},  the number of tuples contributed to the set $\Gamma_{\bar S}^{(t^*_Q)}(Q,D_{{\bar N}^{(i)}}(Q))$ by the set of generative assignments from $Q$ to $D_{{\bar N}^{(i)}}(Q)$, for an {\em arbitrary} $i \geq 1$, can be obtained by substituting, into the formula $\Pi_{j=1}^4$ $N_j$, the values $N^{(i)}_j$ of the variables $N_j$. The values $N^{(i)}_j$ come from the specific vector ${\bar N}^{(i)}$ representing the database $D_{{\bar N}^{(i)}}(Q)$. Recall that $\Pi_{j=1}^4$ $N_j$ is the wave ${\cal P}^{(Q)}_*$ of the query $Q$. 

The variety of possible assignments from the query $Q$ to the above database $D_{{\bar N}^{(i)}}(Q)$ (for the fixed vector ${\bar N}^{(i)}$) stems from the ability of the first subgoal, $g_1$, of $Q$ to map to each of the ground atoms $d_1$ and $d_2$, and from the ability of the second subgoal, $g_2$, of $Q$ to independently also map to each of $d_1$ and $d_2$. (All the above generative assignments map $g_1$ to $d_1$, and map $g_2$ to $d_2$.) We can show that for those assignments from $Q$ to $D_{{\bar N}^{(i)}}(Q)$ that map each of $g_1$ and $g_2$ to $d_1$, the set of all such assignments contributes to the set $\Gamma_{\bar S}^{(t^*_Q)}(Q,D_{{\bar N}^{(i)}}(Q))$ exactly $(N_3^{(i)})^2$ $=$ $4$ distinct tuples. Similarly, mapping $g_1$ to $d_2$ and $g_2$ to $d_1$ contributes to the set $\Gamma_{\bar S}^{(t^*_Q)}(Q,D_{{\bar N}^{(i)}}(Q))$ exactly $N_3^{(i)}$ $\times$ $N_4^{(i)}$ $=$ $6$ distinct tuples, and mapping each of $g_1$ and $g_2$ to $d_2$  contributes to the set $\Gamma_{\bar S}^{(t^*_Q)}(Q,D_{{\bar N}^{(i)}}(Q))$ exactly $(N_4^{(i)})^2$ $=$ $9$ distinct tuples. (Recall that for our fixed database $D_{{\bar N}^{(i)}}(Q)$, $N_1^{(i)}$ $=$ $N_2^{(i)}$ $=$ $1$.) Similarly to the discussion in the previous paragraph, the contribution of each such class of assignments from $Q$ to $D_{{\bar N}^{(i)}}(Q)$  for an {\em arbitrary} $i \geq 1$ can be expressed symbolically using monomials in terms of some of the variables $N_1,$ $\ldots,$ $N_4$. For instance, for the class of all assignments that map each of  $g_1$ and $g_2$ to $d_2$, the number of distinct tuples  contributed by these assignments to $\Gamma_{\bar S}^{(t^*_Q)}(Q,D_{{\bar N}^{(i)}}(Q))$, for each $i \geq 1$, can be obtained using the monomial ${\cal T}_{(d_2)}$ $=$ $N_1$ $\times$ $N_2$ $\times$ $(N_4)^2$ and each specific vector ${\bar N}^{(i)}$. 

In constructing the function ${\cal F}_{(Q)}^{(Q)}$ using all the above monomials in terms of $N_1$ through $N_4$, the problem is that we cannot simply set ${\cal F}_{(Q)}^{(Q)}$ to the sum of the monomial ${\cal P}^{(Q)}_*$ with the monomial ${\cal T}_{(d_2)}$ and with the  monomials that can be built from the other classes of assignments from $Q$ to $D_{{\bar N}^{(i)}}(Q)$ using the above reasoning. The reason is, for our fixed vector ${\bar N}^{(i)}$ $=$ $[ \ 1 \ 1 \ 2 \ 3 \ ]$, the total number of tuples in the set $\Gamma_{\bar S}^{(t^*_Q)}(Q,D_{{\bar N}^{(i)}}(Q))$ -- and thus the total number of copies of the tuple $t^*_Q$ $=$ $(a)$ in the bag ${\sc Res}_C(Q,D_{{\bar N}^{(i)}}(Q))$ -- is the result of substituting the values from the vector ${\bar N}^{(i)}$ into the single monomial ${\cal T}_{(d_2)}$. The problem stems from these different classes  of assignments possibly contributing {\em the same} tuples into the set $\Gamma_{\bar S}^{(t^*_Q)}(Q,D_{{\bar N}^{(i)}}(Q))$. (Recall that the set $\Gamma_{\bar S}^{(t^*_Q)}(Q,D_{{\bar N}^{(i)}}(Q))$ is the result of {\em set}-projecting out the columns for the {\em set} variables of the query $Q$ from the set $\Gamma^{(t^*_Q)}(Q,D_{{\bar N}^{(i)}}(Q))$ (i.e., from the set of all assignments for $Q$ and $D_{{\bar N}^{(i)}}(Q)$ that contribute the tuple $t^*_Q$ to the answer to $Q$ on $D_{{\bar N}^{(i)}}(Q)$). This {\em set} projection ``bundles together,'' into the same extendible assignment, possibly multiple distinct satisfying assignments from the query to the database.) Thus, in general, constructing the function ${\cal F}_{(Q)}^{(Q)}$ from the monomials representing different classes of assignments has to be done in a way that takes into account these overlapping contributions. 

For our specific query $Q$ of Example~\ref{intro-weird-ex}, we show in Example~\ref{again-writeup-weird-ex} in Section~\ref{second-beyond-easy-case-sec} that (1) ${\cal F}_{(Q)}^{(Q)}$ $=$ $N_1$ $\times$ $N_2$ $\times$ $(N_4)^2$ for all vectors ${\bar N}^{(i)}$ where $N^{(i)}_3$ $\leq$ $N^{(i)}_4$, and that (2) ${\cal F}_{(Q)}^{(Q)}$ $=$ $N_1$ $\times$ $N_2$ $\times$ $(N_3)^2$ for all vectors ${\bar N}^{(i)}$ where $N^{(i)}_3$ $\geq$ $N^{(i)}_4$. A compact representation of ${\cal F}_{(Q)}^{(Q)}$ that works for all $i \geq 1$ is ${\cal F}_{(Q)}^{(Q)}$ $=$ $N_1$ $\times$ $N_2$ $\times$ $(max(N_3,N_4))^2$. Clearly, this expression cannot be rewritten equivalently as a multivariate polynomial on the entire domain $\{$ ${\bar N}^{(i)}$, $i \geq 1$ $\}$ of the vector $\bar N$.  

Consider an illustration of the problem with the term ${\cal P}^{(Q)}_*$ not being ``visible'' in any of the above expressions for the function ${\cal F}_{(Q)}^{(Q)}$. Indeed, recall the query $Q'$ of Example~\ref{intro-weird-ex}; we have that $Q \equiv_C Q'$. It turns out that, even though ${\cal P}^{(Q)}_*$ has (technically) contributed to the construction of the function ${\cal F}_{(Q)}^{(Q)}$ {\em for the query} $Q$, there still does not exist a class of assignments from {\em the query} $Q'$ to database $D_{{\bar N}^{(i)}}(Q)$, such that the total number of tuples contributed to the set $\Gamma_{\bar S}^{(t^*_Q)}(Q',D_{{\bar N}^{(i)}}(Q))$ 
by the assignments in this class can be expressed by the monomial ${\cal P}^{(Q)}_*$. (Intuitively, such a class of assignments from $Q'$ to the databases cannot exist because the query $Q'$ has the same set variable in both subgoals, whereas $Q$ has different set variables in the two subgoals.) It is easy to verify that for the queries of Example~\ref{intro-weird-ex}, there does not exist a SCVM from $Q'$ to $Q$. 
\end{example} 

\subsubsection{A Map of the Proof} 
\label{map-sec} 
 
The full detailed proof of Theorem~\ref{magic-mapping-prop} is structured as follows. In Section~\ref{basics-sec}, we lay out the notation, assumptions, and basic results that run through the entire proof of Theorem~\ref{magic-mapping-prop}. In Section~\ref{db-constr-sec}, we show how to construct the family of databases ${\cal D}_{\bar{N}}(Q)$ for an arbitrary CCQ query $Q$. Then, Sections~\ref{valid-map-sec} through \ref{monomial-class-mappings-sec} furnish all the building blocks for the construction of the function ${\cal F}_{(Q)}^{(Q'')}$, for the family of databases ${\cal D}_{\bar{N}}(Q)$ and for a CCQ query $Q''$, such that $Q$ and $Q''$ are an equivalence-compatible CCQ pair. (By Theorem~\ref{not-same-num-multiset-vars-thm}, whenever two CCQ queries are combined-semantics equivalent, they must constitute an equivalence-compatible CCQ pair, please see Section~\ref{containment-mappings-sec} for the details.) 

Among these sections, Section~\ref{monomial-class-mappings-sec} introduces ``the wave'' of a CCQ query $Q$; that ``wave'' notion will play a major role in our reasoning in Section~\ref{q-prime-has-wave-sec} to complete the proof of Theorem~\ref{magic-mapping-prop}. In addition, in Section~\ref{monomial-class-mappings-sec} we prove Propositions~\ref{q-has-wave-prop} and~\ref{q-same-scale-mpng-prop}, which are two of the three main results toward the proof of Theorem~\ref{magic-mapping-prop} (see Section~\ref{main-results-summary-sec}). 

Further in the proof of Theorem~\ref{magic-mapping-prop}, Section~\ref{main-proof-ex-sec} provides an extended example that illustrates in a single flow the contents of Sections~\ref{basics-sec} through~\ref{monomial-class-mappings-sec}. Then, Section~\ref{putting-together-f-sec} puts together the function ${\cal F}_{(Q)}^{(Q'')}$  based on the results of the preceding sections. 

Finally, Section~\ref{q-prime-has-wave-sec} answers the following question: For two CCQ queries $Q$ and $Q'$ such that $Q \equiv_C Q'$, when does $Q'$ have the wave of $Q$? The answer to this question completes the proof of Theorem~\ref{magic-mapping-prop}. In particular, the proof of Proposition~\ref{qprime-goldfish-prop}, which is the third and last main result toward the proof of Theorem~\ref{magic-mapping-prop} (see Section~\ref{main-results-summary-sec}), can be found in Section~\ref{q-prime-has-wave-sec}.

\subsection{Assumptions, Conventions, Basic Results}  
\label{basics-sec} 

To streamline the exposition in the proof of Theorem~\ref{magic-mapping-prop}, we reserve a number of uppercase and lowercase Latin letters, in a variety of fonts and some with sub- and superscripts, to each have ``the standard meaning'' throughout the proof. 
To make referencing the notation easier, every effort has been made to introduce all of the notation into these initial subsections, which are separate from sections for (portions of) the proof of Theorem~\ref{magic-mapping-prop}. 

\subsubsection{Canonical databases of CCQ queries} 
\label{new-ccq-canon-db-sec}

{\bf Set queries.} We first recall the notion of a ``canonical database'' of a CCQ {\em set} query. Every CCQ set query $Q$ can be regarded as a symbolic database $D^{(Q)}$. $D^{(Q)}$ is defined as the result of 
turning each subgoal $p_i(\ldots)$ of $Q$ into a tuple in the relation $P_i$ that corresponds to predicate $p_i$. The procedure is to keep each constant in the body of $Q$, and to replace consistently each variable in the body of $Q$ by a distinct constant different from all constants in  $Q$. The tuples that correspond to the resulting ground atoms are the only tuples in the {\it  canonical database} $D^{(Q)}$ for $Q$, which (database) is unique up to isomorphism. 

{\bf General CCQ queries: Extended canonical data-} {\bf bases:} We now extend the above notion, to define an {\em extended canonical database} for a general (i.e., not necessarily set) CCQ query $Q$. We first partition all the subgoals of $Q$ into equivalence classes ${\cal C}^{(Q)}_1$, $\ldots,$ ${\cal C}^{(Q)}_k$, $k \geq 1$, where two subgoals of $Q$ belong to the same class if and only if the subgoals have the same relational template. We then choose one representative element, $c^{(Q)}_j$, of each class ${\cal C}^{(Q)}_j$, $j$ $\in$ $\{ 1,$ $\ldots,$ $k \}$; if ${\cal C}^{(Q)}_j$ has at least one copy-sensitive atom then $c^{(Q)}_j$ must be a copy-sensitive atom. Finally, an extended canonical database for the query $Q$ is constructed from the subgoals $c^{(Q)}_1$, $\ldots,$ $c^{(Q)}_k$ of $Q$ in the same way as the ``standard'' canonical database is constructed from the condition of a CCQ set query. The only difference is in that whenever $c^{(Q)}_j$ is a copy-sensitive atom, the copy variable of the atom must be replaced by a natural number that is distinct from all the other constants in the database (both in the active domain of the database and among the copy numbers of all ground atoms). The above mapping of terms of $Q$ to the constants in the database, such that the mapping is used to generate the database, can be used to define in a natural way an assignment mapping from the query $Q$ to the database. In defining that assignment mapping, we accept the convention that the mapping maps each copy variable of the query to the constant 1. We call that assignment mapping the {\em generative mapping} for the query and for the database; observe that, by definition, the generative mapping is always a {\em valid} (i.e., satisfying) assignment mapping  from {\em all subgoals} of the query $Q$ to the extended canonical database for $Q$.  We can show that for each CCQ query, its extended canonical database is unique up to isomorphism. 

For instance, an extended canonical database of the query $Q'$ of Example~\ref{intro-weird-ex} is $\{ r(a,b,c,d; 2) \}$. The query generates exactly one equivalence class ${\cal C}^{(Q')}_1$, due to the fact that the two subgoals of the query $Q'$ have the same relational template. We choose arbitrarily the atom $c^{(Q')}_1$ to be the first subgoal of the query $Q'$. The generative mapping in this case is $\{ $ $X_1$ $\rightarrow$ $a,$ $Y_1$ $\rightarrow$ $b,$ $Y_2$ $\rightarrow$ $c,$ $X_2$ $\rightarrow$ $d,$ $i$ $\rightarrow$ $1,$ $j$ $\rightarrow$ $1$   $\}$. 

{\bf General CCQ queries: Copy-neutral canonical data-} {\bf bases:} A database is called a {\em copy-neutral canonical database} for a CCQ query $Q$ if it can be obtained from an extended canonical database for $Q$ by changing, in an unrestricted way, the values of zero or more copy numbers in the latter database. That is, in a copy-neutral canonical database for $Q$, the copy number of each ground atom can (but does not have to) (i) coincide with the copy number of another ground atom in the database, and (ii) coincide with an element of the active domain of the database (in case the active domain includes natural numbers).
Clearly, each CCQ query can be associated with an infinite number (up to isomorphism) of copy-neutral canonical databases. 

\reminder{Must also define ``extended canonical database'' - for a CCQ query (define it on the *regularized* version of the query!); just extends the standard canonical db to ground atoms with copy numbers, as images of copy-sensitive subgoals of the query. Must say that if $Q$ has two or more copy-sensitive subgoals that differ on the value of the copy variable {\em only,} then all these subgoals get mapped, in the extended canonical database, into a {\em single} ground atom. Must also define $adom$ (must call it ``active domain'' -- of extended canonical db!!! [this notion is used in Observation I-2(iv))] NB! must contrast *extended* canonical db with ``standard'' canonical db of past work
}


\subsubsection{The Queries}
\label{basic-queries-sec}

\paragraph{The basics}  
For the fixed (input) CCQ queries $Q$ and $Q'$, as well as for the CCQ query $Q''$ that we use in this proof, we use the notation $Q(\bar{X}) \leftarrow L, M$;  $Q'(\bar{X}') \leftarrow L', M'$;  and $Q''(\bar{X}'') \leftarrow L'', M''$. Denote by $l \geq 0$ the head arity of 
$Q$.  Denote by $P$ (by $P'$, respectively) the (possibly empty) set of all constants in the query $Q$ (in the query $Q'$, respectively). 

\begin{proposition}
Let $Q \equiv_C Q'$. Then $P = P'$. 
\end{proposition}

\begin{proof}
The proof is by contradiction: Assume that $P$ and $P'$ are not the same sets. Pick a constant $c$ in $P - P'$. (If $P - P' = \emptyset$ then pick a constant $c \in P' - P$ and modify the rest of this proof by swapping $Q$ with $Q'$ in all the statements of the proof.) Construct the canonical database $D^{(Q')}$ of $Q'$ in such a way that $adom(D^{(Q')})$ does not have the value $c$. (It is always possible to construct a $D^{(Q')}$  that would satisfy this restriction.) Then it is easy to see that (i) the bag ${\sc Res}_C(Q',D^{(Q')})$ must be nonempty by construction of the database, and (ii) the bag ${\sc Res}_C(Q,D^{(Q')})$ must be empty due to the absence of the constant $c$ in  $adom(D^{(Q')})$. Hence we arrive at a contradiction with the assumption that $Q \equiv_C Q'$. 
\end{proof}

We use the notation $M_{copy}$ ($M'_{copy}$, $M''_{copy}$, respectively) for the set of copy variables of query $Q$ (of $Q'$, of $Q''$, respectively). We use the notation $M_{noncopy}$ ($M'_{noncopy}$, $M''_{noncopy}$, respectively) for the set of multiset noncopy variables of query $Q$ (of $Q'$, of $Q''$, respectively). 

We denote by $m$ the number $|M_{noncopy}|$ of multiset noncopy variables in the query $Q$, and by $r$ the number $|M_{copy}|$ of copy variables in the query $Q$. For  the CCQ query $Q''$ in this proof, we assume that (i) Queries $Q$ and $Q''$ have the same head arity ($l$); (ii) $|M_{noncopy}|$ $=$ $|M''_{noncopy}|$; and (iii) $|M_{copy}|$ $=$ $|M''_{copy}|$. 

The following observation is immediate from the assumption $Q \equiv_C Q'$ (of Theorem~\ref{magic-mapping-prop}) and from Theorem~\ref{not-same-num-multiset-vars-thm}. 

\begin{proposition}
(i) Queries $Q$ and $Q'$ have the same head arity ($l$); (ii) $|M_{noncopy}|$ $=$ $|M'_{noncopy}|$; and (iii) $|M_{copy}|$ $=$ $|M'_{copy}|$. 
\end{proposition} 

The following result is immediate from the containment-mapping theorem of \cite{ChandraM77}. Thus, for the remainder of the proof of Theorem~\ref{magic-mapping-prop}, {\em we assume that the set $M$ of multiset variables of $Q$ is not empty (that is, $m + r \geq 1$).}

\begin{proposition}
Let $M = \emptyset$. Then Theorem~\ref{magic-mapping-prop} holds. 
\end{proposition}

Throughout the proof of Theorem~\ref{magic-mapping-prop} we assume that we are given the regularized version of the query $Q$. 

\paragraph{Equivalence classes of subgoals of query $Q$}
We now introduce the notation that will help us to deal cleanly with the case where query $Q$ has more than one copy-sensitive subgoal for a particular relational template. (For an illustration, see query $Q'$ in Example~\ref{intro-weird-ex}.)  \reminder{Am I keeping that example?}) 


We first partition all the copy-sensitive subgoals (in case $r \geq 1$) of the query $Q$ into equivalence classes: Place two copy-sensitive subgoals  of $Q$ into the same equivalence class if and only if the two subgoals agree on the predicate name and on all the arguments except the copy variable. 
That is, two distinct copy-sensitive subgoals $g_1$ and $g_2$ of $Q$ are in the same equivalence class if and only if the relational templates of $g_1$ and $g_2$ are the same.  Denote by $C'_1,\ldots,C'_w$, $w \geq 1$,  the resulting equivalence classes for all the copy-sensitive subgoals of $Q$. (We have $w \geq 1$ only in case where $r = |M_{copy}| > 0$. Otherwise we set $w := 0$.) 

Further, we assume that each (if any) relational subgoal, $g$, of the query $Q$ is in its own equivalence class $\{ g \}$. 
Let $C_1,\ldots,C_v$, $v \geq 0$, be the resulting equivalence classes of the relational subgoals of the query $Q$. (The case $v = 0$ holds if and only if all subgoals of the query $Q$ are copy-sensitive, rather than relational, atoms.) 

Denote by ${\cal C}_Q = \{ C_1,\ldots,C_v,C'_1,\ldots,C'_w \}$, with $v + w \geq 1$, the set of all equivalence classes of the subgoals of the query $Q$, as defined in the preceding paragraphs. (By Definition~\ref{ccq-def}, the condition of the query $Q$ contains at least one atom, thus $v + w \geq 1$ must hold.) Further, for each class $C \in {\cal C}_Q$, choose one arbitrary element of $C$, call this element $s(C)$, and fix $s(C)$ as {\em the representative element of the class} $C$. 
Denote by  ${\cal S}_{C(Q)} = \{ s(C_1),\ldots,s(C_v),s(C'_1),\ldots,s(C'_w) \}$, with $v + w \geq 1$, the set of the representative-element subgoals of the query $Q$. The following observation is immediate from the definitions and from the fact that we use the regularized version of the query $Q$. (Recall that $L$ is the condition of the query $Q$, and that $r = |M_{copy}|$. Item (v) is immediate from item (iv) and from our assumption $m + r = |M| \geq 1$.)  

\begin{proposition} 
\label{sc-prop}
(i) ${\cal S}_{C(Q)} \subseteq L$. (ii) For an arbitrary (relational or copy-sensitive) atom $g$, the set   ${\cal S}_{C(Q)}$ has at most one element whose relational template is the same as the relational template of $g$. (iii) $r \geq w$ always holds. (iv) If $r > 0$ then $w > 0$. (v) $m + w \geq 1$. 
\end{proposition}

\paragraph{Notation for query variables} 
For ease of exposition, we assume that in case where $l$ $\geq$ $1$, each of $Q$, $Q'$, and $Q''$ has $l$ distinct head variables. (Recall that $l \geq 0$ denotes the head arity of the query $Q$. Handling the cases where ($l$ $\geq$ $1$ and) $Q$, $Q'$, or $Q''$ has either repeated occurrences of the same variable in the head, or has constants in the head, would be straightforward extensions of this proof but would make the exposition considerably harder to follow.) Suppose that $Q$ has $u \geq 0$ nonhead {\em set} variables.  W.l.o.g. denote (in case $l$ $\geq$ $1$) by $X_1,\ldots,X_l$ all the head variables of $Q$ (from left to right in the head vector $\bar{X}$ of $Q$), and (in case $u \geq 1$) denote by $X_{l+1},\ldots,X_{l+u}$ all the $u$ nonhead set variables of $Q$. 

In case $m = |M_{noncopy}| \geq 1$, let $Y_1,\ldots,Y_m$ be w.l.o.g.  all the multiset {\em noncopy} variables of the query $Q$. 
Further, in case $w \geq 1$ 
let $Y_{m+1},\ldots,Y_{m+w}$ be the copy variables of the elements $s(C'_1),\ldots,s(C'_w)$ of the set ${\cal S}_{C(Q)}$. 
By Proposition~\ref{sc-prop} (v), the set of variables $\{ Y_1,\ldots,Y_{m},Y_{m+1},\ldots,Y_{m+w} \}$ is not empty. 

Further, consider the set of all $r$ copy-sensitive subgoals of $Q$. 
Whenever $r > w$ --- that is, in case there exists a copy-sensitive subgoal of $Q$ that is not the representative element of any of the classes $C'_1,\ldots,C'_w$ --- let $Y_{m+w+1},\ldots,Y_{m+r}$ be (w.l.o.g.) the copy variables 
of all the copy-sensitive subgoals of $Q$ that are not the representative elements of any of the classes $C'_1,\ldots,C'_w$. 

In case $m \geq 1$ we denote by $Y'_1,$ $\ldots,$ $Y'_m$ (by  $Y''_1,$ $\ldots,$ $Y''_m$, respectively) the multiset noncopy variables of query $Q'$ (of query $Q''$, respectively). 
In case $r \geq 1$ we denote by $Y'_{m+1},$ $\ldots,$ $Y'_{m+r}$ (by  $Y''_{m+1},$ $\ldots,$ $Y''_{m+r}$, respectively) the copy variables of query $Q'$ (of query $Q''$, respectively).  
Finally, in case $l$ $\geq$ $1$ we denote by $X'_{1},$ $\ldots,$ $X'_{l}$ (by  $X''_{1},$ $\ldots,$ $X''_{l}$, respectively) the (distinct) head variables of query $Q'$ from left to right in the vector $\bar{X}'$ (of query $Q''$  from left to right in the vector $\bar{X}''$, respectively). All of this notation is w.l.o.g. by the basic results and assumptions about $Q'$ and $Q''$ in the beginning of this section.

\subsubsection{Convention for Ground Atoms in Databases}
\label{ground-atoms-convent-sec} 

We use the following convention for ground atoms in databases in this proof. Let $g = p(\bar{Y})$, for some choice of $p$ and $\bar{Y}$, be a ground atom in database $D$, and let $n \geq 1$ be the total number of copies of $g$ in $D$. 
Then we treat the $n$ copies of $g$ in $D$ as a {\em single} ground atom $p(\bar{Y})$ with associated copy number $n$, and represent it as $p(\bar{Y}; n)$, 
as defined in Section~\ref{query-syntax-sec}. 
As a result, every database is a set when using this representation.

\subsection{Constructing Family of Databases ${\cal D}_{\bar{N}}(Q)$}
\label{db-constr-sec}

This section describes the construction of an infinite family of databases based on the input query $Q$. Each database in the family is constructed as a union of copy-neutral canonical databases for the query $Q$.  
In the exposition in this section we use the notation introduced in Section~\ref{basics-sec}. 

Throughout the proof of Theorem~\ref{magic-mapping-prop}, whenever we discuss or use ``the family of databases as constructed in Section~\ref{db-constr-sec}'', we always refer to the family of databases built based on the fixed input query $Q$ (see Section~\ref{basic-queries-sec}),  {\em regardless of the context.} (This convention is lifted in part of Example~\ref{writeup-weird-ex} in Section~\ref{beyond-easy-case-sec}.)

\subsubsection{Mappings $\nu_0$ and $\nu^{(i)}_Q$, vector $\bar{N}$ and its domain $\cal N$, sets $S_0$ and $S_1^{(i)},\ldots,S_{m}^{(i)}$, tuple $t^*_Q$} 
\label{nu-sec}

We begin by defining a bijective mapping called $\nu_0$. The domain of $\nu_0$ is the set of all terms of the query $Q$ that are not multiset variables of the query. The range of $\nu_0$ is a subset of the active domain of each database in the family ${\cal D}_{\bar{N}}(Q)$. {\em The assignment $\nu_0$ is fixed to be the same across all the databases in the family ${\cal D}_{\bar{N}}(Q)$.} 

We define $\nu_0$ as follows: To each head variable and each set variable of $Q$, that is to each of the variables $X_1,\ldots,X_{l+u}$ ($l + u \geq 0$), $\nu_0$ assigns a distinct constant value that is also distinct from all the values in the set $P$ of constants used in the query $Q$. Denote by $S_0$ the set of $l+u$ constant values in the range of $\nu_0$, $S_0 \bigcap P = \emptyset$. Further, to each (if any) constant $c$ used in $Q$, $c \in P$, $\nu_0(c) := c$. The mapping $\nu_0$ is bijective by construction. 

We define tuple $t^*_Q$ as $t^*_Q = \nu_0[\bar{X}]$ in case $l$ $\geq$ $1$, and as the empty tuple in case $l$ $=$ $0$. (Recall that $\bar{X}$ is the vector of head terms of the query $Q$.) By definition, tuple $t^*_Q$ is fixed to be the same across all the databases in the family ${\cal D}_{\bar{N}}(Q)$. 

Define $\bar{N}$ as a vector of variables $N_1, N_2, \ldots, N_{m+w}$, $m+w \geq 1$. (For $m$ and $w$, see Section~\ref{basics-sec}.) 
We assume that the vector $\bar{N}$ accepts values from the Cartesian product of $m+w$ copies of the set ${\mathbb N}_+$ of natural numbers; we denote by $\cal N$ this domain from which $\bar{N}$ accepts values. Fix an arbitrary enumeration (starting with $1$) of the set $\cal N$. By the vector ${\bar{N}}^{(i)} \in {\cal N}$, $i \in {\mathbb N}_+$, we denote the $i$th element of $\cal N$ in this enumeration. We use the notation $N_1^{(i)}, N_2^{(i)}, \ldots, N_{m+w}^{(i)}$ for the natural-number values, from left to right, in the vector ${\bar{N}}^{(i)}$. That is, the natural number $N_j^{(i)}$, $1 \leq j \leq m+w$, in the $j$th position of vector ${\bar{N}}^{(i)}$, is the value, w.r.t. the fixed $i$, of the variable $N_j$ in the vector $\bar{N}$.  

Suppose that the query $Q$ is such that we have $m = |M_{noncopy}| \geq 1$. Fix an $i \in {\mathbb N}_+$ and consider the vector ${\bar N}^{(i)}$. For each value $N^{(i)}_j$ such that $1 \leq j \leq m$, define $S_j^{(i)}$ as a set of  $N^{(i)}_j$ distinct constants such that all the sets $S_j^{(i)}$ ($1 \leq j \leq m$) are pairwise disjoint and such that $S_j^{(i)} \bigcap P = \emptyset$ and $S_j^{(i)} \bigcap S_0 = \emptyset$ for all $j \in \{ 1,\ldots,m \}$. 

We define the set $S^{(i)}_*$ as the empty set in case $m = 0$, and as the union $\bigcup_{j=1}^m S^{(i)}_j$ in case $m \geq 1$. 


Finally, for the vector $\bar{N}^{(i)}$, for each $i \in {\mathbb N}_+$, we define mapping  $\nu^{(i)}_Q$, to be used in Section~\ref{obs-one-sec} and in other parts of the proof of Theorem~\ref{magic-mapping-prop}. Define the domain of the mapping  $\nu^{(i)}_Q$ as the union $P \bigcup S_0 \bigcup S^{(i)}_*$. Define mapping $\nu^{(i)}_Q$ as follows: 
\begin{itemize}
	\item For each $c \in S_0 \bigcup P$, let $\nu^{(i)}_Q(c) := \nu_0^{-1}(c)$ . 
	\item In case $m \geq 1$: For each $j \in \{ 1,\ldots,m \}$ and for each $c \in S_j^{(i)}$, let $\nu^{(i)}_Q(c) := Y_j$ . 
\end{itemize}

\subsubsection{Main Construction Cycle for $D_{\bar{N}^{(i)}}(Q)$ $\in$  ${\cal D}_{\bar{N}}(Q)$}
\label{main-cycle-sec}

The family of databases ${\cal D}_{\bar{N}}(Q)$ that we are about to construct is the infinite set $\{ D_{\bar{N}^{(1)}}(Q),$ $D_{\bar{N}^{(2)}}(Q),$ $\ldots,$ $D_{\bar{N}^{(i)}}(Q),$ $\ldots \}$, where each database $D_{\bar{N}^{(i)}}(Q)$, $i \in {\mathbb N}_+$, is associated with the vector ${\bar{N}}^{(i)} \in {\cal N}$. We will refer to the family ${\cal D}_{\bar{N}}(Q)$ either as $\{ D_{\bar{N}^{(i)}}(Q) \ | \ i \in {\mathbb N}_+ \ \}$ or simply as $\{ D_{\bar{N}^{(i)}}(Q) \}$. 

Fix an $i \in {\mathbb N}_+$ and consider the vector ${\bar N}^{(i)}$. We first define the set ${\cal S}^{(i)}$, to be used in the main construction cycle for creating the database $D_{\bar{N}^{(i)}}(Q)$ for the vector ${\bar N}^{(i)}$. In case that we have $m = 0$, define the set ${\cal S}^{(i)}$ as  a singleton set consisting of a single empty tuple. With that (only) element of the set ${\cal S}^{(i)}$, we associate an {\em empty} assignment $\nu^{noncopy}_t$. Now consider the case where the query $Q$ is such that we have $m \geq 1$.  Let ${\cal S}^{(i)}$ be the cross product $S^{(i)}_{1} \times S^{(i)}_{2} \times \ldots \times S^{(i)}_{m}$. For each tuple $t$ in ${\cal S}^{(i)}$ in this case $m \geq 1$, we treat $t$ as an assignment $\nu^{noncopy}_t$ of values to the multiset noncopy variables (from left to right) $Y_{1},\ldots,Y_m$ of $Q$. 

In case $r \geq 1$ we define a mapping $\nu^{copy}$ on the set $Y_{m+1},\ldots,Y_{m+r}$ of copy variables of the query $Q$. (I) For each copy variable from among $Y_{m+1},\ldots,Y_{m+w}$ of $Q$, define $\nu^{copy}(Y_j) := N^{(i)}_j$ for each $j \in \{ m+1,\ldots,m+w \}$. (II) Whenever $r > w$, for each $j \in \{ m+w+1,\ldots,m+r \}$ we define $\nu^{copy}(Y_j) := \nu^{copy}(Y_{k(j)})$. Here,  by $Y_{k(j)}$, with $k(j) \in \{ m+1,\ldots,m+w \}$, we denote the copy variable of the representative element of the class $C'_{k(j)}$ $\in$ $\{ C'_1,\ldots,C'_w \}$ $\subseteq$ ${\cal C}_Q$, such that the subgoal of $Q$ having copy variable $Y_j$ belongs to that class $C'_{k(j)}$. 

It is convenient to also define here the mapping $\nu^{copy}_Q$, to be used in Section~\ref{valid-map-sec} and beyond, again for the case where  $r \geq 1$. The mapping $\nu^{copy}_Q$ uses ``the same logic'' as the mapping $\nu^{copy}$, except that $\nu^{copy}_Q$ maps each copy variable of $Q$ into the variable {\em name} from among $N_{m+1},\ldots,N_{m+w}$ in the vector $\bar{N}$, whereas $\nu^{copy}$ maps each copy variable of $Q$ into the variable {\em value} from among $N^{(i)}_{m+1},\ldots,N^{(i)}_{m+w}$ in the vector $\bar{N}^{(i)}$ for a fixed $i \in {\mathbb N}_+$. That is, (I) For each copy variable from among $Y_{m+1},\ldots,Y_{m+w}$ of $Q$, define $\nu^{copy}_Q(Y_j) := N_j$ for each $j \in \{ m+1,\ldots,m+w \}$. (II) Whenever $r > w$, for each $j \in \{ m+w+1,\ldots,m+r \}$ we define $\nu^{copy}_Q(Y_j) := \nu^{copy}_Q(Y_{k(j)})$. Here,  by $Y_{k(j)}$, $k(j) \in \{ m+1,\ldots,m+w \}$, we denote the copy variable of the representative element of the class $C'_{k(j)}$ $\in$ $\{ C'_1,\ldots,C'_w \}$ $\subseteq$ ${\cal C}_Q$, such that the subgoal of $Q$ having copy variable $Y_j$ belongs to that class $C'_{k(j)}$. 

We now associate with each tuple $t \in {\cal S}^{(i)}$ exactly one assignment $\nu_t$ of constants to all the variables and constants of the query $Q$, such that (i) $\nu_t[X_j]$ coincides with $\nu_0[X_j]$ for all $j \in \{ 1,\ldots,l+u \}$; (ii) $\nu_t[Y_j]$ (in case $m \geq 1$ only) coincides with $\nu^{noncopy}_t[Y_j]$ for all $j \in \{ 1,\ldots,m \}$; (iii) for each copy variable (in case $r \geq 1$ only) from among $Y_{m+1},\ldots,Y_{m+r}$ of $Q$, define $\nu_t(Y_j) := \nu^{copy}(Y_j)$ for each $j \in \{ m+1,\ldots,m+r \}$; and (iv) for each constant $c$ used in $Q$, $\nu_t(c) := \nu_0(c)$. 

We now construct all the ground atoms in the database $D_{\bar{N}^{(i)}}(Q)$. 
In the construction, 
we use only the elements of the subset ${\cal S}_{C(Q)}$ of the condition of the query $Q$. 

{\em Main construction cycle:} For each tuple $t$ in ${\cal S}^{(i)}$, we add to the database $D_{\bar{N}^{(i)}}(Q)$ the following ground atoms: 
\begin{itemize}
	\item For each relational subgoal of 
	$Q$ of the form $b(\bar{R})$, where $b$ is a predicate name and $\bar{R}$ is a vector of variables from among $X_1,\ldots,X_{l+u}$ and (in case $m \geq 1$) from among $Y_{1},\ldots,Y_m$, and of constants from $P$. Add to the database $D_{\bar{N}^{(i)}}(Q)$ the ground atom $b(\nu_t(\bar{R}); 1)$ (if not already there). Note that the copy number of this ground atom equals one. 
	\item In case $w \geq 1$: For each copy-sensitive atom from among $s(C'_1),\ldots,s(C'_w)$: Suppose the copy-sensitive atom in question is of the form $b(\bar{R};z)$, where $b$ is a predicate name, $\bar{R}$ is 
a vector of variables from among $X_1,\ldots,X_{l+u}$ and (in case $m \geq 1$) from among $Y_{1},\ldots,Y_m$, and of constants from $P$, and $z$ is a copy variable from among $Y_{m+1},\ldots,Y_{m+w}$. Add to the database $D_{\bar{N}^{(i)}}(Q)$ the ground atom $b(\nu_t(\bar{R}); \nu_t(z))$ (if not already there). 
\end{itemize}
Finally, $D_{\bar{N}^{(i)}}(Q)$ has no other ground atoms than those added for the tuples $t$ in ${\cal S}^{(i)}$ as described above. 


\subsubsection{Properties of Databases $D_{\bar{N}^{(i)}}(Q)$} 
\label{obs-one-sec}

We now state some properties of the databases in the set $\{ D_{\bar{N}^{(i)}}(Q) \}$ and of the answer to the query $Q$ on each database. 
The first four results are immediate from the constructions and definitions in this section and in Section~\ref{basics-sec}. 

\begin{proposition} 
Let $i \in {\mathbb N}_+$. (i) The set $adom(D_{\bar{N}^{(i)}}(Q))$ is the union $P \bigcup S_0 \bigcup S^{(i)}_*$. (ii) The mapping $\nu^{(i)}_Q$ is defined on all elements of $adom(D_{\bar{N}^{(i)}}(Q))$. (iii)  For each subset $T \neq \emptyset$ of $adom(D_{\bar{N}^{(i)}}(Q))$ such that (in case $m \geq 1$) the cardinality of $T \bigcap S_j^{(i)}$ does not exceed $1$ for each $j \in \{ 1,\ldots,m \}$,  the mapping $\nu^{(i)}_Q$ is a  bijection from its domain $T$ to the set $\nu^{(i)}_Q(T)$. 
\end{proposition}

\begin{proposition}
\label{Observation_one_one} 
Let $i \in {\mathbb N}_+$. For the database $D_{\bar{N}^{(i)}}(Q)$, we have that: 
\begin{itemize}
	\item[(i)] $D_{\bar{N}^{(i)}}(Q) \neq \emptyset$. 
	\item[(ii)] For each ground atom $h \in D_{\bar{N}^{(i)}}(Q)$, of the form $p(\bar{W}; n)$, for some predicate $p$ and with copy number $n \geq 1$, there exists exactly one element, call it $g$, of the set ${\cal S}_{C(Q)}$ of (representative-element) subgoals of query $Q$, where the relational template of  $g$ is  $p(\bar{Z})$, and such that 
%
		(a) the result of applying the mapping $\nu^{(i)}_Q$ to the vector $\bar{W}$ in $h$ is the vector $\bar{Z}$ in $g$, and (b) $\nu^{(i)}_Q$ is a bijection from $\bar{W}$ to $\bar{Z}$. 
%
\end{itemize}
\vspace{-0.5cm}
\end{proposition} 

For the next result, we introduce some terminology. Let $i \in {\mathbb N}_+$. For a ground atom in $h \in D_{\bar{N}^{(i)}}(Q)$ and for the corresponding atom $g$ $\in$ ${\cal S}_{C(Q)}$ as defined by Proposition~\ref{Observation_one_one} (ii), we call $g$ {\em the $\nu^{(i)}_Q$-image of} $h$. Further, by $adom(h)$ we denote all those terms of ground atom  $h \in D_{\bar{N}^{(i)}}(Q)$ that are elements of $adom(D_{\bar{N}^{(i)}}(Q))$. (That is, for atom $h$ of the form $p(\bar{W}; n)$, for some predicate $p$ and with copy number $n \geq 1$, $adom(h)$ is the set of all elements of the vector $\bar{W}$.) Finally, in case $m \geq 1$, for a $j \in \{ 1,\ldots,m \}$, we denote by $S^{(i)}_{h,j}$ the intersection of the set $adom(h)$ with the set $S^{(i)}_j$. 

The following result holds by construction of the family of databases $\{ D_{\bar{N}^{(i)}}(Q) \}$. 

\begin{proposition} 
\label{all-variety-prop} 
Suppose $m \geq 1$. Let $i \in {\mathbb N}_+$. Let $h$ be an arbitrary ground atom in $D_{\bar{N}^{(i)}}(Q)$, of the form $h = p(\bar{W}; n)$, for some predicate $p$ and with copy number $n \geq 1$. Then we have that: 
\begin{itemize} 
	\item[(i)] For each  $j \in \{ 1,\ldots,m \}$, the set $S^{(i)}_{h,j}$ is either the empty set or a singleton set. 
	\item[(ii)] In case the ground atom $h$ is such that for at least one $j \in \{ 1,\ldots,m \}$, the set $S^{(i)}_{h,j}$ is not the empty set: Let $\mu$ be an arbitrary mapping from all elements of the vector $\bar{W}$ to $adom(D_{\bar{N}^{(i)}}(Q))$, such that (a) $\mu$ is the identity mapping on each element of $\bar{W}$ that does not belong to the set $S^{(i)}_*$, and such that (b) for each element $e$ of $\bar{W}$ such that $e$ belongs to a $S^{(i)}_{h,j}$ for some  $j \in \{ 1,\ldots,m \}$, we have that $\mu(e)$ is an element of   $S^{(i)}_j$ for {\em the same} $j$. Then it holds that: 
	\begin{itemize} 
		\item[(ii-a)] The ground atom $h' = p(\mu(\bar{W}); n)$ is an element of the set  $D_{\bar{N}^{(i)}}(Q)$; and 
		\item[(ii-b)] The $\nu^{(i)}_Q$-image of $h$ and the $\nu^{(i)}_Q$-image of $h'$ are the same element of the set ${\cal S}_{C(Q)}$. 
	\end{itemize} 
\end{itemize} 
\end{proposition} 


\begin{proposition}
\label{Observation_one_two} 
Let $i \in {\mathbb N}_+$. For the database $D_{\bar{N}^{(i)}}(Q)$, we have that: 
\begin{itemize}
	\item[(i)] In the construction of database $D_{\bar{N}^{(i)}}(Q)$, each main construction cycle (for some fixed tuple $t$ in ${\cal S}^{(i)}$) generates in the database $D_{\bar{N}^{(i)}}(Q)$ a copy-neutral canonical database for the query $Q$, call this database $D_t$. For each $t \in {\cal S}^{(i)}$, the mapping $\nu^{(i)}_Q$ induces an isomorphism from $D_t$ to the set ${\cal S}_{C(Q)}$ of  (representative-element)  subgoals of the query $Q$. 
	\item[(ii)] 
	There exists in $D_{\bar{N}^{(i)}}(Q)$ at least one copy-neutral canonical database for query $Q$. 
	\item[(iii)] For each pair $(D_1,D_2)$ of copy-neutral canonical databases for query $Q$ within database $D_{\bar{N}^{(i)}}(Q)$, 
such that each of $D_1$ and $D_2$ was generated using the main construction cycle (using some  $t_1 \in {\cal S}^{(i)}$ to construct $D_1$, and using some  $t_2 \in {\cal S}^{(i)}$ to construct $D_2$), the only difference  (if any) between the values of variables of $Q$, in $D_1$ and $D_2$, is in the values of multiset noncopy variables of $Q$. 
	\item[(iv)] Let $T \neq \emptyset$ be an arbitrary subset of $adom(D_{\bar{N}^{(i)}}(Q))$ such that (a) $P \bigcup S_0$ is a subset of $T$, and (b) in case $m \geq 1$, the cardinality of $T \bigcap S_j^{(i)}$ is exactly $1$ for each $j \in \{ 1,\ldots,m \}$. Then there exists in  $D_{\bar{N}^{(i)}}(Q)$ a copy-neutral canonical database of the query $Q$ such that the active domain of that copy-neutral canonical database is exactly $T$. 
\end{itemize}
\vspace{-0.5cm}
\end{proposition} 

For an $i \in {\mathbb N}_+$, let $\cal T$ be a nonempty set of ground atoms of database $D_{\bar{N}^{(i)}}(Q)$. We denote by $adom({\cal T})$ the set of all values of $adom(D_{\bar{N}^{(i)}}(Q))$ that are used in the atoms of $\cal T$. 

\begin{proposition}
\label{Observation_one_three} 
Let $i \in {\mathbb N}_+$. Let $\cal T$ be a nonempty set of ground atoms of database $D_{\bar{N}^{(i)}}(Q)$. 
Suppose the set $\cal T$ is such that (in case $m = |M_{noncopy}| \geq 1$) for each $j \in \{ 1,\ldots,m \}$,  the cardinality of the intersection of $adom({\cal T})$ with $S_j^{(i)}$ is at most one.\footnote{In case $m = 0$, we require only that ${\cal T} \neq \emptyset$.}  
Then there exists in $D_{\bar{N}^{(i)}}(Q)$ a copy-neutral canonical database for $Q$, call this database $D$, such that ${\cal T} \subseteq D$ (as set of ground atoms {\em with copy numbers}). 
\end{proposition}


\begin{proof}
For each ground atom $h$ in $\cal T$, we label $h$ with the subgoal $s$ of $Q$ such that $s$ ``has  generated'' $h$ in the construction of database  $D_{\bar{N}^{(i)}}(Q)$. (See Proposition~\ref{Observation_one_one} (ii) for the details.) We then partition all the atoms of $\cal T$ into equivalence classes, call them ${\cal E}_1,\ldots,{\cal E}_n$, $n \geq 1$, using the labels. (That is, two atoms of $\cal T$ belong to the same equivalence class if and only if they have the same label.) Now we show that for each possible label (w.r.t. query $Q$), the equivalence class for $\cal T$ and for this label has at most one element. Indeed, recall that for each $j \in \{ 1,\ldots,m \}$,  the cardinality of the intersection of $adom({\cal T})$ with $S_j^{(i)}$ is at most one. By construction of the database $D_{\bar{N}^{(i)}}(Q)$ and of the equivalence classes for $\cal T$, we obtain the desired result. 

Now we use the classes ${\cal E}_1,\ldots,{\cal E}_n$ to construct from $\cal T$ some of the subgoals of $Q$. By the properties of $adom({\cal T})$ and of  the mapping $\nu^{(i)}_Q$, as well as from the fact that each of ${\cal E}_1,\ldots,{\cal E}_n$ is a singleton set, we obtain that the mapping $\nu^{(i)}_Q$ induces an isomorphism from the set $\cal T$ to a nonempty subset, call it $\cal S$, of the set ${\cal S}_{C(Q)}$ of   (representative-element) subgoals of the query $Q$. Denote by $adom({\cal S})$ the set of all terms of the query $Q$ that are not copy variables of $Q$, such that these terms are used in the atoms of $\cal S$. By construction, the mapping $(\nu^{(i)}_Q)^{-1}$ is a bijection from $adom({\cal S})$ to $adom({\cal T})$. By definition of the main construction cycle in constructing the database $D_{\bar{N}^{(i)}}(Q)$, at least one iteration of the main construction cycle in the construction has used (some extension of) the mapping $(\nu^{(i)}_Q)^{-1}$, to generate in $D_{\bar{N}^{(i)}}(Q)$ a copy-neutral canonical database for $Q$. (See Proposition~\ref{Observation_one_two} for the justifications.) It follows that that iteration of the main construction cycle mapped the set ${\cal S}_{C(Q)}$ of  (representative-element)  subgoals of the query $Q$ to a superset of the set $\cal T$.  
%
Q.E.D. 
\end{proof}



We now establish that the fixed tuple $t^*_Q = \nu_0[\bar{X}]$ (where $\bar{X}$ is the vector of head terms of the query $Q$, whenever $l$ $\geq$ $1$) is present in ${\sc Res}_C(Q,D_{\bar{N}^{(i)}}(Q))$, for each $i \geq 1$. (Recall that in case $l$ $=$ $0$, we have that $t^*_Q$ is the empty tuple. For ease of exposition, in the remainder of this proof of Theorem~\ref{magic-mapping-prop}, we omit the discussion of the case $l$ $=$ $0$ and of its implications for the definition of the tuple $t^*_Q$.)  

\begin{proposition}
\label{Observation_one_four} 
Let $i \in {\mathbb N}_+$. Then the bag  ${\sc Res}_C$ $(Q,$ $D_{\bar{N}^{(i)}}(Q))$ has at least one copy of the tuple $t^*_Q$. 
\end{proposition}

\begin{proof}{(sketch)}
This result holds by construction of the databases in the family $\{ D_{\bar{N}^{(i)}}(Q) \}$. 
Indeed, recall that the assignment mapping $\nu_0$ is fixed to be the same across all the databases in $\{ D_{\bar{N}^{(i)}}(Q) \}$. Choose an arbitrary tuple $t$ in the set ${\cal S}^{(i)}$ for $D_{\bar{N}^{(i)}}(Q)$; 
consider the mapping $\nu_t$ associated with $t$. 
It is easy to verify that $\nu_t$ is an assignment mapping from the query $Q$ to the database $D_{\bar{N}^{(i)}}(Q)$, such that $\nu_t$ contributes the tuple $t^*_Q = \nu_0[\bar{X}]$ to the bag ${\sc Res}_C(Q,D_{\bar{N}^{(i)}}(Q))$.  
\end{proof}


\subsection{$t^*_Q$-Valid Assignment Mappings for Query $Q''$ and Databases in $\{ D_{\bar{N}^{(i)}}(Q) \}$} 
\label{valid-map-sec}

 Consider the family $\{ D_{\bar{N}^{(i)}}(Q) \}$ of databases constructed as described in Section~\ref{db-constr-sec}. 
 Let $Q''$ be an arbitrary CCQ query that  satisfies all the requirements (for $Q''$ specifically) of Section~\ref{basic-queries-sec}. In this section, for the query $Q''$ and for an arbitrary $i \in {\mathbb N}_+$,  we provide an algorithm for generating the set of all assignment mappings from $Q''$ to the database $D_{\bar{N}^{(i)}}(Q)$ such that each of the assignments ``generates''\footnote{That is, each of these assignments generates a separate element of the set $\Gamma^{(t^*_Q)}(Q'',D_{\bar{N}^{(i)}}(Q))$.} the fixed tuple $t^*_Q$ (as defined in Section~\ref{db-constr-sec}) in the bag ${\sc Res}_C(Q'',D_{\bar{N}^{(i)}}(Q))$. We call these assignment mappings ``$t^*_Q$-valid assignment mappings''. We also formulate some useful properties of the $t^*_Q$-valid  assignment mappings, for $Q''$ and for each $i \in {\mathbb N}_+$. 
 
The notions introduced in this and previous sections are illustrated by a detailed Example~\ref{main-proof-ex} in Section~\ref{main-proof-ex-sec}.

\subsubsection{Definitions and Construction} 
\label{assoc-defs-sec}

\paragraph{The basics}
\label{assoc-defs-basics-sec}

Denote by $G > 0$ the number of subgoals, call them $g_1,\ldots,g_G$, of the query $Q''$. For convenience, we choose the ordering  $g_1,\ldots,g_G$ of the subgoals of $Q''$ in such a way that all the copy-sensitive atoms, if any, in the ordering precede all the relational atoms, if any, in the ordering, and such that, in case $r \geq 1$, each $j$th atom $g_j$ in the ordering has copy variable $Y''_{m+j}$ of the query $Q''$, $1 \leq j \leq r$. (See Section~\ref{basic-queries-sec} for the notation $Y''_j$.)   
Fix an arbitrary $i \in {\mathbb N}_+$, and denote by $F > 0$ the number of ground atoms in the database $D_{\bar{N}^{(i)}}(Q)$. 

We find all the $t^*_Q$-valid assignment mappings for $Q''$ and $D_{\bar{N}^{(i)}}(Q)$ (i.e., those satisfying assignments, in the terminology of Section~\ref{comb-sem-sec}, that are in the set $\Gamma^{(t^*_Q)}$ $(Q'',$ $D_{\bar{N}^{(i)}}(Q)))$  by enumerating {\em all} associations between the $G$ subgoals of $Q''$ and the $F$ ground atoms in $D_{\bar{N}^{(i)}}(Q)$. 
The definition of ``valid assignment mapping'' is ``as expected''. 
However, since we use associations between atoms to determine which assignment mappings are valid, we provide here the required formal definitions. 

\begin{definition}{Association for $Q$ and $D$} 
\label{assoc-def}
Given a CCQ query $Q$ with $G > 0$ subgoals and a nonempty database $D$, a set of $G$ pairs $\{ (g_1,d_{j1}),$ $(g_2,d_{j2}),$ $\ldots,$ $(g_G,d_{jG}) \}$ between all the $G$ subgoals of $Q$ and some (not necessarily distinct) $G$ ground atoms of $D$ is called an {\em association for $Q$ and $D$.} 
\end{definition} 

\begin{definition}{Candidate assignment mapping}
\label{cand-mpng-def}
Given a CCQ query $Q$ with $G > 0$ subgoals, a nonempty database $D$, and an association ${\cal A} = \{ (g_1,d_{j1}),$ $(g_2,d_{j2}),$ $\ldots,$ $(g_G,d_{jG}) \}$ for $Q$ and $D$. Define a {\em candidate assignment mapping} $\theta$ {\em for $Q$ and $D$ w.r.t. $\cal A$}  as follows:
\begin{itemize}
	\item[(a)] $\theta$ is  the empty mapping in case there exists an integer $k$, $1 \leq k \leq G$, such that $g_k$ and $d_{jk}$ have different predicate names, and 
	
	\item[(b)]  $\theta$ is  the union of the $G$ associations of the terms of each $g_k$, $1 \leq k \leq G$, to the terms of its respective $d_{jk}$, where each association is a set of assignments of the values in $\bar{W}^{(k)}$, from left to right, to the terms in $\bar{Z}^{(k)}$ in the same (from left to right) positions. Here, $\bar{Z}^{(k)}$ is the vector of all terms {\em including the copy variable} (if any) in $g_k$, and $\bar{W}^{(k)}$ is the vector of all arguments of $d_{jk}$. 
	
	We do not include the copy number of $d_{jk}$ {\em per se} as an element of  $\bar{W}^{(k)}$. Instead, we do the following. In case $g_k$ is a relational atom, the copy number of $d_{jk}$ is not used in the construction of $\theta$ for $g_k$. If $g_k$ is a copy-sensitive atom, then we consider the copy variable of $g_k$ to be the last element of vector $\bar{Z}^{(k)}$, and add to the vector $\bar{W}^{(k)}$, as its rightmost element, a natural-number value between (inclusively) $1$ and the copy number of $d_{jk}$. 
\end{itemize}
\vspace{-0.5cm}
\end{definition}


Note that whenever query $Q$ has copy variables {\em and}  database $D$ has ground atoms with nonunity copy numbers, for a single association $\cal A$ for $Q$ and $D$ there could be more than one (but always a finite number, on finite databases) candidate assignment mappings for $Q$ and $D$ w.r.t. $\cal A$. Definition~\ref{cand-mpng-def} provides a straightforward algorithm to generate {\em all} candidate assignment mappings for any CCQ query $Q$, finite database $D$, and association $\cal A$ for $Q$ and $D$. 

\begin{definition}{($t^*_Q$-)Valid assignment mapping}
\label{valid-mpng-def} 
Given a CCQ query $Q$, a nonempty database $D$, an association ${\cal A}$ for $Q$ and $D$, and a candidate assignment mapping $\theta$ for $Q$ and $D$ w.r.t. $\cal A$. Then $\theta$ is a {\em valid assignment mapping from the variables and constants of $Q$ to the values in $adom(D) \cup {\mathbb N}_+$  w.r.t.} $\cal A$ if: (i) $\theta$ is not an empty mapping; (ii) $\theta$ associates all occurrences of each constant of $Q$ with the same constant; and (iii) for each variable of $Q$, $\theta$ associates all occurrences of the variable with the same value in $adom(D) \cup {\mathbb N}_+$. A valid assignment mapping $\theta$ is a {\em $t^*_Q$-valid assignment mapping} for $Q$, $D$, and $\cal A$ if the restriction of $\theta$ to the head vector of the query $Q$ results in the tuple $t^*_Q$.  
\end{definition} 

For brevity, we will refer to each valid assignment mapping from the variables and constants of $Q$ to the values in $adom(D) \cup {\mathbb N}_+$  w.r.t. $\cal A$, for some $Q$, $D$, and $\cal A$, as a ``valid assignment mapping for $Q$, $D$, and $\cal A$''. In addition, we say that $\theta$ is a ``valid assignment mapping for $Q$ and $D$'' if there exists an association, $\cal A$, for $Q$ and $D$, such that $\theta $ is a valid assignment mapping for $Q$, $D$, and $\cal A$. 

By definition, each valid assignment mapping for $Q$ and $D$ is an element of the set $\Gamma(Q,D)$, and each $t^*_Q$-valid assignment mapping for $Q$ and $D$ is an element of the set $\Gamma^{(t^*_Q)}(Q,D)$. If $\theta$ is a  valid assignment mapping for query $Q$, database $D$, and association $\cal A$ (for $Q$ and $D$), then we say that $\cal A$ {\em generates the mapping} $\theta$, or that $\cal A$ {\em contributes the mapping} $\theta$ to the set $\Gamma(Q,D)$. 

\begin{definition}{Unity  ($t^*_Q$-)valid assignment mapping}
Given a CCQ query $Q$, a nonempty database $D$, an association ${\cal A}$ for $Q$ and $D$, and a ($t^*_Q$-)valid assignment mapping $\theta$ for $Q$, $D$, and $\cal A$. Then $\theta$ is a {\em unity ($t^*_Q$-)valid assignment mapping} for $Q$, $D$, and $\cal A$ if $\theta$ maps each (if any) copy variable of the query $Q$ into value $1$.  
\end{definition}

\begin{example}
\label{copy-ex}
Let  CCQ query $Q''$ be defined as 
\begin{tabbing}
$Q''(X) \leftarrow p(X,Y), \ p(X,X; i), \ \{ Y,i \} .$
\end{tabbing}

Let  $g_1$ be the rightmost subgoal of $Q''$ (that is, $p(X,X; i)$); and let $g_2$ be the leftmost subgoal of $Q''$ (that is, $p(X,Y)$); then $G$ for $Q''$ equals 2. 

For ease of exposition, assume that the database, call it $D$, is $D = \{ p(1,1; 3), p(1,2; 5), p(3,3; 7) \}$. We refer to the ground atom $p(1,1; 3)$ of $D$ as $d_1$, to the ground atom $p(1,2; 5)$ of $D$ as $d_2$, and to the ground atom $p(3,3; 7)$ of $D$ as $d_3$. 

Consider three (from the total of nine possible) associations between the subgoals of the query $Q''$ and the ground atoms of the database $D$:
\begin{itemize}
	\item ${\cal A}_1: \{ (g_1, d_2), (g_2, d_1) \}$; 
	\item ${\cal A}_2: \{ (g_1, d_1), (g_2, d_1)  \}$; and 
	\item ${\cal A}_3: \{ (g_1, d_3), (g_2, d_3) \}$. 
\end{itemize}

It is easy to see that each candidate assignment mapping associated with ${\cal A}_1$ is not a valid assignment mapping, as each such candidate mapping maps the two occurrences of $X$ in $g_1$ into distinct values (1 and 2) in $adom(D)$. 

At the same time: 
\begin{itemize}
	\item Given ${\cal A}_2$, $\theta_{21} = \{ X \rightarrow 1, Y \rightarrow 1, i \rightarrow 1 \}$ is a unity valid assignment mapping from the terms of $Q''$ to the elements of $adom(D) \cup {\mathbb N}_+$. 
	\item Given ${\cal A}_3$, $\theta_3 = \{ X \rightarrow 3, Y \rightarrow 3, i \rightarrow 7 \}$ is a valid assignment mapping from the terms of $Q''$ to the elements of $adom(D)  \cup {\mathbb N}_+$; it is not a {\em unity} valid assignment mapping. 
\end{itemize} 

Now observe that in addition to $\theta_{21}$, two more candidate mappings w.r.t. ${\cal A}_2$ would also be valid assignment mappings from the terms of $Q''$ to the elements of $adom(D)  \cup {\mathbb N}_+$. These mappings are $\theta_{22} = \{ X \rightarrow 1, Y \rightarrow 1, i \rightarrow 2 \}$ and $\theta_{23} = \{ X \rightarrow 1, Y \rightarrow 1, i \rightarrow 3 \}$. The only difference between $\theta_{21}$, $\theta_{22}$, and $\theta_{23}$ is in the value assigned to the copy variable $i$ of the query $Q''$. Unlike $\theta_{21}$, neither $\theta_{22}$ nor $\theta_{23}$ is a unity valid assignment mapping for $Q''$, $D$, and ${\cal A}_2$. No other candidate mappings are possible for $Q''$, $D$, and ${\cal A}_2$, because the copy number of $d_1$ in $D$ is three. 


Finally, suppose we 
are given tuple $t^*_Q = (1)$. 
Then $\theta_{21}$ is a 
(unity) 
$t^*_Q$-valid assignment mapping for $Q''$, $D$, and ${\cal A}_2$, because $\theta_{21}[X]$ coincides with $t^*_Q$. At the same time, while being a valid assignment mapping from $Q''$ to $D$, $\theta_{3}$ is {\em not} a 
$t^*_Q$-valid assignment mapping for $Q''$, $D$, and ${\cal A}_3$, 
because $\theta_{3}[X] = 3$ does not coincide with $t^*_Q$. 
\end{example} 

\paragraph{Signatures of associations}

Fix an $i \in {\mathbb N}_+$. Recall Proposition~\ref{Observation_one_one} (ii), which says that by construction of the database $D_{\bar{N}^{(i)}}(Q)$, each ground atom in the database can be ``mapped into'' a {\em unique} subgoal of the subset  ${\cal S}_{C(Q)}$ of the condition of the fixed input query $Q$, using the mapping $\nu^{(i)}_Q$ defined in Section~\ref{nu-sec}. 
%
We denote by $\psi^{gen(Q)}_{\bar{N}^{(i)}}$ this mapping, induced by the mapping $\nu^{(i)}_Q$, from the ground atoms of $D_{\bar{N}^{(i)}}(Q)$ to the elements of the set ${\cal S}_{C(Q)}$. 

\begin{definition}{Atom-signature of association}
\label{atom-sig-def} 
Let $i \in {\mathbb N}_+$, and let ${\cal A}  = \{ (g_1,d_{j1}),\ldots,(g_G,d_{jG}) \} $ be an association for query $Q''$, with $G \geq 1$ subgoals, and for the database $D_{\bar{N}^{(i)}}(Q)$. 
Then the $G$-ary vector $\Psi_a[{\cal A}]$ $=$ $[ \ \psi^{gen(Q)}_{\bar{N}^{(i)}}[d_{j1}],$ $\psi^{gen(Q)}_{\bar{N}^{(i)}}[d_{j2}],$ $\ldots,$ $\psi^{gen(Q)}_{\bar{N}^{(i)}}[d_{jG}] \ ]$ is the {\em atom-signature of} $\cal A$ for $Q''$ and $D_{\bar{N}^{(i)}}(Q)$. 
\end{definition} 

The intuition for Definition~\ref{atom-sig-def} is as follows. Suppose that for an association $\cal A$ for a query $Q''$ and for a database $D_{\bar{N}^{(i)}}(Q)$, the association $\cal A$ contributes at least one valid assignment mapping to the set $\Gamma(Q'',D_{\bar{N}^{(i)}}(Q))$. Then, intuitively, the atom-signature of $\cal A$ shows the pattern in which the condition of the query $Q''$ can be mapped into the (subset ${\cal S}_{C(Q)}$ of the) condition of the query $Q$. 


Recall the notation $Y''_1,\ldots,Y''_{m+r}$ of Section~\ref{basic-queries-sec} for the multiset variables of the query $Q''$. Here, $Y''_1,\ldots,Y''_{m}$ are all the multiset noncopy variables of $Q''$ (in case $m \geq 1$), and  $Y''_{m+1},\ldots,Y''_{m+r}$ are all the copy variables of $Q''$ (in case $r \geq 1$). 

Suppose that the query $Q$ is such that $r = |M_{copy}| \geq 1$. Recall the mapping $\nu^{copy}_Q$ defined in Section~\ref{main-cycle-sec}: $\nu^{copy}_Q$ maps each copy variable of the query $Q$ into one of the variables $N_{m+1},\ldots,N_{m+w}$ in the vector $\bar{N}$. 

We use $\nu^{copy}_Q$ to define the following mapping $\nu^{copy-sig}$ on all atoms in the set ${\cal S}_{C(Q)}$ for the query $Q$: 
\begin{itemize} 
	\item For each (if any) relational atom $h$ in ${\cal S}_{C(Q)}$, $\nu^{copy-sig}(h)$ $:=$ $1$. 
	\item For each (if any) copy-sensitive atom $h$ in ${\cal S}_{C(Q)}$, where $h$ has copy variable of the name $Z \in M_{copy}$, $\nu^{copy-sig}(h) := \nu^{copy}_Q(Z)$. 
\end{itemize} 

Recall (see beginning of Section~\ref{assoc-defs-sec}) that we have fixed an ordering of the subgoals $g_1,\ldots,g_G$ of the query $Q''$ in such a way that: 
\begin{itemize} 
	\item  all the copy-sensitive atoms in the ordering precede all the relational atoms, if any, in the ordering, and 
	\item each $j$th atom $g_j$ in the ordering has copy variable $Y''_{m+j}$ of the query $Q''$, $1 \leq j \leq r$ (in case $r \geq 1$). 
\end{itemize} 

\begin{definition}{Copy-signature of association}
\label{copy-sig-def} 
Let $i \in {\mathbb N}_+$, and let ${\cal A}$ $=$ $\{ (g_1,d_{j1}),$ $\ldots,$ $(g_G,d_{jG}) \} $ be an association for query $Q''$ and for the database $D_{\bar{N}^{(i)}}(Q)$.  
Then vector $\Phi_c[{\cal A}]$, called the {\em copy-signature of} $\cal A$ for $Q''$ and $D_{\bar{N}^{(i)}}(Q)$, is defined as follows: 
\begin{itemize}
	\item In case $r = 0$, $\Phi_c[{\cal A}]$ is the empty vector; and 
	\item In case $r \geq 1$, $\Phi_c[{\cal A}]$ is the $r$-ary vector  
$[ \ \nu^{copy-sig}$ $(\psi^{gen(Q)}_{\bar{N}^{(i)}}[d_{j1}]),$ $\ldots,$ $\nu^{copy-sig}$ $(\psi^{gen(Q)}_{\bar{N}^{(i)}}[d_{jr}]) \ ]$. 
\end{itemize}
\end{definition} 

\begin{definition}{Noncopy-signature of association}
\label{noncopy-sig-def} 
Let $i \in {\mathbb N}_+$, and let ${\cal A}$ $=$ $\{ (g_1,d_{j1}),$ $\ldots,$ $(g_G,d_{jG}) \} $ be an association for query $Q''$ and for the database $D_{\bar{N}^{(i)}}(Q)$, such that there exists a valid assignment mapping, $\theta$, for $Q''$, $D_{\bar{N}^{(i)}}(Q)$, and $\cal A$.  
Then vector $\Phi_n[{\cal A}]$, called {\em noncopy-signature of} $\cal A$ for $Q''$ and $D_{\bar{N}^{(i)}}(Q)$, is defined as follows: 
\begin{itemize}
	\item In case $m = 0$, $\Phi_n[{\cal A}]$ is the empty vector; and 
	\item In case $m \geq 1$, $\Phi_n[{\cal A}]$ is the $m$-ary vector  
$[ \ \nu^{(i)}_Q$ $(\theta(Y''_1)),$ $\ldots,$ $\nu^{(i)}_Q$ $(\theta(Y''_m)) \ ]$. 
\end{itemize}
\end{definition} 

The intuition for Definitions~\ref{copy-sig-def} and \ref{noncopy-sig-def} parallels the intuition for atom-signatures, see the discussion, above, of Definition~\ref{atom-sig-def}. That is, suppose that for an association $\cal A$ for a query $Q''$ and for a database $D_{\bar{N}^{(i)}}(Q)$, the association $\cal A$ contributes at least one valid assignment mapping to the set $\Gamma(Q'',D_{\bar{N}^{(i)}}(Q))$. Then, intuitively, we have that: 
\begin{itemize} 
	\item The copy-signature of $\cal A$ shows how the {\em copy variables} of the query $Q''$ could be mapped, via all the valid assignments associated with $\cal A$, into copy variables of (the subset ${\cal S}_{C(Q)}$ of subgoals of) the query $Q$; 
	
	some of the ``copy variables'' in the vector could equal unity, to indicate those cases where the image, in the sense of $\cal A$ via $\Phi_c[{\cal A}]$, of a copy-sensitive subgoal of $Q''$ is a relational subgoal of $Q$; 
	
	and  
	
	\item The noncopy-signature of $\cal A$ shows how the {\em multiset noncopy variables} of the query $Q''$ could be mapped, via all the valid assignments associated with $\cal A$,  into terms in (the subset ${\cal S}_{C(Q)}$ of subgoals of) the query $Q$. 
\end{itemize} 

In the remainder of this proof, we will use the atom-signatures, the copy-signatures, and the noncopy-signatures of associations for $Q''$ and $D_{\bar{N}^{(i)}}(Q)$, for each natural number $i$, to classify all the $t^*_Q$-valid assignment mappings from $Q''$ to the database $D_{\bar{N}^{(i)}}(Q)$. We will use the classification to construct the closed-form expressions,  ${\cal F}_{(Q)}^{(Q'')}$, for the multiplicity of the tuple $t^*_Q$ in the answer to the query $Q''$ on each such database $D_{\bar{N}^{(i)}}(Q)$. (For an introduction to ${\cal F}_{(Q)}^{(Q'')}$, please refer back to Section~\ref{proof-intuition-sec}.) More specifically, in Section~\ref{monomial-classes-sec} we will find a very practical use for copy-signatures and for noncopy-signatures of associations, in that these signatures of association $\cal A$ will help us compute the number of distinct entries, in the set $\Gamma^{(t^*_Q)}_{\bar{S}}(Q'',D_{\bar{N}^{(i)}}(Q))$, that (entries) are contributed by the (valid) mappings of $\cal A$ for $Q''$ and $D_{\bar{N}^{(i)}}(Q)$. Please see Example~\ref{main-proof-ex} for an extended illustration of copy-signatures and of noncopy-signatures of associations.



The following result is immediate from the definitions. 

\begin{proposition} 
\label{unique-sig-prop}
Let $i \in {\mathbb N}_+$, and let ${\cal A}$ be an association for query $Q''$ and for the database $D_{\bar{N}^{(i)}}(Q)$, such that there exists a valid assignment mapping for $Q''$, $D_{\bar{N}^{(i)}}(Q)$, and $\cal A$.  Then the noncopy-signature of $\cal A$ for $Q''$ and $D_{\bar{N}^{(i)}}(Q)$ exists and is unique. 
\end{proposition} 

Proposition~\ref{unique-sig-prop} lets us refer to  {\em the} noncopy-signature of a given association, for some query and database, provided that the association generates at least one valid assignment mapping.

\subsubsection{Properties of Associations for $Q''$ and $D_{\bar{N}^{(i)}}(Q)$} 
\label{assoc-properties-sec}


The results of Propositions~\ref{same-sigs-prop} through~\ref{fixed-unity-prop} are immediate from the definitions (unless discussed further in this subsection). 

\begin{proposition} 
\label{same-sigs-prop}
Let $i \in {\mathbb N}_+$, and let ${\cal A}_1$ and ${\cal A}_2$ be two associations for query $Q''$ and for the database $D_{\bar{N}^{(i)}}(Q)$, such that ${\cal A}_1$ and ${\cal A}_2$ have the same atom-signature. Then: 
\begin{itemize} 
	\item[(i)] ${\cal A}_1$ and ${\cal A}_2$ have the same copy-signature. 
	\item[(ii)] Suppose that, in addition, there exists a valid assignment mapping, call it $\theta_1$, for $Q''$, $D_{\bar{N}^{(i)}}(Q)$, and ${\cal A}_1$. Let $t_{(\theta_1)}$ be the tuple resulting from restricting $\theta_1$ to the head vector of the query $Q''$. Then we have that: 
	\begin{itemize} 
		\item[(ii-a)] There exists a valid assignment mapping, call it $\theta_2$, for $Q''$, $D_{\bar{N}^{(i)}}(Q)$, and ${\cal A}_2$;
		\item[(ii-b)] ${\cal A}_1$ and ${\cal A}_2$ have the same noncopy-signature; and  
		\item[(ii-c)] If $t_{(\theta_1)}$ is the tuple $t^*_Q$, then restricting $\theta_2$ to the head vector of the query $Q''$ results in the {\em same} tuple $t_{(\theta_1)}$ (which is $t^*_Q$). 
	\end{itemize} 
\end{itemize} 
\end{proposition}

The following result holds due to our convention for ground atoms in databases, see Section~\ref{ground-atoms-convent-sec}, as well as to our definition of the databases $D_{\bar{N}^{(i)}}(Q)$ using the subset ${\cal S}_{C(Q)}$ of subgoals of the query $Q$. (Recall that all the elements of the set ${\cal S}_{C(Q)}$ have pairwise distinct relational templates.) Note that the associations ${\cal A}_1$ and ${\cal A}_2$, in the statement of Proposition~\ref{diff-assoc-prop}, are not restricted to have either the same atom-signature or different  atom-signatures.

\begin{proposition} 
\label{diff-assoc-prop} 
Let $i \in {\mathbb N}_+$, and let ${\cal A}_1$ and ${\cal A}_2$ be two distinct associations for query $Q''$ and for the database $D_{\bar{N}^{(i)}}(Q)$. Let $(\theta_1,\theta_2)$ be an arbitrary pair (if any exists), such that $\theta_1$ is a valid assignment mapping for $Q''$, $D_{\bar{N}^{(i)}}(Q)$, and ${\cal A}_1$, and $\theta_2$ is a valid assignment mapping for $Q''$, $D_{\bar{N}^{(i)}}(Q)$, and ${\cal A}_2$. Then there exists a variable $X$ of $Q''$ such that (i) $X$ is not a copy variable of $Q''$, and (ii) $\theta_1(X) \neq \theta_2(X)$. 
\end{proposition} 

\begin{proposition} 
\label{same-head-unity-prop} 
Let $i \in {\mathbb N}_+$, and let ${\cal A}$ be an association for query $Q''$ and for the database $D_{\bar{N}^{(i)}}(Q)$, such that there exists a valid assignment mapping, $\theta$, for $Q''$, $D_{\bar{N}^{(i)}}(Q)$, and $\cal A$. Let tuple $t_{\theta}$ be the  result of restricting $\theta$ to the head vector of the query $Q''$. Then there exists a {\em unity} valid assignment mapping, call it $\theta^{(u)}$, for $Q''$, $D_{\bar{N}^{(i)}}(Q)$, and $\cal A$, such that the restriction of $\theta^{(u)}$ to the head vector of the query $Q''$ results in the {\em same} tuple $t_{\theta}$. 
\end{proposition} 

\begin{proposition} 
\label{unique-unity-prop} 
Let $i \in {\mathbb N}_+$, and let ${\cal A}$ be an association for query $Q''$ and for the database $D_{\bar{N}^{(i)}}(Q)$, such that there exists a unity valid assignment mapping, $\theta$, for $Q''$, $D_{\bar{N}^{(i)}}(Q)$, and $\cal A$. Then there does not exist any unity valid assignment mapping for $Q''$, $D_{\bar{N}^{(i)}}(Q)$, and $\cal A$, that (unity valid assignment mapping) would be distinct from $\theta$. 
\end{proposition} 

\begin{proposition}
\label{same-head-tuple-prop}
Let $i \in {\mathbb N}_+$, and let ${\cal A}$ be an association for query $Q''$
and for the database $D_{\bar{N}^{(i)}}(Q)$. Suppose $\cal A$ is such that there exists a unity valid assignment mapping, $\theta^{(u)}$, for $Q''$, $D_{\bar{N}^{(i)}}(Q)$, and $\cal A$. Let tuple $t_{\theta^{(u)}}$ be the restriction of $\theta^{(u)}$ to the head vector of the query $Q''$. 
Then we have that for each candidate assignment mapping, call it $\theta'$, for $Q''$, $D_{\bar{N}^{(i)}}(Q)$, and $\cal A$, $\theta'$ is a valid assignment mapping for $Q''$, $D_{\bar{N}^{(i)}}(Q)$, and $\cal A$, and the restriction of $\theta'$ to the head vector of the query $Q''$ is the tuple $t_{\theta^{(u)}}$. 
\end{proposition}

For the next result we introduce some notation. Let $i \in {\mathbb N}_+$, and let ${\cal A}$ be an association for query $Q''$
and for the database $D_{\bar{N}^{(i)}}(Q)$.  In case $r \geq 1$, denote the entries in the copy-signature $\Phi_c[{\cal A}]$ as 
$\Phi_c[{\cal A}]$ $=$ $[V_{j1},\ldots,V_{jr}]$. Here, for each $k \in \{ 1,\ldots,r \}$ we have that $V_{jk}$ $\in$ $\{ 1,$ $N_{m+1},$ $\ldots,$ $N_{m+w}\}$, where the $N$-values are variables in the vector $\bar{N}$.   (In case $r = 0$, $\Phi_c[{\cal A}]$ is the empty vector by definition.) Further, again for each $k \in \{ 1,\ldots,r \}$, by $V^{(i)}_{jk}$ we denote the value of  $V_{jk}$ in the vector $\bar{N}^{(i)}$ whenever  $V_{jk} \in \{ N_{m+1},\ldots,N_{m+r} \}$; we set $V^{(i)}_{jk} := 1$ for the case $V_{jk}$ = 1.   

We also use the notation $\Gamma^{({\cal A})}$ (for all cases $r \geq 0$ of the value $r = |M_{copy}|$): For an $i \in {\mathbb N}_+$ and an association $\cal A$ for $Q''$ and $D_{\bar{N}^{(i)}}(Q)$, $\Gamma^{({\cal A})}$ stands for the set of all tuples contributed to the set $\Gamma^{(t_{\theta^{(u)}})}_{\bar{S}}(Q'',$ $D_{\bar{N}^{(i)}}(Q))$, for some tuple $t_{\theta^{(u)}}$,  by the valid assignment mappings (if any) for the association $\cal A$. (By Proposition~\ref{same-head-tuple-prop}, all the valid assignment mappings (if any) for an association $\cal A$ for $Q''$ and $D_{\bar{N}^{(i)}}(Q)$, agree on their restriction to the head vector of the query $Q''$.) Finally, in case $r \geq 1$, we denote by $\Gamma^{({\cal A})}_c$ the set projection of the set $\Gamma^{({\cal A})}$ on the columns $Y''_{m+1},$ $\ldots,$ $Y''_{m+r}$, for all the copy variables of the query $Q''$.

\begin{proposition} 
\label{fixed-unity-prop} 
Let $i \in {\mathbb N}_+$, and let ${\cal A}$ be an association for query $Q''$
and for the database $D_{\bar{N}^{(i)}}(Q)$. Suppose $\cal A$ is such that there exists a unity valid assignment mapping, $\theta^{(u)}$, for $Q''$, $D_{\bar{N}^{(i)}}(Q)$, and $\cal A$. Let tuple $t_{\theta^{(u)}}$ be the restriction of $\theta^{(u)}$ to the head vector of the query $Q''$. 
Then the following holds: 
\begin{itemize} 
	\item[(i)] The set $\Gamma^{({\cal A})}$ is not empty; 
	\item[(ii)] In case $r = 0$, the association $\cal A$ has exactly one valid assignment mapping and contributes exactly one tuple to the set $\Gamma^{(t_{\theta^{(u)}})}_{\bar{S}}(Q'',D_{\bar{N}^{(i)}}(Q))$;  
	\item[(iii)] In case $r \geq 1$, the set $\Gamma^{({\cal A})}_c$ has the tuple $(n_1,$ $n_2,$ $\ldots,$ $n_r)$ for each $1 \leq n_k \leq V^{(i)}_{jk}$  for each $k \in \{ 1,\ldots,r \}$, and $\Gamma^{({\cal A})}_c$ does not have any other tuples;  
	\item[(iv)] The total number of {\em distinct} valid assignment mappings for $Q''$, $D_{\bar{N}^{(i)}}(Q)$, and $\cal A$, call it $T^{({\cal A})}$, is $1$ in case $r = 0$, and is $\Pi_{k=1}^r V^{(i)}_{jk}$  in case $r \geq 1$; and  
	\item[(v)] The cardinality of the set $\Gamma^{({\cal A})}$ is equal to $T^{({\cal A})}$. 
\end{itemize}
\end{proposition}

\subsection{Sets of Associations for $Q''$ and $D_{\bar{N}^{(i)}}(Q)$} 
\label{total-set-of-assoc-sec}


Consider each of the $F^G$ associations, ${\cal A}_1,\ldots,{\cal A}_{F^G}$, of the form $\cal A$ as defined in Section~\ref{valid-map-sec}, between the $G \geq 1$ subgoals $g_1,\ldots,g_G$ of the fixed query $Q''$ and the $F \geq 1$ 
ground atoms of the database $D_{\bar{N}^{(i)}}(Q)$, for an arbitrary fixed $i \in {\mathbb N}_+$. 
Clearly, enumerating all the $F^G$ associations is a way to find {\em all} satisfiable assignments for $Q''$ and $D_{\bar{N}^{(i)}}(Q)$. In this section we construct a set, ${\mathbb A}^{(i)}_{Q''}$, which includes some of the above associations, and ``captures'', in a very precise sense (see Proposition~\ref{captures-prop}), all of the $t^*_Q$-valid assignment mappings for $Q''$ and $D_{\bar{N}^{(i)}}(Q)$. We will use the set ${\mathbb A}^{(i)}_{Q''}$ in Section~\ref{monomial-classes-sec}, to define a function, ${\cal F}_{(Q)}^{(Q'')}$, in terms of the variables in the vector $\bar N$. For each $i \in {\mathbb N}_+$, the function ${\cal F}_{(Q)}^{(Q'')}$ uses the values (in ${\bar N}^{(i)}$) of the variables in the vector $\bar N$ to return the multiplicity of the tuple $t^*_Q$ in the bag ${\sc Res}_C(Q'',D_{\bar{N}^{(i)}}(Q))$. 

\begin{definition}{Set ${\mathbb A}^{(i)}_{Q''}$ of $t^*_Q$-valid assignment mappings for $Q''$ and $D_{\bar{N}^{(i)}}(Q)$} 
Let $i \in {\mathbb N}_+$. {\em The  set ${\mathbb A}^{(i)}_{Q''}$ of $t^*_Q$-valid assignment mappings for CCQ query $Q''$ and for the database $D_{\bar{N}^{(i)}}(Q)$,} is the set of all of the associations for $Q''$ and  $D_{\bar{N}^{(i)}}(Q)$, such that for each ${\cal A} \in {\mathbb A}^{(i)}_{Q''}$ there exists at least one $t^*_Q$-valid assignment mapping for  $Q''$, $D_{\bar{N}^{(i)}}(Q)$, and $\cal A$. 
\end{definition} 

For each $i \in {\mathbb N}_+$, we denote the cardinality of the set ${\mathbb A}^{(i)}_{Q''}$ by $R^{(i)}_{Q''}$. 
Further, whenever $R^{(i)}_{Q''} \geq 1$, we refer to the individual elements of the set ${\mathbb A}^{(i)}_{Q''}$ as $A^{(i)}_j$, for $1 \leq j \leq R^{(i)}_{Q''}$. (We will avoid the confusion as to which query $A^{(i)}_j$ ``refers to'', by always using the notation $A^{(i)}_j$ in the context of exactly one query.) 

Consider an arbitrary set ${\mathbb A}^*_{Q''}$ of associations for query $Q''$ and for database $D_{\bar{N}^{(i)}}(Q)$. Let $\cal A$ be an association  for $Q''$ and $D_{\bar{N}^{(i)}}(Q)$ such that there exists a valid assignment mapping, $\theta$, for $Q''$, for $D_{\bar{N}^{(i)}}(Q)$, and for $\cal A$. Then  we say that the set ${\mathbb A}^*_{Q''}$ {\em captures the valid assignment mapping} $\theta$ if and only if we have that ${\cal A} \in {\mathbb A}^*_{Q''}$. 
(See Proposition~\ref{diff-assoc-prop} for a justification of this definition.) 

\begin{proposition} 
\label{captures-prop} 
Let $i \in {\mathbb N}_+$. Then (i) The set ${\mathbb A}^{(i)}_{Q''}$ of $t^*_Q$-valid assignment mappings for CCQ query $Q''$ and for the database $D_{\bar{N}^{(i)}}(Q)$ captures all the $t^*_Q$-valid assignment mappings for $Q''$ and $D_{\bar{N}^{(i)}}(Q)$; and (ii) For each  valid assignment mapping $\theta$ for $Q''$ and $D_{\bar{N}^{(i)}}(Q)$ such that ${\mathbb A}^{(i)}_{Q''}$ captures $\theta$, $\theta$ is a  $t^*_Q$-valid assignment mapping for $Q''$ and $D_{\bar{N}^{(i)}}(Q)$. 
\end{proposition} 

\begin{proposition} 
\label{Observation_three_four}
Suppose that there exists an $i^* \in {\mathbb N}_+$ such that for some association ${\cal A}_{j^*}^{(i^*)}$ $\in$ ${\mathbb A}^{(i^*)}_{Q''}$, a valid assignment mapping for ${\cal A}_{j^*}^{(i^*)}$  induces a mapping from all the subgoals of $Q''$ to a {\em single} copy-neutral canonical database for query $Q$.\footnote{That copy-neutral canonical database for $Q$ is within database $D_{\bar{N}^{(i^*)}}(Q)$, see Proposition~\ref{Observation_one_two} in Section~\ref{obs-one-sec}.} Then for {\em each} $i \in {\mathbb N}_+$, there exists a $j$, $1 \leq j \leq R^{(i)}_{Q''}$, such that  a valid assignment mapping for the association ${\cal A}_j^{(i)}$ (exists and) induces a mapping from all the subgoals of the query $Q''$ to a {\em single} copy-neutral canonical database for query $Q$ (within database  $D_{\bar{N}^{(i)}}(Q)$). Moreover, ${\cal A}_j^{(i)}$ and ${\cal A}_{j^*}^{(i^*)}$ have the same atom-signature. 
\end{proposition}

\begin{proof}
We are given that there exists a pair $(i^*, j^*)$, with $i^* \in {\mathbb N}_+$ and with $1 \leq j^* \leq R^{(i^*)}_{Q''}$, such that the  association ${\cal A}_{j^*}^{(i^*)}$ generates a valid assignment mapping from query $Q''$ to a single copy-neutral canonical database, call it $D^*$, for query $Q$  (within database $D_{\bar{N}^{(i^*)}}(Q)$). 
By Proposition~\ref{same-head-unity-prop} in Section~\ref{assoc-properties-sec}, there exists a unity valid 
%
assignment mapping, call it $\theta'$, for $Q''$ and $D_{\bar{N}^{(i^*)}}(Q)$, such that $\theta'$ uses only the atoms of the database $D^*$ to generate the tuple $t^*_Q$ in the answer to $Q''$ on database $D_{\bar{N}^{(i^*)}}(Q)$. 

Now fix an abitrary $i \in {\mathbb N}_+$. By Proposition~\ref{Observation_one_two} in Section~\ref{obs-one-sec}, database $D_{\bar{N}^{(i)}}(Q)$ has at least one copy-neutral canonical database for query $Q$. Choose and fix in $D_{\bar{N}^{(i)}}(Q)$ one arbitrary such copy-neutral canonical database for $Q$, call this database $D$. 
By definition of copy-neutral canonical database, 
there is an isomorphism from all the ground atoms of the database $D^*$ (within $D_{\bar{N}^{(i^*)}}(Q)$, see first paragraph of this proof), to all the ground atoms of database $D$ (within $D_{\bar{N}^{(i)}}(Q)$). 
Moreover, there exists at least one isomorphism from $D^*$ to $D$, call this isomorphism $\iota^*$, such that for each atom $d^*$ in $D^*$, it holds that $\psi^{gen(Q)}_{\bar{N}^{(i)}}[d^*]$ is the same (atom in set ${\cal S}_{C(Q)}$) as $\psi^{gen(Q)}_{\bar{N}^{(i)}}[\iota^*(d^*)]$.    
Then the composition, call it $\varphi$, of $\theta'$ (of the first paragraph of this proof) with $\iota^*$ gives us a valid assignment mapping from query $Q''$ to the copy-neutral canonical database $D$ for query $Q$ in database $D_{\bar{N}^{(i)}}(Q)$, such that the restriction of $\varphi$ to the head variables of $Q''$ is the tuple $t^*_Q$. By Proposition~\ref{captures-prop}, the association represented by $\varphi$ must be in the set ${\mathbb A}^{(i)}_{Q''}$ for database $D_{\bar{N}^{(i)}}(Q)$. By construction, this association and ${\cal A}_{j^*}^{(i^*)}$ have the same atom-signature. By the fact that $i$ has been chosen arbitrarily, Q.E.D. 
\end{proof}


\begin{proposition}
\label{Observation_three_five} 
Suppose that, for some $i \in {\mathbb N}_+$, there exists an association ${\cal A}_j^{(i)}$ in the set ${\mathbb A}^{(i)}_{Q''}$ such that the unity valid assignment mapping, $\theta_j$, for ${\cal A}_j^{(i)}$ 
induces a mapping from the subgoals of $Q''$ into elements of two or more copy-neutral canonical databases of $Q$ in database $D_{\bar{N}^{(i)}}(Q)$. Then there exists an association ${\cal A}_{j'}^{(i)}$ in ${\mathbb A}^{(i)}_{Q''}$, and there exists in $D_{\bar{N}^{(i)}}(Q)$ a copy-neutral canonical database of $Q$, call this database $D^*$, such that the unity valid assignment mapping, $\theta_{j'}$, for ${\cal A}_{j'}^{(i)}$ 
induces a mapping from all the subgoals of $Q''$ into ground atoms belonging to $D^*$ only. Moreover, ${\cal A}_{j}^{(i)}$ and ${\cal A}_{j'}^{(i)}$ have the same atom-signature. 
\end{proposition}

\begin{proof}
We observe first that if $M_{noncopy} = \emptyset$ then database $D_{\bar{N}^{(i)}}(Q)$ comprises exactly one copy-neutral canonical database for $Q$. This observation is immediate from the construction of $D_{\bar{N}^{(i)}}(Q)$. 

Thus we assume for the remainder of this proof that $m = |M_{noncopy}| \geq 1$. For the association ${\cal A}_j^{(i)}$ as in the statement of this Observation, denote by (set of ground atoms) $\cal T$ the image of the condition of query $Q''$ under the mapping $\theta_j$. By Proposition~\ref{Observation_one_three}, the only way the valid assignment mapping $\theta_j$ for ${\cal A}_j^{(i)}$ can map the subgoals of $Q''$ into elements of two or more copy-neutral canonical databases of $Q$ in database $D_{\bar{N}^{(i)}}(Q)$ is when the result of intersecting $adom({\cal T})$ with $S_l^{(i)}$, for at least one $l \in \{1,\ldots,m \}$, has size two or more.  (For the notation $adom({\cal T})$, see note before Proposition~\ref{Observation_one_three}.)

We show an algorithm for producing the association ${\cal A}_{j'}^{(i)}$ and the (copy-neutral canonical) database $D^*$, of the statement of this Proposition, from the association ${\cal A}_{j}^{(i)}$. First, for each $l \in \{1,\ldots,m \}$ such that  the result of intersecting $adom({\cal T})$ with $S_l^{(i)}$ is not empty, fix a single value in that intersection. Let the result be values $v^{(i)}_{l_1},\ldots,v^{(i)}_{l_k}$, where: $1\leq k \leq m$, all the values $l_1,\ldots,l_k$ are distinct, each $l_n$ (for all $1 \leq n \leq k$) satisfies $1 \leq l_n \leq m$, and each $v^{(i)}_{l_n}$ (for all $1 \leq l_n \leq m$, for all $1 \leq n \leq k$) satisfies $v^{(i)}_{l_n} \in S^{(i)}_{l_n}$. Call all values in the set $E = ((\bigcup_{n=1}^m S^{(i)}_n) \bigcap adom({\cal T})) - \{ v^{(i)}_{l_1},\ldots,v^{(i)}_{l_k} \}$ the ``extra multiset noncopy'' values in $adom({\cal T})$. By the assumptions in the statement of this Proposition, the set $E$ is not empty. 

The next step of the algorithm is to modify the mapping $\theta_j$ for the association ${\cal A}^{(i)}_j$, by replacing in $\theta_j$ each value $v$ belonging to $S^{(i)}_n \bigcap E$, $1 \leq n \leq m$, by the value $v^{(i)}_{n}$ that we fixed as described in the previous paragraph. (The intuition is that for each mapping in $\theta_j$ of a variable of $Q''$ to an ``extra multiset noncopy'' value, we ``redirect'' the mapping to a mapping of the same variable into an ``appropriate'' value from among the values $v^{(i)}_{l_1},\ldots,v^{(i)}_{l_k}$ fixed in the previous paragraph.) As we modify $\theta_j$ this way, we also modify the association ${\cal A}_{j}^{(i)}$, by replacing in it all occurrences of each $v \in S^{(i)}_n \bigcap E$, $1 \leq n \leq m$, by the value $v^{(i)}_{n}$. Denote by $\theta_{j'}$ the result of this modification of $\theta_j$, and by ${\cal A}_{j'}^{(i)}$ the result of this modification of ${\cal A}_{j}^{(i)}$. 
By construction, we have that: 
\begin{itemize} 
	\item[(a)] $\theta_{j'}$ is a mapping;
	\item[(b)] for all the terms of $Q''$ that are not copy variables, $\theta_{j'}$ maps all these terms of $Q''$ into the values in $adom({\T}) - E$;
	\item[(c)] all ground atoms mentioned in ${\cal A}_{j'}^{(i)}$ belong to $D_{\bar{N}^{(i)}}(Q)$ (by construction of the database); 
	\item[(d)] $\theta_{j'}$ is a candidate assignment mapping for $Q''$, $D_{\bar{N}^{(i)}}(Q)$, and ${\cal A}_{j'}^{(i)}$; and 
	\item[(e)] ${\cal A}_{j}^{(i)}$ and ${\cal A}_{j'}^{(i)}$ have the same atom-signature. 
\end{itemize} 

We now show that the association ${\cal A}_{j'}^{(i)}$ belongs to the set ${\mathbb A}^{(i)}_{Q''}$. This  is immediate from the fact that $\theta_j$ and $\theta_{j'}$ agree on the images for all the head variables of $Q''$ and from items (a) and (c) of the previous paragraph. Finally, let ${\cal T}'$ be the set of all ground atoms mentioned in the association ${\cal A}_{j'}^{(i)}$. From item (b) of the previous paragraph and from Proposition~\ref{Observation_one_three}, we have that there exists in $D_{\bar{N}^{(i)}}(Q)$ a (single) copy-neutral canonical database for query $Q$, call that database $D^*$, such that ${\cal T}' \subseteq D^*$. We conclude that the unity valid assignment mapping $\theta_{j'}$ for $Q''$, $D_{\bar{N}^{(i)}}(Q)$, and ${\cal A}_{j'}^{(i)}$ maps query $Q''$ into a single copy-neutral canonical database (in $D_{\bar{N}^{(i)}}(Q)$) for the query $Q$. 
\end{proof}

\begin{proposition}
\label{Observation_three_six} 
Suppose there exists an $i \in {\mathbb N}_+$ such that the set ${\mathbb A}^{(i)}_{Q''}$ for database $D_{\bar{N}^{(i)}}(Q)$ is not empty. Then the set  ${\mathbb A}^{(i)}_{Q''}$ is not empty for {\em each} $i \in {\mathbb N}_+$. 
\end{proposition}

\begin{proof}
The proof is immediate from Propositions~\ref{Observation_three_four} and \ref{Observation_three_five}. That is:
\begin{itemize}
	\item Case 1: Suppose that, for some $i \in {\mathbb N}_+$, the set ${\mathbb A}^{(i)}_{Q''}$ for database $D_{\bar{N}^{(i)}}(Q)$ has an association corresponding to a valid assignment mapping from query $Q''$ to a single copy-neutral canonical database for $Q$ in $D_{\bar{N}^{(i)}}(Q)$. Then, by Proposition~\ref{Observation_three_four}, for {\em all}   $i \in {\mathbb N}_+$, the set ${\mathbb A}^{(i)}_{Q''}$ for database $D_{\bar{N}^{(i)}}(Q)$ has an association corresponding to a valid assignment mapping from query $Q''$ to a single copy-neutral canonical database for $Q$ in $D_{\bar{N}^{(i)}}(Q)$. Q.E.D. 
	\item Case 2: Suppose that, for some $i \in {\mathbb N}_+$, the set ${\mathbb A}^{(i)}_{Q''}$ for database $D_{\bar{N}^{(i)}}(Q)$ has an association corresponding to a valid assignment mapping from query $Q''$ to (ground atoms in) two or more copy-neutral canonical databases for $Q$ in $D_{\bar{N}^{(i)}}(Q)$. Then, by Proposition~\ref{Observation_three_five}, the set  ${\mathbb A}^{(i)}_{Q''}$ for {\em the same value of} $i$  has an association corresponding to a valid assignment mapping from query $Q''$ to a single copy-neutral canonical database for $Q$ in $D_{\bar{N}^{(i)}}(Q)$. Thus we have reduced this case to Case 1. Q.E.D. 
\end{itemize}
\end{proof}

\begin{proposition}
\label{Observation_three_seven} 
Suppose there exists an $i^* \in {\mathbb N}_+$ such that the query $Q''$ has no answer $t^*_Q$ on database $D_{\bar{N}^{(i^*)}}(Q)$. Then the multiplicity of the tuple $t^*_Q$ in the bag ${\sc Res}_C(Q'',D_{\bar{N}^{(i)}}(Q))$ equals zero on the database $D_{\bar{N}^{(i)}}(Q)$ for each $i \in {\mathbb N}_+$. 
\end{proposition}

\begin{proof}
It is immediate from Proposition~\ref{Observation_three_six} that if there exists an $i^* \in {\mathbb N}_+$ such that the set ${\mathbb A}^{(i^*)}_{Q''}$ for database $D_{\bar{N}^{(i^*)}}(Q)$ is empty, then the set  ${\mathbb A}^{(i)}_{Q''}$ is empty for {\em each} $i \in {\mathbb N}_+$.   
\end{proof} 

\subsection{Monomials for the Multiplicity of Tuple $t^*_Q$ in Bag  ${\sc Res}_C(Q'',D_{\bar{N}^{(i)}}(Q))$} 
\label{monomial-classes-sec}

In this section we provide an algorithm for constructing monomials for a function, call it ${\cal F}_{(Q)}^{(Q'')}$, defined in terms of the variables in the vector $\bar{N}$.  ${\cal F}_{(Q)}^{(Q'')}$  computes the multiplicity of the tuple $t^*_Q$ in the bag  ${\sc Res}_C(Q'',D_{\bar{N}^{(i)}}(Q))$ for each $i \in {\mathbb N}_+$, by using the values in the vector $\bar{N}^{(i)}$ as values of the variables in the vector $\bar{N}$.

We observe first that, by Proposition~\ref{Observation_three_seven}, ${\cal F}_{(Q)}^{(Q'')}$ either equals zero for all input vectors  $\bar{N}^{(i)}$, or returns a positive-integer value for each $\bar{N}^{(i)}$, $i \in {\mathbb N}_+$. In the remainder of the proof of Theorem~\ref{magic-mapping-prop}, we assume that the function ${\cal F}_{(Q)}^{(Q'')}$ returns a positive-integer value for each $\bar{N}^{(i)}$. By the results of Section~\ref{total-set-of-assoc-sec}, we infer from this assumption that  the cardinality ${R}^{(i)}_{Q''}$ of the set ${\mathbb A}^{(i)}_{Q''}$ is a positive integer for each $i \in {\mathbb N}_+$. 

\subsubsection{Defining the Monomial Classes ${\cal C}^{(Q'')}$}

Fix an $i \in {\mathbb N}_+$. We partition all the elements of the set ${\mathbb A}^{(i)}_{Q''} \neq \emptyset$ into equivalence classes: Two distinct elements (in case ${R}^{(i)}_{Q''} \geq 2$) ${\cal A}_j^{(i)}$ and ${\cal A}_k^{(i)}$ of the set ${\mathbb A}^{(i)}_{Q''}$ belong to the same {\em monomial class} if and only if ${\cal A}_j^{(i)}$ and ${\cal A}_k^{(i)}$ have the same atom-signature. Call all the resulting nonempty monomial classes  ${\cal C}_{1}^{(Q'')(i)}$, ${\cal C}_2^{(Q'')(i)}$, $\ldots,$ ${\cal C}_{n^{(i)}}^{(Q'')(i)}$, $n^{(i)} \leq {R}^{(i)}_{Q''}$. From the definition of the monomial classes, we have that $n^{(i)} \geq 1$, and that $n^{(i)}$ is exactly the number of all the atom-signatures of the elements of the set ${\mathbb A}^{(i)}_{Q''}$. In addition, by Proposition~\ref{same-sigs-prop} and from the definition of the set ${\mathbb A}^{(i)}_{Q''}$ we have that for each $j$, $1 \leq j \leq n^{(i)}$, all the elements of the set ${\cal C}_j^{(Q'')(i)}$ (by having the same atom-signature) have the same noncopy-signature and have the same copy-signature. Hence, for each monomial class ${\cal C}_j^{(Q'')(i)}$ we can refer to  {\em the} atom-signature 
of ${\cal C}_j^{(Q'')(i)}$,  to  {\em the} noncopy-signature 
of ${\cal C}_j^{(Q'')(i)}$, and to  {\em the} copy-signature 
of ${\cal C}_j^{(Q'')(i)}$. 


By Proposition~\ref{captures-prop}, we have that for each $i \in {\mathbb N}_+$ and for each monomial class for $Q''$ and $D_{\bar{N}^{(i)}}(Q)$, all the valid assignment mappings of all the elements of the class contribute tuples to the set  $\Gamma^{(t^*_Q)}_{\bar{S}}(Q'',D_{\bar{N}^{(i)}}(Q))$. That is, for all the valid assignment mappings in each monomial class, the restriction of each valid assignment mapping to the head vector of the query $Q''$ is the tuple $t^*_Q$. 

This result follows from the results of Sections~\ref{assoc-properties-sec} and \ref{total-set-of-assoc-sec}:  

\begin{proposition} 
\label{same-monomial-classes-prop}
Let $\Xi$ be a $G$-ary vector  of (not necessarily distinct) elements of the set ${\cal S}_{C(Q)}$. Suppose   there exists an  $i^* \in {\mathbb N}_+$ such that the monomial class ${\cal C}^{(Q'')(i^*)}$ with atom-signature $\Xi$ is not empty. Then  for {\em all} $i \in {\mathbb N}_+$ it holds that  the monomial class ${\cal C}^{(Q'')(i)}$ with atom-signature $\Xi$ is not empty. 
\end{proposition} 

From Proposition~\ref{same-monomial-classes-prop} it follows that for a fixed query $Q''$, we can drop the $(i)$-superscript from the notation for monomial classes. (That is, the set of nonempty monomial classes for $Q''$, w.r.t. the family  $\{ D_{\bar{N}^{(i)}}(Q) \}$, does not depend on the specific database $D_{\bar{N}^{(i)}}(Q)$ in the family.)  
From now on, when referring to the set of all nonempty monomial classes for query $Q''$ on database $D_{\bar{N}^{(i)}}(Q)$, we will use the notation  ${\cal C}_{1}^{(Q'')}$, ${\cal C}_2^{(Q'')}$, $\ldots,$ ${\cal C}_{n^*}^{(Q'')}$, for a constant (w.r.t. $i$) positive-integer value $n^* \geq 1$.  We will abuse the notation somewhat, by using, in the context of a fixed $i \in {\mathbb N}_+$, the expression ``the set ${\cal C}^{(Q'')}$'' (where ${\cal C}^{(Q'')}$ is one of the ${\cal C}_{1}^{(Q'')}$, ${\cal C}_2^{(Q'')}$, $\ldots,$ ${\cal C}_{n^*}^{(Q'')}$) to refer to the contents of the set ${\cal C}^{(Q'')}$ w.r.t. the set ${\mathbb A}^{(i)}_{Q''}$ for the fixed $i$. 

\subsubsection{Monomials Corresponding to the Monomial \\ Classes for $Q''$ and $D_{\bar{N}^{(i)}}(Q)$: \\ Useful Properties}
\label{multiplicity-monomial-prep-sec} 

In this subsection we set the stage for the introduction, in Section~\ref{multiplicity-monomial-sec}, of ``multiplicity monomials'' for the monomial classes ${\cal C}_1^{(Q'')}$, $\ldots,$ ${\cal C}_{n^*}^{(Q'')}$. 

Assuming a fixed $i \in {\mathbb N}_+$, we recall the mapping $\nu_0$ and the sets $S_j^{(i)}$, which (sets) were introduced (for $1 \leq j \leq m$) for the case $m \geq 1$, see Section~\ref{nu-sec}. We use these constructs to define the domain, on the database $D_{\bar{N}^{(i)}}(Q)$, of each term of the query $Q$ that (term) is not a copy variable of $Q$. 

\begin{definition}{Domain of term of $Q$ in $D_{\bar{N}^{(i)}}(Q)$}
\label{dom-def} 
Let $i \in {\mathbb N}_+$. For each term $s$ of query $Q$ such that $s$ is not a copy variable of $Q$, the domain $Dom^{(i)}_Q(s)$ of the term in the database  $D_{\bar{N}^{(i)}}(Q)$ is defined as follows:  
\begin{itemize} 
	\item If $s$ is a constant, or a head variable of $Q$, or a set variable of $Q$, then $Dom^{(i)}_Q(s) := \{ \nu_0(s) \}$. 
	\item In case $m \geq 1$, for each variable $Y_j$ of the query $Q$, for $1 \leq j \leq m$, $Dom^{(i)}_Q(Y_j) := S_j^{(i)}$. 
\end{itemize} 
\end{definition} 

\begin{proposition} 
\label{dom-properties-prop} 
Let $i \in {\mathbb N}_+$. (i) For each (if any) pair $(s,t)$ of terms of query $Q$ such that $s \neq t$ and such that neither $s$ nor $t$ is a copy variable of $Q$, $Dom^{(i)}_Q(s)$ $\bigcap$ $Dom^{(i)}_Q(t)$ $=$ $\emptyset$. (ii) For each  term $s$ of query $Q$ such that $s$ is not a multiset variable of $Q$, $|Dom^{(i)}_Q(s)|$ $=$ $1$. (iii) In case $m \geq 1$, for each $j \in \{ 1,\ldots,m \}$ we have that $|Dom^{(i)}_Q(Y_j)|$ $=$ $N_j^{(i)}$ (in the vector $\bar{N}^{(i)}$). 
\end{proposition} 

For the next results, we introduce some notation. Given a query $Q''$, an $i \in {\mathbb N}_+$, and a nonempty monomial class ${\cal C}^{(Q'')}$ of associations in the set ${\mathbb A}^{(i)}_{Q''}$ for the query $Q''$ and for the database $D_{\bar{N}^{(i)}}(Q)$, denote by $\Gamma^{(i)}[{\cal C}^{(Q'')}]$ the set of all  tuples contributed to the set $\Gamma^{(t^*_Q)}_{\bar{S}}(Q'',D_{\bar{N}^{(i)}}(Q))$ by all the valid assignment mappings for all the elements of the class  ${\cal C}^{(Q'')}$. 
The following result is immediate from the definitions. 

\begin{proposition}
\label{gamma-i-union-prop} 
Let $i \in {\mathbb N}_+$. Then 
\begin{itemize} 
	\item[(i)] For each $j \in \{ 1,\ldots,n^* \}$, $\Gamma^{(i)}[{\cal C}_j^{(Q'')}]$ $\neq$ $\emptyset$.  
	\item[(ii)] The set $\Gamma^{(t^*_Q)}_{\bar{S}}$ $(Q'',D_{\bar{N}^{(i)}}(Q))$ is the union $\bigcup_{j=1}^{n^*}$ $\Gamma^{(i)}[{\cal C}_j^{(Q'')}]$. 
\end{itemize} 
\end{proposition}

We introduce some further notation: In case $m \geq 1$, for a monomial class ${\cal C}^{(Q'')}$ and  for some $j \in \{ 1,\ldots,m \}$, we denote by $\Gamma^{(i)}_{(Y''_j)}[{\cal C}^{(Q'')}]$ the set projection of the set $\Gamma^{(i)}[{\cal C}^{(Q'')}]$ on the multiset noncopy variable $Y''_j$ of the query $Q''$. 

\begin{proposition} 
\label{noncopy-proj-properties-prop}
Suppose $m \geq 1$. Let $\Xi$ be a $G$-ary vector  of (not necessarily distinct) elements of the set ${\cal S}_{C(Q)}$, such that  the monomial class ${\cal C}^{(Q'')}$ with atom-signature $\Xi$ is not empty. 
Then for each $i \in {\mathbb N}_+$ the following holds: 
\begin{itemize}
	\item[(i)] For each $j \in \{ 1,\ldots,m \}$: Suppose $Z$ is the $j$th component of the noncopy-signature vector of the monomial class ${\cal C}^{(Q'')}$. Then the set $\Gamma^{(i)}_{(Y''_j)}[{\cal C}^{(Q'')}]$: 
	\begin{itemize}
		\item[(i-a)] has all the elements of $Dom_Q^{(i)}(Z)$, and 
		\item[(i-b)] has no values from the set $(adom(D_{\bar{N}^{(i)}}(Q))$ $-$ $Dom_Q^{(i)}(Z))$. 
	\end{itemize}
	\item[(ii)] The set projection of the set $\Gamma^{(i)}[{\cal C}^{(Q'')}]$ on all the multiset noncopy variables $Y''_1,$ $Y''_2,$ $\ldots,$ $Y''_m$ of the query $Q''$ is the Cartesian product of the sets  $\Gamma^{(i)}_{(Y''_1)}$ $[{\cal C}^{(Q'')}]$, $\Gamma^{(i)}_{(Y''_2)}$ $[{\cal C}^{(Q'')}]$, $\ldots,$ $\Gamma^{(i)}_{(Y''_m)}$ $[{\cal C}^{(Q'')}]$.  
\end{itemize}
\end{proposition} 

Now suppose $r \geq 1$. In this case, we denote by $\Gamma^{(i)}_{(Y''_j)}[{\cal C}^{(Q'')}]$ the set projection of the set $\Gamma^{(i)}[{\cal C}^{(Q'')}]$ on the copy variable $Y''_j$ of the query $Q''$, for some $j \in \{ m+1,\ldots,m+r \}$. 

\begin{proposition} 
\label{copy-proj-properties-prop} 
Suppose $r \geq 1$. Let $\Xi$ be a $G$-ary vector  of (not necessarily distinct) elements of the set ${\cal S}_{C(Q)}$, such that  the monomial class ${\cal C}^{(Q'')}$ with atom-signature $\Xi$ is not empty. 
Then for each $i \in {\mathbb N}_+$ the following holds: 
\begin{itemize}
	\item[(i)] For each $j \in \{ 1,\ldots,r \}$: Suppose $Z$ is the $j$th component of the copy-signature vector of the monomial class ${\cal C}^{(Q'')}$. Then the set $\Gamma^{(i)}_{(Y''_{m+j})}[{\cal C}^{(Q'')}]$ is the set $\{ 1,\ldots,Z^{(i)} \}$, where (a) $Z^{(i)}$ is $1$ in case $Z = 1$, and (b) $Z^{(i)}$ is $N^{(i)}_k$ in case $Z = N_k$ for some $k \in \{ {m+1},\ldots,{m+w} \}$. 
	\item[(ii)] The set projection of the set $\Gamma^{(i)}[{\cal C}^{(Q'')}]$ on all the copy variables $Y''_{m+1},$ $Y''_{m+2},$ $\ldots,$ $Y''_{m+r}$ of the query $Q''$ is the Cartesian product of the sets  $\Gamma^{(i)}_{(Y''_{m+1})}[{\cal C}^{(Q'')}]$, $\Gamma^{(i)}_{(Y''_{m+2})}[{\cal C}^{(Q'')}]$, $\ldots,$ $\Gamma^{(i)}_{(Y''_{m+r})}[{\cal C}^{(Q'')}]$.  
\end{itemize}
\end{proposition} 

(The proof is immediate from Proposition~\ref{fixed-unity-prop}, once we recall that all associations in a monomial class share the same copy-signature.) 

For each $i \in {\mathbb N}_+$, we now characterize the set $\Gamma^{(i)}[{\cal C}^{(Q'')}]$ for each nonempty monomial class ${\cal C}^{(Q'')}$ for the query $Q''$ and family of databases $\{ D_{\bar{N}^{(i)}}(Q) \}$, for all combinations of values of $m \geq 0$ and of $r \geq 0$. 

\begin{proposition} 
\label{noncopy-proj-more-properties-prop}
Let $\Xi$ be a $G$-ary vector  of (not necessarily distinct) elements of the set ${\cal S}_{C(Q)}$, such that  the monomial class ${\cal C}^{(Q'')}$ with atom-signature $\Xi$ is not empty. 
Then for each $i \in {\mathbb N}_+$ the following holds: 
\begin{itemize} 
	\item In case $m \geq 1$ and $r \geq 1$, the set $\Gamma^{(i)}[{\cal C}^{(Q'')}]$ is the Cartesian product of two sets: 
	\begin{itemize}
		\item the set projection of $\Gamma^{(i)}[{\cal C}^{(Q'')}]$ on all the multiset noncopy variables $Y''_1,$ $\ldots,$ $Y''_m$ of the query $Q''$, and 
		\item the set projection of $\Gamma^{(i)}[{\cal C}^{(Q'')}]$ on all the copy variables $Y''_{m+1},$ $\ldots,$ $Y''_{m+r}$ of the query $Q''$. 
	\end{itemize} 
	\item In case $r = 0$, $\Gamma^{(i)}[{\cal C}^{(Q'')}]$ is its own set projection on all the multiset noncopy variables of the query $Q''$. 
	\item In case $m = 0$,  $\Gamma^{(i)}[{\cal C}^{(Q'')}]$ is its own set projection on all the copy variables of the query $Q''$. 
	\end{itemize} 
\end{proposition} 

(Recall from Section~\ref{basic-queries-sec} that we assume $m + r \geq 1$; thus in case $r = 0$ we have $m \geq 1$, and in case $m = 0$ we have $r \geq 1$. For a characterization of the set projection of $\Gamma^{(i)}[{\cal C}^{(Q'')}]$ on all the multiset noncopy variables $Y''_1,$ $\ldots,$ $Y''_m$ of the query $Q''$, in case $m \geq 1$, see Proposition~\ref{noncopy-proj-properties-prop}. For a characterization of the set projection of $\Gamma^{(i)}[{\cal C}^{(Q'')}]$ on all the copy variables $Y''_{m+1},$ $\ldots,$ $Y''_{m+r}$ of the query $Q''$, in case $r \geq 1$,   see Proposition~\ref{copy-proj-properties-prop}.)  

\subsubsection{Multiplicity Monomials for the Monomial \\ Classes ${\cal C}_1^{(Q'')},$ $\ldots,$ ${\cal C}_{n^*}^{(Q'')}$}
\label{multiplicity-monomial-sec} 


In this subsection, for each nonempty monomial class ${\cal C}^{(Q'')}$ for the query $Q''$ we construct an expression, 
such that for each $i \in {\mathbb N}_+$, this expression will return the number of distinct tuples contributed by the associations in ${\cal C}^{(Q'')}$ to the set $\Gamma_{\bar{S}}^{(t^*_Q)}(Q'',D_{\bar{N}^{(i)}}(Q))$. That is, we construct an expression that, for each $i \in {\mathbb N}_+$, will provide the cardinality of the set $\Gamma^{(i)}[{\cal C}^{(Q'')}]$. (See Section~\ref{multiplicity-monomial-prep-sec} for the notation  $\Gamma^{(i)}[{\cal C}^{(Q'')}]$.)  For each monomial class ${\cal C}^{(Q'')}$ $\in$ $\{ {\cal C}_1^{(Q'')}$, $\ldots,$ ${\cal C}_{n^*}^{(Q'')} \}$, we call the respective expression ``the multiplicity monomial of the monomial class ${\cal C}^{(Q'')}$.'' Each multiplicity monomial is 
a product of  (some powers of) the elements of the noncopy singature 
and of the copy signature of the corresponding monomial class. These multiplicity monomials, together with the copy-signatures and the noncopy-signatures of the monomial classes, are all that will be needed in Section~\ref{putting-together-f-sec} to construct the function ${\cal F}_{(Q)}^{(Q'')}$. 

We begin by introducing the necessary notation. For each term $s$ of the query $Q$ such that $s$ is not a copy variable of $Q$, by $DomLabel_Q(s)$ we denote (i) variable $N_j$ in case where $m \geq 1$ and where $s$ is (a multiset noncopy variable of $Q$, i.e.,) the variable $Y_j$ of $Q$ for some $1 \leq j \leq m$; and (ii) constant value $1$ in case $s$ is either a constant used in $Q$ or is one of the variables $X_1,\ldots,X_{l+u}$ of $Q$. 

Further, for Propositions~\ref{monomial-classes-i-sizes-prop} and \ref{monomial-classes-sizes-prop}, we use the following notation, for ease of reference to the elements of the noncopy signatures and of the copy signatures of the monomial classes. Let $\Xi$ be a $G$-ary vector  of (not necessarily distinct) elements of the set ${\cal S}_{C(Q)}$, such that  the monomial class ${\cal C}^{(Q'')}$ with atom-signature $\Xi$ is not empty. 

(1) Let ${\Phi}_n^{{\cal C}^{(Q'')}}$ be the noncopy-signature of the class ${\cal C}^{(Q'')}$. Then: 

\begin{itemize}
	\item In case $m \geq 1$, denote the elements of ${\Phi}_n^{{\cal C}^{(Q'')}}$, from left to right, as $Z_{1},Z_{2},\ldots,Z_{m}$. For all $j \in \{ 1,\ldots,m \}$, we have that $Z_{j}$ $\in$ $\{ Y_1,$ $\ldots,$ $Y_m,$ $X_1,$ $\ldots,$ $X_{l+u} \}$ $\bigcup$ $P$. 
	
	
	\item For an $i \in {\mathbb N}_+$, denote by $\Pi_{{\Phi}_n^{{\cal C}^{(Q'')}}}^{(i)}$ the value $1$ in case $m = 0$, and the product $\Pi_{j = 1}^m |Dom^{(i)}_Q(Z_{j})|$
	 in case $m \geq 1$.

	\item Finally, denote by $\Pi_{{\Phi}_n^{{\cal C}^{(Q'')}}}$ the value $1$ in case $m = 0$, and the product $\Pi_{j = 1}^m DomLabel_Q(Z_{j})$ in case $m \geq 1$. 
\end{itemize}

(2) Let ${\Phi}_c^{{\cal C}^{(Q'')}}$ be the copy-signature of the class ${\cal C}^{(Q'')}$. Then: 

\begin{itemize}
	\item In case $r \geq 1$, denote the elements of ${\Phi}_c^{{\cal C}^{(Q'')}}$, from left to right, as $W_{1},W_{2},\ldots,W_{r}$. For all $j \in \{ 1,\ldots,r \}$, we have that $W_{j}$ $\in$ $\{ 1,$ $N_{m+1},$ $\ldots,$ $N_{m+w} \}$. 

	\item For an $i \in {\mathbb N}_+$ and for each $j \in \{ 1,\ldots,r \}$  (still assuming $r \geq 1$), denote by $W_{j}^{(i)}$ the value of the variable $W_{j}$ in the vector $\bar{N}^{(i)}$, in case $W_{j} \neq 1$. (That is, whenever  $W_{j}$ $=$ $N_{m+l}$, for some $l \in \{ 1,\ldots,w \}$, then $W^{(i)}_{j}$ $=$ $N^{(i)}_{m+l}$.) If $W_{j} = 1$ then let $W_{j}^{(i)} := 1$. 
	
	\item For an $i \in {\mathbb N}_+$, denote by $\Pi_{{\Phi}_c^{{\cal C}^{(Q'')}}}^{(i)}$ the value $1$ in case $r = 0$, and the product $\Pi_{j = 1}^r W_{j}^{(i)}$ in case $r \geq 1$.   
	
	\item Finally, denote by $\Pi_{{\Phi}_c^{{\cal C}^{(Q'')}}}$ the value $1$ in case $r = 0$, and the product $\Pi_{j = 1}^r W_{j}$ in case $r \geq 1$.   
\end{itemize}

\begin{proposition} 
\label{monomial-classes-i-sizes-prop}
Let $\Xi$ be a $G$-ary vector  of (not necessarily distinct) elements of the set ${\cal S}_{C(Q)}$, such that  the monomial class ${\cal C}^{(Q'')}$ with atom-signature $\Xi$ is not empty. 
Let $i \in {\mathbb N}_+$. Then the cardinality of the set $\Gamma^{(i)}[{\cal C}^{(Q'')}]$ (that is, the number of distinct tuples contributed, to the set $\Gamma^{(t^*_Q)}_{\bar{S}}(Q'',D_{\bar{N}^{(i)}}(Q))$ for the query $Q''$ and for the database $D_{\bar{N}^{(i)}}(Q)$, by all the valid assignment mappings for all the elements of the class  ${\cal C}^{(Q'')}$) is exactly $\Pi_{{\Phi}_n^{{\cal C}^{(Q'')}}}^{(i)}$ $\times$ $\Pi_{{\Phi}_c^{{\cal C}^{(Q'')}}}^{(i)}$. 
\end{proposition} 

Please see Example~\ref{main-proof-ex} for an illustration. \reminder{Add an illustration to Example~\ref{main-proof-ex}} The proof of Proposition~\ref{monomial-classes-i-sizes-prop} is immediate from Proposition~\ref{noncopy-proj-more-properties-prop}. (See Proposition~\ref{fixed-unity-prop} for the details on the $\Pi_{{\Phi}_c^{{\cal C}^{(Q'')}}}^{(i)}$ part of the computation. The $\Pi_{{\Phi}_n^{{\cal C}^{(Q'')}}}^{(i)}$ part of the computation follows from the construction of the database $D_{\bar{N}^{(i)}}(Q)$, specifically from the definition of the main construction cycle as described in Section~\ref{main-cycle-sec}.) 

We note that the expression $\Pi_{{\Phi}_n^{{\cal C}^{(Q'')}}}^{(i)}$ $\times$ $\Pi_{{\Phi}_c^{{\cal C}^{(Q'')}}}^{(i)}$ in Proposition~\ref{monomial-classes-i-sizes-prop} is in terms of only the elements of the vector $\bar{N}^{(i)}$, and is uniform across all $i \in {\mathbb N}_+$. Thus, we obtain the following result as an easy corollary of Proposition~\ref{monomial-classes-i-sizes-prop}. 

\begin{proposition} 
\label{monomial-classes-sizes-prop}
Let $\Xi$ be a $G$-ary vector  of (not necessarily distinct) elements of the set ${\cal S}_{C(Q)}$, such that  the monomial class ${\cal C}^{(Q'')}$ with atom-signature $\Xi$ is not empty. 
Then, for all $i \in {\mathbb N}_+$, the cardinality of the set $\Gamma^{(i)}[{\cal C}^{(Q'')}]$ (that is, the number of distinct tuples contributed, to the set $\Gamma^{(t^*_Q)}_{\bar{S}}(Q'',D_{\bar{N}^{(i)}}(Q))$ for the query $Q''$ and for the database $D_{\bar{N}^{(i)}}(Q)$, by all the valid assignment mappings for all the elements of the class  ${\cal C}^{(Q'')}$) can be computed by substituting the values in the vector $\bar{N}^{(i)}$ (specifically value $N_j^{(i)}$ as value of variable $N_j$, for each $j \in \{ 1,\ldots,m+w \}$) into the formula $\Pi_{{\Phi}_n^{{\cal C}^{(Q'')}}}$ $\times$ $\Pi_{{\Phi}_c^{{\cal C}^{(Q'')}}}$. 
\end{proposition} 

For a monomial class ${\cal C}^{(Q'')}$ with noncopy signature ${\Phi}_n^{{\cal C}^{(Q'')}}$ and with copy signature ${\Phi}_c^{{\cal C}^{(Q'')}}$, such that  ${\cal C}^{(Q'')}$ is not empty, we call the expression (of Proposition~\ref{monomial-classes-sizes-prop}) $\Pi_{{\Phi}_n^{{\cal C}^{(Q'')}}}$ $\times$ $\Pi_{{\Phi}_c^{{\cal C}^{(Q'')}}}$, in terms of the variables in the vector $\bar{N}$, {\em the multiplicity monomial of the monomial class} ${\cal C}^{(Q'')}$.

\subsection{The Wave Monomial of the Query $Q$} 
\label{monomial-class-mappings-sec} 


In this section we obtain results that are instrumental in proving Theorem~\ref{magic-mapping-prop}. Namely, we show that: 
\begin{itemize} 
	\item[(1)] There exists a (nonempty) monomial class ${\cal C}^{(Q)}$ {\em for the query} $Q$, such that the multiplicity monomial of ${\cal C}^{(Q)}$ is ``the wave of the query $Q$.'' (See Proposition~\ref{q-has-wave-prop}.)  The wave of the query $Q$ is defined in this section (see Definition~\ref{the-wave-def}) based on the vector $\bar{N}$ and on the mapping $\nu^{copy}_Q$ defined in Section~\ref{main-cycle-sec}. 
	\item[(2)] Suppose that for a CCQ query $Q''$, there exists a (nonempty) monomial class ${\cal C}^{(Q'')}$, such that the multiplicity monomial of ${\cal C}^{(Q'')}$ is ``the wave of the query $Q$.'' Then there exists a SCVM from  the query $Q''$ to the query $Q$. (See Proposition~\ref{q-same-scale-mpng-prop}.)  
\end{itemize}  
In Section~\ref{q-prime-has-wave-sec}  we will see that whenever (a) $Q' \equiv_C Q$ for a CCQ query $Q'$ and (b) $Q$ is an explicit-wave CCQ query, then there exists a (nonempty) monomial class ${\cal C}_*^{(Q')}$ for the query $Q'$, such that the multiplicity monomial of ${\cal C}_*^{(Q')}$ is ``the wave of the query $Q$.'' The proof of Theorem~\ref{magic-mapping-prop} is immediate from that result and from Propositions~\ref{q-has-wave-prop} and \ref{q-same-scale-mpng-prop}  of this section. 
(We remind the reader that throughout the proof of Theorem~\ref{magic-mapping-prop}, all monomial classes of all queries are defined w.r.t. the family of databases $\{ D_{\bar{N}^{(i)}}(Q) \}$ for the fixed input query $Q$.) 


We begin the exposition by defining ``the wave of the query $Q$.'' We introduce some notation to make the definition concise: 
\begin{itemize} 
	\item[(A)] Denote by ${\cal P}^{(Q)}_{noncopy}$ (i) the constant $1$ in case $m = 0$, and (ii) the product $\Pi_{j=1}^m N_j$ in case $m \geq 1$. 

	\item[(B)] Denote by ${\cal P}^{(Q)}_{copy}$ (i) the constant $1$ in case $r = 0$, and (ii) the product $\Pi_{j=1}^r \nu_Q^{copy}(Y_{m+j})$ in case $r \geq 1$. (For the notation $\nu^{copy}_Q$, see Section~\ref{main-cycle-sec}.)

\end{itemize}  

\begin{definition}{The wave of CCQ query $Q$} 
\label{the-wave-def} 
For a CCQ query $Q$, {\em the wave ${\cal P}^{(Q)}_*$ of} $Q$ w.r.t. the family of databases  $\{ D_{\bar{N}^{(i)}}(Q) \}$ is the product ${\cal P}^{(Q)}_* = {\cal P}^{(Q)}_{noncopy} \times {\cal P}^{(Q)}_{copy}$. 
\end{definition} 

The intuition for the wave ${\cal P}^{(Q)}_*$ of a CCQ query $Q$ is that ${\cal P}^{(Q)}_*$ reflects (i) the association of each multiset noncopy variable of $Q$ (in case $m \geq 1$) with a separate variable among $N_1,$ $\ldots,$ $N_m$, and (ii) the association, in case $r \geq 1$, of each copy-sensitive subgoal, call it $s$, of $Q$ (via the copy variable of the subgoal) with the unique element, call it $s'$, of the set ${\cal S}_{C(Q)}$ (see Section~\ref{basic-queries-sec}) such that the subgoal $s$ and the element $s'$ have the same relational template. The provenance of each association will become explicit  in the proof of Proposition~\ref{q-has-wave-prop}. As an illustration of the definition, in Example~\ref{main-proof-ex} in Section~\ref{main-proof-ex-sec},  the wave of query $Q$ w.r.t. the family of databases  $\{ D_{\bar{N}^{(i)}}(Q) \}$ is the product ${\cal P}^{(Q)}_*$ $=$ $N_1$ $\times$ $N_2$, where $N_1$ refers to the only multiset noncopy variable of the query $Q$, and $N_2$ refers to its only (multiset) copy variable. 



\begin{proposition} 
\label{wave-has-all-n-j-prop} 
Given a CCQ query $Q$ and the vector $\bar{N} = [ \ N_1 \ N_2 \ \ldots \ N_{m+w}]$ that is used to construct the family of databases  $\{ D_{\bar{N}^{(i)}}(Q) \}$ for the query $Q$. Then each element of the vector $\bar{N}$ occurs in the wave of the query $Q$ w.r.t. $\{ D_{\bar{N}^{(i)}}(Q) \}$.  
\end{proposition} 

The proof of Proposition~\ref{wave-has-all-n-j-prop} is immediate from the definition of the products ${\cal P}^{(Q)}_{noncopy}$ and ${\cal P}^{(Q)}_{copy}$ used in Definition~\ref{the-wave-def}. (In case where $r \geq 1$,  the less obvious part of the proof, that is the presence of each of $N_{m+1},N_{m+2},\ldots,N_{m+w}$ in the product ${\cal P}^{(Q)}_{copy}$, is immediate from the definition of the mapping $\nu_Q^{copy}$, see Section~\ref{main-cycle-sec}.) 

Our next result is immediate from the definition of the wave of the query $Q$. (For each expression of the form $N_j^k$, such that $N_j$ $\in$ $\{ N_1,$ $N_2,$ $\ldots,$ $ N_{m+r} \ \}$ and $k \geq 1$, we say that the expression $N_j^k$ has $k$ occurrences of the variable $N_j$.) 

\begin{proposition} 
\label{wave-is-of-power-r-prop} 
Given a CCQ query $Q$ and the vector $\bar{N} = [ \ N_1 \ N_2 \ \ldots \ N_{m+w}]$ that is used to construct the family of databases  $\{ D_{\bar{N}^{(i)}}(Q) \}$ for the query $Q$. Then the wave of the query $Q$ w.r.t. $\{ D_{\bar{N}^{(i)}}(Q) \}$ has exactly $m$ $+$ $r$ occurrences of the variables from the set $\{ N_1,$ $N_2,$ $\ldots,$ $ N_{m+w} \ \}$. 
\end{proposition} 

\begin{proposition}
\label{q-has-wave-prop} 
Given a CCQ query $Q$, there exists a nonempty monomial class, call it ${\cal C}_*^{(Q)}$, for the query $Q$ w.r.t. the family of databases  $\{ D_{\bar{N}^{(i)}}(Q) \}$, such that the multiplicity monomial of ${\cal C}_*^{(Q)}$ is the wave of the query $Q$ w.r.t. $\{ D_{\bar{N}^{(i)}}(Q) \}$.  
\end{proposition} 

\begin{proof} 
The proof has three parts: 
\begin{itemize} 

	\item[(1)] We first show that for each $i \in {\mathbb N}_+$, there exists an association for the query $Q$ and for the database $D_{\bar{N}^{(i)}}(Q)$, call this association ${\cal A}^{(i)}_*$, such that: 

	\begin{itemize} 

		\item[(i)] The association ${\cal A}^{(i)}_*$ has a $t^*_Q$-valid assignment mapping for the query $Q$ and for the database $D_{\bar{N}^{(i)}}(Q)$; 

		\item[(ii)] In case $m \geq 1$, we have that the noncopy-signature\footnote{Observe that the noncopy-signature of the association ${\cal A}^{(i)}_*$ is well defined, by ${\cal A}^{(i)}_*$ having a valid assignment mapping for the query $Q$ and for the database $D_{\bar{N}^{(i)}}(Q)$.} $\Phi_n[{\cal A}^{(i)}_*]$ of the association ${\cal A}^{(i)}_*$ is the vector $[ Y_1$ $\ \ldots$ $\ Y_m]$; and, finally, 

		\item[(iii)] In case $r \geq 1$, we have that the copy-signature $\Phi_c[{\cal A}^{(i)}_*]$ of the association ${\cal A}^{(i)}_*$ is the vector $[ \nu_Q^{copy}$ $(Y_{m+1})$ $\ \ldots$ $\ \nu_Q^{copy}$ $(Y_{m+r}) ]$. 

	\end{itemize} 
	 
	\item[(2)] We then show that the query $Q$, w.r.t. the family of databases  $\{ D_{\bar{N}^{(i)}}(Q) \}$, has a nonempty monomial class whose noncopy-signature (whose copy-signature, respectively) is the noncopy-signature (the copy-signature, respectively) of the associations ${\cal A}^{(i)}_*$, for all  $i \in {\mathbb N}_+$, of item (1) of this proof. We denote this monomial class by ${\cal C}_*^{(Q)}$. 

	\item[(3)] Finally, we show that the multiplicity monomial of the monomial class ${\cal C}_*^{(Q)}$ of item (2) is the wave of the query $Q$. 

\end{itemize} 

In fact, items (2) and (3) are straightforward: Item (2) is immediate from the definition of monomial classes and from item (1), and item (3) is immediate from item (2) and from Proposition~\ref{monomial-classes-sizes-prop}. Hence, in the remainder of this proof we prove parts (i) through (iii) of item (1) above. 

Recall that we assume $m + r \geq 1$. Hence the set ${\cal S}_{C(Q)}$ (see Section~\ref{basic-queries-sec}) for the query $Q$ is not empty. 

Fix an $i \in {\mathbb N}_+$. Recall the set ${\cal S}^{(i)}$ $\neq$ $\emptyset$ introduced in Section~\ref{main-cycle-sec} to construct the database $D_{\bar{N}^{(i)}}(Q)$. Fix an arbitrary tuple $t \in$ ${\cal S}^{(i)}$. For the tuple $t \in$ ${\cal S}^{(i)}$, Section~\ref{main-cycle-sec} defined a mapping, $\nu_t$, from all the terms of the query $Q$ to constants in the set $adom(D_{\bar{N}^{(i)}}(Q))$ $\bigcup$ ${\mathbb N}_+$. By definition, for the $t \in$ ${\cal S}^{(i)}$ we have that the restriction of $\nu_t$ to all the terms of the query $Q$ occurring in the elements of the set ${\cal S}_{C(Q)}$ induces a bijection from the subset ${\cal S}_{C(Q)}$ of the condition of $Q$ to a set, call it $D_t$, of ground atoms of the database $D_{\bar{N}^{(i)}}(Q)$. By construction of the database $D_{\bar{N}^{(i)}}(Q)$, the set $D_t$ (i) was generated from ${\cal S}_{C(Q)}$ using the mapping $\nu_t$, and (ii) is a copy-neutral canonical database for the query $Q$. 

We now construct the association ${\cal A}^{(i)}_*$. We begin by associating each atom $s$ $\in$ ${\cal S}_{C(Q)}$ with its image (in the set of ground atoms $D_t$ $\subseteq$ $D_{\bar{N}^{(i)}}(Q)$) under $\nu_t$. Now there are two cases: (a) In case all the subgoals of the (regularized version of the) query $Q$ are elements of the set ${\cal S}_{C(Q)}$, we are done with the construction of the association ${\cal A}^{(i)}_*$. We now consider the remaining case (b), where there exist subgoals of the (regularized version of the) query $Q$ that are not elements of the set ${\cal S}_{C(Q)}$. Consider an arbitrary subgoal $s$ of $Q$ such that $s$ $\in$ $(L$ $-$ ${\cal S}_{C(Q)})$, where $L$ is the condition of the regularized version of the query $Q$. By definition of the set ${\cal S}_{C(Q)}$, $s$ is a copy-sensitive atom, such that there exists a unique element, call it $s'$, of the set ${\cal S}_{C(Q)}$, such that $s$ and $s'$ have the same relational template. 

Then in our construction of the association ${\cal A}^{(i)}_*$, for each such subgoal $s$ of the query $Q$, $s$ $\in$ $(L$ $-$ ${\cal S}_{C(Q)})$, we associate (in ${\cal A}^{(i)}_*$) the atom $s$ with the atom $\nu_t(s')$ $\in$ $D_{\bar{N}^{(i)}}(Q)$, for the $s'$ as determined in the previous paragraph. This completes the construction of the association ${\cal A}^{(i)}_*$. Observe that by construction, in both cases (a) and (b) as in the preceding paragraphs, the association ${\cal A}^{(i)}_*$ associates all the elements of the condition of the query $Q$ with exactly the set of ground atoms $D_t$ $\subseteq$ $D_{\bar{N}^{(i)}}(Q)$. 

We now prove claim (1)(i) of the beginning of this proof. We first show that the association ${\cal A}^{(i)}_*$ has a valid assignment mapping for the query $Q$ and for the database $D_{\bar{N}^{(i)}}(Q)$. Indeed, by definition of $\nu_t$ it holds that $\nu_t$ assigns values to {\em all} terms of the query $Q$, consistently across all the pairs in the association ${\cal A}^{(i)}_*$. Denote by $\theta^{(Q)}_*$ the resulting valid assignment mapping for the query $Q$ and for the database $D_{\bar{N}^{(i)}}(Q)$. Now, it is immediate from the definition of $\nu_t$ that the restriction of the mapping $\theta^{(Q)}_*$ to the head vector $[X_1 \ \ldots \ X_l]$ of the query $Q$ is the tuple $t^*_Q$. Thus, $\theta^{(Q)}_*$ is a $t^*_Q$-valid assignment mapping for the query $Q$ and for the database $D_{\bar{N}^{(i)}}(Q)$, which completes our proof of  claim (1)(i). 

We now prove claim (1)(ii) of the beginning of this proof. This claim requires the assumption that $m \geq 1$. Under this assumption, by definition of $\nu_t$ we have that $\nu_t$ maps the variable $Y_j$, for each $j \in \{ 1,\ldots,m \}$, into an element of the set $S_j^{(i)}$. (See Section~\ref{nu-sec} for the definition of $S_j^{(i)}$.) Hence, by definition of the vector $\Phi_n[{\cal A}^{(i)}_*]$ (see Section~\ref{assoc-defs-sec}), the $j$th element of  $\Phi_n[{\cal A}^{(i)}_*]$ is the variable $Y_j$, for each $j \in \{ 1,\ldots,m \}$. Q.E.D. 

To complete the proof of Proposition~\ref{q-has-wave-prop}, it remains to prove claim (1)(iii) of the beginning of this proof. This claim requires the assumption that $r \geq 1$. Under this assumption, the claim is immediate from the construction of the association ${\cal A}^{(i)}_*$ and from the definition of the vector $\Phi_c[{\cal A}^{(i)}_*]$ (see Section~\ref{assoc-defs-sec}). Q.E.D. 
\end{proof} 

\begin{proposition} 
\label{q-same-scale-mpng-prop} 
Given  CCQ queries $Q(\bar{X}) \leftarrow L,M$ and $Q''(\bar{X}'') \leftarrow L'',M''$, such that (i) $Q$ and $Q''$ have the same (nonnegative-integer) head arities, (ii) $|M_{copy}|$ $=$ $|M''_{copy}|$, and (iii) $|M_{noncopy}|$ $=$ $|M''_{noncopy}|$. Suppose that there exists a nonempty monomial class ${\cal C}^{(Q'')}$ for the query $Q''$  w.r.t. the family of databases $\{ D_{\bar{N}^{(i)}}(Q) \}$, such that the multiplicity monomial of ${\cal C}^{(Q'')}$ is the wave of the query $Q$  w.r.t. $\{ D_{\bar{N}^{(i)}}(Q) \}$. Then there exists a SCVM from  the query $Q''$ to the query $Q$.   
\end{proposition} 

The proof of Proposition~\ref{q-same-scale-mpng-prop} is constructive: That is, the proof generates a SCVM from  the query $Q''$ to the query $Q$ of the statement of Proposition~\ref{q-same-scale-mpng-prop}. 

\begin{proof} 
We are given that there exists a nonempty monomial class ${\cal C}^{(Q'')}$ for the query $Q''$  w.r.t. the family of databases $\{ D_{\bar{N}^{(i)}}(Q) \}$, such that the multiplicity monomial of ${\cal C}^{(Q'')}$ is the wave of the query $Q$  w.r.t. $\{ D_{\bar{N}^{(i)}}(Q) \}$. Then, by definition of multiplicity monomials (Section~\ref{multiplicity-monomial-sec}) and of copy-/noncopy-sigature vectors (Section~\ref{assoc-defs-sec}), we have that: 
\begin{itemize} 
	\item In case $m \geq 1$, the vector $\Phi_n[{\cal C}^{(Q'')}]$ must be a permutation of the vector $[ Y_1 \ \ldots \ Y_m ]$; and 
	\item In case $r \geq 1$, the vector $\Phi_c[{\cal C}^{(Q'')}]$ must be a permutation of the vector $[ \nu^{copy}_Q(Y_{m+1})$ $\ $ $\ldots$ $\ $ $\nu^{copy}_Q(Y_{m+r}) ]$. 
\end{itemize} 

By definition of monomial classes and from the fact that the monomial class ${\cal C}^{(Q'')}$ for the query $Q''$  w.r.t. the family of databases $\{ D_{\bar{N}^{(i)}}(Q) \}$ is not empty (see Section~\ref{total-set-of-assoc-sec} for the relevant results), we have that for each $i \in {\mathbb N}_+$, the monomial class ${\cal C}^{(Q'')}$ has an association with at least one $t^*_Q$-valid assignment mapping for $Q''$ and $\{ D_{\bar{N}^{(i)}}(Q) \}$. Fix an arbitrary $i \in {\mathbb N}_+$ and consider an arbitrary such association in ${\cal C}^{(Q'')}$. Denote the association by ${\cal A}^{(init)}_*$, and denote by $\theta^{(init)}_*$ its unity $t^*_Q$-valid assignment mapping for $Q''$ and $\{ D_{\bar{N}^{(i)}}(Q) \}$; the mapping $\theta^{(init)}_*$ exists by Proposition~\ref{same-head-unity-prop}. 

By Proposition~\ref{Observation_three_five}, in case the association ${\cal A}^{(init)}_*$ is such that the mapping $\theta^{(init)}_*$ induces a mapping from the subgoals of the query $Q''$ into two or more copy-neutral canonical databases (in $D_{\bar{N}^{(i)}}(Q)$) for the query $Q$, there must exist an association, call it  ${\cal A}_*$, that has the same atom-signature as  ${\cal A}^{(init)}_*$ and such that the unity valid assignment mapping for ${\cal A}_*$, call this mapping $\theta_*$, induces a mapping from the subgoals of the query $Q''$ into a single copy-neutral canonical database (in $D_{\bar{N}^{(i)}}(Q)$) for the query $Q$, call this database $D_*$. Observe that from the fact that ${\cal A}^{(init)}_*$ and ${\cal A}_*$ have the same  atom-signature, we have that ${\cal A}_*$ belongs to the monomial class ${\cal C}^{(Q'')}$, just as  ${\cal A}^{(init)}_*$ does. In addition, ${\cal A}^{(init)}_*$ and ${\cal A}_*$ have the same copy-signature (which is $\Phi_c[{\cal C}^{(Q'')}]$), as well as the same noncopy-signature (which is $\Phi_n[{\cal C}^{(Q'')}]$). 

If, on the other hand, the association  ${\cal A}^{(init)}_*$ is such that the mapping $\theta^{(init)}_*$ induces a mapping from the subgoals of the query $Q''$ into a single copy-neutral canonical database (in $D_{\bar{N}^{(i)}}(Q)$) for the query $Q$, call this database $D_*$, then for the remainder of the proof we refer to ${\cal A}^{(init)}_*$ as ${\cal A}_*$, and refer to $\theta^{(init)}_*$ as $\theta_*$. 

Now denote by $\nu_{Q''}$ the mapping (i) whose domain in the set of all the terms of the query $Q''$ that are not copy variables of $Q''$, and (ii) such that on the entire domain of $\nu_{Q''}$, the mapping $\nu_{Q''}$ coincides with the mapping $\theta_*$. Further, define $\mu'_{Q''}$ as the composition $\nu^{(i)}_Q \circ \nu_{Q''}$ of the mapping $\nu_{Q''}$ with the mapping $\nu^{(i)}_Q$ defined in Section~\ref{nu-sec}. By definition, $\mu'_{Q''}$ is a mapping from all the terms of the query $Q''$ that are not copy variables of $Q''$ to terms of the query $Q$. Finally, define $\mu_{Q''}$ as the mapping (i) whose domain is the set of all terms of the query $Q''$ (that is, including all the copy variables of $Q''$), and (ii) such that on the entire domain of $\mu'_{Q''}$, the mapping $\mu_{Q''}$ coincides with the mapping $\mu'_{Q''}$. Observe that in case $r$ $=$ $0$, the mapping $\mu_{Q''}$ is fully specified and is unique. 

It remains to define $\mu_{Q''}$ on the copy variables of the query $Q''$, in case, which we refer to as (iii), where $r \geq 1$. In this case, for each $j \in \{ 1,\ldots, r\}$, we define $\mu_{Q''}(Y''_{m+j})$ as follows: Suppose the $j$th element of the vector  $\Phi_c[{\cal C}^{(Q'')}]$, being the variable (in the vector $\bar  N$) $N_{m+k}$ for some $1 \leq k \leq w$,\footnote{By definition of the wave of the query $Q$, the vector $\Phi_c[{\cal C}^{(Q'')}]$ cannot contain the constant $1$.}  occurs in the vector $\Phi_c[{\cal C}^{(Q'')}]$ a total of $n$ times, where $1 \leq n \leq r$. Suppose further that out of these $n$ positions in which this fixed variable $N_{m+k}$ occurs in the vector $\Phi_c[{\cal C}^{(Q'')}]$, our fixed position $j$ is the $l$th such position from the left, $1 \leq l \leq n$. Then, by definition of the wave of the query $Q$, it must be that for the equivalence class, call it ${\cal C}(Y_{m+k})$, for the {\em same} value $k$ as above (i.e., the $k$ in $Y_{m+k}$ is the same as in the $N_{m+k}$), of the copy-sensitive subgoal $s$ of $Q$ where the copy variable of $s$ is $Y_{m+k}$, the class ${\cal C}(Y_{m+k})$ has exactly $n$ copy-sensitive subgoals of the query $Q$. All these $n$ subgoals of the query $Q$ have the same relational template, but have distinct copy variables, call the assortment of these copy variables $Y_{i1},$ $Y_{i2},$ $\ldots,$ $Y_{in}$, with $i1 < i2 < \ldots < in$. (Naturally, the variable $Y_{m+k}$ is one of these $n$ copy variables $Y_{i1},$ $Y_{i2},$ $\ldots,$ $Y_{in}$.) Then define  $\mu_{Q''}(Y''_{m+j})$ to be the copy variable $Y_{il}$ of the query $Q$, where $Y_{il}$ is the $l$th variable in the list $($ $Y_{i1},$ $Y_{i2},$ $\ldots,$ $Y_{in}$ $)$. Note that this assignment algorithm terminates and results in the same assignments, in $\mu_{Q''}$, for all copy variables of the query $Q''$ independently of the order in  which we choose the positions $j$ out of the set $\{ 1,\ldots,r \}$. Observe also that  $\mu_{Q''}$ is still a mapping once we are done with all the assignments in (iii). (Indeed, each copy variable of the query $Q''$, in case $r \geq 1$, is assigned by $\mu_{Q''}$ to a distinct copy variable of the query $Q$.) 
Finally, observe that in case $r = w \geq 1$, the mapping  $\mu_{Q''}(Y''_{m+j})$  is defined by (iii), for each $j \in \{ 1,\ldots, r\}$, as the copy variable $Y_{m+k}$ of query $Q$, where $k$ is such that $N_{m+k}$ is the $j$th element of the vector  $\Phi_c[{\cal C}^{(Q'')}]$. 

We now show that properties (1) through (5) of same-scale covering mappings (SCVMs) in Definition~\ref{magic-mapping-def} are satisfied by  the mapping $\mu_{Q''}$. (Hence, we conclude that the mapping $\mu_{Q''}$ is a SCVM from the query $Q''$ to the query $Q$.) 

\begin{itemize} 
	\item[(1)] This property in Definition~\ref{magic-mapping-def} is satisfied by the mapping $\mu_{Q''}$ due to the fact that the association ${\cal A}_*$ has a (unity) valid assignment mapping, by definition of valid assignment mappings,  and by definition of the mapping $\nu^{(i)}_Q$ used in the construction of the mapping $\mu_{Q''}$. 
	\item[(2)] This property in Definition~\ref{magic-mapping-def} is satisfied by the mapping $\mu_{Q''}$ due to the fact that the association ${\cal A}_*$ has a (unity) $t^*_Q$-valid assignment mapping, by definition of $t^*_Q$-valid assignment mappings, and by definition of the mapping $\nu^{(i)}_Q$  used in the construction of the mapping $\mu_{Q''}$. 
	\item[(3)] This property in Definition~\ref{magic-mapping-def} is satisfied by the mapping $\mu_{Q''}$ due to the facts that: 
		\begin{itemize}
			\item The noncopy-signature of the association  ${\cal A}_*$ is (in case $m \geq 1$) a permutation of the list $[ Y_1 \ \ldots \ Y_m ]$, and by definition of the mapping $\mu'_{Q''}$ (and hence also of the mapping $\mu_{Q''}$) on the set of multiset noncopy variables of the query $Q''$; thus, $\mu_{Q''}$ maps the set of multiset noncopy variables of the query $Q''$ {\em onto} the (same-cardinality) set of multiset noncopy variables of the query $Q$; and 
			\item The copy-signature of the association  ${\cal A}_*$ is (in case $r \geq 1$) a permutation of the list $[ \nu^{copy}_Q(Y_{m+1})$ $\ $ $\ldots$ $\ $ $\nu^{copy}_Q(Y_{m+r}) ]$, and by definition of the mapping $\mu_{Q''}$ on the set of copy variables of the query $Q''$; thus, $\mu_{Q''}$ maps the set of copy variables of the query $Q''$ {\em onto} the (same-cardinality) set of copy variables of the query $Q$. 
		\end{itemize}
	\item[(4)] This property in Definition~\ref{magic-mapping-def} is satisfied by the mapping $\mu_{Q''}$ due to the fact that, by its definition, mapping $\mu_{Q''}$ maps each relational subgoal of the query $Q''$ into a unique element of the set ${\cal S}_{C(Q)}$ of subgoals of the query $Q$.  
	\item[(5)] This property in Definition~\ref{magic-mapping-def} is satisfied by the mapping $\mu_{Q''}$ due to the facts that: 
		\begin{itemize}
			\item by its definition, mapping $\mu_{Q''}$ maps each copy-sensitive subgoal of the query $Q''$ into a subgoal of the query $Q$; and   
			\item the copy-signature of the association  ${\cal A}_*$ does not have occurrences of the constant $1$, hence the mapping $\mu_{Q''}$ maps each copy-sensitive subgoal of the query $Q''$ into a {\em copy-sensitive} subgoal of the query $Q$. 
		\end{itemize} 
	
\end{itemize} 

We conclude that $\mu_{Q''}$ is, by construction, a SCVM from the query $Q''$ to the query $Q$. 
\end{proof}

\subsection{Extended Example: Basic Notation and Constructs} 
\label{main-proof-ex-sec} 

In this section we provide an extended example that illustrates the notions and constructions introduced in  Sections~\ref{basics-sec} through \ref{monomial-class-mappings-sec} 
%
of the proof of Theorem~\ref{magic-mapping-prop}. The example uses three CCQ queries, $Q$, $Q'$, and $Q''$; each of the queries is an explicit-wave query by part (1) of Definition~\ref{expl-wave-def}. 
By the results in this paper, for the queries $Q$ and $Q'$ of Example~\ref{main-proof-ex} we have that $Q \equiv_C Q'$. In the beginning of the example, 
we exhibit a SCVM from $Q'$ to $Q$. (The existence of the mapping is stipulated by Theorem~\ref{magic-mapping-prop}.) 
At the same time, it is easy to ascertain that there does not exist a SCVM from the query $Q''$ to the query $Q$ of  Example~\ref{main-proof-ex}. Thus, by Theorem~\ref{magic-mapping-prop}, $Q \equiv_C Q''$ cannot hold for the queries $Q''$ and $Q$ of  Example~\ref{main-proof-ex}. We build on this example a little later 
 (see Example~\ref{main-proof-cont-ex} in Section~\ref{main-proof-cont-ex-sec}), to show how to use the proof of Theorem~\ref{magic-mapping-prop} to construct a counterexample database to $Q \equiv_C Q''$. \reminder{I have already checked that I really show {\em and discuss!} the counterexample database in the example} 
 
At the end of Example~\ref{main-proof-ex}, we also illustrate the constructs of Section~\ref{monomial-class-mappings-sec}, by discussing ``the wave'' of the query $Q$ (see Definition~\ref{the-wave-def} in Section~\ref{monomial-class-mappings-sec}) and the monomial classes of the queries $Q$ and $Q'$ that ``have the wave'' of $Q$. We also show that query $Q''$ does not ``have the wave'' of the query $Q$, and discuss the implications of this fact.

\begin{example}
\label{main-proof-ex} 
Let CCQ queries $Q$, $Q'$, and $Q''$ be as follows. 
\begin{tabbing}
$Q(X_1) \leftarrow p(X_1,Y_1), p(X_1, X_2; Y_2), \{ Y_1, Y_2 \} .$ \\
$Q'(X'_1) \leftarrow p(X'_1,Y'_1), p(X'_1, X'_2; Y'_2), p(X'_1,X'_3), \{ Y'_1, Y'_2 \} .$ \\
$Q''(X''_1) \leftarrow p(X''_1,X''_2), p(X''_1, Y''_1; Y''_2), \{ Y''_1, Y''_2 \} .$ 
\end{tabbing} 
Observe that by each of the three queries having exactly one copy-sensitive subgoal, each of $Q$, $Q'$, and $Q''$ is an explicit-wave query. (See part (1) of Definition~\ref{expl-wave-def}.) 

By Theorem~\ref{cvm-containm-thm} and by the existence of a SCVM from $Q$ to $Q'$ and of another SCVM from $Q'$ to $Q$, we have that $Q \equiv_C Q'$. A SCVM $\mu$ from $Q'$ to $Q$ is $\mu = \{ X'_1 \rightarrow X_1, Y'_1 \rightarrow Y_1, X'_2 \rightarrow X_2, Y'_2 \rightarrow Y_2,  X'_3 \rightarrow X_2     \}$. 

It is easy to see that there does not exist a SCVM from $Q''$ to $Q$. (Indeed, for each mapping from the terms of $Q''$ to the terms of $Q$, the mapping violates at least one of conditions (3) through (5) of Definition~\ref{magic-mapping-def}.) Thus, by Theorem~\ref{magic-mapping-prop}, $Q \equiv_C Q''$ cannot hold. Later, we build on this example (see Example~\ref{main-proof-cont-ex} in Section~\ref{main-proof-cont-ex-sec}) to show how to use the proof of Theorem~\ref{magic-mapping-prop} to construct a counterexample database to $Q \equiv_C Q''$. 

\mbox{}

We now use queries $Q$ and $Q''$ to illustrate the notation and constructions of the proof of Theorem~\ref{magic-mapping-prop}, sequentially by subsections of the proof. 

\paragraph{Constructing Database $D_{\bar{N}^{(i)}}(Q)$ for $\bar{N}^{(i)}$ $=$ $[2 \ 3]$} 
We first use the notation introduced in Section~\ref{basics-sec} of the proof of Theorem~\ref{magic-mapping-prop}. We have that $m = |M_{noncopy}| = |\{ Y_1 \}| = 1$, and that $r = |M_{copy}| = |\{ Y_2 \}| = 1$. The set ${\cal S}_{C(Q)}$ of the representative-element subgoals of the query $Q$ is ${\cal S}_{C(Q)} = \{ p(X_1,Y_1), p(X_1, X_2; Y_2) \}$, with $w = 1$. The reason is, the only relational subgoal of $Q$, call this subgoal $h_1$, is the representative element of the equivalence class $\{ h_1 \}$, and the only copy-sensitive subgoal of $Q$, call this subgoal $h_2$, is the representative element of the  equivalence class $\{ h_2 \}$. 

We now follow Section~\ref{db-constr-sec} of the proof of Theorem~\ref{magic-mapping-prop}, to illustrate the construction of a database in the family  $\{ D_{\bar{N}^{(i)}}(Q) \}$ for the query $Q$. We define mapping $\nu_0 = \{ X_1 \rightarrow a, X_2 \rightarrow b \}$, for distinct constants $a$ and $b$. Then we have that $S_0 = \{ a, b \}$, and that $t^*_Q = ( a )$. As $m+w = 2$ for the query $Q$, the vector $\bar{N}$ for $Q$ comprises two variables, $N_1$ (intuitively for the multiset noncopy variable $Y_1$ of $Q$) and $N_2$  (intuitively for the copy variable $Y_2$ of $Q$). Let $i$ be a fixed natural number (i.e., we treat $i$ as the same constant throughout this example), and let the vector $\bar{N}^{(i)} = [2 \ 3]$. That is, $N_1^{(i)} = 2$, and $N_2^{(i)} = 3$. For two distinct constants $c$ and $d$, such that $c$ and $d$ are also distinct from the constants $a$ and $b$ used above to form the set $S_0$, let $S^{(i)}_1 = \{ c,d \}$; this set, of cardinality  $N_1^{(i)}$, provides the domain (in the database) of the multiset noncopy variable $Y_1$ of $Q$. Then we have, by the definitions in Section~\ref{db-constr-sec}, that: 
\begin{itemize} 
	\item $S^{(i)}_* = S_0 \bigcup S^{(i)}_1$; 
	\item $\nu^{(i)}_Q = \{ a \rightarrow X_1, b \rightarrow X_2, c \rightarrow Y_1, d \rightarrow Y_1 \}$; 
	\item $\nu^{copy}(Y_2) = N_2^{(i)} = 3$; and 
	\item $\nu^{copy}_Q(Y_2) = N_2$.  
\end{itemize}

For the set ${\cal S}^{(i)} = \{ (c), (d) \}$, we have that $\nu_{(c)} = \{ X_1 \rightarrow a, X_2 \rightarrow b, Y_1 \rightarrow c, Y_2 \rightarrow 3  \}$ and that $\nu_{(d)} = \{ X_1 \rightarrow a, X_2 \rightarrow b, Y_1 \rightarrow d, Y_2 \rightarrow 3  \}$. We use mappings $\nu_{(c)}$ and $\nu_{(d)}$ each in one iteration of the main construction cycle for the database $D_{\bar{N}^{(i)}}(Q)$. The mapping $\nu_{(c)}$ applied to the two atoms in the set ${\cal S}_{C(Q)}$ results in ground atoms $p(a,c;1)$ and $p(a,b;3)$, and the mapping $\nu_{(d)}$ applied to the set ${\cal S}_{C(Q)}$ results in ground atoms $p(a,d;1)$ and (again) $p(a,b;3)$. Therefore, by construction we have the database $D_{\bar{N}^{(i)}}(Q)$ $=$ $\{ \ p(a,c;1),$ $p(a,b;3),$ $p(a,d;1) \ \}$. We will refer to the ground atom $p(a,c;1)$ in the database $D_{\bar{N}^{(i)}}(Q)$ as $d_1$, to the atom  $p(a,b;3)$ as $d_2$, and to the atom $p(a,d;1)$ as $d_3$.

\paragraph{Construction of the Terms for ${\cal F}_{(Q)}^{(Q'')}$}
We now follow Sections~\ref{valid-map-sec} through 
\ref{monomial-class-mappings-sec} 
of the proof of Theorem~\ref{magic-mapping-prop}, to illustrate the construction of the terms for the function ${\cal F}_{(Q)}^{(Q'')}$, for the query $Q''$ and for the database  $D_{\bar{N}^{(i)}}(Q)$ as constructed above in this example.

The number $G$ of subgoals of the query $Q''$ is $G = 2$. Denote by $g_1$ the copy-sensitive subgoal $p(X''_1, Y''_1; Y''_2)$ of $Q''$, and by $g_2$ the relational subgoal $p(X''_1,X''_2)$ of the query. In query $Q$, denote by $h_1$ the subgoal $p(X_1,Y_1)$ and by $h_2$ the subgoal $p(X_1, X_2; Y_2)$. 

There are nine associations between the $G = 2$ subgoals of the query $Q''$ and the three ground atoms ($d_1$, $d_2$, $d_3$) of the database $D_{\bar{N}^{(i)}}(Q)$. We list all the associations in this table:

\noindent
 \begin{tabular}{llllll}
{\bf ID } & \hspace{-0.4cm} {\bf DB} & 
{\bf $\Psi_a[{\cal A}_j]$} & 
{\bf $\Phi_{n}$} & 
{\bf $\Phi_c$}  & {\bf $\Gamma_{\bar{S}}^{(t^*_Q)}(Q'',D_{\bar{N}^{(i)}}(Q))$} \\
\hline
\hline

${\cal A}_1$ & \hspace{-0.4cm}  $[d_1,d_1]$ & $[h_1,h_1]$ & $Y_1$ & $1$ & $(a,c,1)$ \\

${\cal A}_2$ & \hspace{-0.4cm}  $[d_1,d_2]$ & $[h_1,h_2]$ & $Y_1$ & $1$ & $(a,c,1)$ \\

${\cal A}_3$ & \hspace{-0.4cm}  $[d_1,d_3]$ & $[h_1,h_1]$ & $Y_1$ & $1$ & $(a,c,1)$ \\

${\cal A}_4$ & \hspace{-0.4cm}  $[d_2,d_1]$ & $[h_2,h_1]$ & $X_2$ & $N_2$ & $(a,b,1), (a,b,2),(a,b,3)$ \\ 

${\cal A}_5$ & \hspace{-0.4cm}  $[d_2,d_2]$ & $[h_2,h_2]$ & $X_2$ & $N_2$ & $(a,b,1), (a,b,2),(a,b,3)$ \\ 

${\cal A}_6$ & \hspace{-0.4cm}  $[d_2,d_3]$ & $[h_2,h_1]$ & $X_2$ & $N_2$ & $(a,b,1), (a,b,2),(a,b,3)$ \\ 

${\cal A}_7$ & \hspace{-0.4cm}  $[d_3,d_1]$ & $[h_1,h_1]$ & $Y_1$ & $1$ & $(a,d,1)$ \\

${\cal A}_8$ & \hspace{-0.4cm}  $[d_3,d_2]$ & $[h_1,h_2]$ & $Y_1$ & $1$ & $(a,d,1)$ \\

${\cal A}_9$ & \hspace{-0.4cm}  $[d_3,d_3]$ & $[h_1,h_1]$ & $Y_1$ & $1$ & $(a,d,1)$ \\

\hline
\hline
\end{tabular}

The columns of the table, from left to right, refer to: 
\begin{enumerate}
	\item Association ID, ${\cal A}_j$, for each of the associations ${\cal A}_1$ through ${\cal A}_9$ between query $Q''$ and database $D_{\bar{N}^{(i)}}(Q)$; 
	\item List of those ground atoms of the database that are associated by ${\cal A}_j$ with the subgoals of $Q''$; this list is to be read as ``the first item in the list is associated by ${\cal A}_j$ with subgoal $g_1$ of $Q''$,'' and ``the second item in the list is associated by ${\cal A}_j$ with subgoal $g_2$ of $Q''$;'' 
	\item Atom-signature $\Psi_a[{\cal A}_j]$ of the association ${\cal A}_j$; this list is to be read as ``the first item in the list is associated by ${\cal A}_j$ with subgoal $g_1$ of $Q''$,'' and ``the second item in the list is associated by ${\cal A}_j$ with subgoal $g_2$ of $Q''$;'' 
	\item Noncopy-signature $\Phi_{n}[{\cal A}_j]$ of the association ${\cal A}_j$; 
	\item Copy-signature $\Phi_{c}[{\cal A}_j]$ of the association ${\cal A}_j$; and 
	\item All the tuples  contributed by the association ${\cal A}_j$ to the set  $\Gamma_{\bar{S}}^{(t^*_Q)}(Q'',D_{\bar{N}^{(i)}}(Q))$. (We assume that the columns of the relation $\Gamma_{\bar{S}}^{(t^*_Q)}(Q'',D_{\bar{N}^{(i)}}(Q))$ are, from left to right, $X''_1,$ $Y''_1$, and $Y''_2$.)  
\end{enumerate} 

For instance, the next-to-last row of the table is to be read as follows: Association ${\cal A}_8$ for the query $Q''$ and for the database $D_{\bar{N}^{(i)}}(Q)$ as defined above, associates subgoal $g_1$ of $Q''$ with atom $d_3$ of $D_{\bar{N}^{(i)}}(Q)$, and associates subgoal $g_2$ of $Q''$ with atom $d_2$ of $D_{\bar{N}^{(i)}}(Q)$. Therefore, the atom-signature $\Psi_a[{\cal A}_8]$ associates $g_1$ with subgoal $h_1$ of the query $Q$, and associates $g_2$ with subgoal $h_2$ of $Q$. The noncopy-signature of ${\cal A}_8$ maps the multiset noncopy variable $Y''_1$ of the query $Q''$ to the multiset noncopy variable $Y_1$ of the query $Q$, and the copy-signature of ${\cal A}_8$ maps the copy variable $Y''_2$ of the query $Q''$ to the ``copy value'' $1$ of the relational subgoal $h_1$ of the query $Q$. Finally, association ${\cal A}_8$ contributes tuple $(a,d,1)$ to the set $\Gamma_{\bar{S}}^{(t^*_Q)}(Q'',D_{\bar{N}^{(i)}}(Q))$. 

The construction of the table uses the notation and definitions of Section~\ref{assoc-defs-basics-sec}: The mapping $\psi^{gen(Q)}_{\bar{N}^{(i)}}$, as induced by the mapping $\nu^{(i)}_Q$, is defined as $\psi^{gen(Q)}_{\bar{N}^{(i)}} = \{ d_1 \rightarrow h_1, d_2 \rightarrow h_2, d_3 \rightarrow h_1 \}$. Then the atom-signature of, say, association ${\cal A}_8$ is computed as the vector with first element $\psi^{gen(Q)}_{\bar{N}^{(i)}}[d_3] = h_1$ and with second  element $\psi^{gen(Q)}_{\bar{N}^{(i)}}[d_2] = h_2$. 

The computation of the noncopy-signature $\Phi_{n}[{\cal A}_j]$ for each association ${\cal A}_j$ uses an arbitrary valid assignment mapping, call it $\theta$, for $Q''$, $D_{\bar{N}^{(i)}}(Q)$, and ${\cal A}_j$, as well as the mapping $\nu^{(i)}_Q$ defined earlier in this example. 
Then the noncopy-signature of, say, association ${\cal A}_8$ is computed as the unary (because $m = 1$) vector $\Phi_{n}[{\cal A}_8]$ = $[\nu^{(i)}_Q(\theta_8(Y''_1))] = [\nu^{(i)}_Q(d)] = Y_1$. The reason is, ${\cal A}_8$ generates a unique valid assignment mapping $\theta_8 = \{ X''_1 \rightarrow a, Y''_1 \rightarrow d, Y''_2 \rightarrow 1, X''_2 \rightarrow b  \}$ for $Q''$ and $D_{\bar{N}^{(i)}}(Q)$. (By definition, $\theta_8$ is a unity $t^*_Q$-valid assignment mapping for $Q''$, $D_{\bar{N}^{(i)}}(Q)$, and ${\cal A}_8$.) Then we obtain for  $\Phi_{n}[{\cal A}_8]$ that $[\nu^{(i)}_Q(\theta_8(Y''_1))] = [\nu^{(i)}_Q(d)] = Y_1$. 

The computation of the copy-signature $\Phi_{c}[{\cal A}_j]$ for each association ${\cal A}_j$ uses the mapping $\nu^{copy-sig}$, which maps subgoal $h_1$ of the query $Q$ to constant $1$ (because $h_1$ is a relational atom), and maps copy-sensitive subgoal $h_2$ of $Q$, with copy variable $Y_2$, to variable $\nu^{copy}_Q(Y_2) = N_2$ in the vector $\bar{N}$. Then the copy-signature of, say, association ${\cal A}_8$ is computed as the unary (because $r = 1$) vector $\Phi_{c}[{\cal A}_8]$ = $[\nu^{copy-sig}(\psi^{gen(Q)}_{\bar{N}^{(i)}}[d_3])] = [\nu^{copy-sig}(h_1)] = 1$. 

Finally, we use all the $t^*_Q$-valid assignment mappings for $Q''$, $D_{\bar{N}^{(i)}}(Q)$, and each ${\cal A}_j$, to determine the contributions of each association ${\cal A}_j$ to the set  $\Gamma_{\bar{S}}^{(t^*_Q)}(Q'',D_{\bar{N}^{(i)}}(Q))$. For instance, for the association ${\cal A}_8$ we use the  mapping $\theta_8$ (which is the only $t^*_Q$-valid assignment mapping for ${\cal A}_8$) to construct the tuple $(a,d,1)$ for the set $\Gamma_{\bar{S}}^{(t^*_Q)}(Q'',D_{\bar{N}^{(i)}}(Q))$. 

We now illustrate the construction of the set ${\mathbb A}^{(i)}_{Q''}$ for the query $Q''$ and database $D_{\bar{N}^{(i)}}(Q)$, as defined in Section~\ref{total-set-of-assoc-sec}. The set comprises all the nine associations above: ${\mathbb A}^{(i)}_{Q''}$ $=$ $\{ {\cal A}_1,$ ${\cal A}_2,$ $\ldots,$ ${\cal A}_9 \}$. We use Proposition~\ref{captures-prop} to conclude that the tuples shown in the last column of the table in this example are all and the only tuples in the set $\Gamma_{\bar{S}}^{(t^*_Q)}(Q'',D_{\bar{N}^{(i)}}(Q))$ for the query and for the database. 

Now we illustrate the construction of all the monomial classes, as equivalence classes of elements of the set ${\mathbb A}^{(i)}_{Q''}$ for the query $Q''$ and database $D_{\bar{N}^{(i)}}(Q)$, as defined in Section~\ref{monomial-classes-sec}. The classes are: 
\begin{itemize}
	\item ${\cal C}_1^{(Q'')}$ $=$ $\{ {\cal A}_1$, ${\cal A}_3$, ${\cal A}_7,$ ${\cal A}_9 \}$ (for the atom-signature $[h_1,h_1]$ of ${\cal A}_1$, ${\cal A}_3$, ${\cal A}_7$, and ${\cal A}_9$). We have that ${\Phi}_n^{{\cal C}_1^{(Q'')}} = [Y_1]$ and ${\Phi}_c^{{\cal C}_1^{(Q'')}} = [1]$.  
	\item ${\cal C}_2^{(Q'')}$ $=$ $\{ {\cal A}_2$, ${\cal A}_8 \}$ (for the atom-signature $[h_1,h_2]$ of ${\cal A}_2$ and ${\cal A}_8$). We have that ${\Phi}_n^{{\cal C}_2^{(Q'')}} = [Y_1]$ and ${\Phi}_c^{{\cal C}_2^{(Q'')}} = [1]$.  
	\item ${\cal C}_3^{(Q'')}$ $=$ $\{ {\cal A}_4$, ${\cal A}_6 \}$ (for the atom-signature $[h_2,h_1]$ of ${\cal A}_4$ and ${\cal A}_6$).  We have that ${\Phi}_n^{{\cal C}_3^{(Q'')}} = [X_2]$ and ${\Phi}_c^{{\cal C}_3^{(Q'')}} = [N_2]$.  
	\item ${\cal C}_4^{(Q'')}$ $=$ $\{ {\cal A}_5 \}$ (for the atom-signature $[h_2,h_2]$ of ${\cal A}_5$).  We have that ${\Phi}_n^{{\cal C}_4^{(Q'')}} = [X_2]$ and ${\Phi}_c^{{\cal C}_4^{(Q'')}} = [N_2]$.   
\end{itemize} 

Each of the four monomial classes has the noncopy-signature and the copy-signature of all its constituent associations. That is, for ${\cal C}_1^{(Q'')}$, we have above that ${\Phi}_n^{{\cal C}_1^{(Q'')}} = [Y_1]$ and ${\Phi}_c^{{\cal C}_1^{(Q'')}} = [1]$, and so on. 

For those terms of the query $Q$ that are not copy variables, we use the mapping $\nu_0$ and the set $S_1^{(i)}$ to determine that $Dom^{(i)}_Q(X_1) = \{ a \}$, $Dom^{(i)}_Q(X_2) = \{ b \}$, and $Dom^{(i)}_Q(Y_1) = \{ c,d \}$. Further, $DomLabel_Q(X_1)$ $=$ $DomLabel_Q(X_2)$ $ = 1$, and $DomLabel_Q(Y_1) = N_1$. 

We now compute the multiplicity monomial for each of the four monomial classes for $Q''$ and for $D_{\bar{N}^{(i)}}(Q)$. For each of ${\cal C}_1^{(Q'')}$ and ${\cal C}_2^{(Q'')}$, we have that $\Pi_{{\Phi}_n^{{\cal C}_1^{(Q'')}}}$ $=$ $\Pi_{{\Phi}_n^{{\cal C}_2^{(Q'')}}}$ is the product $\Pi_{j = 1}^1 DomLabel_Q(Y_j) = N_1$. Further, we have that $\Pi_{{\Phi}_c^{{\cal C}_1^{(Q'')}}}$ $=$ $\Pi_{{\Phi}_c^{{\cal C}_2^{(Q'')}}}$ is the product $\Pi_{j = 1}^1 1$. Thus, the multiplicity monomial for each of  ${\cal C}_1^{(Q'')}$ and ${\cal C}_2^{(Q'')}$ is the monomial $N_1 \times 1 = N_1$. Observe that $N_1^{(i)} = 2$ in our vector $\bar{N}^{(i)} = [2 \ 3]$, and that the value $N_1^{(i)} = 2$ of the monomial $N_1$ for each of ${\cal C}_1^{(Q'')}$ and ${\cal C}_2^{(Q'')}$ is the correct count of the two tuples, $(a, c, 1)$ and $(a, d, 1)$, contributed by each of the two classes {\em individually} to the set $\Gamma^{t^*_Q}_{\bar{S}}(Q'',D_{\bar{N}^{(i)}}(Q))$. Note that the projection of all these tuples on $Y''_1$ (in $\Gamma^{t^*_Q}_{\bar{S}}(Q'',D_{\bar{N}^{(i)}}(Q))$) is exactly all the elements of the set $Dom_Q^{(i)}(Y_1)$; recall that $Y_1$ is the only element of the vector $\Pi_{{\Phi}_n^{{\cal C}_1^{(Q'')}}}$ and of the vector $\Pi_{{\Phi}_n^{{\cal C}_2^{(Q'')}}}$. (See Proposition~\ref{noncopy-proj-properties-prop} for the details.)

For each of ${\cal C}_3^{(Q'')}$ and ${\cal C}_4^{(Q'')}$, we have that $\Pi_{{\Phi}_n^{{\cal C}_3^{(Q'')}}}$ $=$ $\Pi_{{\Phi}_n^{{\cal C}_4^{(Q'')}}}$ is the product $\Pi_{j = 1}^1 DomLabel_Q(X_2) = 1$. Further, we have that $\Pi_{{\Phi}_c^{{\cal C}_3^{(Q'')}}}$ $=$ $\Pi_{{\Phi}_c^{{\cal C}_4^{(Q'')}}}$ is the product $\Pi_{j = 1}^1 N_2 = N_2$. Thus, the multiplicity monomial for each of  ${\cal C}_3^{(Q'')}$ and ${\cal C}_4^{(Q'')}$ is the monomial $1 \times N_2 = N_2$. Observe that $N_2^{(i)} = 3$ in our vector $\bar{N}^{(i)} = [2 \ 3]$, and that the value $N_2^{(i)} = 3$ of the monomial $N_2$ for each of ${\cal C}_3^{(Q'')}$ and ${\cal C}_4^{(Q'')}$ is the correct count of the three tuples, $(a, b, 1)$, $(a, b, 2)$, and $(a, b, 3)$, contributed by each of the two classes {\em individually}  to the set $\Gamma^{t^*_Q}_{\bar{S}}(Q'',D_{\bar{N}^{(i)}}(Q))$. 
Note that the only element in the projection of all these tuples on $Y''_1$ (in $\Gamma^{t^*_Q}_{\bar{S}}(Q'',D_{\bar{N}^{(i)}}(Q))$) is exactly the only element of the set $Dom_Q^{(i)}(X_2)$; recall that $X_2$ is the only element of the vector $\Pi_{{\Phi}_n^{{\cal C}_3^{(Q'')}}}$ and of the vector $\Pi_{{\Phi}_n^{{\cal C}_4^{(Q'')}}}$. (See Proposition~\ref{noncopy-proj-properties-prop} for the details.) 

For the construction of the function ${\cal F}_{(Q)}^{(Q'')}$ from the above multiplicity monomials, please see Example~\ref{main-proof-cont-ex} (Section~\ref{main-proof-cont-ex-sec}). 

\paragraph{Construction of the Terms for ${\cal F}_{(Q)}^{(Q)}$} 
We now follow Sections~\ref{valid-map-sec} through 
\ref{monomial-class-mappings-sec} 
of the proof of Theorem~\ref{magic-mapping-prop}, to illustrate the construction of the terms for the function ${\cal F}_{(Q)}^{(Q)}$, for the query $Q$ and for the database  $D_{\bar{N}^{(i)}}(Q)$ as constructed above in this example. 
We follow steps similar to those used in the construction of the terms for the function ${\cal F}_{(Q)}^{(Q'')}$ for the query $Q''$, see preceding section of this example. As a result of the steps, we obtain four monomial classes for the query $Q$: 
\begin{itemize} 
	\item Monomial class ${\cal C}_1^{(Q)}$ has noncopy-signature $[Y_1]$ and copy-signature $[1]$; it contributes tuples $(a,c,1)$ and $(a,d,1)$ to the set $\Gamma^{t^*_Q}_{\bar{S}}(Q,D_{\bar{N}^{(i)}}(Q))$. The multiplicity monomial for ${\cal C}_1^{(Q)}$ is the term $N_1$.  
	\item Monomial class ${\cal C}_2^{(Q)}$ has noncopy-signature $[Y_1]$ and copy-signature $[N_2]$; it contributes tuples $(a,c,1)$, $(a,c,2)$, $(a,c,3)$, $(a,d,1)$, $(a,d,2)$, and $(a,d,3)$ to the set $\Gamma^{t^*_Q}_{\bar{S}}(Q,D_{\bar{N}^{(i)}}(Q))$. The multiplicity monomial for ${\cal C}_2^{(Q)}$ is the term $N_1 \times N_2$.  
	\item Monomial class ${\cal C}_3^{(Q)}$ has noncopy-signature $[X_2]$ and copy-signature $[1]$; it contributes tuple $(a,b,1)$ to the set $\Gamma^{t^*_Q}_{\bar{S}}(Q,D_{\bar{N}^{(i)}}(Q))$. The multiplicity monomial for ${\cal C}_3^{(Q)}$ is the term $1$ (i.e., constant $1$).  
	\item Monomial class ${\cal C}_4^{(Q)}$ has noncopy-signature $[X_2]$ and copy-signature $[N_2]$; it contributes tuples $(a,b,1)$, $(a,b,2)$, and $(a,b,3)$ to the set $\Gamma^{t^*_Q}_{\bar{S}}(Q,D_{\bar{N}^{(i)}}(Q))$. The multiplicity monomial for ${\cal C}_4^{(Q)}$ is the term $N_2$.  
\end{itemize} 

For the construction of the function ${\cal F}_{(Q)}^{(Q)}$ from the above multiplicity monomials, please see Example~\ref{main-proof-cont-ex} (Section~\ref{main-proof-cont-ex-sec}). 

\paragraph{Construction of the Terms for ${\cal F}_{(Q)}^{(Q')}$}
The construction of the terms for the function ${\cal F}_{(Q)}^{(Q')}$ is almost identical to that for the function ${\cal F}_{(Q)}^{(Q)}$, because the only difference between $Q$ and $Q'$ is in the presence of an extra subgoal $p(X'_1,X'_3)$ in $Q'$, and this subgoal does not introduce any multiset variables (of $Q'$). Thus, we obtain the same monomial classes for $Q'$ and for $Q$ (modulo renaming all the variables of $Q$ into ``same-name'' variables of $Q'$, for instance variable $X_1$ of $Q$ corresponds to variable $X'_1$ of $Q'$). Please see Example~\ref{main-proof-cont-ex} (Section~\ref{main-proof-cont-ex-sec}) for the construction of the function ${\cal F}_{(Q)}^{(Q')}$ from the above multiplicity monomials. 

\paragraph{The Wave of Query $Q$ w.r.t. $\{ D_{\bar{N}^{(i)}}(Q) \}$} 
We now use the monomial classes of the queries $Q$, $Q'$, and $Q''$ to illustrate the notion of the ``wave of CCQ query,'' which was introduced in Section~\ref{monomial-class-mappings-sec}. By Definition~\ref{the-wave-def}, the wave of the query $Q$ of this example, w.r.t. the family of databases $\{ D_{\bar{N}^{(i)}}(Q) \}$, is the product $N_1$  $\times$ $N_2$. By Proposition~\ref{q-has-wave-prop}, the query $Q$ has a nonempty monomial class w.r.t. $\{ D_{\bar{N}^{(i)}}(Q) \}$, specifically the monomial class  ${\cal C}_2^{(Q)}$, such that the multiplicity monomial of the class  ${\cal C}_2^{(Q)}$ is exactly the wave of the query $Q$ w.r.t. $\{ D_{\bar{N}^{(i)}}(Q) \}$. 

Now the query $Q'$ of this example also has a nonempty monomial class w.r.t. $\{ D_{\bar{N}^{(i)}}(Q) \}$, such that the multiplicity monomial of that monomial class is the wave of the query $Q$. (Recall that in this example we obtained the same monomial classes for $Q'$ and for $Q$, modulo renaming all the variables of $Q$ into ``same-name'' variables of $Q'$.) Thus, by Proposition~\ref{q-same-scale-mpng-prop}, there must exist a SCVM from the query $Q'$ to the query $Q$. Indeed, the same-scale covering  mapping $\mu$ of the beginning of this example is built as specified in the proof of  Proposition~\ref{q-same-scale-mpng-prop}. 

Finally, observe that for the query $Q''$ of this example and for each nonempty monomial class of $Q''$ w.r.t. $\{ D_{\bar{N}^{(i)}}(Q) \}$, the multiplicity monomial of the monomial class is {\em not} the wave of the query $Q$. Thus, Proposition~\ref{q-same-scale-mpng-prop} does not apply. Indeed, as we showed in the beginning of this example, there does not exist a SCVM from $Q''$ to $Q$. Then from Theorem~\ref{magic-mapping-prop} we conclude that $Q \equiv_C Q''$ does not hold for the queries $Q$ and $Q''$ of this example. Please see Example~\ref{main-proof-cont-ex} (Section~\ref{main-proof-cont-ex-sec}) for a discussion of how the database $D_{\bar{N}^{(i)}}(Q)$ constructed earlier (in Example~\ref{main-proof-ex}) is a counterexample database for $Q \equiv_C Q''$. In addition, for the wave ${\cal P}^{(Q)}_*$ $=$ $N_1$ $\times$ $N_2$ of the query $Q$ w.r.t. $\{ D_{\bar{N}^{(i)}}(Q) \}$, Example~\ref{main-proof-cont-ex} points out the presence of  the monomial ${\cal P}^{(Q)}_*$ in the functions ${\cal F}_{(Q)}^{(Q)}$ and ${\cal F}_{(Q)}^{(Q')}$, for the queries $Q$ and $Q'$ of this example, and also points out the absence of the monomial ${\cal P}^{(Q)}_*$ in the function 
${\cal F}_{(Q)}^{(Q'')}$, for the query $Q''$ of this example.  
\end{example}

\subsection{Putting Together the Function ${\cal F}_{(Q)}^{(Q'')}$} 
\label{putting-together-f-sec}

In this section we define the function ${\cal F}_{(Q)}^{(Q'')}$ outlined in the beginning of Section~\ref{monomial-classes-sec}.  The reason that we construct this function in this proof of Theorem~\ref{magic-mapping-prop} is that we want to be able to pinpoint those queries $Q''$ that are associated with the wave (see Definition~\ref{the-wave-def}) of the query $Q$. That is, later in this proof of  Theorem~\ref{magic-mapping-prop} (specifically in Section~\ref{q-prime-has-wave-sec}), for the query $Q'$ specified in the statement of the Theorem we will use the facts that (i) $Q'$ $\equiv_C$ $Q$ (and thus their respective functions ${\cal F}_{(Q)}^{(Q')}$ and ${\cal F}_{(Q)}^{(Q)}$ must return the same value on each database) and that (ii) $Q$ is an explicit-wave query, to infer that the function ${\cal F}_{(Q)}^{(Q')}$ for the query $Q'$ must have as its component the wave monomial of the query $Q$. The claim of Theorem~\ref{magic-mapping-prop} will then follow from Proposition~\ref{q-same-scale-mpng-prop}. 

The only entities that we use in this section to specify the function ${\cal F}_{(Q)}^{(Q'')}$  are  (a) the multiplicity monomials defined in Section~\ref{monomial-classes-sec}, and (b) the noncopy-signatures and the copy-signatures of the monomial classes introduced in Section~\ref{monomial-classes-sec}.   Recall that each of the multiplicity monomials, as well as each of the copy-signatures, is in terms of the variables in the vector $\bar{N}$; we will show in this section how to ``convert'' the noncopy-signatures into collections of variables in the vector $\bar{N}$.  

We specify the function ${\cal F}_{(Q)}^{(Q'')}$ for an arbitrary CCQ query $Q$, for the family $\{ D_{\bar{N}^{(i)}}(Q) \}$ of databases defined using $Q$ (as outlined in Section~\ref{db-constr-sec}), and for an arbitrary CCQ query $Q''$ that satisfies the restrictions (w.r.t. the query $Q$) of Section~\ref{basics-sec}.  

\subsubsection{Notation, Definitions, Basic Results} 
\label{multivar-polyn-basics-sec}

For CCQ queries $Q$ and $Q''$ satisfying the requirements of Section~\ref{basics-sec}, suppose that $Q$ and $Q''$ are also such that (as discussed in the beginning of  Section~\ref{monomial-classes-sec}) the set of all nonempty monomial classes for $Q''$ and for the family of databases $\{ D_{\bar{N}^{(i)}}(Q) \}$ is not empty. That is, suppose that $\{ {\cal C}_1^{(Q'')}$, $\ldots,$ ${\cal C}_{n^*}^{(Q'')} \}$ is the set of all nonempty monomial classes for $Q''$ and for $\{ D_{\bar{N}^{(i)}}(Q) \}$, with $n^* \geq 1$. 

We begin the exposition by making a few straightforward observations. Recall that in Section~\ref{monomial-classes-sec}, we denoted by  $\Gamma^{(i)}[{\cal C}^{(Q'')}]$ the set of all  tuples contributed to the set $\Gamma^{(t^*_Q)}_{\bar{S}}(Q'',D_{\bar{N}^{(i)}}(Q))$ by all the valid assignment mappings for all the elements of the class  ${\cal C}^{(Q'')}$. The following proposition is immediate from Proposition~\ref{noncopy-proj-properties-prop}, for the cases where $n^* \geq 2$. 

\begin{proposition} 
\label{diff-noncopy-prop}
Given CCQ queries $Q$ and $Q''$, suppose that the set ${\cal C}[Q''] = \{ {\cal C}_1^{(Q'')}$, $\ldots,$ ${\cal C}_{n^*}^{(Q'')} \}$ of all nonempty monomial classes for $Q''$ and for the family of databases $\{ D_{\bar{N}^{(i)}}(Q) \}$ has $n^* \geq 2$ elements. Further, let ${\cal C}_n^{(Q'')}$ and ${\cal C}_p^{(Q'')}$, for some $n, p \in \{ 1,\ldots,n^* \}$, be two 
monomial classes in the set ${\cal C}[Q'']$, such that the noncopy-signatures of ${\cal C}_n^{(Q'')}$ and of ${\cal C}_p^{(Q'')}$ are not identical vectors. Then for each $i \in {\mathbb N}_+$, we have that $\Gamma^{(i)}[{\cal C}_n^{(Q'')}]$ $\bigcap$ $\Gamma^{(i)}[{\cal C}_p^{(Q'')}]$ $=$ $\emptyset$. 
\end{proposition} 

For the other observations in this subsection, we will need the following notation. 
For an arbitrary monomial class ${\cal C}_n^{(Q'')} \neq \emptyset$, $n \in \{ 1,\ldots,n^* \}$, in case $r \geq 1$, we denote the elements of the copy-signature vector $\Phi_c[{\cal C}_n^{(Q'')}]$ as $[V_{j1[n]},\ldots,V_{jr[n]}]$. Recall that, by definition of copy-signature, for each $k \in \{ 1,\ldots,r \}$ we have that $V_{jk[n]}$ $\in$ $\{ 1,$ $N_{m+1},$ $\ldots,$ $N_{m+w}\}$, where the $N$-values are variables in the vector $\bar{N}$.   (In case $r = 0$, $\Phi_c[{\cal C}_n^{(Q'')}]$ is the empty vector by definition.) 

\begin{definition}{Unconditional dominance for monomial classes} 
\label{uncond-dom-def}
Let ${\cal C}_n^{(Q'')}$ and  ${\cal C}_p^{(Q'')}$ be  two (not necessarily distinct) monomial classes in the set $\{ {\cal C}_1^{(Q'')}$, $\ldots,$ ${\cal C}_{n^*}^{(Q'')} \}$. (That is, $n, p \in \{ 1,\ldots,n^* \}$.)  Further, let ${\cal C}_n^{(Q'')}$ and  ${\cal C}_p^{(Q'')}$ have the same noncopy-signature. Then we say that {\em monomial class ${\cal C}_n^{(Q'')}$ unconditionally dominates monomial class} ${\cal C}_p^{(Q'')}$ if: 
\begin{itemize}
	\item We have the case $r = 0$; or 
	\item We have the case $r \geq 1$, and for the pair $(V_{jk[p]},V_{jk[n]})$ for each $k \in \{ 1,\ldots,r \}$, we have that either $V_{jk[p]}$ $=$ $1$, or $V_{jk[p]} = V_{jk[n]}$. 
\end{itemize} 
\end{definition} 

We observe that the unconditional-dominance relation is reflexive by definition. 

The following important property of unconditional-dominance holds by the results of Section~\ref{monomial-classes-sec}. 

\begin{proposition} 
\label{uncond-dom-prop} 
Let ${\cal C}_n^{(Q'')}$ and  ${\cal C}_p^{(Q'')}$ be  two monomial classes in the set $\{ {\cal C}_1^{(Q'')}$, $\ldots,$ ${\cal C}_{n^*}^{(Q'')} \}$. (That is, $n, p \in \{ 1,\ldots,n^* \}$.)  
%
Suppose that ${\cal C}_n^{(Q'')}$ unconditionally dominates  ${\cal C}_p^{(Q'')}$. Then for each $i \in {\mathbb N}_+$, we have that $\Gamma^{(i)}[{\cal C}_p^{(Q'')}]$ $\subseteq$ $\Gamma^{(i)}[{\cal C}_n^{(Q'')}]$. 
\end{proposition} 


The following result is immediate from Definition~\ref{uncond-dom-def}. 

\begin{proposition} 
\label{same-copy-sig-prop}
Let ${\cal C}_n^{(Q'')}$ and  ${\cal C}_p^{(Q'')}$ be  two monomial classes in the set $\{ {\cal C}_1^{(Q'')}$, $\ldots,$ ${\cal C}_{n^*}^{(Q'')} \}$. (That is, $n, p \in \{ 1,\ldots,n^* \}$.)  Further, let ${\cal C}_n^{(Q'')}$ and  ${\cal C}_p^{(Q'')}$ have the same noncopy-signature. 
Then we have that (i) ${\cal C}_n^{(Q'')}$ unconditionally dominates  ${\cal C}_p^{(Q'')}$ {\em and}  ${\cal C}_p^{(Q'')}$ unconditionally dominates  ${\cal C}_n^{(Q'')}$, if and only if (ii) ${\cal C}_n^{(Q'')}$ and ${\cal C}_p^{(Q'')}$ have the same copy-signature. 
\end{proposition} 

From reflexivity of unconditional-dominance and from Propositions~\ref{uncond-dom-prop} and~\ref{same-copy-sig-prop}, 
we obtain the following result. 

\begin{proposition} 
\label{identical-answers-prop}
Let ${\cal C}_n^{(Q'')}$ and  ${\cal C}_p^{(Q'')}$ be  two monomial classes in the set $\{ {\cal C}_1^{(Q'')}$, $\ldots,$ ${\cal C}_{n^*}^{(Q'')} \}$. (That is, $n, p \in \{ 1,\ldots,n^* \}$.)  Further, let ${\cal C}_n^{(Q'')}$ and  ${\cal C}_p^{(Q'')}$ have the same noncopy-signature and the same copy-signature. Then  for each $i \in {\mathbb N}_+$, we have that $\Gamma^{(i)}[{\cal C}_p^{(Q'')}]$ $=$ $\Gamma^{(i)}[{\cal C}_n^{(Q'')}]$. 
\end{proposition} 

In Example~\ref{main-proof-ex} in Section~\ref{main-proof-ex-sec}, monomial class  ${\cal C}_1^{(Q'')}$ unconditionally dominates a {\em nonidentical} \footnote{Recall (see Section~\ref{monomial-classes-sec}) that the identity of a monomial class is determined by its atom-signature.} (to ${\cal C}_1^{(Q'')}$) monomial class  ${\cal C}_2^{(Q'')}$, and vice versa (that is, monomial class  ${\cal C}_2^{(Q'')}$ unconditionally dominates monomial class  ${\cal C}_1^{(Q'')}$). Similarly, monomial class  ${\cal C}_3^{(Q'')}$ of the same Example unconditionally dominates a nonidentical (to ${\cal C}_3^{(Q'')}$) monomial class  ${\cal C}_4^{(Q'')}$, and vice versa. 

We now outline an algorithm template that we call {\sc Removal of duplicate monomial classes.} The input is the set  $\{ {\cal C}_1^{(Q'')}$, $\ldots,$ ${\cal C}_{n^*}^{(Q'')} \}$, $n^* \geq 1$, for CCQ query $Q''$ and for family of databases $\{ D_{\bar{N}^{(i)}}(Q) \}$; the output is a subset (denoted by ${\mathbb C}(Q'')$) of the input. The algorithm template involves three steps:  
\begin{itemize} 
	\item[(1)] Partition all  elements of the set  $\{ {\cal C}_1^{(Q'')}$, $\ldots,$ ${\cal C}_{n^*}^{(Q'')} \}$ into equivalence classes, where two distinct (in case $n^* \geq 2$) monomial classes ${\cal C}_n^{(Q'')}$ and ${\cal C}_p^{(Q'')}$, for  $n \neq p \in \{ 1,\ldots,n^* \}$, belong to the same equivalence class if and only if ${\cal C}_n^{(Q'')}$ and ${\cal C}_p^{(Q'')}$ have identical noncopy-signatures and identical copy-signatures. 
	\item[(2)] Use an arbitrary algorithm, call it {\sc Choose-Represent-ative-Element,} to choose one element of each of the equivalence classes as the representative element of the equivalence class. 
	\item[(3)] Return the set ${\mathbb C}(Q'')$ of representative elements (only) of all of the equivalence classes of the set $\{ {\cal C}_1^{(Q'')}$, $\ldots,$ ${\cal C}_{n^*}^{(Q'')} \}$. 
\end{itemize} 

A specific algorithm instantiating the algorithm template {\sc Removal of duplicate monomial classes} is obtained by specifying the algorithm {\sc Choose-Representative-Element.} 
Observe that ${\mathbb C}(Q'') \neq \emptyset$ for all nonempty inputs to {\sc Removal of duplicate monomial classes}  and for all choices of the algorithm {\sc Choose-Representative-Element.}  

Clearly, in general, the contents of the set ${\mathbb C}(Q'')$ depend on the algorithm, {\sc Choose-Representative-Element,} for choosing the representative element of each equivalence class, within the algorithm template {\sc Removal of duplicate monomial classes.} (For instance, given as input the four monomial classes of Example~\ref{main-proof-ex} in Section~\ref{main-proof-ex-sec}, the algorithm template could produce four different outputs.) At the same time, the following two results, Proposition~\ref{correct-drop-duplicates-prop} and Proposition~\ref{correct-isomorphic-drop-dupl-prop}, hold regardless of the choice of the algorithm {\sc Choose-Representative-Element} when instantiating the algorithm template {\sc Removal of duplicate monomial} {\sc classes.} 

\begin{proposition} 
\label{correct-drop-duplicates-prop} 
Given the set \ ${\cal C}[Q'']$ of all nonempty monomial classes for CCQ query $Q''$ and for family of databases $\{ D_{\bar{N}^{(i)}}(Q) \}$, and given an algorithm {\sc Choose-Representative-Element} to instantiate the algorithm template {\sc Removal of duplicate monomial classes.}   Then for the output ${\mathbb C}(Q'')$ of the resulting algorithm 
given the input ${\cal C}[Q'']$, the following two facts hold:  
\begin{itemize}
	\item[(i)] For each pair $(e_1,e_2)$ of distinct (i.e., $e_1 \neq e_2$) elements of the set ${\mathbb C}(Q'')$, $e_1$ and $e_2$ either have different noncopy-signatures or have different copy-signatures; and  
	\item[(ii)] For all $i \in {\mathbb N}_+$, we have that: 
$$
\bigcup_{{\cal C} \in {\cal C}[Q''] } \Gamma^{(i)}[{\cal C}]
 = 
\bigcup_{{\cal C}' \in {\mathbb C}(Q'')} \Gamma^{(i)}[{\cal C}'] . 
$$ 
\end{itemize} 
\end{proposition} 

Proposition~\ref{correct-drop-duplicates-prop} is immediate from Proposition~\ref{identical-answers-prop} and from the construction of the algorithm template {\sc Removal of duplicate monomial classes}.   

In the next result, Proposition~\ref{correct-isomorphic-drop-dupl-prop}, we denote by $|S|$ the cardinality of set $S$. Proposition~\ref{correct-isomorphic-drop-dupl-prop} holds by construction of the algorithm template {\sc Removal of duplicate monomial classes}. 

\begin{proposition} 
\label{correct-isomorphic-drop-dupl-prop} 
Let ${\cal C}[Q'']$ be the set of all nonempty monomial classes for CCQ query $Q''$ and for family of databases $\{ D_{\bar{N}^{(i)}}(Q) \}$. Let $a_1$ and $a_2$ be two instantiations of the algorithm template {\sc Removal of duplicate monomial classes}, where $a_1$ and $a_2$ may use different ways of choosing the representative element of each equivalence class generated by the algorithm. Let ${\mathbb C}_j(Q'')$ be the output of algorithm $a_j$ on the input ${\cal C}[Q'']$, for $j \in \{ 1,2 \}$. Then we have that: 
\begin{itemize} 
	\item[(1)] $|{\mathbb C}_1(Q'')|$ $=$ $|{\mathbb C}_2(Q'')|$; and 
	\item[(2)] There exists an isomorphism,  call it $\mu$, from the set ${\mathbb C}_1(Q'')$ to the set ${\mathbb C}_2(Q'')$, such that for each element $e$ of the set ${\mathbb C}_1(Q'')$, $e$ and $\mu(e)$ have the same copy-signature as well as the same noncopy-signature. 
\end{itemize} 
\end{proposition} 

For our purpose of constructing a function that would return the multiplicity of the tuple $t^*_Q$ in the bag ${\sc Res}_C$ $(Q'',$ $D_{\bar{N}^{(i)}}(Q))$, for all $i \in {\mathbb N}_+$, Propositions~\ref{correct-drop-duplicates-prop} and \ref{correct-isomorphic-drop-dupl-prop}  let us refer to the output of an arbitrary instantiation of the algorithm template {\sc Removal of duplicate monomial classes} as {\em the} output, ${\mathbb C}(Q'')$, of the algorithm (template). We can show that the unconditional-dominance relation of Definition~\ref{uncond-dom-def} is reflexive, antisymmetric, and transitive on the set ${\mathbb C}(Q'')$. It follows that the unconditional-dominance relation is a partial order on that set. 

We now use any standard algorithm\footnote{We use the observation that the unconditional-dominance relation of Definition~\ref{uncond-dom-def}  is a partial order on the set ${\mathbb C}(Q'')$.} for removal of all those monomial classes from the set ${\mathbb C}(Q'')$ that (monomial classes) are unconditionally dominated by some other monomial class in the set ${\mathbb C}(Q'')$. Clearly, the output of that algorithm is a unique subset, call it ${\mathbb C}^{nondom}(Q'')$, of the set ${\mathbb C}(Q'')$. We say that the set ${\mathbb C}^{nondom}(Q'')$ is {\em the result of dropping unconditionally-dominated monomial classes from the set}  ${\mathbb C}(Q'')$. 

Using Propositions~\ref{correct-drop-duplicates-prop} and~\ref{correct-isomorphic-drop-dupl-prop}, it is straightforward to show the following. 

\begin{proposition} 
\label{nondom-prop} 
Let ${\cal C}[Q'']$ be the set of all nonempty monomial classes for CCQ query $Q''$ and for family of databases $\{ D_{\bar{N}^{(i)}}(Q) \}$. Let $a_1$ and $a_2$ be two instantiations of the algorithm template {\sc Removal of duplicate monomial classes}, where $a_1$ and $a_2$ may use different ways of choosing the representative element of each equivalence class generated by the algorithm. Let ${\mathbb C}_j(Q'')$ be the output of algorithm $a_j$ on the input ${\cal C}[Q'']$, for $j \in \{ 1,2 \}$. Further, let ${\mathbb C}^{nondom}_j(Q'')$ be the result of dropping unconditionally-dominated monomial classes from the set ${\mathbb C}_j(Q'')$, for $j \in \{ 1,2 \}$. 
Then for all $i \in {\mathbb N}_+$,  we have that: 
$$
\bigcup_{{\cal C} \in {\cal C}[Q''] } \Gamma^{(i)}[{\cal C}]
 = 
\bigcup_{{\cal C}' \in {\mathbb C}^{nondom}_1(Q'')} \Gamma^{(i)}[{\cal C}'] $$
$$= 
\bigcup_{{\cal C}'' \in {\mathbb C}^{nondom}_2(Q'')} \Gamma^{(i)}[{\cal C}'']  .
$$ 
\end{proposition} 

As a result of Proposition~\ref{nondom-prop}, for our purpose (of constructing a function that would return the multiplicity of the tuple $t^*_Q$ in the bag ${\sc Res}_C(Q'',D_{\bar{N}^{(i)}}(Q))$) we can refer to each set  ${\mathbb C}^{nondom}(Q'')$ as {\em the} set ${\mathbb C}^{nondom}(Q'')$ for the set of all nonempty monomial classes for CCQ query $Q''$ and for family of databases $\{ D_{\bar{N}^{(i)}}(Q) \}$, regardless of the identity of the exact set ${\mathbb C}(Q'')$ as discussed above. 

\subsubsection{The Easy Case of Constructing ${\cal F}_{(Q)}^{(Q'')}$}
\label{easy-funct-case-sec} 

The following observation lets us finalize the construction of the function ${\cal F}_{(Q)}^{(Q'')}$ for the case where all elements of the set  ${\mathbb C}^{nondom}(Q'')$ have different noncopy-signatures. In this case, we have that the function  ${\cal F}_{(Q)}^{(Q'')}$  is always a multivariate polynomial in terms of the variables in the vector $\bar{N}$ and with integer coefficients, on the entire domain ${\cal N}$ of the function. 
The result of Proposition~\ref{easy-final-funct-prop} is immediate from Propositions~\ref{diff-noncopy-prop}, \ref{correct-drop-duplicates-prop}, and~\ref{correct-isomorphic-drop-dupl-prop}. For the definition of multiplicity monomial for monomial class, see Section~\ref{multiplicity-monomial-sec}. 

\begin{proposition}
\label{easy-final-funct-prop} 
Given the set ${\mathbb C}^{nondom}(Q'')$ for a CCQ query $Q''$ and for a family of databases $\{ D_{\bar{N}^{(i)}}(Q) \}$, such that the elements of the set  ${\mathbb C}^{nondom}(Q'')$ have $|{\mathbb C}^{nondom}(Q'')|$ distinct noncopy-signatures. Then, for each $i \in {\mathbb N}_+$, the cardinality of the set $\Gamma^{(t^*_Q)}_{\bar{S}}$ $(Q'',D_{\bar{N}^{(i)}}(Q))$, can be computed exactly, by substituting the values in the vector $\bar{N}^{(i)}$ (specifically value $N_j^{(i)}$ as value of variable $N_j$, for each $j \in \{ 1,\ldots,m+w \}$) into the formula 
$$\Sigma_{{\cal C} \in {\mathbb C}^{nondom}(Q'')} {\cal M}[{\cal C}]$$
where ${\cal M}[{\cal C}]$ is the multiplicity monomial of monomial class $\cal C$. 
\end{proposition} 
 
That is, under the conditions of Proposition~\ref{easy-final-funct-prop}, the function ${\cal F}_{(Q)}^{(Q'')}$ is given  by the formula 
$$\Sigma_{{\cal C} \in {\mathbb C}^{nondom}(Q'')} {\cal M}[{\cal C}]$$
in the Proposition. 
  
For instance, if we choose the set $\{ {\cal C}_1(Q''), {\cal C}_3(Q'') \}$ as the set ${\mathbb C}^{nondom}(Q'')$ for Example~\ref{main-proof-ex}  in Section~\ref{main-proof-ex-sec}, then, for query $Q''$ of the Example, the function  ${\cal F}_{(Q)}^{(Q'')}$  is the following multivariate polynomial in terms of the variables in the vector $\bar{N}$: ${\cal F}_{(Q)}^{(Q'')}$ $=$ $N_1 + N_2$. For the vector ${\bar{N}}^{(i)}$ of Example~\ref{main-proof-ex}, ${\cal F}_{(Q)}^{(Q'')}$ returns the correct multiplicity, $5$, of the tuple $t^*_Q = (a)$ in the bag ${\sc Res}_C(Q'',D_{\bar{N}^{(i)}}(Q))$. (For the details, please see Example~\ref{main-proof-cont-ex} in Section~\ref{main-proof-cont-ex-sec}.) 

\begin{corollary} 
\label{easy-case-corol}
In case $r \leq 1$, given a CCQ query $Q''$ and a family of databases $\{ D_{\bar{N}^{(i)}}(Q) \}$. Then, for each $i \in {\mathbb N}_+$, the cardinality of the set $\Gamma^{(t^*_Q)}_{\bar{S}}$ $(Q'',D_{\bar{N}^{(i)}}(Q))$, can be computed exactly, by substituting the values in the vector $\bar{N}^{(i)}$ (specifically value $N_j^{(i)}$ as value of variable $N_j$, for each $j \in \{ 1,\ldots,m+w \}$) into the formula 
$$\Sigma_{{\cal C} \in {\mathbb C}^{nondom}(Q'')} {\cal M}[{\cal C}]$$
where ${\cal M}[{\cal C}]$ is the multiplicity monomial of monomial class $\cal C$. 
\end{corollary} 

\begin{proof}{(sketch)} 
The reason that Corollary~\ref{easy-case-corol} of Proposition~\ref{easy-final-funct-prop}  holds is that in case $r \leq 1$, either $r = 0$ holds and then the copy-signature of each monomial class for $Q''$ is the empty vector, or $r = 1$ holds and then  the copy-signature of each monomial class for $Q''$ is  either the vector $[1]$ or the vector $[N]$, for exactly one variable name $N$ across all the copy-signatures. Then the unconditional-dominance relation of Definition~\ref{uncond-dom-def} holds for each pair of the monomial classes for the query $Q''$ such that the classes in the pair have the same noncopy-signature, and hence  all elements of the set  ${\mathbb C}^{nondom}(Q'')$ have different noncopy-signatures. 
\end{proof} 

Observe that it is not obvious how to generalize the statement of Corollary~\ref{easy-case-corol} to the case $r \geq 2$. Indeed, even when $w \leq 1$ (and hence we still have exactly one variable name $N$ as the only possible variable across all the noncopy-signatures),\footnote{By Proposition~\ref{sc-prop}, $r \geq 1$ implies $w \geq 1$, hence from $r \geq 2$ and $w \leq 1$ we have the exact equality $w = 1$.} the case $r = 2$ already presents us with the (theoretical) possibility where two monomial classes for $Q''$, with the same noncopy-signature, might have respective copy-signatures $[1 \ N]$ and $[N \ 1]$, for which unconditional-dominance does not hold in either direction.

\reminder{Do I really need to do this??? [It seems that Proposition~\ref{easy-final-funct-prop} is already *the* correctness statement] --- Must put in the proposition for ``correctness'' of the output of the function ${\cal F}_{(Q)}^{(Q'')}$, in terms of Proposition~\ref{gamma-i-union-prop}.}

\reminder{Do I really need to do this??? [Into the subsection about the *general* case of building ${\cal F}_{(Q)}^{(Q'')}$ (as opposed to the ``easy'' case of Section~\ref{easy-funct-case-sec}): Must put in the proposition for ``correctness'' of the output of the function ${\cal F}_{(Q)}^{(Q'')}$, in terms of Proposition~\ref{gamma-i-union-prop}.]}

\subsubsection{Illustration} 
\label{main-proof-cont-ex-sec}

In this subsection we build on Example~\ref{main-proof-ex} (of Section~\ref{main-proof-ex-sec}), to show the construction of the functions ${\cal F}_{(Q)}^{(Q)}$, ${\cal F}_{(Q)}^{(Q')}$, and ${\cal F}_{(Q)}^{(Q'')}$, for the three queries of Example~\ref{main-proof-ex} and for the database constructed in Example~\ref{main-proof-ex}. We also show how that database is a counterexample to $Q \equiv_C Q''$, for the queries $Q$ and $Q''$ of Example~\ref{main-proof-ex}. Finally, we continue our discussion (started in Example~\ref{main-proof-ex}) of ``the wave of'' the query $Q$, and explore the relationship between that entity and the multivariate polynomials for  ${\cal F}_{(Q)}^{(Q)}$, ${\cal F}_{(Q)}^{(Q')}$, and ${\cal F}_{(Q)}^{(Q'')}$.

\begin{example} 
\label{main-proof-cont-ex} 
Recall the queries $Q$, $Q'$, and $Q''$ of Example~\ref{main-proof-ex} (of Section~\ref{main-proof-ex-sec}). Recall also the database $D_{\bar{N}^{(i)}}(Q)$ that we constructed in Example~\ref{main-proof-ex} for the query $Q$. In this example we build functions ${\cal F}_{(Q)}^{(Q)}$, ${\cal F}_{(Q)}^{(Q')}$, and ${\cal F}_{(Q)}^{(Q'')}$, for the three queries $Q$, $Q'$, and $Q''$ and for the database $D_{\bar{N}^{(i)}}(Q)$. We also show how the database $D_{\bar{N}^{(i)}}(Q)$ is a counterexample to $Q \equiv_C Q''$. Finally, we continue our discussion (started in Example~\ref{main-proof-ex}) of ``the wave of'' the query $Q$, and explore the relationship between that entity and the multivariate polynomials for  ${\cal F}_{(Q)}^{(Q)}$, ${\cal F}_{(Q)}^{(Q')}$, and ${\cal F}_{(Q)}^{(Q'')}$. 

\paragraph{Construction of Function ${\cal F}_{(Q)}^{(Q)}$} 
For the query $Q$ and database $D_{\bar{N}^{(i)}}(Q)$, we came up in Example~\ref{main-proof-ex} with four monomial classes ${\cal C}_1^{(Q)},$ ${\cal C}_2^{(Q)}$, ${\cal C}_3^{(Q)}$, and ${\cal C}_4^{(Q)}$. By Definition~\ref{uncond-dom-def},  monomial class ${\cal C}_2^{(Q)}$ unconditionally-dominates monomial class ${\cal C}_1^{(Q)}$. The reason is, ${\cal C}_1^{(Q)}$ and ${\cal C}_2^{(Q)}$ have identical noncopy-signatures, the copy-signature of the monomial class ${\cal C}_1^{(Q)}$ is  $[1]$, and the copy-signature of the monomial class ${\cal C}_2^{(Q)}$ is  $[N_2]$. (See Section~\ref{multivar-polyn-basics-sec} for further details on unconditional-dominance.) Similarly, monomial class ${\cal C}_4^{(Q)}$ unconditionally-dominates monomial class ${\cal C}_3^{(Q)}$.   Thus, the set $\{ {\cal C}_2^{(Q)}, {\cal C}_4^{(Q)} \}$ is {\em the} set ${\mathbb C}^{nondom}(Q)$ as defined in Section~\ref{multivar-polyn-basics-sec}. 

Then, by Proposition~\ref{easy-final-funct-prop}, the function  ${\cal F}_{(Q)}^{(Q)}$  is the following multivariate polynomial in terms of the variables in the vector $\bar{N}$: ${\cal F}_{(Q)}^{(Q)}$ $=$ $N_1 \times N_2 + N_2$. For the vector ${\bar{N}}^{(i)} = [2 \ 3]$ that we fixed in Example~\ref{main-proof-ex}, ${\cal F}_{(Q)}^{(Q)}$ returns the correct multiplicity, $9$, of the tuple $t^*_Q = (a)$ in the bag ${\sc Res}_C(Q,D_{\bar{N}^{(i)}}(Q))$. 

\paragraph{Construction of Function ${\cal F}_{(Q)}^{(Q')}$} 
For the query $Q'$ and database $D_{\bar{N}^{(i)}}(Q)$ of Example~\ref{main-proof-ex}, we use the reasoning similar to that for constructing the function ${\cal F}_{(Q)}^{(Q)}$ earlier in this example, to obtain the function ${\cal F}_{(Q)}^{(Q')}$ $=$ $N_1 \times N_2 + N_2$. As the multivariate polynomials ${\cal F}_{(Q)}^{(Q)}$ and ${\cal F}_{(Q)}^{(Q')}$ are identical to each other, they output the same answer for each $\bar{N}^{(i)} \in {\cal N}$. 
 
\paragraph{Construction of Function ${\cal F}_{(Q)}^{(Q'')}$} 
For the query $Q''$ and database $D_{\bar{N}^{(i)}}(Q)$, we came up in Example~\ref{main-proof-ex} with four monomial classes ${\cal C}_1^{(Q'')},$ ${\cal C}_2^{(Q'')}$, ${\cal C}_3^{(Q'')}$, and ${\cal C}_4^{(Q'')}$. By Definition~\ref{uncond-dom-def}, monomial classes ${\cal C}_1^{(Q'')}$ and ${\cal C}_2^{(Q'')}$ unconditionally-dominate each other. (The reason is, ${\cal C}_1^{(Q'')}$ and ${\cal C}_2^{(Q'')}$ have identical noncopy-signatures, and have identical copy-signatures.) Similarly, monomial classes ${\cal C}_3^{(Q'')}$ and ${\cal C}_4^{(Q'')}$ unconditionally-dominate each other. Suppose that we choose the set $\{ {\cal C}_1^{(Q'')}, {\cal C}_3^{(Q'')} \}$ as the set ${\mathbb C}^{nondom}(Q'')$ as defined in Section~\ref{multivar-polyn-basics-sec}. (See Section~\ref{multivar-polyn-basics-sec} for the discussion of possible choices for the set ${\mathbb C}^{nondom}(Q'')$.) 

Then, by Proposition~\ref{easy-final-funct-prop}, the function ${\cal F}_{(Q)}^{(Q'')}$  is the following multivariate polynomial in terms of the variables in the vector $\bar{N}$: ${\cal F}_{(Q)}^{(Q'')}$ $=$ $N_1 + N_2$. For the vector ${\bar{N}}^{(i)} = [2 \ 3]$ that we fixed in  Example~\ref{main-proof-ex}, ${\cal F}_{(Q)}^{(Q'')}$ returns the correct multiplicity, $5$, of the tuple $t^*_Q = (a)$ in the bag ${\sc Res}_C(Q'',D_{\bar{N}^{(i)}}(Q))$. 

\paragraph{Database  $D_{\bar{N}^{(i)}}(Q)$ Is a Counterexample to $Q \equiv_C Q''$} 
From the different sizes of the sets $\Gamma^{t^*_Q}_{\bar{S}}(Q,D_{\bar{N}^{(i)}}(Q))$ and $\Gamma^{t^*_Q}_{\bar{S}}(Q'',D_{\bar{N}^{(i)}}(Q))$ on the database $D_{\bar{N}^{(i)}}(Q)$, as discussed earlier in this example, we have that the database $D_{\bar{N}^{(i)}}(Q)$ is a counterexample to $Q \equiv_C Q''$. 

\paragraph{The Wave of $Q$ in the Functions ${\cal F}_{(Q)}^{(Q)}$, ${\cal F}_{(Q)}^{(Q')}$} 
Recall from Example~\ref{main-proof-ex} (Section~\ref{main-proof-ex-sec}) our discussion of ``the wave of'' the query $Q$ of that example. By Definition~\ref{the-wave-def}, the wave of that query $Q$ w.r.t. the family of databases $\{ D_{\bar{N}^{(i)}}(Q) \}$ is the monomial $N_1$ $\times$ $N_2$. 

For the queries $Q$, $Q'$ and $Q''$ discussed in this example (see Example~\ref{main-proof-ex} for their definitions), we now contrast the functions ${\cal F}_{(Q)}^{(Q)}$ and ${\cal F}_{(Q)}^{(Q')}$, on the one hand, with the function ${\cal F}_{(Q)}^{(Q'')}$, on the other hand. Recall that ${\cal F}_{(Q)}^{(Q)}$ $=$ ${\cal F}_{(Q)}^{(Q')}$ $=$ $N_1 \times N_2 + N_2$. Observe that each of ${\cal F}_{(Q)}^{(Q)}$ and ${\cal F}_{(Q)}^{(Q')}$ has a term that is exactly the wave of the query $Q$ (w.r.t. $\{ D_{\bar{N}^{(i)}}(Q) \}$), that is the term $N_1$ $\times$ $N_2$. In contrast, the function  ${\cal F}_{(Q)}^{(Q'')}$ $=$ $N_1 + N_2$ clearly does not have a term that is the wave $N_1$ $\times$ $N_2$ of the query $Q$ w.r.t. $\{ D_{\bar{N}^{(i)}}(Q) \}$. 
\end{example}

\subsubsection{Beyond the Easy Case: Example} 
\label{beyond-easy-case-sec}

In this subsection we exhibit a CCQ query $Q$, such that the function ${\cal F}_{(Q)}^{(Q)}$, w.r.t. the family of databases $\{ D_{\bar{N}^{(i)}}(Q) \}$, cannot be computed using the results of Section~\ref{easy-funct-case-sec}, specifically using Proposition~\ref{easy-final-funct-prop}. This example motivates the development, in Section~\ref{hard-funct-case-sec}, of a more general (as compared to that of Section~\ref{easy-funct-case-sec}) approach  toward constructing the function ${\cal F}_{(Q)}^{(Q)}$ for CCQ query $Q''$ and family of databases $\{ D_{\bar{N}^{(i)}}(Q) \}$. 

\begin{example}
\label{writeup-weird-ex} 
Let CCQ query $Q$ be as follows. 
\begin{tabbing} 
Hehetab b \= hehe \kill
$Q(X_1) \leftarrow r(X_1,Y_1,Y_2,X_2; Y_3), r(X_1,Y_1,Y_2,X_3; Y_4),$ \\ 
\> $\{ Y_1,Y_2,Y_3,Y_4 \} .$ 
\end{tabbing} 

\noindent 
(This is the query $Q$ of Example~\ref{intro-weird-ex}, rendered here using somewhat different notation.) 


\paragraph{Toward Constructing Function ${\cal F}_{(Q)}^{(Q)}$ for the Databases $\{$ $D_{\bar{N}^{(i)}}(Q)$ $\}$}
We show here how to develop a database in the family of databases $\{ D_{\bar{N}^{(i)}}(Q) \}$ for the query $Q$, and start constructing function ${\cal F}_{(Q)}^{(Q)}$ w.r.t. the databases in the family. (We will complete the construction in Example~\ref{again-writeup-weird-ex} in Section~\ref{second-beyond-easy-case-sec}.)

We begin by following Section~\ref{db-constr-sec} of the proof of Theorem~\ref{magic-mapping-prop}. Fix an $i \in {\mathbb N}_+$. Let the vector ${\bar{N}}^{(i)}$, for this fixed $i$, of values of the variables in the vector $\bar{N} = [N_1 \ N_2 \ N_3 \ N_4]$, be $\bar{N}^{(i)} = [1 \ 2 \ 3 \ 5]$.  Here, each $N_j$ in $\bar{N}$ is generated for the variable $Y_j$ of $Q$, for $j \in \{ 1,2,3,4 \}$. We use $\nu_0(X_1) = a$ (hence $t^*_Q = (a)$), $\nu_0(X_2) = b$, and $\nu_0(X_3) = c$. Let $S^{(i)}_1 = \{ e \}$, and let $S^{(i)}_2 = \{ f,g \}$. These setting generate, for the fixed $i$, the database $D_{\bar{N}^{(i)}}(Q) = \{ \ r(a,e,f,b; 3),$ $r(a,e,g,b; 3),$ $r(a,e,f,c; 5),$ $r(a,e,g,c; 5) \  \}$. We will refer to the ground atoms in the set $D_{\bar{N}^{(i)}}(Q)$, from left to right, as $d_1$ through $d_4$. Denote by $h_1$ the first subgoal of the query $Q$, and by $h_2$ its second subgoal. By construction of $D_{\bar{N}^{(i)}}(Q)$, we have that $\psi^{gen(Q)}_{\bar{N}^{(i)}}[d_1]$ $=$ $\psi^{gen(Q)}_{\bar{N}^{(i)}}[d_2]$ $=$ $h_1$, and that  $\psi^{gen(Q)}_{\bar{N}^{(i)}}[d_3]$ $=$ $\psi^{gen(Q)}_{\bar{N}^{(i)}}[d_4]$ $=$ $h_2$. 

We now follow Sections~\ref{valid-map-sec} through \ref{multivar-polyn-basics-sec} of the proof of Theorem~\ref{magic-mapping-prop}, to construct the monomial classes for the function ${\cal F}_{(Q)}^{(Q)}$, for the query $Q$ and for the database  $D_{\bar{N}^{(i)}}(Q)$ as generated above in this example. 
As a result of the construction steps,\footnote{These steps are outlined in significant detail in Example~\ref{main-proof-ex}, albeit using queries that are different from the queries of the current example.} we obtain four monomial classes for the query $Q$: 
\begin{itemize} 
	\item Monomial class ${\cal C}_1^{(Q)}$ has noncopy-signature $[Y_1 \ Y_2]$ and copy-signature $[N_3 \ N_3]$; it contributes  to the set $\Gamma^{t^*_Q}_{\bar{S}}$ $(Q,D_{\bar{N}^{(i)}}(Q))$, with columns (from left to right) $X_1$ $Y_1$ $Y_2$ $Y_3$ $Y_4$, nine tuples $(a,e,f,1,1)$ through $(a,e,f,3,3)$ (that is, tuples $(a,e,f,1,1)$, $(a,e,f,1,2)$, $(a,e,f,1,3)$, $(a,e,f,2,1)$, $\ldots,$ $(a,e,f,3,2)$, $(a,e,f,3,3)$), as well as nine tuples $(a,e,g,1,1)$ through $(a,e,g,3,3)$. 
	\item Monomial class ${\cal C}_2^{(Q)}$ has noncopy-signature $[Y_1 \ Y_2]$ and copy-signature $[N_3 \ N_4]$; it contributes  to the set $\Gamma^{t^*_Q}_{\bar{S}}$ $(Q,D_{\bar{N}^{(i)}}(Q))$ fifteen tuples $(a,e,f,1,1)$ through $(a,e,f,3,5)$, as well as fifteen tuples $(a,e,g,1,1)$ through $(a,e,g,3,5)$. 

	\item Monomial class ${\cal C}_3^{(Q)}$ has noncopy-signature $[Y_1 \ Y_2]$ and copy-signature $[N_4 \ N_3]$; it contributes  to the set $\Gamma^{t^*_Q}_{\bar{S}}$ $(Q,D_{\bar{N}^{(i)}}(Q))$ fifteen tuples $(a,e,f,1,1)$ through $(a,e,f,5,3)$, as well as fifteen tuples $(a,e,g,1,1)$ through $(a,e,g,5,3)$. 
	\item Monomial class ${\cal C}_4^{(Q)}$ has noncopy-signature $[Y_1 \ Y_2]$ and copy-signature $[N_4 \ N_4]$; it contributes  to the set $\Gamma^{t^*_Q}_{\bar{S}}$ $(Q,D_{\bar{N}^{(i)}}(Q))$ twenty five tuples $(a,e,f,1,1)$ through $(a,e,f,5,5)$, as well as twenty five tuples $(a,e,g,1,1)$ through $(a,e,g,5,5)$. 
\end{itemize} 

While all four of the above monomial classes have the same noncopy-signature, none of the classes unconditionally dominates (see Definition~\ref{uncond-dom-def}) any other monomial class in the set $\{ {\cal C}_1^{(Q)},$ ${\cal C}_2^{(Q)},$ ${\cal C}_3^{(Q)}$, ${\cal C}_4^{(Q)} \}$. 
\end{example}

\subsubsection{The General Case of Constructing ${\cal F}_{(Q)}^{(Q'')}$}
\label{hard-funct-case-sec} 

\reminder{In this subsection and beyond, need to synchronize with Example~\ref{writeup-weird-ex} the following terminology: 

\begin{itemize}
	\item implicit-wave CCQ query
	\item explicit-wave CCQ query
	\item characteristic-wave monomial class 
	\item characteristic-wave monomial 
	\item total-order vector 
	\item multivariate polynomial ${\cal F}_{(Q)}^{[1 \ N_3 \ N_4]}(Q'')$ ??? 
	\item this monomial class order-dominates all the elements of the set $\{ {\cal C}_1^{(Q)},$ ${\cal C}_2^{(Q)},$ ${\cal C}_3^{(Q)}$, ${\cal C}_4^{(Q)} \}$ w.r.t. the total order associated with the total-order vector $[1 \ N_3 \ N_4]$. 
\end{itemize}
} 

In this subsection we address the construction of the function ${\cal F}_{(Q)}^{(Q'')}$ for the general case, as opposed to the (easy) case considered in Section~\ref{easy-funct-case-sec}. That is, we introduce an approach to computing,  for a query $Q''$ and database $D_{\bar{N}^{(i)}}(Q)$,  the cardinality of the set $\Gamma^{(t^*_Q)}(Q'',D_{\bar{N}^{(i)}}(Q))$ -- and therefore the multiplicity of the tuple $t^*_Q$ in the bag ${\sc Res}_C(Q'',D_{\bar{N}^{(i)}}(Q))$ -- for those cases  where  at least two distinct elements of the set  ${\mathbb C}^{nondom}$ $(Q'')$ could have the same noncopy-signature.  
The approach introduced in this subsection is applicable to constructing the function ${\cal F}_{(Q)}^{(Q'')}$ for {\em all} cases, including the special case of  Section~\ref{easy-funct-case-sec}.  

Consider Example~\ref{writeup-weird-ex} of Section~\ref{beyond-easy-case-sec}: 
For the CCQ query $Q$ and for the family of databases $\{ D_{\bar{N}^{(i)}}(Q) \}$ of the example, the set ${\cal C}(Q) = \{ {\cal C}_1^{(Q)}, {\cal C}_2^{(Q)}, {\cal C}_3^{(Q)}, {\cal C}_4^{(Q)} \}$ has four monomial classes with the same noncopy-signature $[Y_1 \ Y_2]$ and with the respective copy-signatures $[N_3 \ N_3]$, $[N_3 \ N_4]$, $[N_4 \ N_3]$, and $[N_4 \ N_4]$. Clearly, no unconditional-dominance of Definition~\ref{uncond-dom-def} holds for any pair of monomial classes in the set ${\cal C}(Q)$. Hence the set ${\mathbb C}^{nondom}(Q)$ (see Section~\ref{multivar-polyn-basics-sec}) is the set ${\cal  C}(Q)$. Further, as  the set ${\mathbb C}^{nondom}(Q)$ does not satisfy the conditions of Proposition~\ref{easy-final-funct-prop},  the function ${\cal F}_{(Q)}^{(Q)}$ for the example {\em cannot} be constructed using Proposition~\ref{easy-final-funct-prop}. Indeed, it is easy to see that  ${\cal F}_{(Q)}^{(Q)}$ for  Example~\ref{writeup-weird-ex} is {\em not} the sum of the multiplicity monomials for the elements of the above set ${\mathbb C}^{nondom}(Q)$. Specifically, Example~\ref{writeup-weird-ex}  shows that w.r.t. the fixed database $D_{\bar{N}^{(i)}}(Q)$ used in the example, each element of the set ${\mathbb C}^{nondom}(Q)$ contributes to the set $\Gamma^{(t^*_Q)}(Q,D_{\bar{N}^{(i)}}(Q))$ {\em the same} tuple $(a,e,f,1,1)$. 





We summarize that the problem with the general case considered in this subsection 
is that the multiplicity of the tuples contributed, to the set $\Gamma^{(t^*_Q)}(\ldots)$, by distinct monomial classes for the query in question, 
 cannot always be added up to obtain the correct total contribution of the classes to that set. At the same time, we know from Proposition~\ref{gamma-i-union-prop} 
 that for each $i \in {\mathbb N}_+$, the size of the set $\Gamma^{(t^*_Q)}(Q'',D_{\bar{N}^{(i)}}(Q))$ is the size of the union $\bigcup_{j=1}^{n^*}$  $\Gamma^{(i)}[{\cal C}_j^{(Q'')}]$, over all the nonempty monomial classes ${\cal C}_1^{(Q'')},$ $\ldots,$ ${\cal C}_{n^*}^{(Q'')}$ for the query $Q''$ w.r.t. the family of databases $\{ D_{\bar{N}^{(i)}}(Q) \}$. We also know, from Proposition~\ref{diff-noncopy-prop}, 
that for each pair $({\cal C}_n^{(Q'')},{\cal C}_p^{(Q'')})$ of distinct monomial classes among ${\cal C}_1^{(Q'')},$ $\ldots,$ ${\cal C}_{n^*}^{(Q'')}$, such that ${\cal C}_n^{(Q'')}$ and ${\cal C}_p^{(Q'')}$ have distinct noncopy signatures, it holds that the intersection of the sets $\Gamma^{(i)}[{\cal C}_n^{(Q'')}]$ and $\Gamma^{(i)}[{\cal C}_p^{(Q'')}]$ is empty for each $i \in {\mathbb N}_+$. Thus, to obtain the function ${\cal F}_{(Q)}^{(Q'')}$ for the general case, it remains to consider the (perhaps nonempty) intersections of the sets $\Gamma^{(i)}[{\cal C}_n^{(Q'')}]$ and $\Gamma^{(i)}[{\cal C}_p^{(Q'')}]$ {\em only} for those pairs $({\cal C}_n^{(Q'')},{\cal C}_p^{(Q'')})$ where ${\cal C}_n^{(Q'')}$ and ${\cal C}_p^{(Q'')}$ have {\em the same} noncopy signature. 

Thus, the two last missing links in (finally) constructing the function ${\cal F}_{(Q)}^{(Q'')}$ for the general case, are based on the following two results, Propositions~\ref{bigunion-prop} and~\ref{multip-monomial-prop}. Example~\ref{again-writeup-weird-ex} in Section~\ref{second-beyond-easy-case-sec} provides an illustration of the construction of the function ${\cal F}_{(Q)}^{(Q)}$ for the query $Q$  of Example~\ref{writeup-weird-ex} in Section~\ref{beyond-easy-case-sec}.

\begin{proposition} 
\label{bigunion-prop}
Suppose the monomial classes in the set ${\cal C}[Q'']$ $=$ $\{ {\cal C}_1^{(Q'')},$ $\ldots,$ ${\cal C}_{n^*}^{(Q'')} \}$ are indexed (by $1$, $2$, $\ldots,$ $n^*$) in such a way that for all triples $({\cal C}_{j_1}^{(Q'')},$ ${\cal C}_{j_2}^{(Q'')},$ ${\cal C}_{j_3}^{(Q'')})$,  with $1$ $\leq$ $j_1$ $<$ $j_2$ $<$ $j_3$ $\leq$ $n^*$, it cannot be that (a) ${\cal C}_{j_1}^{(Q'')}$ and ${\cal C}_{j_3}^{(Q'')}$ have the same noncopy signature, and (b) ${\cal C}_{j_1}^{(Q'')}$ and ${\cal C}_{j_2}^{(Q'')}$ have different noncopy signatures. 
Further, let $n \in \{ 1,\ldots,n^* \}$ be such that $n$ is the number of distinct noncopy signatures of all the elements of the set ${\cal C}[Q'']$. Finally, let $k_0$ $=$ $0$ and, for this value of $n$, let $1$ $\leq$ $k_1$ $<$ $k_2$ $<$ $\ldots$ $<$ $k_n$ $=$ $n^*$ be such that for each $j$ $\in$ $\{ 1,$ $2,$ $\ldots,$ $n \}$, all monomial classes ${\cal C}_{k_{j-1}+1}^{(Q'')}$, ${\cal C}_{k_{j-1}+2}^{(Q'')}$, $\ldots$, ${\cal C}_{k_j}^{(Q'')}$ have the same noncopy signature.\footnote{All of the above conditions together just say that the elements of the set ${\cal C}[Q'']$ are indexed in such a way that, in the sequence ${\cal C}_1^{(Q'')},$ $\ldots,$ ${\cal C}_{n^*}^{(Q'')}$, we first have all the monomial classes with some noncopy-signature $NS_1$, then all the monomial classes with a different noncopy-signature $NS_2$, and so on. That is, for each noncopy-signature, $NS$, of at least one element of the set ${\cal C}[Q'']$, all monomial classes in ${\cal C}[Q'']$ that have the noncopy-signature $NS$ are ``grouped together'' in the sequence  ${\cal C}_1^{(Q'')},$ $\ldots,$ ${\cal C}_{n^*}^{(Q'')}$.} 

Let $i \in {\mathbb N}_+$. Then the cardinality of the set $\Gamma^{(t^*_Q)}(Q'',$ $D_{\bar{N}^{(i)}}(Q))$ is given exactly as the sum 
$$\sum_{j=0}^{n-1} \hspace{1cm} | \bigcup_{l=1}^{(k_{j+1})-(k_j)} \Gamma^{(i)}[{\cal C}_{(k_j)+l}^{(Q'')}] \hspace{0.3cm} | \ .$$ 
(Here, each ${\cal C}_{(k_j)+l}^{(Q'')}$ referenced in the formula is the$ {(k_j+l)}$'th element of the set  ${\cal C}[Q'']$ $=$ $\{ {\cal C}_1^{(Q'')},$ $\ldots,$ ${\cal C}_{n^*}^{(Q'')} \}$, under a fixed ordering of the elements of the set as specified in the beginning of the statement of this result.)  
\end{proposition} 

(As usual, we denote by $|S|$ the cardinality of the set $S$. The result of Proposition~\ref{bigunion-prop} is immediate from Propositions~\ref{gamma-i-union-prop} and~\ref{diff-noncopy-prop}.) 

Now we will be able to compute correctly the function  ${\cal F}_{(Q)}^{(Q'')}$ for each $i \in {\mathbb N}_+$, as soon as we are able to evaluate the formulas of the form 
\begin{equation} 
\label{union-eq}
| \bigcup_{l=1}^{(k_{j+1})-(k_j)} \Gamma^{(i)}[{\cal C}_{(k_j)+l}^{(Q'')}] \hspace{0.3cm} | , 
 \end{equation} 
as introduced in Proposition~\ref{bigunion-prop}. We compute the value of such formulas using the basic {\em inclusion-exclusion principle} for computing the cardinality of the union of several sets. 
All that the inclusion-exclusion principle requires as inputs is the cardinalities of the {\em intersections} of the relevant (groups of) sets. (We handle the case of determining the size of each individual set,  $S$, in the input to the cardinality-of-union formula, as the special case of ``intersection of $S$ with itself.'' As will be clear from the statement of Proposition~\ref{multip-monomial-prop}, this special case is captured correctly -- as expected -- by Proposition~\ref{monomial-classes-sizes-prop}.) 

Thus, our next result, Proposition~\ref{multip-monomial-prop},  is the final missing link in the construction of the function  ${\cal F}_{(Q)}^{(Q'')}$, as Proposition~\ref{multip-monomial-prop} tells us how to compute correctly the cardinalities of the intersections of  sets of the form $\Gamma^{(i)}[{\cal C}^{(Q'')}]$, using {\em only the elements of the vector $\bar{N}$,}  that is only variables $N_1$ through $N_{m+w}$ and nothing else. (More precisely, Proposition~\ref{multip-monomial-prop} gives us a formula where, for each specific $i \in {\mathbb N}_+$, we can compute the cardinalities of all the requisite intersections by using the specific values, in ${\bar N}^{(i)}$ for this value $i$, of the respective variables in $\bar{N}$. The formula itself is in terms of $\bar{N}$ only, and does not use ${\bar N}^{(i)}$.) 

For the formulation of Proposition~\ref{multip-monomial-prop}, assume that in the set ${\cal C}[Q'']$ $=$ $\{ {\cal C}_1^{(Q'')},$ $\ldots,$ ${\cal C}_{n^*}^{(Q'')} \}$ there exist (at least) $k$ monomial classes, for some $k \in \{ 1,\ldots,n^* \}$, whose noncopy signature is a given vector $\Xi$ of length $m$. Suppose that for some fixed $i \in {\mathbb N}_+$, we want to compute the cardinality of the intersection of the sets $\Gamma^{(i)}$ (using the notation of Proposition~\ref{bigunion-prop}) for exactly these $k$ elements of the set ${\cal C}[Q'']$. To make easier the notation in the formal results to follow, assume w.l.o.g. that the elements of the set ${\cal C}[Q'']$ are indexed in such a way that for all these chosen $k$ elements of ${\cal C}[Q'']$ that have noncopy-signature $\Xi$, these $k$ monomial classes {\em are the first $k$ elements of the set} ${\cal C}[Q'']$. (That is, these $k$ monomial classes are the elements ${\cal C}_1^{(Q'')},$ $\ldots,$ ${\cal C}_{k}^{(Q'')}$ of the set ${\cal C}[Q'']$.) 

Now let us refer to the vector $\Xi$ as ${\Phi}_n^{{\cal C}_1^{(Q'')}}$. (By our indexing of the elements of the set ${\cal C}[Q'']$ , as introduced in the previous paragraph, the noncopy-signature of the monomial class ${\cal C}_1^{(Q'')}$ is exactly $\Xi$.) The reason that we want to refer to the vector $\Xi$ as ${\Phi}_n^{{\cal C}_1^{(Q'')}}$ is that we want, in the formal results to follow, to use the notation $\Pi_{{\Phi}_n^{{\cal C}_1^{(Q'')}}}$ introduced in Section~\ref{multiplicity-monomial-sec}. 

We also use the following notation of Section~\ref{multivar-polyn-basics-sec}: For an arbitrary monomial class ${\cal C}_l^{(Q'')}$ $\in$ $\{ {\cal C}_1^{(Q'')},$ $\ldots,$ ${\cal C}_{k}^{(Q'')} \}$, for the $k$ $\in$ $\{ 1,\ldots,n^* \}$ fixed as explained above, in case where $r \geq 1$, we denote the elements of the copy-signature vector $\Phi_c[{\cal C}_l^{(Q'')}]$ as $[V_{j1[l]},V_{j2[l]},\ldots,V_{jr[l]}]$. (In case where $r = 0$, the copy-signature vector of each of ${\cal C}_1^{(Q'')},$ $\ldots,$ ${\cal C}_{k}^{(Q'')}$ is the empty vector by definition.) Further, in case $r \geq 1$, for the element $V_{js[l]}$ of the vector $\Phi_c[{\cal C}_l^{(Q'')}]$ (for an arbitrary $s \in \{ 1,\ldots,r \}$) and for an $i \in {\mathbb N}_+$, we denote by $V^{(i)}_{js[l]}$ (a) the constant $1$ in case $V_{js[l]} = 1$, and (b) the value $N_u^{(i)}$ from the vector $\bar{N}^{(i)}$ in case $V_{js[l]}$ is the element $N_u$, for an $u \in \{ m+1,\ldots,m+w \}$, of the vector $\bar{N}$. 

We are finally ready to phrase the final formal result needed in the construction of the function ${\cal F}_{(Q)}^{(Q'')}$. As has been noted earlier in this subsection, Example~\ref{again-writeup-weird-ex} in Section~\ref{second-beyond-easy-case-sec} provides an illustration of the construction of the function ${\cal F}_{(Q)}^{(Q)}$ for the query $Q$ and for the database of Example~\ref{writeup-weird-ex} in Section~\ref{beyond-easy-case-sec}. 


\begin{proposition} 
\label{multip-monomial-prop} 
In the set ${\cal C}[Q'']$ $=$ $\{ {\cal C}_1^{(Q'')},$ $\ldots,$ ${\cal C}_{n^*}^{(Q'')} \}$, let (at least) the first $k$ elements, for some $k \in \{ 1,\ldots,n^* \}$, have the same noncopy-signature ${\Phi}_n^{{\cal C}_1^{(Q'')}}$. Then, for an arbitrary $i \in {\mathbb N}_+$, the cardinality of the set 
$$\bigcap_{s=1}^{k} \Gamma^{(i)}[{\cal C}_{s}^{(Q'')}]$$ 
is provided by substituting the constants in ${\bar N}^{(i)}$ as the values of the respective variables in $\bar{N}$,  into the formula: 
\begin{itemize} 
	\item $\Pi_{{\Phi}_n^{{\cal C}_1^{(Q'')}}}$, in case where $r = 0$; and 
	\item $\Pi_{{\Phi}_n^{{\cal C}_1^{(Q'')}}}$ $\times$ $\Pi_{u=1}^r \ min(V_{ju[1]}, V_{ju[2]},\ldots,V_{ju[k]})$, in case where $r \geq 1$. 
\end{itemize} 
\end{proposition} 

\begin{proof}  
For the case where $r = 0$, observe that for each pair of monomial classes among ${\cal C}_1^{(Q'')},$ $\ldots,$ ${\cal C}_{k}^{(Q'')}$, the monomial classes in the pair unconditionally dominate each other, by Definition~\ref{uncond-dom-def}. Therefore,  the result of Proposition~\ref{multip-monomial-prop}  is immediate from  Proposition~\ref{uncond-dom-prop}.  
 
For the case where $r \geq 1$, the result of Proposition~\ref{multip-monomial-prop}  is immediate from Lemma~\ref{multip-monomial-lemma}. 
\end{proof} 

To formulate Lemma~\ref{multip-monomial-lemma}, we use the following terminology. For an element ${\cal C}^{(Q'')}$ of the set ${\cal C}[Q'']$ $=$ $\{ {\cal C}_1^{(Q'')},$ $\ldots,$ ${\cal C}_{n^*}^{(Q'')} \}$, and for an $i \in {\mathbb N}_+$, consider the set $\Gamma^{(i)}[{\cal C}^{(Q'')}]$. In case where $m \geq 1$, let an $m$-tuple $t^{(M)}$ be an arbitrary tuple in the projection of the set $\Gamma^{(i)}[{\cal C}^{(Q'')}]$ on all the multiset noncopy variables of the query $Q''$ (in some arbitrary fixed order of these variables). Then we say that $z \geq 1$ tuples $t_1$, $t_2$, $\ldots,$ $t_z$ in 
the set $\Gamma^{(i)}[{\cal C}^{(Q'')}]$ {\em agree on the multiset-noncopy projection} $t^{(M)}$, if we have that the set projection of the subset $\{ t_1, t_2, \ldots, t_z \}$ of the set $\Gamma^{(i)}[{\cal C}^{(Q'')}]$ on all the multiset noncopy variables of the query $Q''$ (in the same fixed order) is a singleton set $\{ t^{(M)} \}$. In case where $m = 0$, we say that (by default) all the tuples in the set $\Gamma^{(i)}[{\cal C}^{(Q'')}]$ agree on the multiset-noncopy projection, which is the empty tuple when $m$ $=$ $0$. 
 
\begin{lemma} 
\label{multip-monomial-lemma} 
Suppose $r \geq 1$. 
In the set ${\cal C}[Q'']$ $=$ $\{ {\cal C}_1^{(Q'')},$ $\ldots,$ ${\cal C}_{n^*}^{(Q'')} \}$, let (at least) the first $k$ elements, for some $k \in \{ 1,\ldots,n^* \}$, have the same noncopy-signature ${\Phi}_n^{{\cal C}_1^{(Q'')}}$. Let $i \in {\mathbb N}_+$. Let $t^{(M)}$ be an arbitrary tuple in the projection of the set $${\cal S} = \bigcap_{s=1}^{k} \Gamma^{(i)}[{\cal C}_{s}^{(Q'')}]$$ on all the multiset noncopy variables of the query $Q''$, in case $m \geq 1$, and let  $t^{(M)}$ be the empty tuple in case $m = 0$. Then, for the number $K$ of all those tuples in the set ${\cal S} = \bigcap_{s=1}^{k} \Gamma^{(i)}[{\cal C}_{s}^{(Q'')}]$ that agree on the multiset-noncopy projection $t^{(M)}$, we have that the value of $K$ is provided by substituting the constants in ${\bar N}^{(i)}$ as the values of the respective variables in $\bar{N}$,  into the formula 
$$K = \Pi_{u=1}^r \ min(V_{ju[1]}, V_{ju[2]},\ldots,V_{ju[k]}) .$$  
\end{lemma}

\begin{proof}{(sketch)}
Assume a fixed $i \in {\mathbb N}_+$. The proof of Lemma~\ref{multip-monomial-lemma} is immediate from Proposition~\ref{copy-proj-properties-prop}, 
which exhibits the structure of the projection of the set  $\Gamma^{(i)}[{\cal C}^{(Q'')}]$ (for an arbitrary monomial class ${\cal C}^{(Q'')}$ in the set $\{ {\cal C}_1^{(Q'')},$ $\ldots,$ ${\cal C}_{n^*}^{(Q'')} \}$)  on the set of all copy variables of the query $Q''$, and from  Proposition~\ref{noncopy-proj-more-properties-prop}, 
which explores the ``symmetries'' of the set  $\Gamma^{(i)}[{\cal C}^{(Q'')}]$ on the databases in the family $\{ D_{\bar{N}^{(i)}}(Q) \}$. Specifically, we have that: 
\begin{itemize} 
	\item The values in the projection of the set  $\Gamma^{(i)}[{\cal C}^{(Q'')}]$ on the set of all copy variables of the query $Q''$ are natural numbers in a specified range according to the copy signature of the monomial class  ${\cal C}^{(Q'')}$. More precisely, let the copy signature for the monomial class ${\cal C}^{(Q'')}$ be $[V_{j1}$ $\ V_{j2}$ $\ \ldots$ $\ V_{jr}]$. Then for each $u \in \{ 1,\ldots,r  \}$, each value in the projection of $\Gamma^{(i)}[{\cal C}^{(Q'')}]$ onto the copy variable $Y''_{m+u}$ of $Q''$ is a natural number belonging to the set $\{ 1,\ldots, V_{ju}^{(i)}\}$. Moreover, for each   $u \in \{ 1,\ldots,r  \}$ and for each value $v_u \in \{ 1,\ldots, V_{ju}^{(i)}\}$, the tuple $(v_1,v_2,\ldots,v_r)$ is in the projection of $\Gamma^{(i)}[{\cal C}^{(Q'')}]$ onto all the copy variables $Y''_{m+1},$ $Y''_{m+2},$ $\ldots,$ $Y''_{m+r}$ of $Q''$, in this order. 
	\item Consider now the monomial classes ${\cal C}_1^{(Q'')},$ $\ldots,$ ${\cal C}_{k}^{(Q'')}$ in the statement of Lemma~\ref{multip-monomial-lemma}. For the fixed $i \in {\mathbb N}_+$ and for each $u \in \{ 1,\ldots,r  \}$, denote by $Z_u$ the value $min(V^{(i)}_{ju[1]}, V^{(i)}_{ju[2]},\ldots,V^{(i)}_{ju[k]})$. Then we can show that: 
	
	\begin{itemize} 
		\item For each   $u \in \{ 1,\ldots,r  \}$ and for each value $v_u \in \{ 1,\ldots, Z_u\}$, the tuple $(v_1,v_2,\ldots,v_r)$ is in the projection of the set $\bigcap_{s=1}^{k} \Gamma^{(i)}[{\cal C}_{s}^{(Q'')}]$ onto all the copy variables $Y''_{m+1},$ $Y''_{m+2},$ $\ldots,$ $Y''_{m+r}$ of $Q''$, in this order; and 

		\item Whenever, for at least one $u \in \{ 1,\ldots,r \}$, the value $v_u$ is {\em not} an element of the set $\{ 1,\ldots, Z_u\}$, then we have that the tuple $(v_1,v_2,\ldots,v_r)$ is {\em not} in the projection of the set $\bigcap_{s=1}^{k} \Gamma^{(i)}[{\cal C}_{s}^{(Q'')}]$ onto all the copy variables $Y''_{m+1},$ $Y''_{m+2},$ $\ldots,$ $Y''_{m+r}$ of $Q''$, in this order. 
	\end{itemize} 
\end{itemize} 
\end{proof}


At the conclusion of this subsection, we observe that by the inclusion-exclusion principle for unions of sets, the value of the function ${\cal F}_{(Q)}^{(Q'')}$ for each $i \in {\mathbb N}_+$, that is the cardinality of the set $\Gamma^{(t^*_Q)}(Q'',D_{\bar{N}^{(i)}}(Q))$, can also be computed exactly using the set ${\mathbb C}^{nondom}(Q'')$ of Section~\ref{multivar-polyn-basics-sec}. That is, we can use  the set ${\mathbb C}^{nondom}(Q'')$, rather than the set $\{ {\cal C}_1^{(Q'')},$ $\ldots,$ ${\cal C}_{n^*}^{(Q'')} \}$ of {\em all} nonempty monomial classes for query $Q''$ and database $D_{\bar{N}^{(i)}}(Q)$ (cf. Proposition~\ref{bigunion-prop}):  

\begin{proposition} 
\label{bigunion-nondom-prop}
Suppose the monomial classes in the set ${\mathbb C}^{nondom}(Q'')$ $=$ $\{ {\cal C}_1^{(Q'')},$ $\ldots,$ ${\cal C}_{p}^{(Q'')} \}$ are indexed (by $1$, $2$, $\ldots,$ $p$) in such a way that for all triples $({\cal C}_{j_1}^{(Q'')},$ ${\cal C}_{j_2}^{(Q'')},$ ${\cal C}_{j_3}^{(Q'')})$,  with $1$ $\leq$ $j_1$ $<$ $j_2$ $<$ $j_3$ $\leq$ $p$, it cannot be that (a) ${\cal C}_{j_1}^{(Q'')}$ and ${\cal C}_{j_3}^{(Q'')}$ have the same noncopy signature, and (b) ${\cal C}_{j_1}^{(Q'')}$ and ${\cal C}_{j_2}^{(Q'')}$ have different noncopy signatures. Further, let $n \in \{ 1,\ldots,p \}$ be such that $n$ is the number of distinct noncopy signatures of all the elements of the set ${\cal C}[Q'']$. Finally, let $k_0$ $=$ $0$ and, for this value of $n$, let $1$ $\leq$ $k_1$ $<$ $k_2$ $<$ $\ldots$ $<$ $k_n$ $=$ $p$ be such that for each $j$ $\in$ $\{ 1,$ $2,$ $\ldots,$ $n \}$, all monomial classes ${\cal C}_{k_{j-1}+1}^{(Q'')}$, ${\cal C}_{k_{j-1}+2}^{(Q'')}$, $\ldots$, ${\cal C}_{k_j}^{(Q'')}$ have the same noncopy signature.\footnote{Similarly to the condition of Proposition~\ref{bigunion-prop}, all of the above conditions together just say that the elements of the set ${\mathbb C}^{nondom}(Q'')$ are indexed in such a way that, in the sequence ${\cal C}_1^{(Q'')},$ $\ldots,$ ${\cal C}_{p}^{(Q'')}$, we first have all the monomial classes with some noncopy-signature $NS_1$, then all the monomial classes with a different noncopy-signature $NS_2$, and so on. That is, for each noncopy-signature, $NS$, of at least one element of the set ${\mathbb C}^{nondom}(Q'')$, all monomial classes in ${\mathbb C}^{nondom}(Q'')$ that have the noncopy-signature $NS$ are ``grouped together'' in the sequence  ${\cal C}_1^{(Q'')},$ $\ldots,$ ${\cal C}_{p}^{(Q'')}$.} 

Let $i \in {\mathbb N}_+$. Then the cardinality of the set $\Gamma^{(t^*_Q)}(Q'',$ $D_{\bar{N}^{(i)}}(Q))$ is given exactly as the sum 
$$\sum_{j=0}^{n-1} \hspace{1cm} | \bigcup_{l=1}^{(k_{j+1})-(k_j)} \Gamma^{(i)}[{\cal C}_{(k_j)+l}^{(Q'')}] \hspace{0.3cm} | \ .$$ 
(Here, each ${\cal C}_{(k_j)+l}^{(Q'')}$ referenced in the formula is an element of the set ${\mathbb C}^{nondom}(Q'')$ $=$ $\{ {\cal C}_1^{(Q'')},$ $\ldots,$ ${\cal C}_{p}^{(Q'')} \}$.) 
\end{proposition} 

Observe that Proposition~\ref{easy-final-funct-prop}, which constructs the function ${\cal F}_{(Q)}^{(Q'')}$ for a special ``easy'' case as considered in Section~\ref{easy-funct-case-sec},    
is an immediate corollary of Proposition~\ref{bigunion-nondom-prop} and of the definition of the set ${\mathbb C}^{nondom}(Q'')$.  

For an illustration, consider again the function ${\cal F}_{(Q)}^{(Q)}$ of Example~\ref{main-proof-cont-ex} in Section~\ref{main-proof-cont-ex-sec}. When we construct function ${\cal F}_{(Q)}^{(Q)}$ using all four monomial classes of the example, the inclusion-exclusion formulae of this current subsection correctly account for the fact that the $\Gamma^{(i)}()$ for the monomial class ${\cal C}_1^{(Q)}$ is a subset of the $\Gamma^{(i)}()$ for the monomial class ${\cal C}_2^{(Q)}$ on all the databases in question. We observe the similar effect when considering how the  inclusion-exclusion  formulae  account for the relationship between the monomial classes ${\cal C}_3^{(Q)}$ and ${\cal C}_4^{(Q)}$ of the example. Hence, by the inclusion-exclusion principle for unions of sets, the construction of the function ${\cal F}_{(Q)}^{(Q)}$ using all four monomial classes of the example results in the same function as the construction using the set ${\mathbb C}^{nondom}(Q)$, as shown in the example. 
 
 \subsubsection{Illustration of the general construction}
 \label{second-beyond-easy-case-sec}
 
In this subsection, we 
provide an illustration of the results of Section~\ref{hard-funct-case-sec}, by following the construction of the function ${\cal F}_{(Q)}^{(Q)}$ for the query $Q$ and for the database of Example~\ref{writeup-weird-ex} in Section~\ref{beyond-easy-case-sec}. As discussed in the beginning of  Section~\ref{hard-funct-case-sec}, the construction cannot be carried out correctly when using just the results of the ``easy-case'' Section~\ref{easy-funct-case-sec}. 

\begin{example} 
\label{again-writeup-weird-ex} 
We refer to the query $Q$ and database  $D_{\bar{N}^{(i)}}(Q)$ of Example~\ref{writeup-weird-ex} in Section~\ref{beyond-easy-case-sec}. In this example we construct the function ${\cal F}_{(Q)}^{(Q)}$ for that query $Q$ and for the entire family of databases  $\{ D_{\bar{N}^{(i)}}(Q) \}$, $i \geq 1$. In addition, we illustrate the correctness of the construction, by using the multiplicity of the tuple $t^*_Q$ of Example~\ref{writeup-weird-ex} in the combined-semantics answer to the query $Q$ on the specific database $D_{\bar{N}^{(i)}}(Q)$  of Example~\ref{writeup-weird-ex}. 

Recall that the query $Q$ has four nonempty monomial classes, ${\cal C}_{1}^{(Q)}$, ${\cal C}_{2}^{(Q)}$, ${\cal C}_{3}^{(Q)}$, and ${\cal C}_{4}^{(Q)}$,  w.r.t. the family of databases  $\{ D_{\bar{N}^{(i)}}(Q) \}$. (Refer to Example~\ref{writeup-weird-ex} for the details.) Each of the monomial classes has noncopy-signature $[Y_1 \ Y_2]$; the copy-signatures of the four monomial classes are $[N_3 \ N_3]$, $[N_3 \ N_4]$, $[N_4 \ N_3]$, and $[N_4 \ N_4]$, in this order. 

To construct function   ${\cal F}_{(Q)}^{(Q)}$ for the query $Q$ and for the family of databases  $\{ D_{\bar{N}^{(i)}}(Q) \}$, we use 
	Proposition~\ref{bigunion-prop},  to establish that for each $i \in {\mathbb N}_+$, the cardinality of the set $\Gamma^{(t^*_Q)}(Q,D_{\bar{N}^{(i)}}(Q))$ is given exactly as the sum 
$$| \hspace{0.3cm} \bigcup_{l=1}^{4} \Gamma^{(i)}[{\cal C}_{l}^{(Q)}] \hspace{0.3cm} | \ .$$ 


For greater succinctness of the formulae to follow, we label more compactly each of the sets $\Gamma^{(i)}[{\cal C}_{1}^{(Q)}]$ through $\Gamma^{(i)}[{\cal C}_{4}^{(Q)}]$ used in the above formula, as follows:  Denote $\Gamma^{(i)}[{\cal C}_{1}^{(Q)}]$ by $A$, $\Gamma^{(i)}[{\cal C}_{2}^{(Q)}]$ by $B$, $\Gamma^{(i)}[{\cal C}_{3}^{(Q)}]$ by $C$, and $\Gamma^{(i)}[{\cal C}_{4}^{(Q)}]$ by $D$. Then, by the inclusion-exclusion principle for unions of sets, we have that the above union formula can be rewritten as follows: 

$$|A \ \cup \ B \ \cup \ C \ \cup \ D| = |A| + |B| + |C| + |D| - |A \cap B| - |A \cap C|$$ 
$$- |A \cap D| - |B \cap C| - |B \cap D| - |C \cap D| + |A \cap B \cap C| + |A \cap B \cap D|$$
$$+ |A \cap C \cap D| + |B \cap C \cap D| - |A \cap B \cap C \cap D| .$$

We now use Proposition~\ref{multip-monomial-prop} to obtain the cardinality of each of the set intersections in the right-hand side of this formula. First, observe that the multipliers  $\Pi_{{\Phi}_n^{{\cal C}_l^{(Q)}}}$, for $l \in \{ 1,2,3,4 \}$, are all equal to each other, and are each the product $N_1$ $\times$ $N_2$. (This is due to the fact that all the four monomial classes have the same noncopy signature.) 

Thus, what remains to be done, in the construction of the formula ${\cal F}_{(Q)}^{(Q)}$, is to compute the products $$\Pi_{u=1}^k \ min(V_{ju[1]}, V_{ju[2]}, \ldots, V_{ju[k]})$$ of Proposition~\ref{multip-monomial-prop} for all the above set intersections, for all $k$ between 2 (for $|A \cap B|$, $|A \cap C|$, $\ldots,$ $|C \cap D|$) and 4 (for $|A \cap B \cap C \cap D|$). (For convenience in the statement of Proposition~\ref{multip-monomial-prop}, the indexing in the product $\Pi_{u=1}^k \ min(V_{ju[1]}, V_{ju[2]},$ $\ldots, V_{ju[k]})$ assumes that each time we look at the cardinality of the intersection of the sets $\Gamma^{(i)}()$ for {\em the first} $k$ {\em consecutive} elements of the set $\{ {\cal C}_1^{(Q)}$, $\ldots,$ ${\cal C}_4^{(Q)} \}$. That is, the statement of the Proposition assumes reindexing of the elements of the set $\{ {\cal C}_1^{(Q)}$, $\ldots,$ ${\cal C}_4^{(Q)} \}$ ``as needed.'' This assumption needs to be kept in mind when understanding the consecutive indexing by $u$ in the formula $\Pi_{u=1}^k \ min(V_{ju[1]}, V_{ju[2]}, \ldots, V_{ju[k]})$ in this example.) Then, by multiplying each of these products by $N_1$ $\times$ $N_2$ and by ``putting the multiplication results back correctly'' into our inclusion-exclusion formula for the cardinality of the union of Proposition~\ref{bigunion-prop}, we will obtain the expression for the function ${\cal F}_{(Q)}^{(Q)}$. 

We make the basic observation that each $min()$ expression for this example will result in $N_3$ (when the only value in the $min$ expression is $N_3$ -- that is, when all arguments of the $min$ expression are the same variable $N_3$), in $N_4$  (when the only value in the $min$ expression is $N_4$), or in $min(N_3,N_4)$ (in all the remaining cases, regardless of the number of times each of $N_3$ and $N_4$ is an argument of the min expression). To make the writeup more concise, we refer to the latter minimum expression as $Z$. (That is, we denote by $Z$ the expression $min(N_3,N_4)$.) In addition, we denote by $T$ the term $N_1$ $\times$ $N_2$. {\em As the union expressions of Proposition~\ref{bigunion-prop} are uniform (as expressed using the elements of the vector ${\bar N}^{(i)}$) across all values of $i$ $\in$ ${\mathbb N}_+$, in the remainder of this example we switch to the elements of the vector $\bar{N}$ as basic blocks in the construction of the formula ${\cal F}_{(Q)}^{(Q)}$, and refrain from clarifying all the time that the values for specific $i$ can be obtained by substituting the elements of the vector ${\bar{N}}^{(i)}$ for the respective elements of $\bar N$ in the expressions that we are to obtain.} 

By the formula of Proposition~\ref{multip-monomial-prop}, we obtain that: 
\begin{tabbing} 
$|A|$ $=$ $T \times (N_3)^2$; $|B|$ $=$ $|C|$ $=$ $T \times N_3 \times N_4$; \\
$|D|$ $=$ $T \times (N_4)^2$; $|A \cap B|$ $=$ $|A \cap C|$ $=$ $T \times N_3 \times Z$; \\
 $|B \cap D|$ $=$ $|C \cap D|$ $=$ $T \times N_4 \times Z$ . 
\end{tabbing} 
Further, it is easy to check that each of the remaining cardinalities, in the inclusion-exclusion formula for $|A \cup B \cup C \cup D|$, equals $T \times Z^2$. 

Thus, we obtain that 

$$(|A \cup B \cup C \cup D|) / T = (N_3)^2 + 2 N_3 N_4 + (N_4)^2 - 2 Z N_3 - 2 Z N_4 + Z^2 .$$ 

That is, we obtain that, by Propositions~\ref{bigunion-prop} and~\ref{multip-monomial-prop},  

$${\cal F}_{(Q)}^{(Q)} = N_1 N_2 \times [ (N_3)^2 + 2 N_3 N_4 + (N_4)^2 - 2 Z N_3 - 2 Z N_4 + Z^2 ] .$$ 

Recall that $Z$ here denotes the expression $min(N_3,N_4)$. Observe that in this formula for ${\cal F}_{(Q)}^{(Q)}$, for each of the terms 

$$- 2 N_1 \times N_2 \times min(N_3,N_4) \times N_3 ,$$ 

$$- 2 N_1 \times N_2 \times min(N_3,N_4) \times N_4, \ and$$ 

$$+ N_1 \times N_2 \times (min(N_3,N_4))^2 , $$ 
\noindent
we have that none of the three terms corresponds to monomial classes for the query $Q$. Thus, none of these terms is ``backed up'' by assignments from the query $Q$ to any database $D_{{\bar N}^{(i)}}(Q)$. 

Due to the presence of the term $min(N_3,N_4)$ in the above expression for the function ${\cal F}_{(Q)}^{(Q)}$, the function is not a multivariate polynomial (in terms of the elements of the vector $\bar N$) on the entire domain ${\cal N}$ of the function. 
At the same time: 

\begin{itemize} 
	\item For all $i \in {\mathbb N}_+$ such that $N_3^{(i)} \leq N_4^{(i)}$ in the vector $\bar{N}^{(i)}$, we have that (after we substitute $Z = min(N_3,N_4) = N_3$ and then cancel out in the resulting formula) the function ${\cal F}_{(Q)}^{(Q)}$ on this subdomain of ${\cal N}$ is the following multivariate polynomial in terms of the elements of the vector $\bar N$: 
	$${\cal F}_{(Q)}^{(Q)} = N_1 \times N_2 \times (N_4)^2 .$$ 
	\item Similarly, for all $i \in {\mathbb N}_+$ such that $N_3^{(i)} \geq N_4^{(i)}$ in the vector $\bar{N}^{(i)}$, we have that (after we substitute $Z = min(N_3,N_4) = N_4$ and then cancel out in the resulting formula) the function ${\cal F}_{(Q)}^{(Q)}$  on this subdomain of ${\cal N}$ is the following multivariate polynomial in terms of the elements of the vector $\bar N$: 
	$${\cal F}_{(Q)}^{(Q)} = N_1 \times N_2 \times (N_3)^2 .$$ 
\end{itemize} 

Specifically,  for the vector $\bar{N}^{(i)} = [1 \ 2 \ 3 \ 5]$ of Example~\ref{writeup-weird-ex}, we have that $N_3 = 3 \leq N_4 = 5$. Hence, for this $i$ we have that ${\cal F}_{(Q)}^{(Q)}({\bar N}^{(i)}) = N_1^{(i)} \times N_2^{(i)} \times (N_4^{(i)})^2 .$ Observe that the result of evaluating this expression ${\cal F}_{(Q)}^{(Q)}({\bar N}^{(i)})$ for this $i$ is $1 \times 2 \times (5)^2 = 50 .$ This value $50$ is the correct multiplicity of the tuple $t^*_Q = (a)$ of Example~\ref{writeup-weird-ex} in the combined-semantics answer to the query $Q$ on the specific database $D_{\bar{N}^{(i)}}(Q)$  of Example~\ref{writeup-weird-ex}. (Please refer to Example~\ref{writeup-weird-ex} for the specific 50 tuples in the set $\Gamma_{\bar{S}}(Q,D_{\bar{N}^{(i)}}(Q))$ that generate the tuple $t^*_Q$ in the answer to the query on the database.) 
\end{example} 

\subsubsection{Function for the query $Q'$ of Example 4.1}  
\label{qprime-writeup-weird-sec} 


In this section we illustrate by example that, when the condition (i) of Theorem~\ref{magic-mapping-prop} (i.e., the condition of $Q$ being an explicit-wave query) is not satisfied, then even in case where $Q \equiv_C Q'$ does hold, there does not have to exist a monomial class for the query $Q'$ such that the multiplicity monomial of that class is the wave of the query $Q$. 

\begin{example} 
\label{qprime-writeup-weird-ex} 
Let CCQ query $Q'$ be as follows. 
\begin{tabbing} 
Hehetab b \= hehe \kill
$Q'(X_1) \leftarrow r(X_1,Y_1,Y_2,X_2; Y_3), r(X_1,Y_1,Y_2,X_2; Y_4),$ \\ 
\> $\{ Y_1,Y_2,Y_3,Y_4 \} .$ 
\end{tabbing} 

\noindent 
(This is the query $Q'$ of Example~\ref{intro-weird-ex}, rendered here using somewhat different notation.) 


\paragraph{Toward Constructing the Function ${\cal F}_{(Q)}^{(Q')}$ for the Databases $\{$ $D_{\bar{N}^{(i)}}(Q)$ $\}$ of Example~\ref{writeup-weird-ex}} 

We briefly outline here how we construct the function ${\cal F}_{(Q)}^{(Q')}$ for the databases $\{$ $D_{\bar{N}^{(i)}}(Q)$ $\}$, which (databases) we constructed in Example~\ref{writeup-weird-ex} of Section~\ref{beyond-easy-case-sec}. 
The two monomial classes for the function ${\cal F}_{(Q)}^{(Q')}$, for the query $Q'$ and for the databases  $D_{\bar{N}^{(i)}}(Q)$, are as follows: 
\begin{itemize} 
	\item Monomial class ${\cal C}_1^{(Q')}$ has noncopy-signature $[Y_1 \ Y_2]$ and copy-signature $[N_3 \ N_3]$; intuitively, this monomial class results from mapping the query $Q'$ in all possible ways into those atoms of databases $D_{\bar{N}^{(i)}}(Q)$ that (atoms) are associated with the copy variable of the first subgoal of the query $Q$ (in the specific database $D_{\bar{N}^{(i)}}(Q)$ of Example~\ref{writeup-weird-ex}, these atoms would be ground atoms $d_1$ an $d_2$ -- recall that the copy number in each of those atoms is $3$ $=$ $N_3^{(i)}$); 
	
	and 
	\item Monomial class ${\cal C}_2^{(Q')}$ has noncopy-signature $[Y_1 \ Y_2]$ and copy-signature $[N_4 \ N_4]$; intuitively, this monomial class results from mapping the query $Q'$ in all possible ways into those atoms of databases $D_{\bar{N}^{(i)}}(Q)$ that (atoms) are associated with the copy variable of the {\em second} subgoal of the query $Q$ (in the specific database $D_{\bar{N}^{(i)}}(Q)$ of Example~\ref{writeup-weird-ex}, these atoms would be ground atoms $d_3$ an $d_4$). 
\end{itemize} 

There are no other monomial classes for the query $Q'$ and for the databases  $D_{\bar{N}^{(i)}}(Q)$. The reason is, both subgoals of the query $Q'$ have the same relational template. (Specifically, unlike the subgoals of the query $Q$, the two subgoals of the query $Q'$ use the same set variable $X_2$.) Thus, the two subgoals of the query $Q'$ can only be mapped into one ground atom at a time. At the same time, each database $D_{\bar{N}^{(i)}}(Q)$ always has at least two ``relational templates'' for its ground atoms. (To gain the intuition, recall that by construction of the databases $D_{\bar{N}^{(i)}}(Q)$, the two set variables $X_2$ and $X_3$ used in the two subgoals of the query $Q$ are mapped in the databases into two distinct fixed values.)

As a result, we obtain in a process that is similar to that of Example~\ref{again-writeup-weird-ex} (in Section~\ref{second-beyond-easy-case-sec}) that the closed-form expression function ${\cal F}_{(Q)}^{(Q')}$ is 

\begin{tabbing} 
${\cal F}_{(Q)}^{(Q')}$ $=$ $N_1$ $\times$ $N_2$ $\times$ ($N_3^2 + N_4^2 - (min(N_3, N_4))^2$) . 
\end{tabbing} 

\noindent 
This is the expression for the multiplicity of the tuple $t^*_Q$ $=$ $(a)$ in the answer to {\em the query} $Q'$ on the databases $D_{\bar{N}^{(i)}}(Q)$ {\em constructed for the query} $Q$, hence the expression is in terms of the elements of the vector $\bar N$ {\em for the query} $Q$. We can obtain a more compact expression for ${\cal F}_{(Q)}^{(Q')}$, which is 

\begin{tabbing} 
${\cal F}_{(Q)}^{(Q')}$ $=$ $N_1$ $\times$ $N_2$ $\times$ $max(N_3, N_4))^2$ . 
\end{tabbing} 

\noindent 
Not surprisingly, the latter expression is identical to what we can obtain for ${\cal F}_{(Q)}^{(Q)}$ (for the query $Q$), see the end of Example~\ref{again-writeup-weird-ex} in Section~\ref{second-beyond-easy-case-sec}. (Recall that we have proved that the two queries $Q$ and $Q'$ are combined-semantics equivalent.) 
 
Observe that while $Q$ $\equiv_C$ $Q'$ has been proved to hold, neither monomial class for the query $Q'$ on the databases for the query $Q$ is associated with the wave monomial $\Pi_{j = 1}^4 N_j$ of the query $Q$. Indeed, based on the intuition that we discussed for each monomial class of $Q'$ earlier in this example, we cannot map the two subgoals of the query $Q'$ {\em onto} both subgoals of the query $Q$. (Such a mapping would be the only possibility for a SCVM from $Q'$ to $Q$, as a SCVM must map each copy variable of $Q'$ into a distinct copy variable of $Q$.) The reason that we cannot find such a mapping is simple: The query $Q'$ has the same relational template for both subgoals (note the same set variable in both subgoals of $Q'$), and hence these subgoals are only mappable into one subgoal of $Q$ at a time. 

Finally, the intuition discussed in the preceding paragraph can also be used to understand why there exists a SCVM from $Q$ to $Q'$ (as opposed to from $Q'$ to $Q$). Indeed, as $Q$ has different set variables in its two subgoals, it is easy to see that the condition of $Q$ can be mapped {\em onto} the condition of $Q'$ -- as a result, we can construct a SCVM from $Q$ to $Q'$. 
\end{example}

\subsection{For the $Q$ and $Q'$ such that $Q \equiv_C Q'$, When Does $Q'$ Have the Wave of $Q$?}
\label{q-prime-has-wave-sec}

In Section~\ref{putting-together-f-sec} we learned how to construct, for CCQ queries $Q$ and $Q''$ as specified in Section~\ref{basic-queries-sec}, a function ${\cal F}_{(Q)}^{(Q'')}$. For each $i \in {\mathbb N}_+$, the function ${\cal F}_{(Q)}^{(Q'')}$ returns the multiplicity of the tuple $t^*_Q$ in the bag ${\sc Res}_C(Q'',$ $D_{\bar{N}^{(i)}}(Q))$.  The main result of this current section, Proposition~\ref{qprime-goldfish-prop}, 
shows that, whenever 
\begin{itemize} 
	\item[(a)] $Q \equiv_C Q'$ for CCQ queries $Q$ and $Q'$, and 
	\item[(b)] $Q$ is an explicit-wave CCQ query (as specified by Definition~\ref{expl-wave-def}), 
\end{itemize} 
then there exists a (nonempty) monomial class ${\cal C}_*^{(Q')}$ for the query $Q'$ and for the family of databases $\{ D_{\bar{N}^{(i)}}(Q) \}$, such that the multiplicity monomial of ${\cal C}_*^{(Q')}$ is ``the wave of the query $Q$'' (as specified in Definition~\ref{the-wave-def}). We show the result of Proposition~\ref{qprime-goldfish-prop} using the properties of the functions ${\cal F}_{(Q)}^{(Q)}$ and ${\cal F}_{(Q)}^{(Q')}$. The proof of Theorem~\ref{magic-mapping-prop} is immediate from Proposition~\ref{qprime-goldfish-prop} and from Propositions~\ref{q-has-wave-prop} and \ref{q-same-scale-mpng-prop}  of Section~\ref{monomial-class-mappings-sec}. 

Example~\ref{qprime-writeup-weird-ex} of Section~\ref{qprime-writeup-weird-sec} illustrates that, when the above condition (b) (of $Q$ being an explicit-wave query) is not satisfied, then such a monomial class ${\cal C}_*^{(Q')}$ does not have to exist, and hence a SCVM from  the query $Q'$ to the query $Q$ does not have to exist even in case $Q \equiv_C Q'$.

We begin the technical exposition by stating a useful auxiliary result in Section~\ref{polyn-equiv-sec}. 

\subsubsection{Equivalence of Multivariate Polynomials} 
\label{polyn-equiv-sec}

\begin{proposition}
\label{polyn-prop}
For a positive integer $n$, let $X_1,$ $X_2,$ $\ldots,$ $X_n$ be $n$ distinct variables, where each variable accepts values from (at least) an infinite-cardinality subset of the set ${\mathbb Z}$ of all integers.\footnote{For different variables $X_i$, $X_j$, $i \neq j$, in the set $\{ X_1,$ $\ldots,$ $X_n \}$, the domains of $X_i$ and of $X_j$ may include nonidentical (infinite-cardinality) subsets of the set $\mathbb Z$.} Let each of ${\cal P}_1$ and ${\cal P}_2$ be a finite-degree multivariate polynomial in terms of the variables $X_1,$ $\ldots,$ $X_n$ and with integer coefficients. Further, assume that (w.l.o.g.) ${\cal P}_1 \equiv \hspace{-0.3cm} / \hspace{0.3cm} 0$. Then ${\cal P}_1 - {\cal P}_2 \equiv 0$ if and only if for each term $\Pi_{i = 1}^{n} X_i^{l_i}$, where $l_i \in \{ 0 \} \bigcup {\mathbb N}_+$ for all\footnote{When $l_i = 0$ for all $i \in \{ 1,\ldots,n \}$ in the term $\Pi_{i = 1}^{n} X_i^{l_i}$, we set $\Pi_{i = 1}^{n} X_i^{l_i}$ to the constant 1.} $i \in \{ 1,\ldots,n \}$,  
the term has the same integer coefficient in ${\cal P}_1$ and ${\cal P}_2$. 
\end{proposition} 

\begin{proof}
{\em If:} Immediate from the definitions. 

{\em Only-If:} The proof is by contradiction: Assume that for the finite-degree multivariate polynomial ${\cal P}_1 - {\cal P}_2$, call it $\cal P$, we have that ${\cal P} \equiv 0$.  Assume further that there exists a term, call it $\cal T$, of the form $\Pi_{i = 1}^{n} X_i^{l_i}$, such that the polynomial $\cal P$ has a nonzero integer coefficient for $\cal T$. We will show that in this case, ${\cal P} \equiv 0$ cannot hold, hence we arrive at a contradiction with the assumption ${\cal P} \equiv 0$. 

Case 1: $\cal T$ is the only term with nonzero coefficient in the polynomial $\cal P$, and $l_i = 0$ for all $i \in \{ 1,\ldots,n \}$ in $\cal T$. Then $\cal P$ is equivalent to a nonzero-valued constant function, and the contradiction with the assumption ${\cal P} \equiv 0$ is immediate; Q.E.D. 

Case 2: There exists a nonzero-coefficient term in $\cal P$, call this term ${\cal T}'$, such that there exists a $j \in \{ 1,\ldots,n \}$, where the power $l_j$ of variable $X_j$ in ${\cal T}'$ is a positive integer. Then for each $X_l$ such that $l \in \{ 1,\ldots,n \}$ $-$ $\{ j \}$, fix one arbitrary integer value $x_l \neq 0$ in the domain of $X_l$. (Clearly, it is possible to find a nonzero integer domain value for each $X_l$.) The result of substituting all the values $x_l$, $l \in \{ 1,\ldots,n \}$ $-$ $\{ j \}$,  into the polynomial $\cal P$ is a finite-degree  univariate polynomial with integer coefficients, call it ${\cal P}_{(X_j)}$, in terms of the variable $X_j$ and with at least one term with a nonzero (integer) coefficient. (One term  with a nonzero coefficient in ${\cal P}_{(X_j)}$ results from ${\cal T}'$.)  By our assumption that ${\cal P} \equiv 0$, the value of ${\cal P}_{(X_j)}$ equals zero on the entire infinite integer-valued domain of the variable $X_j$. This is impossible, hence we have arrived at a contradiction with the assumption that ${\cal P} \equiv 0$; Q.E.D. 
\end{proof}

\subsubsection{When the Query $Q'$ Has the Wave of $Q$} 
\label{qprime-goldfish-sec} 

We now state and prove the main result of Section~\ref{q-prime-has-wave-sec}. 

\begin{proposition} 
\label{qprime-goldfish-prop} 
Let $Q$ and $Q'$ be two CCQ queries, such that 
\begin{itemize} 
	\item[(a)] we have that $Q \equiv_C Q'$, and 
	\item[(b)] $Q$ is an explicit-wave CCQ query.  
\end{itemize} 
Then  for the query $Q'$ and for the family of databases $\{ D_{\bar{N}^{(i)}}(Q) \}$, there exists a nonempty monomial class ${\cal C}_*^{(Q')}$, such that the multiplicity monomial of ${\cal C}_*^{(Q')}$ is the wave of the query $Q$. 
\end{proposition} 

The proof of Proposition~\ref{qprime-goldfish-prop}, to be given in Section~\ref{sec-hinge-proof-sec}, hinges 
on several results, which we now proceed to introduce. For the entire exposition, please keep in mind that  throughout the proof of Theorem~\ref{magic-mapping-prop}, all monomial classes of all queries, as well as each of the functions ${\cal F}_{(Q)}^{(Q)}$ and ${\cal F}_{(Q)}^{(Q')}$, are defined w.r.t. the family of databases $\{ D_{\bar{N}^{(i)}}(Q) \}$ for the fixed input query $Q$. 

\subsubsection{Multivariate polynomials on total orders} 
\label{multiv-polyn-sec}

Recall that in general, for CCQ query $Q''$ and for the family of databases $\{ D_{\bar{N}^{(i)}}(Q) \}$ (for CCQ query $Q$), the function ${\cal F}_{(Q)}^{(Q'')}$ for $Q''$ and for $\{ D_{\bar{N}^{(i)}}(Q) \}$  is not a multivariate polynomial on its entire domain ${\cal N}$. (See Example~\ref{again-writeup-weird-ex} for an illustration.) At the same time, it turns out  that the set ${\cal N}$ can be represented as a union of infinite-cardinality sets, such that for each set $S$ in the union, the function ${\cal F}_{(Q)}^{(Q'')}$, for all the elements of the set $S$, can be rewritten equivalently as a multivariate polynomial in terms of the elements of the vector $\bar N$ and with integer coefficients. (That is, for each  ${\bar N}^{(i)}$ $\in$ $S$, the value of ${\cal F}_{(Q)}^{(Q'')}({\bar N}^{(i)})$ can be obtained by substituting the values in ${\bar N}^{(i)}$ into the relevant multivariate polynomial in terms of the elements of the vector $\bar N$ and with integer coefficients.) 

In fact, as we know already for the case $r \leq 1$ (that is, for the input CCQ query $Q$ that has $r = |M_{copy}| \leq 1$), the function  ${\cal F}_{(Q)}^{(Q'')}$ for this case is a multivariate polynomial in terms of the elements of the vector $\bar N$ and with integer coefficients, on the entire domain $\cal N$ of the function. (See Section~\ref{easy-funct-case-sec}.) Hence we proceed to prove the above claim for the case $r \geq 2$. Recall from Proposition~\ref{sc-prop}(iv) that for all $r > 0$ we have that $w > 0$. We conclude that whenever $r \geq 2$, the vector $\bar N$ has at least one element in the sequence $N_{m+1}$ $\ N_{m+2}$ $\ \ldots$ $\ N_{m+w}$. Of these cases, we first consider  the special case $w = 1$, and then the general case $w \geq 1$. 

\paragraph{The special case: $r \geq 2$ and $w = 1$} 
We first consider all those cases (for the input CCQ query $Q$) where $r \geq 2$ and $w = 1$. In all such cases, the copy signatures of all relevant monomial classes (for both $Q$ and $Q''$) are composed, by their definition, of the elements of the set $\{ 1,N_{m+1} \}$. Clearly, then, each $min$ expression of Proposition~\ref{multip-monomial-prop}  in terms of elements of all the relevant copy signatures (in each case where $w = 1$) evaluates to either $1$ or $N_{m+1}$, independently of the value $N_{m+1}^{(i)}$ of $N_{m+1}$ in each vector ${\bar N}^{(i)}$ $\in$ $\cal N$. (Recall that $1 \leq N^{(i)}_{m+1}$ holds for all vectors ${\bar N}^{(i)}$, by definition of the set $\cal N$.) Using the results of Section~\ref{hard-funct-case-sec}, we conclude that in all cases of query $Q$ for which $r \geq 2$ and $w = 1$,   the function  ${\cal F}_{(Q)}^{(Q'')}$ for each such case is a multivariate polynomial in terms of the elements of the vector $\bar N$ and with integer coefficients, on the entire domain $\cal N$ of the function. 

\paragraph{The general case: $r \geq 2$ (and $w \geq 1$)} 
Now consider all those cases (for the input CCQ query $Q$) where $r \geq 2$, and therefore, by Proposition~\ref{sc-prop}(iv), $w \geq 1$. 
We will define the sets $S$ suggested above (such that ${\cal N}$ is a union of such infinite-cardinality sets) using total orders on the elements of the vector ${\bar N}_w$ $=$ $[ \ 1$ $\ N_{m+1}$ $\ N_{m+2}$ $\ \ldots$ $\ N_{m+w} \ ]$. By $ w \geq 1$, we have that the vector ${\bar N}_w$ has at least two elements. (Note: We will see that in the above special case of $r \geq 2$ and $w = 1$, the set $\cal N$ is a union of only one such set $S$.) 

Let vector ${\bar K}_w$ $=$ $[ \ 1$ $\ K_1$ $\ K_2$ $\ \ldots$ $\ K_w \  ]$ be an arbitrary fixed permutation of the vector ${\bar N}_w$ that satisfies the condition that the first element of ${\bar K}_w$ is always the constant $1$. (That is, in each vector ${\bar K}_w$ we have that the sequence $K_1$ $\ K_2$ $\ \ldots$ $\ K_w$ is a permutation of the sequence $N_{m+1}$ $\ N_{m+2}$ $\ \ldots$ $\ N_{m+w}$ in the vector  ${\bar N}_w$.)  We refer to each such vector ${\bar K}_w$ as a {\em copy-variable-ordering vector for the vector} ${\bar N}$. 

Let a total order $\cal O$ on the set $\{ 1,$ $N_{m+1},$ $N_{m+2},$ $\ldots,$ $N_{m+w} \}$ be defined as the reflexive transitive closure on the relation $\{$ $(1, K_1),$ $(K_1, K_2),$ $(K_2, K_3),$ $\ldots,$ $(K_{j}, K_{j+1}),$ $\ldots,$ $(K_{w-1}, K_w)$ $\}$, using the fixed vector ${\bar K}_w$. (We interpret the pair $(1, K_{1})$ in $\cal O$ as $1 \leq K_{1}$. Further, whenever $w \geq 2$,  for $1 \leq j \leq w-1$, we interpret the pair $(K_{j}, K_{j+1})$ in $\cal O$ as $K_j \leq K_{j+1}$. That is, $\cal O$ is the $\leq$ relation.) Then we say that the vector ${\bar K}_w$ {\em determines the total-order relation ${\cal O}$ on} ${\bar N}_w$, and use the notation ${\cal O}^{({\bar K}_w)}$ for that total-order relation $\cal O$. 

Now for a vector ${\bar K}_w$ as above and for an arbitrary vector ${\bar N}^{(i)}$ $\in$ $\cal N$, we define {\em the interpretation of each element of} ${\bar K}_w$ {\em w.r.t.} ${\bar N}^{(i)}$, as follows: 
\begin{itemize} 
	\item[(1)] We define the interpretation of the first element of  ${\bar K}_w$ (that is, of the constant $1$)  to be the constant $1$; and 
	\item[(2)] For each $j$ $\in$ $\{ 2,$ $\ldots,$ $w+1 \}$, suppose that the $j$th element of the vector ${\bar K}_w$ (that is, $K_{j-1}$ in our notation for the vector ${\bar K}_w$) is the variable $N_k$ of $\bar N$, for some $k$ $\in$ $\{ m+1, $ $\ldots,$ $ m+w \}$. Then the interpretation of $K_{j-1}$ w.r.t. ${\bar N}^{(i)}$ is the value $N^{(i)}_k$ in ${\bar N}^{(i)}$ of the variable $N_k$ in $\bar N$. We denote this interpretation of $K_{j-1}$ w.r.t. ${\bar N}^{(i)}$ as  $K^{(i)}_{j-1}$. 
	\end{itemize} 

\begin{example} 
\label{k-w-ex} 
For $m = 1$ and for $w = 3$, the vector ${\bar N}_w$ is ${\bar N}_w$ $=$ $[ \ 1$ $\ N_2$ $\ N_3$ $\ N_4 \ ]$. Let  ${\bar K}_w$ $:=$ $[ \ 1$ $\ N_3$ $\ N_4$ $\ N_2 \ ]$. Let ${\bar N}^{(i)}$, for some fixed natural number $i$, be $[ \ 5$ $\ 3$ $\ 7$ $\ 6 \ ]$. Then the interpretation of the elements of the vector ${\bar K}_w$ w.r.t. the vector ${\bar N}^{(i)}$ is as follows:  $1$ in ${\bar K}_w$ is interpreted as $1$, the element $K_1$ $=$ $N_3$ in ${\bar K}_w$ is interpreted as $K_1^{(i)}$ $=$ $7$,  the element $K_2$ $=$ $N_4$ in ${\bar K}_w$ is interpreted as $K_2^{(i)}$ $=$ $6$, and, finally,  the element $K_3$ $=$ $N_2$ in ${\bar K}_w$ is interpreted as $K_3^{(i)}$ $=$  $3$. 
\end{example} 

Now suppose that we are given a vector ${\bar K}_w$ as defined above, and are given a vector ${\bar N}^{(i)}$ $\in$ $\cal N$. Then we say that the vector ${\bar N}^{(i)}$  {\em agrees with the total order}  ${\cal O}^{({\bar K}_w)}$ if and only if the interpretation of the elements of the vector ${\bar K}_w$ w.r.t. the vector ${\bar N}^{(i)}$ results in all (i.e., in {\em only}) true inequalities, on the set of natural numbers, in the reflexive transitive closure of the relation $\{$ $(1, K^{(i)}_1),$ $(K^{(i)}_1, K^{(i)}_2),$ $(K^{(i)}_2, K^{(i)}_3),$ $\ldots,$ $(K^{(i)}_{j}, K^{(i)}_{j+1}),$ $\ldots,$ $(K^{(i)}_{w-1}, K^{(i)}_w)$ $\}$. (We interpret the pair $(1, K^{(i)}_1)$ as $1 \leq K^{(i)}_1.$ Further, in case $w \geq 2$, for each $j \in \{ 1,\ldots,w-1 \}$, the pair $(K^{(i)}_j, K^{(i)}_{j+1})$ in this relation is interpreted as $K^{(i)}_j \leq K^{(i)}_{j+1}$ \ .) The latter relation is obtained by replacing each $K_{j}$ with $K^{(i)}_{j}$, for $j \in \{ 1,\ldots,w \}$, in the relation that  defines the total order ${\cal O}^{({\bar K}_w)}$, that is in the relation $\{$ $(1, K_1),$ $(K_1, K_2),$ $(K_2, K_3),$ $\ldots,$ $(K_{w-1}, K_w)$ $\}$. 

\begin{example} 
For the vectors ${\bar K}_w$  and ${\bar N}^{(i)}$ of Example~\ref{k-w-ex}, we have that the reflexive transitive closure of the relation $\{$ $(1, K^{(i)}_1),$ $(K^{(i)}_1, K^{(i)}_2),$ $(K^{(i)}_2, K^{(i)}_3)$ $\}$, that is of the relation $ \{$ $(1, 7),$ $(7, 6),$ $(6, 3)$ $\}$, has elements violating true inequalities on natural numbers. For instance, one of the violations comes from the pair  $(K^{(i)}_1, K^{(i)}_2)$ in this relation, that is from the pair $(7, 6)$. (Recall that we interpret $(K^{(i)}_j, K^{(i)}_{j+1})$ as $K^{(i)}_j \leq K^{(i)}_{j+1}$\ .) We conclude that the vector ${\bar N}^{(i)}$  does not agree with the total order  ${\cal O}^{({\bar K}_w)}$. 

Now consider a different vector  ${\bar K}'_w$ $:=$ $[ \ 1$ $\ N_2$ $\ N_4$ $\ N_3 \ ]$. The interpretation of the elements of the vector ${\bar K}'_w$ w.r.t. the vector ${\bar N}^{(i)}$ (which is the same as before) is as follows:  $1$ in ${\bar K}'_w$ is interpreted as $1$, the element $K'_1$ $=$ $N_2$ in ${\bar K}'_w$ is interpreted as $K_1^{' \ (i)}$ $=$ $3$,  the element $K'_2$ $=$ $N_4$ in ${\bar K}'_w$ is interpreted as $K_2^{' \ (i)}$ $=$ $6$, and, finally,  the element $K'_3$ $=$ $N_3$ in ${\bar K}'_w$ is interpreted as $K_3^{' \ (i)}$ $=$  $7$. Then we have that the reflexive transitive closure of the relation $\{$ $(1, K^{' \ (i)}_1),$ $(K^{' \ (i)}_1, K^{' \ (i)}_2),$ $(K^{' \ (i)}_2, K^{' \ (i)}_3)$ $\}$, that is of the relation $ \{$ $(1, 3),$ $(3, 6),$ $(6, 7)$ $\}$, does {\em not} have elements violating true inequalities on natural numbers. Thus, we conclude that the vector ${\bar N}^{(i)}$ agrees with the total order  ${\cal O}^{({\bar K}'_w)}$.  
\end{example} 

We now define the sets $S$ as suggested in the beginning of this subsection. For a CCQ query $Q$ for which $r \geq 2$ (and thus $w \geq 1$), and for the vector $\bar N$ for the query $Q$ (as defined in Section~\ref{nu-sec}), let  ${\bar K}_w$ be an arbitrary copy-variable-ordering vector for the vector ${\bar N}$. Then we define a subset ${\cal N}^{({\bar K}_w)}$ of the set $\cal N$ as 
\begin{tabbing} 
${\cal N}^{({\bar K}_w)} = \{ {\bar N}^{(i)} \in {\cal N} \ | \ {\bar N}^{(i)}$ agrees with the total order 
 ${\cal O}^{({\bar K}_w)} \} .$ 
\end{tabbing} 
(Note that in the case $w = 1$, there is only one possible vector ${\bar K}_w$ $=$ $[ \ 1$ $\ N_{m+1} \ ]$. Therefore, it is not hard to see that in the case $w = 1$, we have that the only possible set ${\cal N}^{({\bar K}_w)}$ coincides with the entire set $\cal N$.) 

The following result captures straightforward observations about the sets ${\cal N}^{({\bar K}_w)}$. 

\begin{proposition} 
\label{n-subsets-prop}
Given a CCQ query $Q$ for which $r \geq 2$, with vector $\bar N$ for the query $Q$ constructed as defined in Section~\ref{nu-sec}. Then we have that: 
\begin{itemize} 
	\item For each element ${\bar N}^{(i)}$ of the set $\cal N$, there exists at least one copy-variable-ordering vector, ${\bar K}_w$, for the vector ${\bar N}$, such that ${\bar N}^{(i)}$ belongs to the set ${\cal N}^{({\bar K}_w)}$; 
	\item For each copy-variable-ordering vector, ${\bar K}_w$, for the vector ${\bar N}$, the set ${\cal N}^{({\bar K}_w)}$ is an infinite-cardinality subset of the set $\cal N$.  
\end{itemize} 
\end{proposition} 

We conclude from Proposition~\ref{n-subsets-prop}  that the set $\cal N$ is a union of the infinite-cardinality sets ${\cal N}^{({\bar K}_w)}$ ranging over all possible vectors ${\bar K}_w$ for the vector $\bar N$ (via the vector ${\bar N}_w$). 

\begin{proposition} 
\label{n-subsets-funct-prop} 
Given a CCQ query $Q$ for which $r \geq 2$, with vector $\bar N$ for the query $Q$ constructed as defined in Section~\ref{nu-sec}; and given a CCQ query $Q''$ satisfying the restrictions of Section~\ref{basic-queries-sec} w.r.t. the query $Q$. Then   for each copy-variable-ordering vector, ${\bar K}_w$, for the vector ${\bar N}$, there exists a multivariate polynomial $f^{(Q'')}_{(Q)}[{\bar K}_w]$ in terms of the elements of the vector $\bar N$ and with integer coefficients, such that: 
\begin{itemize} 
	\item The polynomial $f^{(Q'')}_{(Q)}[{\bar K}_w]$ is defined on (at least)  the set ${\cal N}^{({\bar K}_w)}$ $\subseteq$ $\cal N$; and 
	\item For each element ${\bar N}^{(i)}$ of the set ${\cal N}^{({\bar K}_w)}$, we have that $f^{(Q'')}_{(Q)}[{\bar K}_w]({\bar N}^{(i)})$ $=$ ${\cal F}^{(Q'')}_{(Q)}({\bar N}^{(i)})$ . 
\end{itemize} 
\end{proposition} 

The proof of Proposition~\ref{n-subsets-funct-prop} is constructive and generalizes the intuition that we gained by considering (earlier in this subsection) the special case $r = 2$ and $w = 1$. Indeed, for each copy-variable-ordering vector, ${\bar K}_w$, for the vector ${\bar N}$, each $min$ expression of Proposition~\ref{multip-monomial-prop}  in terms of elements of all the relevant copy signatures evaluates to one fixed argument among all its arguments, regardless of the identity of specific elements ${\bar N}^{(i)}$ of the set ${\cal N}^{({\bar K}_w)}$. Hence the result of replacing all the $min$ expressions in ${\cal F}^{(Q'')}_{(Q)}$ with these fixed elements of the set $\{ 1,$ $N_{m+1},$ $\ldots,$ $N_{m+w} \}$ is the required function $f^{(Q'')}_{(Q)}[{\bar K}_w]$ for the subset ${\cal N}^{({\bar K}_w)}$  of the domain $\cal N$ of the function ${\cal F}^{(Q'')}_{(Q)}$. The end of Example~\ref{again-writeup-weird-ex} (in Section~\ref{second-beyond-easy-case-sec}) provides an illustration.

For each vector ${\bar K}_w$ in case $r \geq 2$, in the remainder of this proof we will refer to the polynomial $f^{(Q'')}_{(Q)}[{\bar K}_w]$ as ${\cal F}^{(Q'')}_{(Q)}[{\bar K}_w]$; we will use similar notation for the respective functions for the queries $Q$ and $Q'$. (See Proposition~\ref{n-subsets-funct-prop} for a justification.) Further, for uniformity of notation, in the case where $r \leq 1$ we will refer to the function ${\cal F}^{(Q'')}_{(Q)}$ (which we showed in Section~\ref{easy-funct-case-sec} to be a multivariate polynomial in terms of the elements of the vector $\bar N$ and with integer coefficients, on the entire domain $\cal N$ of the function)  as ${\cal F}^{(Q'')}_{(Q)}[{\cal N}]$.  In the latter case (of $r \leq 1$), the only subdomain ${\cal N}^{({\bar K}_w)}$ of the domain $\cal N$ is the set $\cal N$; hence we can use the notations ${\cal F}^{(Q'')}_{(Q)}[{\bar K}_w]$ and ${\cal F}^{(Q'')}_{(Q)}[{\cal N}]$ interchangeably when $r \leq 1$. (Formally, in case $r = 0$, we define the vector ${\bar N}_w$ as ${\bar N}_w$ $=$ $[ \ 1 \ ]$, and in case $r = 1$, we define ${\bar N}_w$ as ${\bar N}_w$ $=$ $[ \ 1$ $\ N_{m+1} \ ]$. In either case, there exists exactly one permutation ${\bar K}_w$ of the vector ${\bar N}_w$, such that the first element of the vector ${\bar K}_w$ is the constant $1$. Thus, the domain ${\cal N}^{({\bar K}_w)}$ does indeed coincide with the domain $\cal N$ in all cases where $r \leq 1$.) 

\subsubsection{Query equivalence implies identical multivariate polynomials} 

\reminder{This is (4) in the flow of reasoning of the reminders}

We now show that for two CCQ queries $Q$ and $Q'$ such that $Q \equiv_C Q'$, the functions ${\cal F}^{(Q)}_{(Q)}$ and ${\cal F}^{(Q')}_{(Q)}$ can be expressed, on each well-defined (as in Sections~\ref{easy-funct-case-sec} and~\ref{multiv-polyn-sec}) infinite-cardinality subdomain of the set $\cal N$, as {\em identical} multivariate polynomials in terms of the elements of the vector $\bar N$ and with integer coefficients. This result, Proposition~\ref{equiv-funct-prop}, is immediate from the results of Section~\ref{hard-funct-case-sec} and from Propositions~\ref{polyn-prop} and~\ref{n-subsets-funct-prop}.

\begin{proposition} 
\label{equiv-funct-prop} 
Given a CCQ query $Q$, with vector $\bar N$ for the query $Q$ constructed as defined in Section~\ref{nu-sec}; and given a CCQ query $Q'$ such that $Q \equiv_C Q'$ holds. 
Then 
%
	 for each copy-variable-ordering vector, ${\bar K}_w$, for the vector ${\bar N}$, we have that  the multivariate polynomials ${\cal F}^{(Q)}_{(Q)}[{\bar K}_w]$ and ${\cal F}^{(Q')}_{(Q)}[{\bar K}_w]$, in terms of the elements of the vector $\bar N$ and with integer coefficients, are identical functions on the domain ${\cal N}^{({\bar K}_w)}$. 
\end{proposition}


\subsubsection{Solid terms and phantom terms of the polynomials for the multiplicity functions} 
\label{solid-phantom-sec}

\reminder{[(2)] Observation that for query $Q''$, for each total order $\cal O$, and for each solid term $\cal M$ in the multivariate polynomial ${\cal F}_{(Q)}^{(Q'')}[\cal O]$, we have that $\cal M$ is the multiplicity monomial for a nonempty monomial class for $Q''$ and for $\{ D_{\bar{N}^{(i)}}(Q) \}$.  
} 

To proceed with the proof of Proposition~\ref{qprime-goldfish-prop} (see Section~\ref{qprime-goldfish-sec}), we need to introduce technical denotations for the terms of the multivariate polynomials that we have been considering. Suppose that for CCQ queries $Q$ and $Q''$, we are given the multivariate polynomial ${\cal F}^{(Q'')}_{(Q)}[{\bar K}_w]$ for an arbitrary ${\bar K}_w$, as defined earlier in Section~\ref{q-prime-has-wave-sec}. Consider a monomial (excluding the nonzero integer coefficient of the term) $\cal T$ in ${\cal F}^{(Q'')}_{(Q)}[{\bar K}_w]$. We say that: 
\begin{itemize} 
	\item $\cal T$ is a {\em solid term of} ${\cal F}^{(Q'')}_{(Q)}[{\bar K}_w]$ if there exists a nonempty monomial class, ${\cal C}^{(Q'')}$, for the query $Q''$ and for the family of databases  $\{ D_{\bar{N}^{(i)}}(Q) \}$, such that $\cal T$ is the multiplicity monomial for the class ${\cal C}^{(Q'')}$. 
	\item Conversely, we say that $\cal T$ is a {\em phantom term of} ${\cal F}^{(Q'')}_{(Q)}[{\bar K}_w]$ whenever $\cal T$ is not a solid term of ${\cal F}^{(Q'')}_{(Q)}[{\bar K}_w]$. 
\end{itemize} 

For instance, in Example~\ref{again-writeup-weird-ex} (in Section~\ref{second-beyond-easy-case-sec}), the multivariate polynomial ${\cal F}^{(Q)}_{(Q)}[[ \ 1 \ N_3 \ N_4 \ ]]$, that is the polynomial $N_1 \times N_2 \times (N_4)^2$, has one solid term, ${\cal T}$ $=$ $N_1 \times N_2 \times (N_4)^2$. The reason for the term $\cal T$ being a solid term of the polynomial ${\cal F}^{(Q)}_{(Q)}[[ \ 1 \ N_3 \ N_4 \ ]]$  is that the query $Q$ w.r.t. the family of databases $\{ D_{\bar{N}^{(i)}}(Q) \}$ has a monomial class ${\cal C}_4^{(Q)}$ whose multiplicity monomial is exactly the term $\cal T$. (See Example~\ref{writeup-weird-ex}  in Section~\ref{beyond-easy-case-sec} for the details on the monomial class ${\cal C}_4^{(Q)}$.) 

Now consider an abstract example of a phantom term. (We do not know whether there are CCQ queries whose associated functions $\cal F$ have phantom terms, hence this example is abstract. However, to prove Theorem~\ref{magic-mapping-prop}, we have to prove that phantom terms do not arise in certain parts of  ${\cal F}^{(Q')}_{(Q)}$ for those CCQ queries $Q'$ that are combined-semantics equivalent to explicit-wave CCQ query $Q$.) 

\begin{example} 
\label{phantom-ex}
Consider the case where $m = 0$ and $r = w = 2$. Suppose that for these values of $m$, $r$, and $w$, for some hypothetical CCQ query $Q''$ and for some hypothetical family of databases $\{ D_{\bar{N}^{(i)}}(Q) \}$, it holds that the set of all monomial classes in this setting consists of two monomial classes, say ${\cal C}_1^{(Q'')}$ and ${\cal C}_2^{(Q'')}$. By $m = 0$, we have that the noncopy signature of each of ${\cal C}_1^{(Q'')}$ and ${\cal C}_2^{(Q'')}$ is the empty vector. Suppose that the copy signature of ${\cal C}_1^{(Q'')}$ is $[ \ N_1 \ N_2 \ ]$, and that the copy signature of ${\cal C}_2^{(Q'')}$ is $[ \ N_2 \ N_1 \ ]$. 
 
The noncopy signatures of the two monomial classes are identical. Thus, by our results of Section~\ref{hard-funct-case-sec}, the function for the multiplicity of the tuple $t^*_Q$, in the answer to the query $Q''$ on the databases  $\{ D_{\bar{N}^{(i)}}(Q) \}$, will be computed as the cardinality of the union of the sets for the copy signatures of the two monomial classes. That is, the multiplicity function will be 

$$2 \times N_1 \times N_2 - (min(N_1,N_2))^2 \ .$$ 

Then for the vector ${\bar K}_w$ $=$ $[ \ 1 \ N_1 \ N_2 \ ]$, the multivariate polynomial for this multiplicity function on the domain ${\cal N}^{({\bar K}_w)}$ will be 

$$2 \times N_1 \times N_2 - (N_1)^2 \ .$$ 

In this polynomial, (i) the term $N_1 \times N_2$ is a solid term of the polynomial, because the monomial class ${\cal C}_1^{(Q'')}$ (as well as ${\cal C}_2^{(Q'')}$) has the term $N_1 \times N_2$ as its multiplicity monomial. In contrast, (ii) the term $(N_1)^2$ is by definition a phantom term of the polynomial. 
\end{example} 



\subsubsection{The multiplicity function for the query $Q$} 
\label{mult-fun-sec} 

We now make several observations concerning the solid and phantom terms in the polynomials ${\cal F}^{(Q)}_{(Q)}[{\bar K}_w]$ for the function ${\cal F}^{(Q)}_{(Q)}$ of the query $Q$. First, as usual, we will need some terminology. For CCQ queries $Q$ and $Q''$ (where $Q''$, as usual, may or may not be the query $Q$) and for the  vector $\bar N$ constructed for the query $Q$ as defined in Section~\ref{nu-sec}, fix an arbitrary copy-variable-ordering vector, ${\bar K}_w$,  for $\bar N$.  In case $m \geq 1$, consider all the terms in the function ${\cal F}^{(Q'')}_{(Q)}[{\bar K}_w]$ such that each term  has the product $N_1$ $\times$ $N_2$ $\times$ $\ldots$ $\times$ $N_m \ $ . Call all these terms collectively ``the $m$-covering part of the function ${\cal F}^{(Q'')}_{(Q)}[{\bar K}_w] \ $.''  For the case $m = 0$, we say that {\em all} the terms of the function ${\cal F}^{(Q'')}_{(Q)}[{\bar K}_w]$ constitute the $m$-covering part of the function. 

The observations of this subsection, Propositions~\ref{m-covering-prop} and~\ref{goldfish-properties-prop}, are made for the polynomials ${\cal F}^{(Q)}_{(Q)}[{\bar K}_w]$, that is for the case where the query $Q''$ coincides with the query $Q$. 

\begin{proposition} 
\label{m-covering-prop} 
Given an explicit-wave CCQ query $Q$, with vector $\bar N$ constructed as defined in Section~\ref{nu-sec}. Then we have that: 
\begin{itemize} 
	\item For each copy-variable-ordering vector, ${\bar K}_w$,  for the vector $\bar N$, the  $m$-covering part of the function ${\cal F}^{(Q)}_{(Q)}[{\bar K}_w]$ has at least one term with a nonzero coefficient; and 
	\item For each pair $({\bar K}_w, {\bar K}'_w)$ of copy-variable-ordering vectors  for the vector $\bar N$, the $m$-covering part of the function ${\cal F}^{(Q)}_{(Q)}[{\bar K}_w]$ and the $m$-covering part of the function ${\cal F}^{(Q)}_{(Q)}[{\bar K}'_w]$ are identical multivariate polynomials (in terms of the elements of the vector $\bar N$ and with integer coefficients). 
\end{itemize} 
\end{proposition} 

We note that {\em implicit-wave} CCQ queries are not covered by the above result. In fact, the implicit-wave query $Q$ of Example~\ref{intro-weird-ex} does not satisfy Proposition~\ref{m-covering-prop}. (Please see Example~\ref{again-writeup-weird-ex} of Section~\ref{second-beyond-easy-case-sec} for the details on this query.)  

The observations of Proposition~\ref{m-covering-prop} are immediate from the relevant definitions, from the construction of the function ${\cal F}^{(Q)}_{(Q)}$, and from the definition of explicit-wave query. (Intuitively, we use the definition of explicit-wave CCQ queries, Definition~\ref{expl-wave-def}, to argue that all the monomial classes in question that have the same noncopy signature, must also have the same copy signature. In fact, Definition~\ref{expl-wave-def} has been formulated so that Proposition~\ref{m-covering-prop}, and consequently Theorem~\ref{magic-mapping-prop}, could go through.) Specifically, the second bullet of Proposition~\ref{m-covering-prop} follows from the fact that among all the monomial classes for the query $Q$ whose (classes') noncopy signature is a permutation of the vector $[ \ N_1 \ N_2 \ \ldots \ N_m \ ]$ in case $m \geq 1$, or is the empty vector in case $m = 0$, for each pair of such monomial classes with the same noncopy signature, there is {\em unconditional} dominance between the two classes. This fact holds by the definition of explicit-wave query and by the definition of its wave function, Definition~\ref{the-wave-def}. We then use the results of Section~\ref{easy-funct-case-sec} (specifically of Proposition~\ref{easy-final-funct-prop}) to establish that the respective $m$-covering polynomials do not depend on the vector ${\bar K}_w$, because the relevant sets ${\mathbb C}^{nondom}$ do not depend on the vector ${\bar K}_w$. 

The results of Proposition~\ref{m-covering-prop} permit us to talk about ``the $m$-covering part of the function ${\cal F}^{(Q)}_{(Q)} \ $,'' for every explicit-wave CCQ query $Q$. We denote by ${\cal G}^{(Q)}_{(Q)}$ this nonempty multivariate polynomial in terms of the elements of the vector $\bar N$ and with integer coefficients. Notice that the reference to copy-variable-ordering vectors has been dropped from the notation for ${\cal G}^{(Q)}_{(Q)}$. 

The next observations, in Proposition~\ref{goldfish-properties-prop}, are about the properties of the function ${\cal G}^{(Q)}_{(Q)}$ for explicit-wave queries $Q$. The results of Proposition~\ref{goldfish-properties-prop} hold by the relevant definitions and by Proposition~\ref{q-has-wave-prop}. The intuition for Proposition~\ref{goldfish-properties-prop} is the same as the intuition for Proposition~\ref{m-covering-prop}.

\reminder{[(3)] In ${\cal F}_{(Q)}^{(Q)}$ {\em for the query $Q$}, for each total order $\cal O$ we have that in ${\cal F}_{(Q)}^{(Q)}[{\cal O}]$, all the terms of ${\cal F}_{(Q)}^{(Q)}[\cal O]$ that have $\Pi_{j = 1}^m N_j$ (or, in case $m = 0$, {\em all} terms of ${\cal F}_{(Q)}^{(Q)}[\cal O]$) are {\em solid terms} of ${\cal F}_{(Q)}^{(Q)}[\cal O]$. (This is immediate from the fact that $Q$ is an explicit-wave query.) 
} 

\begin{proposition} 
\label{goldfish-properties-prop}
Given an explicit-wave CCQ query $Q$, we have that: 
\begin{itemize} 
	\item In the function ${\cal G}^{(Q)}_{(Q)}$, each term is a solid term; 
	\item Each term of the function ${\cal G}^{(Q)}_{(Q)}$ has a positive coefficient; and 
	\item The multivariate polynomial ${\cal G}^{(Q)}_{(Q)}$  has a (solid) term which is the wave ${\cal P}^{(Q)}_*$ of the query $Q$ w.r.t. the family of databases  $\{ D_{\bar{N}^{(i)}}(Q) \}$. 
\end{itemize} 
\end{proposition} 

\subsubsection{Intuition for the proof of Proposition~\ref{qprime-goldfish-prop}} 
\label{intuition-sec} 

In this subsection we provide the intuition for the remainder of this proof of Theorem~\ref{magic-mapping-prop}. Specifically, we explain how we plan, in the remainder of this proof, to prove Proposition~\ref{qprime-goldfish-prop} of Section~\ref{qprime-goldfish-sec}. 

We have just established (see Proposition~\ref{goldfish-properties-prop}) that for the query $Q$ given in Theorem~\ref{magic-mapping-prop}, the wave ${\cal P}^{(Q)}_*$ of the query $Q$ is present as a solid term in the $m$-covering part of the function ${\cal F}^{(Q)}_{(Q)}$. 
We now make an easy observation that follows trivially from Propositions~\ref{equiv-funct-prop} and~\ref{goldfish-properties-prop}: 

\begin{proposition} 
\label{same-funct-in-q-prime-prop} 
Given an explicit-wave CCQ query $Q$ and a CCQ query $Q'$ such that $Q \equiv_C Q'$. Then for each copy-variable-ordering vector, ${\bar K}_w$,  for the vector $\bar N$, the  $m$-covering part of the function ${\cal F}^{(Q')}_{(Q)}[{\bar K}_w]$ has, with a positive-integer coefficient, the wave ${\cal P}^{(Q)}_*$ of the query $Q$ w.r.t. the family of databases  $\{ D_{\bar{N}^{(i)}}(Q) \}$. 
\end{proposition} 

What remains to be shown, to complete the proof of Proposition~\ref{qprime-goldfish-prop} of Section~\ref{qprime-goldfish-sec}, is the following claim:  

\begin{conjecture} 
\label{goldfish-in-q-prime-conj} 
Given an explicit-wave CCQ query $Q$ and a CCQ query $Q'$ such that $Q \equiv_C Q'$. Then for each copy-variable-ordering vector, ${\bar K}_w$,  for the vector $\bar N$, the  $m$-covering part of the function ${\cal F}^{(Q')}_{(Q)}[{\bar K}_w]$ has, as a solid term, the wave ${\cal P}^{(Q)}_*$ of the query $Q$ w.r.t. the family of databases  $\{ D_{\bar{N}^{(i)}}(Q) \}$. 
\end{conjecture}

The only difference between the formulations of Proposition~\ref{same-funct-in-q-prime-prop} and of Conjecture~\ref{goldfish-in-q-prime-conj} is that Conjecture~\ref{goldfish-in-q-prime-conj} claims that the term  ${\cal P}^{(Q)}_*$  is a {\em solid} term in the $m$-covering part of the function ${\cal F}^{(Q')}_{(Q)}[{\bar K}_w]$, for each vector ${\bar K}_w$. 

Observe that if Conjecture~\ref{goldfish-in-q-prime-conj} is true, then it implies the result of Proposition~\ref{qprime-goldfish-prop} of Section~\ref{qprime-goldfish-sec}. That is, proving Conjecture~\ref{goldfish-in-q-prime-conj} would complete immediately the proof of Theorem~\ref{magic-mapping-prop}, please see Section~\ref{main-results-summary-sec}. 

Conjecture~\ref{goldfish-in-q-prime-conj} does turn out to be true. (We will introduce its recasting, in identical wording, as Proposition~\ref{goldfish-in-q-prime-prop} of Section~\ref{sec-hinge-proof-sec}.) In the remainder of this subsection, we provide the intuition for the planned proof. 

Our plan of attack for the proof of the result of Conjecture~\ref{goldfish-in-q-prime-conj} is as follows. We obtain that result by contradiction, by assuming that ${\cal P}^{(Q)}_*$  is a {\em phantom} term in the $m$-covering part of the function ${\cal F}^{(Q')}_{(Q)}[{\bar K}_w]$, for at least one vector ${\bar K}_w$. Our key point in the proof is that if ${\cal P}^{(Q)}_*$ is a phantom term for some fixed vector ${\bar K}_w$, then the expression for ${\cal P}^{(Q)}_*$ using a product of $min$-expressions (i.e., in a way that does not factor in any specific orderings ${\bar K}_w$) must include a $min$-expression involving at least two distinct variables among the variables $N_{m+1}, \ldots, N_{m+w}$ of the query $Q$. (Hence our assumption that $w$ $\geq$ $2$ would be essential in that part of the remainder of the proof of Theorem~\ref{magic-mapping-prop}.) It is rather straightforward to obtain from this fact that for some vector ${\bar K}'_w$ different from the fixed vector ${\bar K}_w$, we could get a {\em different} monomial for the term that ``looked like'' ${\cal P}^{(Q)}_*$ under the ordering ${\bar K}_w$, because we would obtain ${\bar K}'_w$ essentially by swapping the relative order of some two distinct variables $N_A$ and $N_B$ in the vector ${\bar K}_w$. Then, {\em unless} the latter monomial (i.e., the monomial that results from  ${\cal P}^{(Q)}_*$ when we replace ${\bar K}_w$ with ${\bar K}'_w$) gets canceled out by other terms in the polynomial resulting from ${\cal F}^{(Q')}_{(Q)}$ under ${\bar K}'_w$, we can achieve the desired contradiction by observing that the overall $m$-covering parts of the tuple-multiplicity function for $Q'$ are different under the vectors ${\bar K}_w$ and ${\bar K}'_w$. (The contradiction follows immediately from the results of Section~\ref{mult-fun-sec} and from Proposition~\ref{equiv-funct-prop}, which together say that for each pair $({\bar K}_w$, ${\bar K}'_w)$, the $m$-covering part of the tuple-multiplicity function for the query $Q'$ under ${\bar K}_w$ and the $m$-covering part of the tuple-multiplicity function for $Q'$ under ${\bar K}'_w$ must be identical polynomials.) 

However, it is not clear at all how to prove the claim that the requisite terms in the functions for $Q'$ do not get canceled out in a way that is ``convenient for the function $Q'$.'' Here is the source of the difficulty: Observe that if ${\cal P}^{(Q)}_*$ is indeed a phantom term in the tuple-multiplicity polynomial for some fixed ${\bar K}_w$, then it must be that ${\cal P}^{(Q)}_*$ under that ${\bar K}_w$ is a term in a polynomial, $\cal R$, such that $\cal R$ arises from the inclusion-exclusion principle for counting the cardinality of a union of sets of tuples for some monomial classes for the query $Q'$ and for databases $D_{N^{(i)}}(Q)$. Further, ${\cal P}^{(Q)}_*$ under that ${\bar K}_w$ being a {\em phantom} term of that polynomial $\cal R$ means that: 

(*) The polynomial for the union that we are speaking about must be for a union of the sets of tuples for {\em at least two}  such monomial classes. 

(This fact is immediate from our assumption that ${\cal P}^{(Q)}_*$ is a phantom term, that is, ${\cal P}^{(Q)}_*$ cannot be the multiplicity monomial for any monomial class for $Q'$ and for databases $D_{N^{(i)}}(Q)$). 

As a result, we have that: 

(**) The polynomial $\cal R$ for the above union must have terms with positive coefficients, as well as terms with negative coefficients. 

The fact (**) furnishes a significant challenge in the proof of Conjecture~\ref{goldfish-in-q-prime-conj}. Indeed, if ${\cal P}^{(Q)}_*$ is a phantom term in such a polynomial for some fixed ${\bar K}_w$, then (by definition of phantom terms) ${\cal P}^{(Q)}_*$ must be a phantom term in the relevant polynomials for {\em all possible} vectors ${\bar K}_w$, and thus the polynomials as in (*) {\em for all possible vectors} ${\bar K}_w$ must have terms with positive coefficients, as well as terms with negative coefficients. As a result, there is a theoretical possibility, for the query $Q'$, that under each different vector ${\bar K}_w$, ${\cal P}^{(Q)}_*$ is a {\em different} term in some such union-of-multiplicities polynomial. Thus, it is theoretically possible that in the overall polynomial that is  the  $m$-covering part of the function ${\cal F}^{(Q')}_{(Q)}[{\bar K}_w]$, the remaining terms (i.e., those terms that do not ``look like'' ${\cal P}^{(Q)}_*$ under the given vector ${\bar K}_w$) always cancel each other out in a convenient way. This way, we can make no conclusion about the wave ${\cal P}^{(Q)}_*$ of the query $Q$ being backed up by assignments from the query $Q'$ to the databases $D_{{\bar N}^{(i)}}(Q)$, and thus can make no conclusion about the result of Theorem~\ref{magic-mapping-prop} for (explicit-wave CCQ queries $Q$ and) arbitrary CCQ queries $Q'$. 


Our plan of attack for overcoming this major challenge is in developing a class of signatures, for multivariate polynomials with integer coefficients in the context of the proof of Theorem~\ref{magic-mapping-prop}, in such a way that the signature for each such polynomial is always a set that features (some combinations of) variables used in those polynomials, always with positive coefficients. (As a result, no part of such a signature can cancel out another part of the signature.) Further, whenever the signatures for two such polynomials are different, then the polynomials must also be different. We explain how to construct the signatures in Section~\ref{bqkw-sec}.  

Then, in Section~\ref{sec-hinge-proof-sec}, we prove  Conjecture~\ref{goldfish-in-q-prime-conj}, by essentially following the idea-of-the-proof argument of this subsection, but by looking at the signatures instead of at the polynomials that generate those signatures. This way, we overcome the challenge that arises from the presence of both positive- and negative-coefficient terms under our assumption-toward-contradiction in the proof of Conjecture~\ref{goldfish-in-q-prime-conj}, and thus complete the proof of Theorem~\ref{magic-mapping-prop}. 

\subsubsection{Signatures for the polynomials ${\cal F}^{(Q')}_{(Q)}[{\bar K}_w]$} 
\label{bqkw-sec} 

In this subsection, for the multivariate polynomials with integer coefficients in the context of the proof of Theorem~\ref{magic-mapping-prop}, we define a class of signatures. One property of these signatures is that the signature for each such polynomial is always a set that features (some combinations of) variables used in those polynomials, always with positive coefficients. (As a result, no part of such a signature can cancel out another part of the signature.) Another property of these signatures is that, whenever the signatures for two such polynomials are different, then the polynomials must also be different. We use these signatures in Section~\ref{sec-hinge-proof-sec}, to complete the proof of Theorem~\ref{magic-mapping-prop}. 

The general idea of our proposed signatures is to produce, for a given multivariate polynomial, in terms of the variables $N_{m+1}$ through $N_{m+w}$ and with integer coefficients, the ``overall count'' of the number of occurrences of each variable in the polynomial, separately for each of the variables. (We count as positive occurrences all such occurrences of variables in terms with positive coefficients, and count as negative occurrences all such occurrences of variables in terms with negative coefficients.) Intuitively, we do the counting on a version of the polynomial where (i) each power of the form $X^p$, with $X$ a variable of interest and with $p$ an integer value greater than unity (i.e., greater than $1$), is expanded as $X \times X \times \ldots \times X$, with $X$ occurring $p$ times in the product, and where (ii) each term of the form $C \times \Pi$, with $\Pi$ being a product of variables with the unity coefficient and with the absolute value $|C|$ of $C$ being strictly greater than unity, is expanded as $sign(C)$ $\times$ $\Pi$ $\times$ $\Pi$ $\times$ $\ldots$ $\times$ $\Pi$, with $\Pi$ occurring $|C|$ times in the product. 

To construct the signatures, we recall the ``original'' formulations of the $m$-covering part (only) of the functions  ${\cal F}^{(Q'')}_{(Q)}$, as given via Propositions~\ref{bigunion-prop} and~\ref{multip-monomial-prop} of Section~\ref{hard-funct-case-sec}.  For the functions  ${\cal F}^{(Q'')}_{(Q)}$, these are the formulations  that do not take into account any fixed copy-variable-ordering vectors ${\bar K}_w$, and hence the formulations may use $min$-expressions. To construct the signatures in this section, in each such formulation of a function  ${\cal F}^{(Q'')}_{(Q)}$, we order the elements of each term according to the provenance of each element from the respective copy signature. That is, whenever a $min$-expression, $\cal M$, in any such term $\cal T$, is a $min$-expression for some fixed value of $u$ between $1$ and $r$, as given in the notation of Proposition~\ref{multip-monomial-prop} of Section~\ref{hard-funct-case-sec}, then the $min$-expression $\cal M$ is exactly the $u$th element, from left to right, in the term $\cal T$. (In any special case where the $min$-expression $\cal M$ evaluates either to the constant $1$ or to a single variable name $N_j$ for some $j$th element, $j$ $\in$ $\{$ $m+1$, $\ldots$, $m+w$ $\}$, of the vector $\bar N$, -- that is, in all cases where the $min$-expression contains only one element, -- the respective $u$th element in the term $\cal T$ is exactly the value $1$ or that variable name, rather than the (unevaluated) $min$-expression.) 

In the expression that results from the procedure described in the previous paragraph, we further group together, by using parentheses, each separate union-of-multiplicity-terms expression as in the format of Proposition~\ref{bigunion-prop} of Section~\ref{hard-funct-case-sec}. 

The purpose of the above ordering is to make it drastically easier, as we will see in Proposition~\ref{key-sig-prop}, to compute our proposed signatures. To make it even easier to compute the signatures, in case where $m$ $\geq$ $1$ we divide each such ordered expression as above by the product $\Pi_{i=1}^{m} N_i$. Thus, the expression that we use to construct a signature is always an (ordered) expression in terms of {\em only} the variables $N_{m+1}$ through $N_{m+w}$. 

For a fixed CCQ query $Q''$, the first input to our signature constructor is always the $m$-covering part of the function ${\cal F}^{(Q'')}_{(Q)}$, which we then preprocess by (a) dividing the entire expression by $\Pi_{i=1}^{m} N_i$ in case where $m$ $\geq$ $1$, and by then (b) ordering, as explained in the third and fourth paragraphs of this subsection, the elements of all the terms in  the output of the step (a). (All the ``expansions'' outlined in the second paragraph of this subsection result naturally from the preprocessing step (b).) The second input to our signature constructor is always some copy-variable-ordering vector ${\bar K}_w$ for the query $Q$.

In the remainder of this subsection, we denote by $f^{(\cal S)}(Q'')$ the output of our preprocessing steps (a) and (b) when given as input the function ${\cal F}^{(Q'')}_{(Q)}$ for a CCQ query $Q''$ and for a family of databases $D_{{\bar N}^{(i)}}(Q)$. 

Consider an illustration. 

\begin{example} 
\label{sig-ex} 
Consider the function 

\begin{equation} 
\label{ex-nine-eq}  
2 \times N_1 \times N_2 - (min(N_1,N_2))^2 
\end{equation} 

\noindent 
that was constructed in Example~\ref{phantom-ex} for a hypothetical CCQ query $Q''$.   As $m$ $=$ $0$ in the context of that example, Eq. (\ref{ex-nine-eq}) defines exactly the $m$-covering part of the function ${\cal F}^{(Q'')}_{(Q)}$. Observe that Eq. (\ref{ex-nine-eq}) contains terms with nonunity coefficients (see the first term in the Equation) and, as we will see in this example, does not observe the ordering that we proposed in the third paragraph of this subsection. 

Our preprocessing step (a) is to divide the entire expression by $\Pi_{i=1}^{m} N_i$, in case where $m$ $\geq$ $1$. As $m$ $=$ $0$ for this query $Q''$, the output of this step for the function of Eq. (\ref{ex-nine-eq}) is exactly the input of this step (a). 

We now begin the preprocessing step (b), by ordering, as explained in the third paragraph of this subsection, the elements of all the terms in  the output of the step (a). Specifically, the first term $2 \times N_1 \times N_2$ in Eq. (\ref{ex-nine-eq}) stands for two separate identical expressions $N_1 \times N_2$, with different meanings, as follows. The first expression $N_1 \times N_2$ is the expression for the cardinality of the set $\bigcap_{s=1}^{1} \Gamma^{(i)}[{\cal C}_{s}^{(Q'')}]$ (see Proposition~\ref{multip-monomial-prop} of Section~\ref{hard-funct-case-sec}), that is  for the cardinality of the intersection with itself of the tuple-multiplicity set for the monomial class ${\cal C}_{1}^{(Q'')}$ of Example~\ref{phantom-ex}. The copy signature of that monomial class is $[ N_1 \ N_2 ]$, hence this first expression is reordered (vacuously) in this step (b) as the product of $N_1$ and $N_2$, in this order. (Trivially, when we compute the cardinality of the intersection of a cardinality set with itself, in each product-of-$min$-expressions term as given by Proposition~\ref{multip-monomial-prop} of Section~\ref{hard-funct-case-sec} we have that each $min$-expression has only one argument, and hence is evaluated immediately to the respective element of the copy signature.) 

The meaning of the second copy of the expression $N_1 \times N_2$ in Eq. (\ref{ex-nine-eq}) is different from the meaning of the first copy of that expression. Specifically, the second copy of the expression $N_1 \times N_2$  is the expression for the cardinality of the set $\bigcap_{s=2}^{2} \Gamma^{(i)}[{\cal C}_{s}^{(Q'')}]$ (see Proposition~\ref{multip-monomial-prop} of Section~\ref{hard-funct-case-sec}), that is  for the cardinality of the intersection with itself of the tuple-multiplicity set for the monomial class ${\cal C}_{2}^{(Q'')}$ of Example~\ref{phantom-ex}. The copy signature of that monomial class is $[ N_2 \ N_1 ]$, hence this expression (unlike the first expression $N_1$ $\times$ $N_2$) is reordered (nontrivially) in this step (b) as the product of $N_2$ and $N_1$, in this order. 

Now the term ``$- (min(N_1,N_2))^2$'' of Eq. (\ref{ex-nine-eq}) is reformatted in this preprocessing step (b) as the product of two copies of the expression ``$min(N_1,N_2)$,'' with the negative sign then passed over to the output from the input expression. That is, the step (b) reformats ``$- (min(N_1,N_2))^2$'' into ``$- min(N_1,N_2)$ $\times$ $min(N_1,N_2)$.'' In this product expression, the leftmost element of the product stands for the minimum between the first element ($N_1$) of the copy signature of the monomial class ${\cal C}_{1}^{(Q'')}$ and the first element ($N_2$) of the copy signature of the monomial class ${\cal C}_{2}^{(Q'')}$. Similarly, the rightmost element of the product stands for the minimum between the second element ($N_2$) of the copy signature of the monomial class ${\cal C}_{1}^{(Q'')}$ and the second element ($N_1$) of the copy signature of the monomial class ${\cal C}_{2}^{(Q'')}$. 

Finally, all three expressions above, that is,  ``$N_1 \times N_2$,''  ``$N_2 \times N_1$,'' and ``$- min(N_1,N_2)$ $\times$ $min(N_1,N_2)$,'' all belong to the same cardinality expression for a union of monomial classes. Hence, as explained in the fourth paragraph of this subsection, the output  $f^{(\cal S)}(Q'')$ of the step (b) and of the overall preprocessing of this subsection is the single parenthesized expression 

\begin{tabbing} 
$f^{(\cal S)}(Q'')$ $=$ ($N_1 \times N_2$ $+$  $N_2 \times N_1$ \\ 
$-$ $min(N_1,N_2)$ $\times$ $min(N_1,N_2)$) . 
\end{tabbing} 
\vspace{-0.6cm} 
\end{example} 

We make the following observation, which holds trivially by construction of the expression $f^{(\cal S)}(Q'')$: 

\begin{proposition} 
Given the $m$-covering part, call it ${\cal F}^{(Q'')}_{m}$, of the function ${\cal F}^{(Q'')}_{(Q)}$ for CCQ queries $Q$ and $Q''$. Denote by ${\cal G}^{(Q'')}_{m}$ the result of dividing the expression ${\cal F}^{(Q'')}_{m}$ by the product $\Pi_{j=1}^m N_j$ in case $m$ $\geq$ $1$; in case where $m$ $=$ $0$ we denote by ${\cal G}^{(Q'')}_{m}$ the original expression ${\cal F}^{(Q'')}_{m}$. Then for each $i$ $\geq$ $1$, the expressions $f^{(\cal S)}(Q'')$ and ${\cal G}^{(Q'')}_{m}$ evaluate to the same value when the elements of the vector ${\bar N}^{(i)}$ are substituted into each expression as values for the elements of $\bar N$. 
\end{proposition} 

We now define the signatures, which will be the key element of the proof in Section~\ref{sec-hinge-proof-sec}. Given the function $f^{(\cal S)}(Q'')$ for a CCQ query $Q''$ and given a copy-variable-ordering vector ${\bar K}_w$ for the vector ${\bar N}$ of the fixed explicit-wave query query $Q$, we say that {\em the signature ${\cal S}(Q'',{\bar K}_w)$ of} $f^{(\cal S)}(Q'')$ {\em with respect to} ${\bar K}_w$ is the set  resulting from (i) evaluating the value of each $min$-expression in $f^{(\cal S)}(Q'')$ using the total ($\leq$) order ${\bar K}_w$ (that is, given a fixed ${\bar K}_w$, each $min$-expression as in Proposition~\ref{multip-monomial-prop} of Section~\ref{hard-funct-case-sec} evaluates to a single element of the set $\{$ $1$, $N_{m+1}$, $\ldots$, $N_{m+w}$ $\}$), and then from (ii) counting the number of occurrences of each variable name separately in the resulting expression, taking into account the sign of each term (that is, an occurrence of a variable name in a term with the negative sign is a negative occurrence, whereas an occurrence of a variable name in a term with the positive sign is a positive occurrence).  The resulting signature ${\cal S}(Q'',{\bar K}_w)$ is a set that has a separate element for each variable name in the set $\{$ $N_{m+1}$, $\ldots$, $N_{m+w}$ $\}$, giving the total number of occurrences of that variable name in  $f^{(\cal S)}(Q'')$ under the vector ${\bar K}_w$; we omit from ${\cal S}(Q'',{\bar K}_w)$ the mention of all the variable names whose total occurrence in $f^{(\cal S)}(Q'')$ under ${\bar K}_w$ is zero. For instance, in the context of Example~\ref{sig-ex}, the signature ${\cal S}(Q'',{\bar K}_w)$ of $f^{(\cal S)}(Q'')$ with respect to ${\bar K}_w$ $=$ $[$ $1 \ N_1 \ N_2 $ $]$ is ${\cal S}(Q'',{\bar K}_w)$ $=$ $\{$ $2 \times N_2$ $\}$, as is easy to ascertain using the expression for $f^{(\cal S)}(Q'')$ that we obtained in Example~\ref{sig-ex}. 

In evaluating ${\cal S}(Q'',{\bar K}_w)$, we must take care not to ``bundle together'' into exponents repeated occurrences of the same variable name; that is, for instance, the product $N_{m+1}$ $\times$ $N_{m+1}$, which is the same as $N_{m+1}^2$, counts as two occurrences of $N_{m+1}$ rather than one. Further, we do not count any occurrences of the constant $1$ in evaluating ${\cal S}(Q'',{\bar K}_w)$. The reason is, all we need to complete 
the proof of Theorem~\ref{magic-mapping-prop} is to show that a hypothetical function of the form $f^{(\cal S)}(Q'')$, in the context of two distinct vectors ${\bar K}_w$ and ${\bar K}'_w$, results in two nonidentical multivariate polynomials. We will do this in our proof in Section~\ref{sec-hinge-proof-sec} by arguing that the signatures of those polynomials are not the same. Clearly, by construction of the signatures we have that: 

\begin{proposition} 
\label{not-same-sigs-prop} 
Given expressions $f^{(\cal S)}(Q')$ and $f^{(\cal S)}(Q'')$, where $Q'$ and $Q''$ may or may not be the same CCQ query, and given two distinct vectors ${\bar K}_w$ and ${\bar K}'_w$. Then, whenever the signatures ${\cal S}(Q',{\bar K}_w)$ and ${\cal S}(Q'',{\bar K}'_w)$ are nonidentical sets, then the polynomial that results from evaluating the value of each $min$-expression in $f^{(\cal S)}(Q')$ under the total ($\leq$) order ${\bar K}_w$ is not identical to the polynomial that results from evaluating the value of each $min$-expression in $f^{(\cal S)}(Q'')$ under the total ($\leq$) order ${\bar K}'_w$. 
\end{proposition}

In the remainder of this subsection, we make a key observation, Proposition~\ref{key-sig-prop}. The intuition for that result is that it is very easy to compute the signature of each expression $f^{(\cal S)}(Q'')$, w.r.t. each fixed vector ${\bar K}_w$, by first forming a single positive-sign $max$-expression for each position (from $1$st to $r$th) of the copy signatures associated with each parenthesized expression of $f^{(\cal S)}(Q'')$,  by then evaluating the resulting $max$-expressions under the total ($\leq$) order ${\bar K}_w$, and by finally summing up all the resulting {\em positive-sign-only} occurrences of the variable names. 

\begin{example} 
\label{second-sig-ex} 
We continue Example~\ref{sig-ex}. In the expression for $f^{(\cal S)}(Q'')$ that we obtained in that Example, 

\begin{tabbing} 
$f^{(\cal S)}(Q'')$ $=$ ($N_1 \times N_2$ $+$  $N_2 \times N_1$ \\ 
$-$ $min(N_1,N_2)$ $\times$ $min(N_1,N_2)$) . 
\end{tabbing} 

\noindent 
the ``{\em first} position'' in the only parenthesized expression in $f^{(\cal S)}(Q'')$ is the expression 

\begin{tabbing} 
$f_1^{(\cal S)}(Q'')$ $=$ ($N_1$ $+$  $N_2$ $-$ $min(N_1,N_2)$ ) . 
\end{tabbing} 

\noindent 
Intuitively, that first position in the only parenthesized expression in $f^{(\cal S)}(Q'')$ was obtained by keeping only the {\em first} element of each term in the original parenthesized expression. These correspond to a ``union expression'' for exactly the first element of each of the two copy signatures ($[ N_1 \ N_2]$ and $[ N_2 \ N_1]$) associated with this parenthesized expression. 

Similarly, the ``{\em second} position'' in the only parenthesized expression in $f^{(\cal S)}(Q'')$ is the expression 

\begin{tabbing} 
$f_2^{(\cal S)}(Q'')$ $=$ ($N_2$ $+$  $N_1$ $-$ $min(N_1,N_2)$ ) . 
\end{tabbing} 

\noindent 
Intuitively, that second position in the only parenthesized expression in $f^{(\cal S)}(Q'')$ was obtained by keeping only the {\em second} element of each term in the original parenthesized expression. These correspond to a ``union expression'' for exactly the second element of each of the two copy signatures ($[ N_1 \ N_2]$ and $[ N_2 \ N_1]$) associated with this parenthesized expression. 
\end{example}  

Interestingly, a more concise expression for $f_1^{(\cal S)}(Q'')$ in Example~\ref{second-sig-ex} is $f_1^{(\cal S)}(Q'')$ $=$ $max(N_1, N_2)$. Indeed, the expression given as $f_1^{(\cal S)}(Q'')$ in Example~\ref{second-sig-ex} is the equivalent expansion of the expression $max(N_1, N_2)$ using the inclusion-exclusion principle for unions of sets. This is exactly the intuition for our promised easy evaluation for signatures: The first position in the expression for $f^{(\cal S)}(Q'')$ in Example~\ref{sig-ex} corresponds to a formula, using the inclusion-exclusion principle, for the cardinality of the union of the sets $\{$ $1$, $\ldots$, $N_1$ $\}$ and $\{$ $1$, $\ldots$, $N_2$ $\}$. The former set is exactly the range of the values that are provided by the assignments in the monomial class ${\cal C}_{1}^{(Q'')}$ to the column for the first copy variable of the query $Q''$ (see Example~\ref{sig-ex} for the first element of the copy signature of that monomial class), when $Q''$ is evaluated on any of the databases $D_{{\bar N}^{(i)}}(Q)$. Similarly, the latter set, i.e., $\{$ $1$, $\ldots$, $N_2$ $\}$,   is exactly the range of the values that are provided by the assignments in the monomial class ${\cal C}_{2}^{(Q'')}$ to the column for the first copy variable of the query $Q''$ (see Example~\ref{sig-ex} for the first element of the copy signature of that monomial class), when $Q''$ is evaluated on any of the databases $D_{{\bar N}^{(i)}}(Q)$. That is, the expression for $f_1^{(\cal S)}(Q'')$ in Example~\ref{second-sig-ex} is an expression for evaluating the cardinality of the union of these two sets; the latter union is exactly the set of values of the first copy variable of the query $Q''$, when $Q''$ is evaluated on any of the databases $D_{{\bar N}^{(i)}}(Q)$ under all the assignments provided by the monomial classes ${\cal C}_{1}^{(Q'')}$ and ${\cal C}_{2}^{(Q'')}$ together. We can obtain a similar intuition for the expression for $f_2^{(\cal S)}(Q'')$ in Example~\ref{second-sig-ex}, that is again as a $max(N_1, N_2)$. 

\begin{example} 
\label{third-sig-ex} 
We continue with Example~\ref{second-sig-ex}. Fix a vector ${\bar K}_w$ $=$ $\{$ $N_1 \ N_2$ $\}$. Then the expression  for $f_1^{(\cal S)}(Q'')$ in Example~\ref{second-sig-ex} evaluates to $N_2$, and so does the expression  for $f_2^{(\cal S)}(Q'')$ in Example~\ref{second-sig-ex}. As a result, the set ${\cal S}(Q'',{\bar K}_w)$ for this vector ${\bar K}_w$ is the result of adding up these two positive-sign occurrences of $N_2$, from $f_1^{(\cal S)}(Q'')[{\bar K}_w]$ $=$ $max(N_1, N_2)[{\bar K}_w]$ $=$ $N_2$ and similarly from $f_2^{(\cal S)}(Q'')[{\bar K}_w]$ $=$ $N_2$. That is, the set ${\cal S}(Q'',{\bar K}_w)$ for this vector ${\bar K}_w$ is ${\cal S}(Q'',{\bar K}_w)$ $=$ $\{$ $2 N_2$ $\}$. 
\end{example}

We now formalize the above intuition in Proposition~\ref{key-sig-prop}. To formulate Proposition~\ref{key-sig-prop}, we use the following notation: Let the function $f^{(\cal S)}(Q'')$, for a CCQ query $Q''$, have exactly $q$ $\geq$ $1$ parenthesized subexpressions as defined by the fourth paragraph of this subsection; we enumerate these parenthesized subexpressions from left to right as the $1$st through $q$th parenthesized subexpression of the function $f^{(\cal S)}(Q'')$. In the $k$th such parenthesized subexpression, for any $k$ between $1$ and $q$ inclusively, let the $l_k$ $\geq$ $1$ monomial classes involved in the union in this subexpression be ${\cal C}_{k1}^{(Q'')}$ through ${\cal C}_{kl_k}^{(Q'')}$. Further, for each $p$ $\in$ $\{$ $1$, $\ldots,$ $l_k$ $\}$, let the copy signature of the class  ${\cal C}_{kp}^{(Q'')}$ be denoted by $[$ $V^{(1)}_{kp}$, $V^{(2)}_{kp}$, $\ldots$, $V^{(r)}_{kp}$ $]$, where each element of the vector is (by definition of copy signature) an element of the set $\{$ $1$, $N_{m+1}$, $\ldots$, $N_{m+w}$ $\}$. For a fixed $j$ $\in$ $\{$ $1$, $\ldots$, $r$ $\}$, consider the result of keeping in the $k$th parenthesized subexpression of $f^{(\cal S)}(Q'')$ only the $j$th element of each term, with the (positive or negative) sign of the original term kept around. 
We refer to the resulting expression as {\em the $j$th projection of the $k$th parenthesized subexpression of} $f^{(\cal S)}(Q'')$. 
 (For instance, for the function $f^{(\cal S)}(Q'')$ that we obtained in Example~\ref{sig-ex}, the $1$st projection of the first (and only) parenthesized subexpression of $f^{(\cal S)}(Q'')$ is the expression $f_1^{(\cal S)}(Q'')$ $=$ ($N_1$ $+$  $N_2$ $-$ $min(N_1,N_2)$); see Example~\ref{second-sig-ex} for the details.) 

We are now ready to formulate Proposition~\ref{key-sig-prop}. 

\begin{proposition} 
\label{key-sig-prop} 
Let the function $f^{(\cal S)}(Q'')$, for a CCQ query $Q''$, have $q$ $\geq$ $1$ parenthesized subexpressions. Fix a value $k$ $\in$ $\{$ $1$, $\ldots$, $q$ $\}$ and a value $j$ $\in$ $\{$ $1$, $\ldots$, $r$ $\}$. Denote by ${\cal C}_{k1}^{(Q'')}$ through ${\cal C}_{kl_k}^{(Q'')}$, for some $l_k$ $\geq$ $1$, all the monomial classes contributing to the $k$th parenthesized subexpression of $f^{(\cal S)}(Q'')$. Then the  $j$th projection of the $k$th parenthesized subexpression of $f^{(\cal S)}(Q'')$ is equivalent to the expression $f^{(\cal S)}_{kj}(Q'')$ $=$ $max(V^{(j)}_{k1}, \ldots, V^{(j)}_{kl_k})$, where each of $V^{(j)}_{k1}$, $\ldots$, $V^{(j)}_{kl_k}$ is the $j$th element of the copy signature of the respective monomial class. 
\end{proposition} 

The result of Proposition~\ref{key-sig-prop} is by construction of the union expressions for Proposition~\ref{bigunion-prop} of Section~\ref{hard-funct-case-sec}, as given by the inclusion-exclusion principle for unions of sets and by the $min$-expressions of Proposition~\ref{multip-monomial-prop} of Section~\ref{hard-funct-case-sec} for intersections of sets. 

Proposition~\ref{key-sig-prop} gives us immediately a linear-time procedure for constructing the signature ${\cal S}(Q'',{\bar K}_w)$ for each function $f^{(\cal S)}(Q'')$ and copy-variable-ordering vector ${\bar K}_w$: 

\begin{corollary} 
\label{key-sig-one-corol} 
Let the function $f^{(\cal S)}(Q'')$, for a CCQ query $Q''$, have  $q$ $\geq$ $1$ parenthesized subexpressions. 
Then for each copy-variable-ordering vector ${\bar K}_w$, the signature ${\cal S}(Q'',{\bar K}_w)$ of $f^{(\cal S)}(Q'')$ with respect to ${\bar K}_w$ can be constructed by (1) obtaining all the $q$ $\times$ $r$ positive-sign $max$-expressions given by Proposition~\ref{key-sig-prop}, by then (2) evaluating each $max$-expression under the total ($\leq$) order ${\bar K}_w$, and by finally (3) adding up all the positive-sign occurrences of all the variables $N_{m+1}$, $\ldots$, $N_{m+w}$, separately for each variable, in these results of evaluating the $max$-expressions under ${\bar K}_w$.  
\end{corollary} 

The result of Corollary~\ref{key-sig-one-corol} is straightforward from Proposition~\ref{key-sig-prop} and from the following Lemma~\ref{excl-incl-lemma}. 

\begin{lemma} 
\label{excl-incl-lemma} 
Given $n \geq 1$ natural numbers $a_1$, $a_2$, $\ldots,$ $a_n$ such that $a_1$ $\leq$ $a_2$ $\leq$ $\ldots$ $\leq$ $a_n$. For each $j$ $\in$ $\{ 1,\ldots,n \}$, let the set ${\cal A}_j$ be ${\cal A}_j$ $:=$ $\{ 1,$ $2,$ $\ldots,$ $a_{j}-1,$ $a_j \}$. Then the cardinality of the set $\bigcup_{j=1}^n {\cal A}_j$ is the natural number $a_n$. 
\end{lemma} 

The claim of Lemma~\ref{excl-incl-lemma} is trivial. (Observe that for all $n$ we have that (a) ${\cal A}_j \subseteq {\cal A}_n$ for all $j$ $\in$ $\{ 1,\ldots,n \}$, and that (b) the cardinality of the set ${\cal A}_n$ is exactly $a_n$.) 

We end this subsection by observing that all elements of the set ${\cal S}(Q',{\bar K}_w)$ have positive coefficients, which add up to exactly $q$ $\times$ $r$. This result is immediate from  Proposition~\ref{key-sig-prop}. 

\begin{corollary} 
\label{key-sig-two-corol} 
Given a CCQ query $Q'$ and a copy-variable-ordering vector ${\bar K}_w$ as specified in the proof of Theorem~\ref{magic-mapping-prop}. Denote by $q$ $\geq$ $1$ the total number of parenthesized expressions in the function $f^{(\cal S)}(Q')$. Then  we have that: 
\begin{enumerate} 
	\item Each element of the signature ${\cal S}(Q',{\bar K}_w)$ of $f^{(\cal S)}(Q')$ w.r.t. ${\bar K}_w$ has a positive coefficient; and 
	\item All the positive coefficients of all the elements of the set ${\cal S}(Q',{\bar K}_w)$ add up to $q$ $\times$ $r$. 
\end{enumerate}  
\vspace{-0.4cm} 
\end{corollary}


\subsubsection{Proof of Proposition~\ref{goldfish-in-q-prime-prop}} 
\label{sec-hinge-proof-sec} 

We are now ready to prove Conjecture~\ref{goldfish-in-q-prime-conj} (of Section~\ref{intuition-sec}), that is the following Proposition~\ref{goldfish-in-q-prime-prop}. This proof completes the proof of Proposition~\ref{qprime-goldfish-prop} of Section~\ref{qprime-goldfish-sec} (see Section~\ref{intuition-sec}), and hence the proof of Theorem~\ref{magic-mapping-prop}, as outlined in Section~\ref{main-results-summary-sec}. 

\begin{proposition} 
\label{goldfish-in-q-prime-prop} 
Given an explicit-wave CCQ query $Q$ and a CCQ query $Q'$ such that $Q \equiv_C Q'$. Then for each copy-variable-ordering vector, ${\bar K}_w$,  for the vector $\bar N$, the  $m$-covering part of the function ${\cal F}^{(Q')}_{(Q)}[{\bar K}_w]$ has, as a solid term, the wave ${\cal P}^{(Q)}_*$ of the query $Q$ w.r.t. the family of databases  $\{ D_{\bar{N}^{(i)}}(Q) \}$. 
\end{proposition} 

In the proof of Proposition~\ref{goldfish-in-q-prime-prop}, we will use the following observation. 

\begin{proposition} 
\label{phantom-min-two-prop} 
Consider two CCQ queries $Q$ and $Q'$ as in the statement of Theorem~\ref{magic-mapping-prop}. Let ${\bar K}_w$ be a copy-variable-ordering vector, for the vector $\bar N$ of the query $Q$, and let ${\cal P}[{\bar K}_w]$ be an arbitrary phantom term in the $m$-covering part of ${\cal F}^{(Q')}_{(Q)}[{\bar K}_w]$. Then for the product-of-$min$-expressions, $\cal P$, in ${\cal F}^{(Q')}_{(Q)}$, such that $\cal P$ gives rise to ${\cal P}[{\bar K}_w]$ under the total ($\leq$) order ${\bar K}_w$, we have that for at least one $j$ $\in$ $\{$ $1$, $\ldots$, $r$ $\}$, the $min$-expression that is the $j$th projection of $\cal P$ has (the $min$-expression) at least two distinct arguments from the set $\{$ $1$, $N_{m+1}$, $\ldots$, $N_{m+w}$ $\}$. 
\end{proposition} 

(Please see Section~\ref{bqkw-sec} for both the $j$th-projection terminology and for the idea of product-of-$min$-expressions generating terms of $m$-covering parts of ${\cal F}^{(Q')}_{(Q)}[{\bar K}_w]$ under the total ($\leq$) orders ${\bar K}_w$.) 

\begin{proof} 
The proof is by contradiction: Assume that for no $j$ between $1$ and $r$ inclusively does the ($min$-expression that is the) $j$th projection of $\cal P$ have more than one argument. Then $\cal P$ must be a solid term, and so must be ${\cal P}^{(Q)}_*$ generated from $\cal P$ under ${\bar K}_w$, a contradiction. (It is immediate from Proposition~\ref{multip-monomial-prop} of Section~\ref{hard-funct-case-sec} that, if the $min$-expression that is the $j$th projection of a term has only one argument, for all $j$ from $1$ to $r$ inclusively, then the term must be a multiplicity monomial for some nonempty monomial class for the CCQ query in question.) 
\end{proof}

\begin{proof}{(Proposition~\ref{goldfish-in-q-prime-prop})} 
We consider first the special case where $Q'$ is under the jurisdiction of Proposition~\ref{easy-final-funct-prop} (of Section~\ref{easy-funct-case-sec}; our query $Q'$ would be the $Q''$ in the statement of  Proposition~\ref{easy-final-funct-prop}). In that case, for each vector ${\bar K}_w$ we clearly have that all terms in the $m$-covering part of the function ${\cal F}^{(Q')}_{(Q)}[{\bar K}_w]$ are solid terms. Hence, by 
Proposition~\ref{same-funct-in-q-prime-prop} we have the desired result of Proposition~\ref{goldfish-in-q-prime-prop} for this case. Observe that this special case includes all cases where $w$ $=$ $0$. 

Note that in all cases where $Q'$ is {\em not} under the jurisdiction of Proposition~\ref{easy-final-funct-prop}, it holds that we have $r \geq 2$. (Recall Corollary~\ref{easy-case-corol} of Section~\ref{easy-funct-case-sec}; the Corollary outlines a special case of Proposition~\ref{easy-final-funct-prop} and covers all cases where $r \leq 1$.) In the remainder of the proof, we consider this case of $Q'$ not satisfying the conditions of Proposition~\ref{easy-final-funct-prop} (and hence $r \geq 2$ for this case). 

The proof for this case is by contradiction: We assume that for some vector ${\bar K}_w$, the $m$-covering part of the multivariate polynomial ${\cal F}^{(Q')}_{(Q)}[{\bar K}_w]$ has the wave ${\cal P}^{(Q)}_*$ of the query $Q$, w.r.t. the family of databases  $\{ D_{\bar{N}^{(i)}}(Q) \}$, as a {\em phantom} term. (The fact that the $m$-covering part of ${\cal F}^{(Q')}_{(Q)}[{\bar K}_w]$ has the wave of the query $Q$ is by Proposition~\ref{same-funct-in-q-prime-prop}.) Note that it follows immediately from the definitions of solid and phantom terms (see Section~\ref{solid-phantom-sec}) that for {\em all} vectors ${\bar K}_w$, the $m$-covering part of the multivariate polynomial ${\cal F}^{(Q')}_{(Q)}[{\bar K}_w]$ has the wave ${\cal P}^{(Q)}_*$ of the query $Q$ (w.r.t. the family of databases  $\{ D_{\bar{N}^{(i)}}(Q) \}$) as a {\em phantom} term. 

For the remainder of the proof, we fix any one vector ${\bar K}_w$ such that the $m$-covering part of the multivariate polynomial ${\cal F}^{(Q')}_{(Q)}[{\bar K}_w]$ has the wave ${\cal P}^{(Q)}_*$ of the query $Q$, w.r.t. the family of databases  $\{ D_{\bar{N}^{(i)}}(Q) \}$, as a phantom term. In addition, we denote by $\cal P$ the product-of-$min$-expressions in ${\cal F}^{(Q')}_{(Q)}$, such that $\cal P$ gives rise to ${\cal P}[{\bar K}_w]$ under the total ($\leq$) order ${\bar K}_w$. (We use here the same terminology as in the statement of Proposition~\ref{phantom-min-two-prop}.)


In the proof by contradiction, we consider separately the special case where $w$ $=$ $1$, and then the case where $w$ $\geq$ $2$. (Recall that we dealt with the case $w$ $=$ $0$ in the first paragraph of the proof.) 

(1) Suppose $w$ $=$ $1$. For the $\cal P$ fixed as explained above, by Proposition~\ref{phantom-min-two-prop} there must be a $j$ between $1$ and $r$ inclusively, such that the $min$-expression  that is the $j$th projection of $\cal P$ has (the $min$-expression) at least two distinct arguments from the set $\{$ $1$, $N_{m+1}$, $\ldots$, $N_{m+w}$ $\}$. Fix an arbitrary value of $j$ that satisfies this condition. Observe that in case where $w$ $=$ $1$, the set $\{$ $1$, $N_{m+1}$, $\ldots$, $N_{m+w}$ $\}$ has exactly two elements, $1$ and $N_{m+1}$. Thus, the only case where the $j$th projection of $\cal P$ can have at least two arguments is when that $min$-expression (which is the $j$th projection of $\cal P$) has exactly two arguments, one of them $1$ and the other $N_{m+1}$. 

We now recall that for all copy-variable-ordering vectors ${\bar K}_w$, the value of $N_{m+1}$ is always greater than or equal to $1$. Thus, any $min$-expression as above would always evaluate immediately to $1$ in $\cal P$. With this conclusion, we arrive at the desired contradiction with the assumption that under the fixed vector ${\bar K}_w$, the term $\cal P$ could give rise, in the function ${\cal F}^{(Q')}_{(Q)}[{\bar K}_w]$, to the wave ${\cal P}^{(Q)}_*$ of the query $Q$. Indeed, the wave of the query $Q$ is by definition a product of $m+r$ (not necessarily distinct) variable names from the set $\{$ $1$, $N_{1}$, $\ldots$, $N_{m+w}$ $\}$. At the same time, each term generated from the function ${\cal F}^{(Q')}_{(Q)}$ under each vector ${\bar K}_w$ cannot be a product of more than $m+r$ elements, please see Proposition~\ref{multip-monomial-prop} of Section~\ref{hard-funct-case-sec}. That is, whenever any element of the product in any such term is the constant $1$, then the total number of the remaining multipliers constituting the term cannot exceed $m+r-1$. This contradiction concludes our proof for the case where $w$ $=$ $1$. 

(2) Suppose now that $w$ $\geq$ $2$; that is, we have that the set $\{$ $1$, $N_{m+1}$, $\ldots$, $N_{m+w}$ $\}$ has at least two distinct variable names corresponding to the copy variables of the query $Q$. Similarly to our argument for the case where $w$ $=$ $1$, for the $\cal P$ fixed as explained above, by Proposition~\ref{phantom-min-two-prop} there must be a $j$ between $1$ and $r$ inclusively, such that the $min$-expression  that is the $j$th projection of $\cal P$ has (the $min$-expression) at least two distinct arguments from the set $\{$ $1$, $N_{m+1}$, $\ldots$, $N_{m+w}$ $\}$. 

Fix an arbitrary value of $j$ that satisfies this condition. Again similarly to our argument for the case where $w$ $=$ $1$, we obtain that for the (at least)  two distinct arguments from the set $\{$ $1$, $N_{m+1}$, $\ldots$, $N_{m+w}$ $\}$ in the $min$-expression that is the $j$th projection of $\cal P$, these distinct arguments must include at least two distinct elements of the set $\{$ $N_{m+1}$, $\ldots$, $N_{m+w}$ $\}$. That is, the $j$th projection of $\cal P$ must have at least two distinct variable names, from among $N_{m+1}$, $\ldots$, $N_{m+w}$, as arguments of its $min$-expression.  Denote the set of all those distinct arguments of that $min$-expression that are elements of the set $\{$ $N_{m+1}$, $\ldots$, $N_{m+w}$ $\}$ as the set $S^*$. By our argument, the cardinality $|S^*|$ of the  set $S^*$ is at least two. 


Now consider the function $f^{(\cal S)}(Q')$ constructed for the $m$-covering part of the function ${\cal F}^{(Q')}_{(Q)}$ as described in Section~\ref{bqkw-sec}. Suppose that the term $\cal P$ in $f^{(\cal S)}(Q')$, where $\cal P$  is fixed as explained before item (1) of this proof, belongs to the $k$th parenthesized expression of $f^{(\cal S)}(Q')$, for some $k$ $\geq$ $1$. In the remainder of the proof, we keep this value $k$ fixed. Denote by (w.l.o.g.) ${\cal C}_1^{(Q')}$ through ${\cal C}_p^{(Q')}$, for some $p$ $\geq$ $1$, all the $p$ distinct monomial classes that have contributed to the construction of this $k$th parenthesized expression of $f^{(\cal S)}(Q')$. From the set $S^*$, of cardinality at least two, being the set of arguments of the $min$-expression in the $j$th projection of $\cal P$, we infer that $p$ $\geq$ $2$. Further, for the $j$ fixed as above, consider the set, call it $S^{**}$, of all items that occur as the $j$th element of the copy signature of ${\cal C}_l^{(Q')}$, for all $l$ $\in$ $\{$ $1$, $\ldots,$ $p$ $\}$. From the fact that $\cal P$ is a term in the $k$th parenthesized expression of $f^{(\cal S)}(Q')$ to whose (the $k$th parenthesized expression) construction all of ${\cal C}_1^{(Q')}$ through ${\cal C}_p^{(Q')}$ have contributed, we have that $S^{**}$ is a superset of the set $S^*$. Hence, the set $S^{**}$ contains at least two distinct elements of the set $\{$ $N_{m+1}$, $\ldots$, $N_{m+w}$ $\}$, and the cardinality $|S^{**}|$ of the set $S^{**}$ is at least two. 

The contradiction that we are to arrive at is going to be as follows. We will show that, in addition to the vector ${\bar K}_w$ fixed as explained above, there exists a different copy-variable-ordering vector ${\bar K}'_w$ for the query $Q$, with the following property: The signature ${\cal S}(Q',{\bar K}_w)$ of the function $f^{(\cal S)}(Q')$ w.r.t. the vector ${\bar K}_w$ is not identical to the signature ${\cal S}(Q',{\bar K}'_w)$ of the function $f^{(\cal S)}(Q')$ w.r.t. the vector ${\bar K}'_w$. 
By Proposition~\ref{not-same-sigs-prop}, we will then obtain that 
the $m$-covering part of the function ${\cal F}^{(Q')}_{(Q)}[{\bar K}_w]$ and the $m$-covering part of the function ${\cal F}^{(Q')}_{(Q)}[{\bar K}'_w]$ (for the query $Q'$ and for the two distinct vectors ${\bar K}_w$ and ${\bar K}'_w$) are nonidentical polynomials. That allows us to immediately obtain a contradiction with Propositions~\ref{equiv-funct-prop} 
and~\ref{m-covering-prop}, which together claim that for all pairs  $({\bar K}_w$, ${\bar K}'_w)$ of copy-variable-ordering vectors  for the query $Q$, the $m$-covering part of the function ${\cal F}^{(Q')}_{(Q)}[{\bar K}_w]$ and the $m$-covering part of the function ${\cal F}^{(Q')}_{(Q)}[{\bar K}'_w]$ must be identical polynomials, whenever $Q$ $\equiv_C$ $Q'$ and $Q$ is an explicit-wave query. 

Thus, to complete this proof it remains to show that there exists a copy-variable-ordering vector ${\bar K}'_w$ for the query $Q$ such that ${\bar K}'_w$ is not identical to the fixed vector ${\bar K}_w$, and such that the signature ${\cal S}(Q',{\bar K}_w)$ of the function $f^{(\cal S)}(Q')$ w.r.t. the vector ${\bar K}_w$ is not identical to the signature ${\cal S}(Q',{\bar K}'_w)$ of the function $f^{(\cal S)}(Q')$ w.r.t. the vector ${\bar K}'_w$. 

We construct one vector ${\bar K}'_w$ with this desired property from the vector ${\bar K}_w$, as follows. Recall the set $S^{**}$, which contains at least two distinct elements of the set $\{$ $N_{m+1}$, $\ldots$, $N_{m+w}$ $\}$. In the set $S^{**}$, let $N_A$ be the element that is the minimal element of $S^{**}$ under the total ($\leq$) order ${\bar K}_w$. Similarly, in the set $S^{**}$, let $N_B$ be the element that is the maximal element of $S^{**}$ under the total ($\leq$) order ${\bar K}_w$. (Actually, picking $N_A$ to be any element of $S^{**}$ that is distinct from variable name $N_B$ would be sufficient for our purpose in this proof. We choose $N_A$ as the minimal element to make it easy to see how we construct ${\bar K}'_w$ from ${\bar K}_w$ -- by swapping, in the relative $\leq$ order of the variable names, the minimal element with the maximal element of $S^{**}$ in the vectors ${\bar K}_w$ and ${\bar K}'_w$.)  

By construction of  the set $S^{**}$, it is immediate that: 
\begin{itemize} 
	\item $N_A$ and $N_B$ are two distinct elements of the set $\{$ $N_{m+1}$, $\ldots$, $N_{m+w}$ $\}$; and 
	\item $N_B$ is the variable name that  the $j$th projection of the $k$th parenthesized expression of $f^{(\cal S)}(Q')$ contributes to the signature ${\cal S}(Q',{\bar K}_w)$ of the function $f^{(\cal S)}(Q')$ w.r.t. the vector ${\bar K}_w$. (This fact is immediate from Proposition~\ref{key-sig-prop}.) 
\end{itemize} 
	
We now take as the desired vector ${\bar K}'_w$ an arbitrary copy-variable-ordering vector for $Q$ such that $N_B$ is the minimal element of ${\bar K}'_w$, with the exception of the element $1$ of ${\bar K}'_w$, and such that $N_A$ is the maximal element of ${\bar K}'_w$. That is, we take as ${\bar K}'_w$ an arbitrary copy-variable-ordering vector for $Q$ whose first two elements are $1$ and $N_B$, in this order, and whose last element is $N_A$. Clearly, by $w$ $\geq$ $2$ we have that at least one such vector ${\bar K}'_w$ must exist. 

Consider the variable name that  the $j$th projection of the $k$th parenthesized expression of $f^{(\cal S)}(Q')$ contributes to the signature ${\cal S}(Q',{\bar K}'_w)$ of the function $f^{(\cal S)}(Q')$ w.r.t. the vector ${\bar K}'_w$. By construction of the set $S^{**}$ and of the vector ${\bar K}'_w$, we have that this variable name is $N_A$. 

Now for the two signatures ${\cal S}(Q',{\bar K}_w)$ and ${\cal S}(Q',{\bar K}'_w)$ to still be identical sets, it must be that some other, fixed projection of some fixed parenthesized expression of $f^{(\cal S)}(Q')$  would contribute: 
\begin{itemize} 
	\item the variable name $N_B$ to the signature ${\cal S}(Q',{\bar K}'_w)$, and 
	\item some variable name  $N_C$ that is distinct from $N_B$, to the signature ${\cal S}(Q',{\bar K}_w)$. ($N_C$ may or may not be identical to $N_A$. Intuitively, for the two signatures to be identical sets ``on the balance,'' there is to be a cycle going through the variable names $N_A$ and $N_B$, and that cycle may or may not involve other elements of the set $\{$ $N_{m+1}$, $\ldots$, $N_{m+w}$ $\}$.) 
\end{itemize} 
	
	This observation gives is the desired contradiction. Indeed, if some  projection of some parenthesized expression of $f^{(\cal S)}(Q')$ contributes the variable name $N_B$ to the signature ${\cal S}(Q',{\bar K}'_w)$, then the $max$-expression for that projection can contain, as its arguments, only $1$ in addition to $N_B$. (Recall that $N_B$ is the minimal-value variable name under the total ($\leq$) order  ${\bar K}'_w$.) Thus, it is not possible for the signatures ${\cal S}(Q',{\bar K}_w)$ and ${\cal S}(Q',{\bar K}'_w)$ to be identical sets. This observation completes the proof for the case $w$ $\geq$ $2$, and hence completes the proof of the entire Proposition~\ref{goldfish-in-q-prime-prop}.  
\end{proof} 


\subsubsection{Theorem~\ref{magic-mapping-prop}: End of the proof} 
\label{hinge-proof-sec}




As summarized in Section~\ref{proof-intuition-sec}, the proof of Theorem~\ref{magic-mapping-prop} is immediate from three results, as follows. 
\begin{itemize} 
	\item Proposition~\ref{q-has-wave-prop} of Section~\ref{monomial-class-mappings-sec} states the following: Given a CCQ query $Q$, there exists a nonempty monomial class, call it ${\cal C}_*^{(Q)}$, for the query $Q$ w.r.t. the family of databases  $\{ D_{\bar{N}^{(i)}}(Q) \}$, such that the multiplicity monomial of ${\cal C}_*^{(Q)}$ is the wave of the query $Q$ w.r.t. $\{ D_{\bar{N}^{(i)}}(Q) \}$. 
	\item Proposition~\ref{q-same-scale-mpng-prop} of Section~\ref{monomial-class-mappings-sec} states the following: 
Given  CCQ queries $Q(\bar{X}) \leftarrow L,M$ and $Q'(\bar{X}') \leftarrow L',M'$, such that (i) $Q$ and $Q'$ have the same (positive-integer) head arities, (ii) $|M_{copy}|$ $=$ $|M'_{copy}|$, and (iii) $|M_{noncopy}|$ $=$ $|M'_{noncopy}|$. Suppose that there exists a nonempty monomial class ${\cal C}_*^{(Q')}$ for the query $Q'$  w.r.t. the family of databases $\{ D_{\bar{N}^{(i)}}(Q) \}$, such that the multiplicity monomial of ${\cal C}_*^{(Q')}$ is the wave of the query $Q$  w.r.t. $\{ D_{\bar{N}^{(i)}}(Q) \}$. Then there exists a SCVM from  the query $Q'$ to the query $Q$.    
	\item Proposition~\ref{qprime-goldfish-prop} of Section~\ref{q-prime-has-wave-sec} states the following: Whenever 
	\begin{itemize} 
		\item[(a)] $Q \equiv_C Q'$ for CCQ queries $Q$ and $Q'$, and 
		\item[(b)] $Q$ is an explicit-wave CCQ query (as specified by Definition~\ref{expl-wave-def}), 
	\end{itemize} 
then there exists a (nonempty) monomial class ${\cal C}_*^{(Q')}$ for the query $Q'$ and for the family of databases $\{ D_{\bar{N}^{(i)}}(Q) \}$, such that the multiplicity monomial of ${\cal C}_*^{(Q')}$ is the wave of the query $Q$ w.r.t. $\{ D_{\bar{N}^{(i)}}(Q) \}$. 
\end{itemize}  

We can finally conclude that the result of Theorem~\ref{magic-mapping-prop} holds. Q.E.D.

\reminder{Must move *all* the continuations and finalizations of Example~\ref{writeup-weird-ex} to *before* this Section~\ref{q-prime-has-wave-sec}}

\reminder{ 

\subsection{Showing that the Exposed Wave of $Q$ Induces an Explicit Wave of $Q'$ in Case $Q \equiv_C Q'$} 

\subsubsection{True Proof?} 
\label{goldfish-sec}

Suppose that for a CCQ query $Q$ and for the set $M_{noncopy}$ of multiset noncopy variables of $Q$, $|M_{noncopy}| = m \geq 1$. Then let  list $L^{(Q)}_{nc}$ be an arbitrary fixed ordering of the set of multiset noncopy variables of $Q$, and let  $\pi^{(Q)}_{nc}$ be an arbitrary permutation of the list $L^{(Q)}_{nc}$. For a containment mapping $\mu$ from the query $Q_{H^d}$ (see Section~\ref{minimiz-sec}) to itself, we say that $\mu$ {\em agrees with the permutation} $\pi^{(Q)}_{nc}$ if for each $i$th position in the list $L^{(Q)}_{nc}$, $1 \leq i \leq m$, $\mu$ maps the variable in that position into the variable in the $i$th position of $\pi^{(Q)}_{nc}$. (As an easy example, for the case where the permutation $\pi^{(Q)}_{nc}$ is the identity permutation of the list $L^{(Q)}_{nc}$, for  a containment mapping $\mu$ from the query $Q_{H^d}$ to itself, we say that $\mu$ agrees with the permutation $\pi^{(Q)}_{nc}$ $=$ $L^{(Q)}_{nc}$ whenever $\mu$ is an identity mapping w.r.t. each multiset noncopy variable of the query $Q$.) 

In case where CCQ query $Q$ does not have multiset noncopy variables, that is where $|M_{noncopy}| = m = 0$, the only possible list $L^{(Q)}_{nc}$ of multiset noncopy variables of $Q$ is an empty list by definition. In this case, the set of all permutations of the empty list $L^{(Q)}_{nc}$ is a singleton set comprising one empty list. Whenever the query $Q$ is such that $|M_{noncopy}| = m = 0$, we say that for each containment mapping $\mu$ from the query $Q_{H^d}$ to itself,  $\mu$ {\em agrees with the (only) permutation} $\pi^{(Q)}_{nc} = []$ of the list $L^{(Q)}_{nc} = []$. 

\begin{definition}{Explicit-wave CCQ query} 
\label{old-expl-wave-def}
A CCQ query $Q$ is an {\em explicit-wave (CCQ) query} if one of the following two conditions holds: 
\begin{itemize} 
	\item[(1)] It is the case that $r \leq 1$; or 
	\item[(2)] It is the case that $r \geq 2$, and for each permutation $\pi^{(Q)}_{nc}$ of the list $L^{(Q)}_{nc}$ of multiset noncopy variables of $Q$ and for each pair of containment mappings $(\mu_1,\mu_2)$ from $Q_{H^d}$ to itself such that each of $\mu_1$ and $\mu_2$ agrees with the permutation $\pi^{(Q)}_{nc}$, it holds that for each multiset-on subgoal, $s$, of the query $Q_{H^d}$, $\mu_1(s)$ and $\mu_2(s)$ are identical atoms. 
\end{itemize} 
\vspace{-0.4cm}
\end{definition} 

It is easy to see that in case (2) of Definition~\ref{old-expl-wave-def}, a CCQ query $Q$ is or is not an explicit-wave query regardless of the choice of the ordering $L^{(Q)}_{nc}$ of the variables in the set  $M_{noncopy}$. 

\reminder{Do I absolutely need the ``{\em for each} permutation $\pi^{(Q)}_{nc}$ of the list $L^{(Q)}_{nc}$'' in Definition~\ref{expl-wave-def}, as opposed to ``{\em there exists a} permutation $\pi^{(Q)}_{nc}$ of the list $L^{(Q)}_{nc}$''?  I guess I do need ``for each permutation'' to prove Theorem~\ref{magic-mapping-prop}, in order to have no negative terms in the {\em entire} function (in terms of $N_{m+1}$ $\ldots$ $N_{m+w}$) for $Q$, call that function $f^{(Q)}_p$, where that function is multiplied by $\Pi_{j=1}^m N_j$, that is, by the term representing all permutations of the list of all (if any) multiset noncopy variables of $Q$. If the function $f^{(Q)}_p$ has no negative terms, then it must be {\em the same} multivariate polynomial (in terms of $N_{m+1}$ $\ldots$ $N_{m+w}$) on {\em all} the total-order subdomains of the domain $\cal N$ of the vector $\bar N$. Then the corresponding function $f^{(Q')}_p$ {\em for the query} $Q'$, $Q' \equiv_C Q$, would also be forced to be the same  multivariate polynomial (in terms of $N_{m+1}$ $\ldots$ $N_{m+w}$) on {\em all} the total-order subdomains of the domain $\cal N$ of the vector $\bar N$, and hence would also be forced to have no negative terms. {\em Whenever $f^{(Q')}_p$ has no negative terms, it must have -- as a ``real''  term (as opposed to ``imaginary'' term, which would be coming from the union-of-sets formula by the inclusion-exclusion principle) -- the characteristic wave of $Q$ (by $f^{(Q')}_p$ being the same multivariate polynomial as $f^{(Q)}_p$ on the entire domain $\cal N$ of $\bar N$, and by $f^{(Q)}_p$ having the characteristic wave of $Q$ [because $Q$ is an explicit-wave CCQ query]). Hence we can construct a SCVM from $Q'$ to $Q$, and hence such a mapping exists.} 

Proving that function $f^{(Q')}_p$ does not have negative terms: 
First we observe that in case $m = |M_{noncopy}| \leq 1$, function $f^{(Q)}_p$ has exactly one term, the characteristic-wave monomial of $Q$, on the entire domain $\cal N$ of $\bar N$. (This is immediate from Definition~\ref{expl-wave-def}.) Hence, by $f^{(Q')}_p$ $\equiv$ $f^{(Q)}_p$ on all of $\cal N$, we conclude that $f^{(Q')}_p$ does not have negative terms in case $m = 0$. Hence, in the remainder of this proof we assume $m \geq 2$. 
Under this assumption, the proof is by contradiction: 
Suppose that function $f^{(Q')}_p$ has negative terms. 
This is only possible when (at least) one of the size-of-union terms in $f^{(Q')}_p$  is for a union of two (or more) sets where the two sets do not unconditionally-dominate each other.  
Thus, there must exist at least two monomial classes, ${\cal C}_1^{(Q')}$ and ${\cal C}_2^{(Q')}$, such that: 
\begin{itemize} 
	\item[(i)]  the noncopy-signatures of ${\cal C}_1^{(Q')}$ and ${\cal C}_2^{(Q')}$ are the same, and this (shared) noncopy-signature includes exactly once each variable in $M_{noncopy}$ (and thus does not include any other terms of $Q$); 
	\item[(ii)] the copy-signature  of ${\cal C}_1^{(Q')}$ is distinct from the copy-signature of ${\cal C}_2^{(Q')}$, and neither copy-signature unconditionally-dominates the other. 
\end{itemize} 

From (ii) above, there must exist two distinct positions $j \neq k$, such that $1 \leq j \leq r$ and $1 \leq k \leq r$, and such that the copy-signature  of ${\cal C}_1^{(Q')}$ has term $T_1^{(j)} \neq 1$ in position $j$ and has term $T_1^{(k)} \neq 1$ in position $k$, and such that the copy-signature  of ${\cal C}_2^{(Q')}$ has term $T_2^{(j)} \neq 1$ in position $j$ and has term $T_2^{(k)} \neq 1$ in position $k$, where there exists a total-order vector ${\cal O}_1$ such that $T_1^{(j)} < T_2^{(j)}$ and $T_1^{(k)} > T_2^{(k)}$ under the {\em strict} total order for ${\cal O}_1$, and, conversely, there exists a total-order vector ${\cal O}_2$ such that $T_1^{(j)} > T_2^{(j)}$ and $T_1^{(k)} < T_2^{(k)}$ under the {\em strict} total order for ${\cal O}_2$. (This condition is immediate from the fact that there is no unconditional-dominance ``in either direction'' between the copy-signatures of ${\cal C}_1^{(Q')}$ and of ${\cal C}_2^{(Q')}$.) It follows immediately that $T_1^{(j)}$ and $T_2^{(j)}$ are two {\em distinct} variables in $\{ N_{m+1},$ $\ldots,$ $N_{m+r} \}$, and, similarly, $T_1^{(k)}$ and $T_2^{(k)}$ are two {\em distinct} variables in $\{ N_{m+1},$ $\ldots,$ $N_{m+r} \}$.

W.l.o.g., choose the above (fixed) $j$, and consider the subgoal of the query $Q'$ that has the copy variable $Y'_{m+j}$; call this subgoal $s'_j$. Clearly, $s'_j$ is a copy-sensitive (as opposed to relational) subgoal of $Q'$. By the above reasoning, the monomial-class mapping for the monomial class ${\cal C}_1^{(Q')}$, call this mapping $\mu'_1$, maps $s'_j$ to a copy-sensitive subgoal of $Q$, say to $s_1$ with copy variable $Y^{(1)}$ $\in$ $\{ Y_{m+1},$ $\ldots,$ $Y_{m+w} \}$. Similarly, the monomial-class mapping for the monomial class ${\cal C}_2^{(Q')}$, call this mapping $\mu'_2$, maps $s'_j$ to a copy-sensitive subgoal of $Q$, say to $s_2$ with copy variable $Y^{(2)}$ $\in$ $\{ Y_{m+1},$ $\ldots,$ $Y_{m+w} \}$, such that $Y^{(1)}$ $\neq$ $Y^{(2)}$. From $Y^{(1)}$ $\neq$ $Y^{(2)}$ we have that $s_1$ and $s_2$ are two distinct elements of the set ${\cal S}_{C(Q)}$; hence $s_1$ and $s_2$ have {\em distinct} relational templates. 

Observe that by the reasoning in the previous paragraph (specifically from $Y^{(1)}$ $\neq$ $Y^{(2)}$), it must be that $w \geq 2$, and hence $r \geq 2$, whenever we assume that $f^{(Q')}_p$ has negative terms. 

From the monomial classes  ${\cal C}_1^{(Q')}$ and  ${\cal C}_2^{(Q')}$ having the same noncopy-signature, 
we conclude that the subgoal $s'_j$ of the query $Q'$ must have at least one set variable. (Indeed, suppose toward contradiction that for all terms of $s'_j$ that are not the copy variable of $s'_j$, each such term is a head variable of $Q'$, a constant used in $Q'$, or a multiset noncopy variable of $Q'$. Then the relational templates of  $\mu'_1(s'_j)$ and of $\mu'_2(s'_j)$ must coincide, hence $Y^{(1)}$ $\neq$ $Y^{(2)}$ cannot hold. We establish this result by recalling that each monomial-class mapping maps each constant of $Q'$ into itself, maps each head variable of $Q'$ into the ``corresponding'' head variable of $Q$, and maps each multiset-noncopy variable of $Q'$ ``according to'' the noncopy-signature of the monomial class for the mapping.) Denote by $X'_*$ this set variable of $Q'$ in $s'_j$; let $n \geq 1$ be the position of $X'_*$ in the relational template of $s'_j$. 

We now obtain that neither subgoal $s_1$ $=$ $\mu'_1(s'_j)$ nor subgoal $s_2$ $=$ $\mu'_2(s'_j)$ of the query $Q$ can have a set variable in the position $n$. The proof is by contradiction: assuming that  (w.l.o.g.) $s_1$ has a  a set variable in the position $n$

} 

Queries $Q$, $Q'$, and $Q''$ of Example~\ref{main-proof-ex} (in Section~\ref{main-proof-ex-sec}) are all explicit-wave CCQ queries, by case (1) of Definition~\ref{expl-wave-def}. In addition, query $Q'$ of Example~\ref{writeup-weird-ex} (in Section~\ref{writeup-weird-ex-sec}) is an explicit-wave CCQ query, by case (2) of Definition~\ref{expl-wave-def}.  Indeed, every containment mapping from query $Q'_{H^d}$ (see Example~\ref{writeup-weird-ex}) to itself is an identity mapping on each of $Y'_1$ and $Y'_2$ and maps each subgoal of the query into a copy of itself. 

On the other hand, query $Q$ of Example~\ref{writeup-weird-ex} is {\em not} an explicit-wave CCQ query. Indeed, let $L^{(Q)}_{nc} = [Y_1 \ Y_2]$. Consider containment mappings $\mu_1^{(Q)}$ and $\mu_2^{(Q)}$ from $Q_{H^d}$ (see Example~\ref{writeup-weird-ex}) to itself: Let $\mu_1^{(Q)}$ be the identity mapping, and let $\mu_2^{(Q)}$ be the result of replacing $X_2 \rightarrow X_2$ in $\mu_1^{(Q)}$ with $X_2 \rightarrow X_3$. Clearly, $\mu_1^{(Q)}$ and $\mu_2^{(Q)}$ agree on the $\pi^{(Q)}_{nc}$ $=$ $L^{(Q)}_{nc}$. It is easy to see that $\mu_1^{(Q)}$ and $\mu_2^{(Q)}$ map the first subgoal of $Q_{H^d}$ into different 
subgoals (the first and the second subgoal, respectively) of $Q_{H^d}$. 

{\em Sufficient condition I for CCQ query $Q$ to be an explicit-wave query:} For the class, call it ${\cal Q}_{r \leq 1}$, of CCQ queries where each query has at most one copy-sensitive subgoal, each query in the class is an explicit-wave query. This sufficient condition is immediate from part (1) of Definition~\ref{expl-wave-def}. By this sufficient condition, all set-semantics CCQ queries and all bag-set-semantics CCQ queries are explicit-wave queries. (For less trivial examples of CCQ queries that are explicit-wave queries by this sufficient condition, see queries $Q$, $Q'$, and $Q''$ in Example~\ref{main-proof-ex} in Section~\ref{main-proof-ex-sec}.) 

{\em Sufficient condition II for CCQ query $Q$ to be an explicit-wave query:} For the following class of CCQ queries, each query in the class is an explicit-wave query.  
	The class ${\cal Q}_{cszs}$ of all CCQ queries whose each copy-sensitive subgoal has zero set variables. (For the proof, see Proposition~\ref{explicit-wave-sufficient-prop}.) This class encompasses: 
	\begin{itemize} 
		\item The class of all queries that do not have set variables. This class encompasses the class of all bag queries, as well as the class of all bag-set queries. 
		\item The class of all queries that have no copy-sensitive subgoals. This class encompasses the class of all set queries, as well as the class of all bag-set queries. 
	\end{itemize} 

From the above sufficient conditions I and II for a CCQ query to be an explicit-wave query, we have that Theorem~\ref{magic-mapping-prop} has the following special case: 

\begin{theorem} 
\label{cszs-magic-thm} 
Given CCQ queries $Q$ and $Q'$, such that (i) $Q$ $\in$ ${\cal Q}_{r \leq 1}$ $\bigcup$ ${\cal Q}_{cszs}$, and (ii) $Q \equiv_C Q'$. Then  there exists a SCVM from $Q'$ to $Q$. 
\end{theorem} 

For query $Q''$ in Example~\ref{main-proof-ex} (Section~\ref{main-proof-ex-sec}), we have that $Q''$ $\in$ ${\cal Q}_{r \leq 1}$ $\bigcap$ ${\cal Q}_{cszs}$ . However, as exemplified by query $Q$ in Example~\ref{main-proof-ex}, ${\cal Q}_{r \leq 1}$ $\subseteq \hspace{-0.4cm} / \hspace{0.15cm}$ ${\cal Q}_{cszs}$ . 
Further, trivially we have that ${\cal Q}_{cszs}$ $\subseteq \hspace{-0.4cm} / \hspace{0.15cm}$ ${\cal Q}_{r \leq 1}$ . 
(For instance, query $Q(X)$ $\leftarrow$ $p(X,Y_1),$ $p(X,Y_2; Y_3),$ $p(X,Y_2; Y_4),$ $\{ Y_1,$ $Y_2,$ $Y_3,$ $Y_4 \}$ is in ${\cal Q}_{cszs}$ $-$ ${\cal Q}_{r \leq 1}$.) 
Finally, the set of all explicit-wave CCQ queries is not within the set ${\cal Q}_{cszs}$ $\bigcup$ ${\cal Q}_{r \leq 1}$ . Indeed, for the explicit-wave CCQ query $Q'$ of Example~\ref{writeup-weird-ex}, $Q'$ $\notin$  ${\cal Q}_{cszs}$ $\bigcup$ ${\cal Q}_{r \leq 1}$ . 

The proof of Theorem~\ref{cszs-magic-thm} is immediate from Theorem~\ref{magic-mapping-prop} and from Proposition~\ref{explicit-wave-sufficient-prop}. 

We now further special-case Theorem~\ref{cszs-magic-thm} using the three traditional semantics for conjunctive queries. 

\begin{theorem} 
\label{bag-magic-thm} 
Given bag CQ queries $Q$ and $Q'$, such that $Q \equiv_B Q'$. Then  there exists an isomorphic mapping from $Q'$ to $Q$. 
\end{theorem} 

Theorem~\ref{bag-magic-thm} is the ``if'' part of Theorem \reminder{??? -- put in here the theorem ID} (from  \cite{ChaudhuriV93}). Its proof is straightforward from the facts that (i) each bag CQ query is a CCQ query that belongs to the class ${\cal Q}_{cszs}$,  hence by Theorem~\ref{magic-mapping-prop} (using the fact that for bag CQ queries $Q$ and $Q'$, $Q \equiv_B Q'$ implies $Q \equiv_C Q'$) there must exist a SCVM from $Q'$ to $Q$, and that (ii) for every SCVM from a bag CQ query to another bag CQ query, the mapping is an isomorphism. (The part (ii) is immediate from three straightforward observations, as follows. (1) Two bag-semantics-equivalent\footnote{Recall that for bag CQ queries $Q$ and $Q'$, $Q \equiv_B Q'$ implies $Q \equiv_C Q'$.} bag CQ queries must have the same number of copy variables (see Theorem~\ref{not-same-num-multiset-vars-thm}); (2) A bag CQ query has zero relational subgoals (and hence, in particular, the regularized version of a bag CQ query is unique and is ``the query itself''); and (3) A SCVM from CCQ query $Q'$ (from CCQ query $Q$, respectively) to CCQ query $Q$ (to CCQ query $Q'$, respectively) induces a surjection from the set of copy-sensitive subgoals of $Q'$ (of $Q$, respectively) to the set of copy-sensitive subgoals of $Q$ (of $Q'$, respectively). From (3) we infer that $Q$ and $Q'$ have the same number of (copy-sensitive) subgoals. Thus, every SCVM from $Q'$ to $Q$ must, by property (6) in Definition~\ref{magic-mapping-def}, be an isomorphism from the set of subgoals of $Q'$ to the set of subgoals of $Q$. From the existence of at least one SCVM from $Q'$ to $Q$ (by (i) here), Q.E.D.) 

\begin{theorem} 
\label{bag-set-magic-thm} 
Given bag-set CQ queries $Q$ and $Q'$, such that $Q \equiv_{BS} Q'$. Then  there exists an isomorphic mapping from the regularized version of the query $Q'$ to the regularized version of the query $Q$. 
\end{theorem} 

Theorem~\ref{bag-set-magic-thm} is the ``if'' part of Theorem \reminder{??? -- put in here the theorem ID} (from \cite{ChaudhuriV93}). Its proof is straightforward from the facts that: 
\begin{itemize} 
	\item[(i)] Each bag-set CQ query is a CCQ query that belongs to the class ${\cal Q}_{cszs}$, hence by Theorem~\ref{magic-mapping-prop} (using the fact that for bag-set CQ queries $Q$ and $Q'$, $Q \equiv_{BS} Q'$ implies $Q \equiv_C Q'$) there must exist a SCVM from $Q'$ to $Q$, and
	\item[(ii)] For every SCVM from bag-set CQ query $Q'$ to bag-set CQ query $Q$, the mapping maps each subgoal of the regularized version of $Q'$ into a distinct (in terms of relational template) subgoal of the regularized version of $Q$.\footnote{We can use {\em the regularized version} of $Q$ here, because the regularized version of a bag-set CQ query $Q$, by definition, has all subgoals of $Q$ albeit without duplication.} (Indeed, by definition, each SCVM from $Q'$ to $Q$ maps each element of the set ${\cal X}'$ $\bigcup$ $P'$ $\bigcup$ $M'_{noncopy}$ into a {\em distinct} element of the set ${\cal X}$ $\bigcup$ $P$ $\bigcup$ $M_{noncopy}$. Here, $\cal X$ is the set of all variables in the head of query $Q$, $P$ is the set of all constants used in $Q$, and $M_{noncopy}$ is the set of all multiset noncopy variables of $Q$; similarly for $Q'$.) As a result, for every SCVM $\mu$ from $Q'$ to $Q$, $\mu$ maps each subgoal of the regularized  version of $Q'$ to a distinct subgoal of the regularized  version of  $Q$. (Recall that in the regularized version of $Q$, all subgoals have different relational templates; we make the same observation for $Q'$.) 
	
	We make the same observation about all SCVMs {\em from $Q$ to $Q'$,} and thus obtain that the regularized versions of $Q$ and of $Q'$ must have the same number of subgoals. We conclude that  for every SCVM $\mu$ from $Q'$ to $Q$, $\mu$ is an isomorphic mapping from the regularized  version of $Q'$ to the regularized  version of  $Q$. Hence, from (i) Q.E.D. 
\end{itemize} 
\reminder{Need to recheck the above proof of (ii)!} 

\begin{theorem} 
\label{set-magic-thm} 
Given set CQ queries $Q$ and $Q'$, such that $Q \equiv_{S} Q'$. Then  there exists a containment mapping from the query $Q'$ to the query $Q$. 
\end{theorem} 

Theorem~\ref{set-magic-thm} is the ``if'' part of the containment-mapping theorem \reminder{??? -- put in here the theorem ID} of \cite{ChandraM77}. Its proof is immediate from the facts that (i) each set CQ query is a CCQ query that belongs to the class ${\cal Q}_{cszs}$,  hence by Theorem~\ref{magic-mapping-prop} (using the fact that  for set CQ queries $Q$ and $Q'$, $Q \equiv_S Q'$ implies $Q \equiv_C Q'$) there must exist a SCVM from $Q'$ to $Q$, and that (ii) for every SCVM from set CQ query $Q'$ to set CQ query $Q$, the mapping is a containment mapping from $Q'$ to $Q$. (Item (ii) is immediate from Definition~\ref{magic-mapping-def}, because set CQ queries have no multiset variables.)

\begin{proposition} 
\label{explicit-wave-sufficient-prop}
Let $Q$ be a CCQ query, such that each copy-sensitive subgoal of $Q$ has zero set variables. Then $Q$ is an explicit-wave query. 
\end{proposition} 
	
\begin{proof} 
Let $Q$ be a CCQ query whose each copy-sensitive subgoal has zero set variables. In case where $Q$ does not have copy-sensitive subgoals, the query $Q_{H^d}$ does not have multiset-on subgoals by definition. Thus, it is immediate from Definition~\ref{expl-wave-def} that in this case, $Q$ is an explicit-wave query. Hence in the remainder of this proof we assume that $Q$ has at least one copy-sensitive subgoal. 

Let list $L^{(Q)}_{nc}$ be an arbitrary fixed ordering of all the elements of the set $M_{noncopy}$ of multiset noncopy variables of $Q$. Fix an arbitrary permutation $\pi^{(Q)}_{nc}$ of the list $L^{(Q)}_{nc}$. (In case $M_{noncopy} = \emptyset$, we have that the only possibility for $L^{(Q)}_{nc}$ is $L^{(Q)}_{nc}$ $=$ $[]$, that is the empty list, and that the only possible permutation of that empty list is $\pi^{(Q)}_{nc}$ $=$ $[]$. In this case $M_{noncopy} = \emptyset$, let the empty list be (the only possibility for) the  permutation $\pi^{(Q)}_{nc}$ that we fix for the remainder of this proof.) 

In case where for the fixed permutation $\pi^{(Q)}_{nc}$ there exists at most one  containment mapping, $\mu$, from the query $Q_{H^d}$ to itself,  such that $\mu$ agrees with $\pi^{(Q)}_{nc}$, we have that w.r.t. the permutation $\pi^{(Q)}_{nc}$, query $Q$ is an explicit-wave query by definition (by default). Hence for the remainder of this proof, for the chosen permutation $\pi^{(Q)}_{nc}$, assume that there exist at least two distinct containment mappings from $Q_{H^d}$ to itself,  such that all of the mappings agree with the permutation $\pi^{(Q)}_{nc}$. Let $\mu_1$ and $\mu_2$ be any two distinct such containment mappings. 


By our assumption that query $Q$ has copy-sensitive subgoals, the set of multiset-on subgoals of the  query $Q_{H^d}$ is not empty. Let $s$ be an arbitrary multiset-on subgoal of the query $Q_{H^d}$. We show that the atoms $\mu_1(s)$ and $\mu_2(s)$ are identical relational atoms. 

Indeed, consider the set ${\cal S}^*$ of all those terms of the query $Q_{H^d}$ that are not set variables of the query $Q$. By definition of $Q_{H^d}$, the set ${\cal S}^*$ is the union of three sets: (i) the set $\cal X$ of all distinguished variables of the query $Q$; (ii) the set $M_{noncopy}$ of all multiset noncopy variables of the query $Q$; and (iii) the set $P$ of all constants used in the query $Q$. (Recall that none of the copy variables of the query $Q$ is a term of the query $Q_{H^d}$.) 

By definition of containment mapping, each containment mapping from the query $Q_{H^d}$ to itself maps each element of ${\cal X} \bigcup P$ to itself. Further, in case $M_{noncopy} \neq \emptyset$, for each containment mapping $\mu$ from $Q_{H^d}$ to itself such that $\mu$ agrees with the permutation  $\pi^{(Q)}_{nc}$, the restriction of $\mu$ to the list $L^{(Q)}_{nc}$ (which is a fixed ordering of the elements of the set $M_{noncopy}$) is the list $\pi^{(Q)}_{nc}$. 

Recall that all terms of the query $Q_{H^d}$ used in the subgoal $s$ of  $Q_{H^d}$ belong to the set ${\cal S}^*$. (Indeed, we have that for each copy-sensitive subgoal $s'$ of the query $Q$ and for each term $t$ of $Q$ used in $s'$ such that $t$ is not a copy variable, $t$ is not a set variable of $Q$ and hence must belong to the set ${\cal S}^*$.) Therefore, the mappings $\mu_1$ and $\mu_2$ yield the same atom when applied to the multiset-on subgoal $s$ of the query $Q_{H^d}$. Hence $Q$ is an explicit-wave query w.r.t. the permutation $\pi^{(Q)}_{nc}$. 

From the arbitrary choice of (a) the containment mappings $\mu_1$ and $\mu_2$ each agreeing with the fixed permutation $\pi^{(Q)}_{nc}$ of the list $L^{(Q)}_{nc}$, and of (b) the multiset-on subgoal $s$ of the query $Q_{H^d}$, we conclude that query $Q$ is an explicit-wave  query w.r.t. the permutation $\pi^{(Q)}_{nc}$ of the list $L^{(Q)}_{nc}$. From the arbitrary choice of the permutation $\pi^{(Q)}_{nc}$ of the list $L^{(Q)}_{nc}$, we conclude that query $Q$ is an explicit-wave CCQ query by definition. Q.E.D. 
\end{proof}

\reminder{Define the ``noncopy-permutation function'' of $Q''$ (for the $Q''$ and for the family of databases $\{ D_{\bar{N}^{(i)}}(Q) \}$): This is the sum of all the size-of-set-union terms for all the permutations of the (arbitrarily fixed) list of all multiset noncopy variables of the query. Also define how a ``noncopy-permutation function of $Q''$ has a grouping'', define ``the explicit wave of $Q$''.} 

\begin{proposition} 
\label{goldfish-prop} 
Let $Q$ be an explicit-wave CCQ query, and let $Q'$ be a CCQ query such that $Q' \equiv_C Q$. Then $Q'$ has a {\em noncopy-permutation monomial class} whose multiplicity monomial is the {\em explicit-wave monomial} of $Q$. 
\end{proposition} 

\subsubsection{Old Reasoning that Did Not Work}

Consider all the terms of ${\cal F}_{(Q)}^{(Q)}$ that have the product of all of $N_1,\ldots,N_m$ (in case $m \geq 1$); in case $m = 0$ consider {\em all} the monomials of ${\cal F}_{(Q)}^{(Q)}$. Call these monomials collectively ${\cal F}_{(Q)}^{(m)}(Q)$. By our assumption, $Q$ is an explicit-wave CCQ query.  \reminder{In case $m = 0$, the entire  ${\cal F}_{(Q)}(Q)$ must have just the exposed wave and {\em nothing else,} for some $\cal O$ etc. Hence in the remainder of this reasoning, we assume $m \geq 1$.}
Hence there does exist a permutation $\pi_0^{(m)}$ of $N_1,\ldots,N_m$, there does exist a monomial class ${\cal C}_0^{(Q)}$ with noncopy-signature $\pi_0^{(m)}$ and with a characteristic-wave copy-signature, and there does exist  a total-order vector $\cal O$ such that for all the elements $\bar{N}^{(i)}$ of $\cal N$ that satisfy all the inequalities of $Ord({\cal O})$, the multiplicity-monomial term \reminder{???} for $\pi_0^{(m)}$ is exactly the multiplicity monomial for the monomial class ${\cal C}_0^{(Q)}$. 

We show that it follows that for each CCQ query $Q'$ such that $Q \equiv_C Q'$, there exists a monomial class, call it ${\cal C}_0^{(Q')}$, whose noncopy-signature is a permutation of $N_1,\ldots,N_m$ and whose copy-signature is a characteristic wave. In ${\cal F}_{(Q)}^{(m)}(Q)$ {\em for the query $Q$,} group 

\reminder{Idea: Group all the elements of ${\cal F}_{(Q)}^{(m)}(Q)$ by permutation of $N_1,\ldots,N_m$. Within each such permutation, we have the size-of-Cartesian-product formula by the inclusion-exclusion principle. The exposed wave will be by itself. The group ${\cal F}_{(Q)}^{(m)}(Q')$ {\em for the query $Q'$} has exactly the same polynomial, in terms of $N_{m+1},\ldots,N_{m+w}$, as does ${\cal F}_{(Q)}^{(m)}(Q)$; hence we can group that polynomial (for $Q'$, in terms of $N_{m+1},\ldots,N_{m+w}$) exactly the same way as we did for $Q$. Then the exposed wave of $Q$ will be ``by itself'' in the result of that grouping and will be ``the odd one out''. Then, by assuming that $Q'$ does not have a characteristic-wave monomial class, we have to group that exposed wave with one of the grouping, but it will be ``the odd one out'' always. } 

\begin{itemize} 
	\item Take the polynomials ${\cal F}_{(Q)}^{(m)}(Q)$ and ${\cal F}_{(Q)}^{(m)}(Q')$; they must return the same value on all inputs in the domain for the total order $Ord({\cal O})$; $\cal O$ is the total-order vector for which $Q$ has the exposed wave. 
	\item Represent each of   ${\cal F}_{(Q)}^{(m)}(Q)$ and ${\cal F}_{(Q)}^{(m)}(Q')$ as sum of groups, where each group is the multivariate polynomial for the size of a union of {\em pairwise $\cal O$- nondominated sets}; make the {\em maximal number of groups} -- that is, if the size of a union of sets can be represented as a positive sum of two sizes of unions of (contributor) sets, then do represent as that sum. In other words, each group is for the size of union of non-disjoint sets. Clearly, the exposed wave for $Q$ is a separate group whose multivariate polynomial has exactly one monomial -- the characteristic-wave monomial for $Q$. 
	\item Prove that if, for two collections of pairwise $\cal O$-nondominated sets $A_1,\ldots,A_k$ and $B_1,\ldots,B_n$ it holds that the size of the union of the $A_j$'s is the same as the size of the union of the $B_l$'s then the two collections are identical to each other, up to permutation on the Cartesian products. (Proof: By contradiction. If the two collections are different then they have different top-level sum (that is, sum $|A_1| + \ldots + |A_k|$ is different from the sum $|B_1| + \ldots + |B_n|$), then the multivariate polynomials for the two top-level sums (which have the highest degrees in the original size-of-union polynomials) are not identical, then the original multivariate polynomials are not identical. Then contradiction.) 
	\item In the equation ${\cal F}_{(Q)}^{(m)}(Q)$ $=$ ${\cal F}_{(Q)}^{(m)}(Q')$ on the domain of $Ord({\cal O})$, cancel out all same-size unions. Then what remains is the sum on the left side (this sum has the characteristic-wave monomial group for $Q$) such that the sum on the right side does not have any element identical to the characteristic-wave monomial group for $Q$. Call the resulting equation $f_{\bar{N}}^{(m)}(Q)$ $=$ $f_{\bar{N}}^{(m)}(Q')$. 
	\item Prove that when $N_{m+1}$ $=$ $N_{m+2}$ $=$ $\ldots$ $=$ $N_{m+w}$, then for all groups of each of  $f_{\bar{N}}^{(m)}(Q)$ and $f_{\bar{N}}^{(m)}(Q')$, whenever a group's top-level sum has an elementary set with degree without ones ($1$-s), then the value of the size of the group on $N_{m+1}$ $=$ $N_{m+2}$ $=$ $\ldots$ $=$ $N_{m+w}$ $=$ $N^*$ is always $(N^*)^w$. 
	\item From the previous item, the number of such groups in $f_{\bar{N}}^{(m)}(Q)$ is the same as the number of such groups in $f_{\bar{N}}^{(m)}(Q')$. (To prove, just set $N_{m+1}$ $=$ $N_{m+2}$ $=$ $\ldots$ $=$ $N_{m+w}$ $=$ $N^*$) and count the left-hand-side coefficient and the right-hand-side coefficient for the term $(N^*)^w$. 
	\item Prove that for each group, setting each of $N_{m+1}$ $=$ $N_{m+2}$ $=$ $\ldots$ $=$ $N_{m+w}$ to the value {\em unity} results in the same size value ($= 1$) for each group in each of $f_{\bar{N}}^{(m)}(Q)$ and $f_{\bar{N}}^{(m)}(Q')$. Conclude that the total number of all groups in $f_{\bar{N}}^{(m)}(Q)$ equals the total number of all groups in $f_{\bar{N}}^{(m)}(Q')$. Conclude further (from previous item and from this item) that the total number of groups whose top-level sum has an elementary set with degree without ones, is the same in $f_{\bar{N}}^{(m)}(Q)$ and $f_{\bar{N}}^{(m)}(Q')$. 
	
	\item In $f_{\bar{N}}^{(m)}(Q)$ and $f_{\bar{N}}^{(m)}(Q')$, the {\em all-positive-coefficient} sum of all the top-degree (i.e., those not having {\em unity,} i.e., $= 1$, literals in the product) monomials is the same on both sides (otherwise $f_{\bar{N}}^{(m)}(Q)$ and $f_{\bar{N}}^{(m)}(Q')$ as totals would not be the same multivariate polynomials). Then, from the groups in $f_{\bar{N}}^{(m)}(Q)$ and $f_{\bar{N}}^{(m)}(Q')$ being all not-pairwise-the-same, we conclude that it is the {\em groupings} of these top-degree monomials that is the difference between $f_{\bar{N}}^{(m)}(Q)$ and $f_{\bar{N}}^{(m)}(Q')$. 
	
	\item By the number of ``plus'' and ``minus'' signs in the two polynomials $f_{\bar{N}}^{(m)}(Q)$ and $f_{\bar{N}}^{(m)}(Q')$ (after we ``open the parentheses'' of all the groups), we conclude that the number of sets being unioned within each group is pairwise the same between $f_{\bar{N}}^{(m)}(Q)$ and $f_{\bar{N}}^{(m)}(Q')$. (That is, there is a bijection from the set of groups in $f_{\bar{N}}^{(m)}(Q)$ to the set of groups in $f_{\bar{N}}^{(m)}(Q')$, such that the number of sets being unioned between each group and its image under the bijection is the same.) 
	
	\item By contrapositive, we assume that $f_{\bar{N}}^{(m)}(Q')$ does not have the characteristic-wave monomial of $Q$ (note that it has only $N_{m+1}$ through $N_{m+w}$ in the product, as the variables $N_1,\ldots,N_m$ are {\em not} used in either $f_{\bar{N}}^{(m)}(Q)$ or $f_{\bar{N}}^{(m)}(Q')$), call this monomial ${\cal P}^*$, either as a standalone grouping or as a top-level set in any of the size-of-nontrivial-union grouping. Then it must be that there are negative-coefficient terms in both of $f_{\bar{N}}^{(m)}(Q)$ and $f_{\bar{N}}^{(m)}(Q')$. 
	 In addition, by $f_{\bar{N}}^{(m)}(Q')$ being identical to $f_{\bar{N}}^{(m)}(Q)$ after all the groupings have been removed (i.e., after all the parentheses for the groupings have been removed in each polynomial), there are two cases: 

	\begin{itemize} 
		\item Suppose that 	the characteristic-wave monomial ${\cal P}^*$ of $Q$ is {\em absent} in both  $f_{\bar{N}}^{(m)}(Q)$ and  $f_{\bar{N}}^{(m)}(Q')$  after all the groupings have been removed. Then it must be that one of the size-of-union-of-sets groupings in $f_{\bar{N}}^{(m)}(Q')$, call this grouping $G^*$,  has ${\cal P}^*$ with the negative coefficient. This means that ${\cal P}^*$ is present in that grouping as a {\em below-the-top-level term.} From ${\cal P}^*$ being absent from  $f_{\bar{N}}^{(m)}(Q')$  after all the groupings have been removed, we conclude that ${\cal P}^*$ must also be present in one of the size-of-union-of-sets groupings in $f_{\bar{N}}^{(m)}(Q')$, where that grouping $G'$ is different from $G^*$ (!), as a {\em below-the-top-level term}  (from out assumption that ${\cal P}^*$ is not a top-level term in $f_{\bar{N}}^{(m)}(Q')$) with a positive coefficient. 
 	
		\item Suppose that the characteristic-wave monomial ${\cal P}^*$ of $Q$ is {\em present} in both  $f_{\bar{N}}^{(m)}(Q)$ and  $f_{\bar{N}}^{(m)}(Q')$  after all the groupings have been removed. Then it must be that one of the size-of-union-of-sets groupings in $f_{\bar{N}}^{(m)}(Q')$, call this grouping $G^*$, has ${\cal P}^*$ with the positive coefficient. By our assumption that ${\cal P}^*$ is not present in $f_{\bar{N}}^{(m)}(Q')$ either as a standalone grouping or as a top-level set in any of the size-of-nontrivial-union grouping,  this means that ${\cal P}^*$ is present in $G^*$ as a {\em below-the-top-level term.} 

	\end{itemize} 
	
	In both cases above, we discover ${\cal P}^*$ in $f_{\bar{N}}^{(m)}(Q')$, in some group, ``under'' some negative-coefficient level. In $f_{\bar{N}}^{(m)}(Q)$, this same (!!!) negative-coefficient level has under it a (positive-coefficient) term, call it ${\cal P}'$, such that ${\cal P}'$ is different from ${\cal P}^*$ (because ${\cal P}^*$ in $f_{\bar{N}}^{(m)}(Q)$ is not under any negative-coefficient level). 
	
	We now look at an arbitrary term, $T$, being directly under a negative-coefficient level, $L$ in some group. That term's (i.e., $T$'s) variables are directly dependent on the variables in $L$, because each $i$th (i.e., in position $i$) variable in $T$ is a $min$ (under the total order $Ord({\cal O})$) of all the (appropriate) $i$th variables in $L$. Hence we cannot replace the ${\cal P}'$ in $G^*$ by ${\cal P}^*$ because from ${\cal P}'$ and ${\cal P}^*$ not being identical to each other it follows that ${\cal P}'$ and ${\cal P}^*$ must be different in at least one position, under each ``vectorization/linearlization'' of each of ${\cal P}'$ and ${\cal P}^*$. Thus we have the desired contradiction, Q.E.D.


\end{itemize}

\subsection{Extended Example 2: Implicit Waves} 
\label{writeup-weird-ex-sec} 

Interestingly, in general, the function ${\cal F}_{(Q)}^{(Q'')}$ is {\em not} a multivariate polynomial (in terms of the variables in the vector $\bar{N}$) on the {\em entire} domain of the vector $\bar{N}$. For instance, in Example~\ref{writeup-weird-ex} in Section~\ref{writeup-weird-ex-sec}, for the query $Q$ of the example, we use vector $\bar{N}$ $=$ $[ N_1$ $\ N_2$ $\ N_3$ $\ N_4 ]$ to construct the family of databases $\{ D_{\bar{N}^{(i)}}(Q) \}$. We go on in that example to show that for the query $Q$ and for the family of databases $\{ D_{\bar{N}^{(i)}}(Q) \}$, (i) ${\cal F}_{(Q)}^{(Q)}$ $=$ $N_1 \times N_2 \times N_4^2$ for all vectors ${\bar{N}}^{(i)}$ where $N_3^{(i)} \leq N_4^{(i)}$, and (ii) ${\cal F}_{(Q)}^{(Q)}$ $=$ $N_1 \times N_2 \times N_3^2$  for all vectors ${\bar{N}}^{(i)}$ where $N_3^{(i)} \geq N_4^{(i)}$. That is, on the entire domain $\cal N$ of the vector $\bar N$, the function ${\cal F}_{(Q)}^{(Q)}$ in the example can be expressed as ${\cal F}_{(Q)}^{(Q)}$ $=$ $N_1 \times N_2 \times (max(N_3,N_4))^2$. (The text of Example~\ref{writeup-weird-ex} provides the details.) At the same time, as can be seen in the function ${\cal F}_{(Q)}^{(Q)}$ in Example~\ref{writeup-weird-ex}, on certain subsets of the domain $\cal N$ of the vector $\bar{N}$, the function ${\cal F}_{(Q)}$ ``behaves as'' a multivariate polynomial in terms of the variables in the vector $\bar{N}$.

\reminder{In this beginning of this subsection, need to summarize briefly the point(s) of the example(s) of this subsection. Would those points be that: 
\begin{itemize} 
	\item For the queries $Q$ and $Q'$ of Example~\ref{writeup-weird-ex}, there exists a same-scale containment mapping from $Q$ to $Q'$. The reason is, $Q'$ is an explicit-wave query, hence Theorem~\ref{magic-mapping-prop} holds.   
	\item At the same time, there does not exist a same-scale containment mapping from $Q'$ to $Q$, for the queries $Q$ and $Q'$ of Example~\ref{writeup-weird-ex}. 
	The reason is, the query $Q$ of Example~\ref{writeup-weird-ex} is not an explicit-wave query; thus, Theorem~\ref{magic-mapping-prop} does not apply. 
\end{itemize} 
} 

\begin{example}
\label{pre-writeup-weird-ex}
We continue working with the queries $Q$ and $Q'$ of Example~\ref{writeup-weird-ex}. 
For the reader's convenience, we repeat the queries here.  
\begin{tabbing} 
Hehetab b \= hehe \kill
$Q(X_1) \leftarrow r(X_1,Y_1,Y_2,X_2; Y_3), r(X_1,Y_1,Y_2,X_3; Y_4),$ \\
\> $\{ Y_1,Y_2,Y_3,Y_4 \} .$ \\
$Q'(X'_1) \leftarrow r(X'_1,Y'_1,Y'_2,X'_2; Y'_3), r(X'_1,Y'_1,Y'_2,X'_2; Y'_4),$ \\
\> $\{ Y'_1,Y'_2,Y'_3,Y'_4 \} .$ 
\end{tabbing} 

The only difference between the queries $Q$ and $Q'$, besides the straightforward variable renaming, is that the two subgoals of the query $Q$ have different set variables, $X_2$ and $X_3$, whereas the two subgoals of the query $Q'$ have the same set variable $X'_2$. 

We show in Proposition~\ref{weird-equiv-prop} that $Q \equiv_C Q'$. 
\end{example} 


\begin{example}
\label{second-writeup-weird-ex} 
We continue working with the queries $Q$ and $Q'$ of Example~\ref{writeup-weird-ex}. 
For the reader's convenience, we repeat the queries here.  
\begin{tabbing} 
Hehetab b \= hehe \kill
$Q(X_1) \leftarrow r(X_1,Y_1,Y_2,X_2; Y_3), r(X_1,Y_1,Y_2,X_3; Y_4),$ \\
\> $\{ Y_1,Y_2,Y_3,Y_4 \} .$ \\
$Q_{H^d}(X_1) \leftarrow r(X_1,Y_1,Y_2,X_2), r(X_1,Y_1,Y_2,X_3), \emptyset .$ \\
$Q'(X'_1) \leftarrow r(X'_1,Y'_1,Y'_2,X'_2; Y'_3), r(X'_1,Y'_1,Y'_2,X'_2; Y'_4),$ \\
\> $\{ Y'_1,Y'_2,Y'_3,Y'_4 \} .$ \\
$Q'_{H^d}(X'_1) \leftarrow r(X'_1,Y'_1,Y'_2,X'_2), r(X'_1,Y'_1,Y'_2,X'_2), \emptyset .$ 
\end{tabbing} 
(The query $Q_{H^d}$ for the query $Q$, as well as the query $Q'_{H^d}$ for the query $Q'$, is constructed as defined in Section~\ref{minimiz-query-def}.)

The summary of this example is as follows:
\begin{itemize} 
	\item We know that $Q \equiv_C Q'$ (see Proposition~\ref{weird-equiv-prop}).  
	\item We establish here that $Q'$ is an explicit-wave CCQ query, and $Q$ is {\em not} an explicit-wave CCQ query. (To establish these facts, we will need the above queries $Q_{H^d}$ and $Q'_{H^d}$.) 
	\item By Theorem~\ref{magic-mapping-prop}, we conclude that there exists a SCVM from the query $Q$ to the query $Q'$.  (One such mapping is $\{$ $X_1 \rightarrow X'_1$, $Y_1 \rightarrow Y'_1$, $Y_2 \rightarrow Y'_2$, $X_2 \rightarrow X'_2$, $Y_3 \rightarrow Y'_3$, $X_3 \rightarrow X'_2$, $Y_4 \rightarrow Y'_4$ $\}$.) 
	\item At the same time, because $Q$ is not an explicit-wave CCQ query, Theorem~\ref{magic-mapping-prop} does not apply w.r.t. the existence of a SCVM from the query $Q'$ to the query $Q$.  Indeed, it is easy to see that a SCVM from $Q'$ to $Q$ cannot exist: The two subgoals of $Q'$ must be mapped by any such mapping into distinct subgoals of the query $Q$. This condition is impossible to satisfy because, unlike the subgoals of $Q$, the subgoals of $Q'$ have the same relational pattern. 
\end{itemize}

We now begin the detailed exposition, by establishing that $Q'$ is an explicit-wave CCQ query, and $Q$ is {\em not} an explicit-wave CCQ query. 
Recall that the only difference between the queries $Q$ and $Q'$, besides the straightforward variable renaming, is that the two subgoals of the query $Q$ have different set variables, $X_2$ and $X_3$, whereas the two subgoals of the query $Q'$ have the same set variable $X'_2$. 

By Definition~\ref{expl-wave-def}, query $Q'$ is an explicit-wave CCQ query. Indeed, every containment mapping from query $Q'_{H^d}$ (see above) to itself is an identity mapping on each of $Y'_1$ and $Y'_2$ and maps each subgoal of the query into a copy of itself. 

At the same time, 
query $Q$ is {\em not} an explicit-wave CCQ query. Indeed, let $L^{(Q)}_{nc} = [Y_1 \ Y_2]$. (For the notation, see discussion of Definition~\ref{expl-wave-def}.) Consider containment mappings $\mu_1^{(Q)}$ and $\mu_2^{(Q)}$ from $Q_{H^d}$ to itself: Let $\mu_1^{(Q)}$ be the identity mapping on all the terms of $Q_{H^d}$, and let $\mu_2^{(Q)}$ be the result of replacing $X_2 \rightarrow X_2$ in $\mu_1^{(Q)}$ with $X_2 \rightarrow X_3$. Clearly, $\mu_1^{(Q)}$ and $\mu_2^{(Q)}$ agree on the $\pi^{(Q)}_{nc}$ $=$ $L^{(Q)}_{nc}$. It is easy to see that $\mu_1^{(Q)}$ and $\mu_2^{(Q)}$ map the first subgoal of $Q_{H^d}$ into different 
subgoals (the first and the second subgoal, respectively) of $Q_{H^d}$.

\paragraph{Functions ${\cal F}_{(Q)}^{(Q)}$ and ${\cal F}_{(Q)}^{(Q')}$ for the Databases\\ $\{ D_{\bar{N}^{(i)}}(Q) \}$} 
We continue here the discussion that we began in Example~\ref{writeup-weird-ex} in Section~\ref{beyond-easy-case-sec}. 

Observe that for each of  ${\cal C}_2^{(Q)}$ and ${\cal C}_3^{(Q)}$, the multiplicity monomial of that monomial class is the wave $\Pi_{j=1}^4 N_j$ of the query $Q$. (See Definition~\ref{the-wave-def}.)

At the same time, on the  above fixed database $D_{\bar{N}^{(i)}}(Q)$, the monomial class ${\cal C}_4^{(Q)}$ generates {\em all} the tuples in the set $\Gamma^{t^*_Q}_{\bar{S}}(Q,D_{\bar{N}^{(i)}}(Q))$. The reason is, the values $N_3^{(i)} = 3$ and $N_4^{(i)} = 5$ satisfy the total order whose total-order vector \reminder{Is this term correct?} is $[1 \ N_3 \ N_4]$. Hence, by the results of \reminder{which sections?} , the function ${\cal F}_{(Q)}^{(Q)}$ for all the value vectors (in $\cal N$) for this total-order vector (and hence for our fixed value vector $\bar{N}^{(i)}$) is the multivariate polynomial ${\cal F}_{(Q)}^{(Q)}[1 \ N_3 \ N_4]$ $=$ $N_1 \times N_2 \times N_4^2$. The multivariate polynomial ${\cal F}_{(Q)}^{(Q)}[1 \ N_3 \ N_4]$ is exactly the multiplicity monomial for the monomial class ${\cal C}_4^{(Q)}$, because the monomial class ${\cal C}_4^{(Q)}$ order-dominates all the elements of the set $\{ {\cal C}_1^{(Q)},$ ${\cal C}_2^{(Q)},$ ${\cal C}_3^{(Q)}$, ${\cal C}_4^{(Q)} \}$ w.r.t. the total order associated with the total-order vector $[1 \ N_3 \ N_4]$. 

 For all the value vectors (in $\cal N$) for the only other possible total-order vector $[1 \ N_4 \ N_3]$, the function ${\cal F}_{(Q)}^{(Q)}$ is the multivariate polynomial ${\cal F}_{(Q)}^{(Q)}[1 \ N_4 \ N_3]$ $=$ $N_1 \times N_2 \times N_3^2$, that is, ${\cal F}_{(Q)}^{(Q)}[1 \ N_4 \ N_3]$ is exactly the multiplicity monomial for the monomial class ${\cal C}_1^{(Q)}$. Observe that the characteristic-wave monomial $\Pi_{j=1}^4 N_j$ cannot order-dominate \reminder{Is ``dominate'' a correct term?} \reminder{Do I need domination by the $<$ relation, as opposed to the $\leq$ relation?} either multivariate polynomial unless $N_3 = N_4$. 

Function ${\cal F}_{(Q)}^{(Q')}$ for the query $Q'$ (rather than $Q$), on the same family of databases $\{ D_{\bar{N}^{(i)}}(Q) \}$, is identical to the function ${\cal F}_{(Q)}^{(Q)}$. That is,  ${\cal F}_{(Q)}^{(Q')}[1 \ N_3 \ N_4]$ $=$ $N_1 \times N_2 \times N_4^2$, and ${\cal F}_{(Q)}^{(Q')}[1 \ N_4 \ N_3]$ $=$ $N_1 \times N_2 \times N_3^2$. 

About the above functions  ${\cal F}_{(Q)}^{(Q)}$ and  ${\cal F}_{(Q)}^{(Q')}$, each built on the  family of databases $\{ D_{\bar{N}^{(i)}}(Q) \}$, we summarize that neither function is a multivariate polynomial (in terms of the variables in the vector $\bar{N}$) on the {\em entire} domain $\cal N$ of the vector $\bar{N}$. Indeed, on the entire domain $\cal N$ of the vector $\bar{N}$ we have that ${\cal F}_{(Q)}^{(Q)}$ $=$ ${\cal F}_{(Q)}^{(Q')}$ $=$ $N_1$ $\times$ $N_2$ $\times$ $max(N_3,N_4)^2$.

\paragraph{Functions ${\cal F}_{(Q)}^{(Q)}$ and ${\cal F}_{(Q)}^{(Q')}$ for the Databases\\ $\{ D_{\bar{N}^{(i)}}(Q') \}$} 
In this section of the example, we show how to construct the family of databases $\{ D_{\bar{N}^{(i)}}(Q') \}$ for the query $Q'$ (rather than for $Q$ as above), and will develop functions ${\cal F}_{(Q)}^{(Q)}$ and ${\cal F}_{(Q)}^{(Q')}$ w.r.t. the databases in the family. 


We begin by following Section~\ref{db-constr-sec} of the proof of Theorem~\ref{magic-mapping-prop}. Fix an $i \in {\mathbb N}_+$. Let the vector ${\bar{N}}^{(i)}$, for this fixed $i$, of values of the variables in the vector $\bar{N} = [N_1 \ N_2 \ N_3]$, be $\bar{N}^{(i)} = [1 \ 2 \ 3]$.  Here, each $N_j$ in $\bar{N}$ is generated for the variable $Y'_j$ of $Q'$, for $j \in \{ 1,2,3 \}$. 
The important difference between this construction and the  construction (discussed earlier in this example) of the family of databases $\{ D_{\bar{N}^{(i)}}(Q) \}$ for the query $Q$ is that for the query $Q'$ (but not for $Q$), the set ${\cal S}_{C(Q')}$ is a singleton set, due to the fact that the two subgoals of the query $Q'$ have the same relational template. We assign arbitrarily ${\cal S}_{C(Q')}$ $=$ $\{ r(X'_1,Y'_1,Y'_2,X'_2; Y'_3) \}$. (The only other choice would be to assign ${\cal S}_{C(Q')}$ $=$ $\{ r(X'_1,Y'_1,Y'_2,X'_2; Y'_4) \}$.)

As discussed above, assume $\bar{N}^{(i)} = [1 \ 2 \ 3]$ for the fixed $i \in {\mathbb N}_+$. We use $\nu_0(X'_1) = a$ (hence $t^*_{Q'} = (a)$) and $\nu_0(X'_2) = b$. Let $S^{(i)}_1 = \{ e \}$, and let $S^{(i)}_2 = \{ f,g \}$. These setting generate, for the fixed $i$, the database $D_{\bar{N}^{(i)}}(Q') = \{ r(a,e,f,b; 3),$ $r(a,e,g,b; 3)  \}$. We will refer to the ground atoms in the set $D_{\bar{N}^{(i)}}(Q')$, from left to right, as $d_1$ and $d_2$. Denote by $g_1$ the first subgoal of the query $Q$, and by $g_2$ its second subgoal. By construction of $D_{\bar{N}^{(i)}}(Q')$, we have that $\psi^{gen(Q')}_{\bar{N}^{(i)}}[d_1]$ $=$ $\psi^{gen(Q')}_{\bar{N}^{(i)}}[d_2]$ $=$ $g_1$. 

We now follow Sections~\ref{valid-map-sec} through \ref{easy-funct-case-sec} of the proof of Theorem~\ref{magic-mapping-prop}, to construct the monomial classes for the function ${\cal F}_{(Q)}^{(Q')}$, for the query $Q'$ and for the database  $D_{\bar{N}^{(i)}}(Q')$ as constructed above in this example. 
As a result of the steps, we obtain exactly one monomial class for the query $Q'$: Monomial class ${\cal C}_1^{(Q')}$ has noncopy-signature $[Y'_1 \ Y'_2]$ and copy-signature $[N_3 \ N_3]$; it contributes  to the set $\Gamma^{t^*_{Q'}}_{\bar{S}}(Q',D_{\bar{N}^{(i)}}(Q'))$, with columns (from left to right) $X'_1$ $Y'_1$ $Y'_2$ $Y'_3$ $Y'_4$, nine tuples $(a,e,f,1,1)$ through $(a,e,f,3,3)$ (that is, tuples $(a,e,f,1,1)$, $(a,e,f,1,2)$, $(a,e,f,1,3)$, $\ldots,$ $(a,e,f,3,2)$, $(a,e,f,3,3)$), as well as nine tuples $(a,e,g,1,1)$ through $(a,e,g,3,3)$. 

Observe that \reminder{By which definition?} ${\cal C}_1^{(Q')}$ is a characteristic-wave monomial class \reminder{Is this term correct?} for the query $Q'$ and for the family of databases $\{ D_{\bar{N}^{(i)}}(Q') \}$. The characteristic-wave monomial \reminder{Is this term correct?} for ${\cal C}_1^{(Q')}$ is $N_1 \times N_2 \times (N_3)^2$. Hence, by the results of \reminder{which sections?} , the function ${\cal F}_{(Q)}^{(Q')}$ for all the value vectors in $\cal N$ is the multivariate polynomial ${\cal F}_{(Q)}^{(Q')}$ $=$ $N_1 \times N_2 \times (N_3)^2$. 

Function ${\cal F}_{(Q)}^{(Q)}$ for the query $Q$ (rather than $Q'$), on the same family of databases $\{ D_{\bar{N}^{(i)}}(Q') \}$, is identical to the function ${\cal F}_{(Q)}^{(Q')}$. That is,  ${\cal F}_{(Q)}^{(Q)}$ $=$ $N_1 \times N_2 \times (N_3)^2$. 

About the above functions  ${\cal F}_{(Q)}^{(Q)}$ and  ${\cal F}_{(Q)}^{(Q')}$, each built on the  family of databases $\{ D_{\bar{N}^{(i)}}(Q') \}$, we summarize that each function is a multivariate polynomial (in terms of the variables in the vector $\bar{N}$) on the {\em entire} domain $\cal N$ of the vector $\bar{N}$. 
\end{example}

\subsection{Old: Showing that the Exposed Wave of $Q$ Induces an Explicit Wave of $Q'$ in Case $Q \equiv_C Q'$} 

\subsubsection{True Proof?} 

\reminder{This subsubsection is labeled as ``old'' because it seems to have problems with Definition~\ref{old-expl-wave-def}: See ``reminder'' inside that definition: What is in the reminder seems to be wrong. See corrections of this entire subsection in Section~\ref{goldfish-sec}.}

Suppose that for a CCQ query $Q$ and for the set $M_{noncopy}$ of multiset noncopy variables of $Q$, $|M_{noncopy}| = m \geq 1$. Then let  list $L^{(Q)}_{nc}$ be an arbitrary fixed ordering of the set of multiset noncopy variables of $Q$, and let  $\pi^{(Q)}_{nc}$ be an arbitrary permutation of the list $L^{(Q)}_{nc}$. For a containment mapping $\mu$ from the query $Q_{H^d}$ (see Section~\ref{minimiz-sec}) to itself, we say that $\mu$ {\em agrees with the permutation} $\pi^{(Q)}_{nc}$ if for each $i$th position in the list $L^{(Q)}_{nc}$, $1 \leq i \leq m$, $\mu$ maps the variable in that position into the variable in the $i$th position of $\pi^{(Q)}_{nc}$. (As an easy example, for the case where the permutation $\pi^{(Q)}_{nc}$ is the identity permutation of the list $L^{(Q)}_{nc}$, for  a containment mapping $\mu$ from the query $Q_{H^d}$ to itself, we say that $\mu$ agrees with the permutation $\pi^{(Q)}_{nc}$ $=$ $L^{(Q)}_{nc}$ whenever $\mu$ is an identity mapping w.r.t. each multiset noncopy variable of the query $Q$.) 

In case where CCQ query $Q$ does not have multiset noncopy variables, that is where $|M_{noncopy}| = m = 0$, the only possible list $L^{(Q)}_{nc}$ of multiset noncopy variables of $Q$ is an empty list by definition. In this case, the set of all permutations of the empty list $L^{(Q)}_{nc}$ is a singleton set comprising one empty list. Whenever the query $Q$ is such that $|M_{noncopy}| = m = 0$, we say that for each containment mapping $\mu$ from the query $Q_{H^d}$ to itself,  $\mu$ {\em agrees with the (only) permutation} $\pi^{(Q)}_{nc} = []$ of the list $L^{(Q)}_{nc} = []$. 

\begin{definition}{Explicit-wave CCQ query} 
\label{old-expl-wave-def}
A CCQ query $Q$ is an {\em explicit-wave (CCQ) query} if, for each permutation $\pi^{(Q)}_{nc}$ of the list $L^{(Q)}_{nc}$ of multiset noncopy variables of $Q$ and for each pair of containment mappings $(\mu_1,\mu_2)$ from $Q_{H^d}$ to itself such that each of $\mu_1$ and $\mu_2$ agrees with the permutation $\pi^{(Q)}_{nc}$, it holds that for each multiset-on subgoal, $s$, of the query $Q_{H^d}$, $\mu_1(s)$ and $\mu_2(s)$ are identical atoms, {\em whenever} each of $\mu_1(s)$ and $\mu_2(s)$ is a multiset-on subgoal of the query $Q_{H^d}$. \reminder{Recheck whether the ``whenever'' part of this definition is necessary. That is, when at least one of $\mu_1(s)$ and $\mu_2(s)$ is {\em not} a multiset-on subgoal of $Q_{H^d}$, then there is no danger in $\mu_1$ and $\mu_2$ mapping $s$ into {\em different} subgoals of $Q_{H^d}$, right?} 
\end{definition} 

It is easy to see that a CCQ query $Q$ is or is not an explicit-wave query regardless of the choice of the ordering $L^{(Q)}_{nc}$ of the variables in the set  $M_{noncopy}$. 

\reminder{Do I absolutely need the ``{\em for each} permutation $\pi^{(Q)}_{nc}$ of the list $L^{(Q)}_{nc}$'' in Definition~\ref{old-expl-wave-def}, as opposed to ``{\em there exists a} permutation $\pi^{(Q)}_{nc}$ of the list $L^{(Q)}_{nc}$''?  I guess I do need ``for each permutation'' to prove Theorem~\ref{magic-mapping-prop}, in order to have no negative terms in the {\em entire} function (in terms of $N_{m+1}$ $\ldots$ $N_{m+w}$) for $Q$, call that function $f^{(Q)}_p$, where that function is multiplied by $\Pi_{j=1}^m N_j$, that is, by the term representing all permutations of the list of all (if any) multiset noncopy variables of $Q$. If the function $f^{(Q)}_p$ has no negative terms, then it must be {\em the same} multivariate polynomial (in terms of $N_{m+1}$ $\ldots$ $N_{m+w}$) on {\em all} the total-order subdomains of the domain $\cal N$ of the vector $\bar N$. Then the corresponding function $f^{(Q')}_p$ {\em for the query} $Q'$, $Q' \equiv_C Q$, would also be forced to be the same  multivariate polynomial (in terms of $N_{m+1}$ $\ldots$ $N_{m+w}$) on {\em all} the total-order subdomains of the domain $\cal N$ of the vector $\bar N$, and hence would also be forced to have no negative terms. {\em Whenever $f^{(Q')}_p$ has no negative terms, it must have -- as a ``real''  term (as opposed to ``imaginary'' term, which would be coming from the union-of-sets formula by the inclusion-exclusion principle) -- the characteristic wave of $Q$ (by $f^{(Q')}_p$ being the same multivariate polynomial as $f^{(Q)}_p$ on the entire domain $\cal N$ of $\bar N$, and by $f^{(Q)}_p$ having the characteristic wave of $Q$ [because $Q$ is an explicit-wave CCQ query]). Hence we can construct a SCVM from $Q'$ to $Q$, and hence such a mapping exists.} 

Proving that function $f^{(Q')}_p$ does not have negative terms: 
First we observe that in case $m = |M_{noncopy}| \leq 1$, function $f^{(Q)}_p$ has exactly one term, the characteristic-wave monomial of $Q$, on the entire domain $\cal N$ of $\bar N$. (This is immediate from Definition~\ref{old-expl-wave-def}.) Hence, by $f^{(Q')}_p$ $\equiv$ $f^{(Q)}_p$ on all of $\cal N$, we conclude that $f^{(Q')}_p$ does not have negative terms in case $m = 0$. Hence, in the remainder of this proof we assume $m \geq 2$. 
Under this assumption, the proof is by contradiction: 
Suppose that function $f^{(Q')}_p$ has negative terms. 
This is only possible when (at least) one of the size-of-union terms in $f^{(Q')}_p$  is for a union of two (or more) sets where the two sets do not unconditionally-dominate each other.  
Thus, there must exist at least two monomial classes, ${\cal C}_1^{(Q')}$ and ${\cal C}_2^{(Q')}$, such that: 
\begin{itemize} 
	\item[(i)]  the noncopy-signatures of ${\cal C}_1^{(Q')}$ and ${\cal C}_2^{(Q')}$ are the same, and this (shared) noncopy-signature includes exactly once each variable in $M_{noncopy}$ (and thus does not include any other terms of $Q$); 
	\item[(ii)] the copy-signature  of ${\cal C}_1^{(Q')}$ is distinct from the copy-signature of ${\cal C}_2^{(Q')}$, and neither copy-signature unconditionally-dominates the other. 
\end{itemize} 

From (ii) above, there must exist two distinct positions $j \neq k$, such that $1 \leq j \leq r$ and $1 \leq k \leq r$, and such that the copy-signature  of ${\cal C}_1^{(Q')}$ has term $T_1^{(j)} \neq 1$ in position $j$ and has term $T_1^{(k)} \neq 1$ in position $k$, and such that the copy-signature  of ${\cal C}_2^{(Q')}$ has term $T_2^{(j)} \neq 1$ in position $j$ and has term $T_2^{(k)} \neq 1$ in position $k$, where there exists a total-order vector ${\cal O}_1$ such that $T_1^{(j)} < T_2^{(j)}$ and $T_1^{(k)} > T_2^{(k)}$ under the {\em strict} total order for ${\cal O}_1$, and, conversely, there exists a total-order vector ${\cal O}_2$ such that $T_1^{(j)} > T_2^{(j)}$ and $T_1^{(k)} < T_2^{(k)}$ under the {\em strict} total order for ${\cal O}_2$. (This condition is immediate from the fact that there is no unconditional-dominance ``in either direction'' between the copy-signatures of ${\cal C}_1^{(Q')}$ and of ${\cal C}_2^{(Q')}$.) It follows immediately that $T_1^{(j)}$ and $T_2^{(j)}$ are two {\em distinct} variables in $\{ N_{m+1},$ $\ldots,$ $N_{m+r} \}$, and, similarly, $T_1^{(k)}$ and $T_2^{(k)}$ are two {\em distinct} variables in $\{ N_{m+1},$ $\ldots,$ $N_{m+r} \}$.

W.l.o.g., choose the above (fixed) $j$, and consider the subgoal of the query $Q'$ that has the copy variable $Y'_{m+j}$; call this subgoal $s'_j$. Clearly, $s'_j$ is a copy-sensitive (as opposed to relational) subgoal of $Q'$. By the above reasoning, the monomial-class mapping for the monomial class ${\cal C}_1^{(Q')}$, call this mapping $\mu'_1$, maps $s'_j$ to a copy-sensitive subgoal of $Q$, say to $s_1$ with copy variable $Y^{(1)}$ $\in$ $\{ Y_{m+1},$ $\ldots,$ $Y_{m+w} \}$. Similarly, the monomial-class mapping for the monomial class ${\cal C}_2^{(Q')}$, call this mapping $\mu'_2$, maps $s'_j$ to a copy-sensitive subgoal of $Q$, say to $s_2$ with copy variable $Y^{(2)}$ $\in$ $\{ Y_{m+1},$ $\ldots,$ $Y_{m+w} \}$, such that $Y^{(1)}$ $\neq$ $Y^{(2)}$. From $Y^{(1)}$ $\neq$ $Y^{(2)}$ we have that $s_1$ and $s_2$ are two distinct elements of the set ${\cal S}_{C(Q)}$; hence $s_1$ and $s_2$ have {\em distinct} relational templates. 

Observe that by the reasoning in the previous paragraph (specifically from $Y^{(1)}$ $\neq$ $Y^{(2)}$), it must be that $w \geq 2$, and hence $r \geq 2$, whenever we assume that $f^{(Q')}_p$ has negative terms. 

From the monomial classes  ${\cal C}_1^{(Q')}$ and  ${\cal C}_2^{(Q')}$ having the same noncopy-signature, 
we conclude that the subgoal $s'_j$ of the query $Q'$ must have at least one set variable. (Indeed, suppose toward contradiction that for all terms of $s'_j$ that are not the copy variable of $s'_j$, each such term is a head variable of $Q'$, a constant used in $Q'$, or a multiset noncopy variable of $Q'$. Then the relational templates of  $\mu'_1(s'_j)$ and of $\mu'_2(s'_j)$ must coincide, hence $Y^{(1)}$ $\neq$ $Y^{(2)}$ cannot hold. We establish this result by recalling that each monomial-class mapping maps each constant of $Q'$ into itself, maps each head variable of $Q'$ into the ``corresponding'' head variable of $Q$, and maps each multiset-noncopy variable of $Q'$ ``according to'' the noncopy-signature of the monomial class for the mapping.) Denote by $X'_*$ this set variable of $Q'$ in $s'_j$; let $n \geq 1$ be the position of $X'_*$ in the relational template of $s'_j$. 

We now obtain that neither subgoal $s_1$ $=$ $\mu'_1(s'_j)$ nor subgoal $s_2$ $=$ $\mu'_2(s'_j)$ of the query $Q$ can have a set variable in the position $n$. The proof is by contradiction: assuming that  (w.l.o.g.) $s_1$ has a  a set variable in the position $n$

} 

Query $Q'$ of Example~\ref{writeup-weird-ex} (in Section~\ref{writeup-weird-ex-sec}) is an explicit-wave CCQ query.  Indeed, every containment mapping from query $Q'_{H^d}$ (see Example~\ref{writeup-weird-ex}) to itself is an identity mapping on each of $Y'_1$ and $Y'_2$ and maps each subgoal of the query into a copy of itself. 

On the other hand, query $Q$ of Example~\ref{writeup-weird-ex} is {\em not} an explicit-wave CCQ query. Indeed, let $L^{(Q)}_{nc} = [Y_1 \ Y_2]$. Consider containment mappings $\mu_1^{(Q)}$ and $\mu_2^{(Q)}$ from $Q_{H^d}$ (see Example~\ref{writeup-weird-ex}) to itself: Let $\mu_1^{(Q)}$ be the identity mapping, and let $\mu_2^{(Q)}$ be the result of replacing $X_2 \rightarrow X_2$ in $\mu_1^{(Q)}$ with $X_2 \rightarrow X_3$. Clearly, $\mu_1^{(Q)}$ and $\mu_2^{(Q)}$ agree on the $\pi^{(Q)}_{nc}$ $=$ $L^{(Q)}_{nc}$. It is easy to see that $\mu_1^{(Q)}$ and $\mu_2^{(Q)}$ map the first subgoal of $Q_{H^d}$ into different {\em multiset-on} subgoals (the first and the second subgoal, respectively) of $Q_{H^d}$. 

Sufficient conditions for CCQ query $Q$ to be an explicit-wave query: For the following two classes of queries, each query in each class is an explicit-wave query: 
\begin{itemize} 
	\item The class of all CCQ queries where each query has at most one copy-sensitive subgoal. (The proof is immediate from Definition~\ref{old-expl-wave-def}.)  
	This class encompasses the class of all set queries, as well as the class of all bag-set queries. (In each of of the latter query classes, each query has zero copy-sensitive subgoals.) 
	\item The class of all CCQ queries whose each copy-sensitive subgoal has zero set variables. (For the proof, see Proposition~\ref{old-explicit-wave-sufficient-prop}.) This class, which we denote by  ${\cal Q}_{cszs}$, encompasses: 
	\begin{itemize} 
		\item The class of all queries that do not have set variables. This class encompasses the class of all bag queries, as well as the class of all bag-set queries. 
		\item The class of all queries that have no copy-sensitive subgoals. This class encompasses the class of all set queries, as well as the class of all bag-set queries. 
	\end{itemize} 
\end{itemize} 

From the above sufficient conditions for a CCQ query to be an explicit-wave query, we have that Theorem~\ref{magic-mapping-prop} has the following special case:\footnote{Another special case of Theorem~\ref{magic-mapping-prop} results from replacing item (i) in the statement of Theorem~\ref{old-cszs-magic-thm} with ``(i) $Q$ has at most one copy-sensitive subgoal.''}  

\begin{theorem} 
\label{old-cszs-magic-thm} 
Given CCQ queries $Q$ and $Q'$, such that (i) $Q$ $\in$ ${\cal Q}_{cszs}$, and (ii) $Q \equiv_C Q'$. Then  there exists a SCVM from $Q'$ to $Q$. 
\end{theorem} 

The proof of Theorem~\ref{old-cszs-magic-thm} is immediate from Theorem~\ref{magic-mapping-prop} and from Proposition~\ref{old-explicit-wave-sufficient-prop}. 

We now further special-case Theorem~\ref{old-cszs-magic-thm} using the three traditional semantics for conjunctive queries. 

\begin{theorem} 
\label{old-bag-magic-thm} 
Given bag CQ queries $Q$ and $Q'$, such that $Q \equiv_B Q'$. Then  there exists an isomorphic mapping from $Q'$ to $Q$. 
\end{theorem} 

Theorem~\ref{old-bag-magic-thm} is the ``if'' part of Theorem \reminder{??? -- put in here the theorem ID} of \cite{ChaudhuriV93}. Its proof is immediate from the facts that (i) each bag CQ query is a CCQ query that belongs to the class ${\cal Q}_{cszs}$,  hence by Theorem~\ref{magic-mapping-prop} (using $Q \equiv_B Q'$) there must exist a SCVM from $Q'$ to $Q$, and that (ii) for every SCVM from a bag CQ query to another bag CQ query, the mapping is an isomorphism. (The part (ii) is immediate from three straightforward observations, as follows. (1) Two equivalent bag CQ queries must have the same number of copy variables (see Theorem~\ref{not-same-num-multiset-vars-thm}); (2) A bag CQ query has zero relational subgoals (and hence, in particular, the regularized version of a bag CQ query is unique and is ``the query itself''); and (3) A SCVM from CCQ query $Q'$ (from CCQ query $Q$, respectively) to CCQ query $Q$ (to CCQ query $Q'$, respectively) induces a surjection from the set of copy-sensitive subgoals of $Q'$ (of $Q$, respectively) to the set of copy-sensitive subgoals of $Q$ (of $Q'$, respectively).) 

\begin{theorem} 
\label{old-bag-set-magic-thm} 
Given bag-set CQ queries $Q$ and $Q'$, such that $Q \equiv_{BS} Q'$. Then  there exists an isomorphic mapping from the regularized version of the query $Q'$ to the regularized version of the query $Q$. 
\end{theorem} 

Theorem~\ref{old-bag-set-magic-thm} is the ``if'' part of Theorem \reminder{??? -- put in here the theorem ID} of \cite{ChaudhuriV93}. Its proof is immediate from the facts that: 
\begin{itemize} 
	\item[(i)] Each bag-set CQ query is a CCQ query that belongs to the class ${\cal Q}_{cszs}$, hence by Theorem~\ref{magic-mapping-prop} (using $Q \equiv_{BS} Q'$) there must exist a SCVM from $Q'$ to $Q$, and
	\item[(ii)] For every SCVM from bag-set CQ query $Q'$ to bag-set CQ query $Q$, the mapping maps each subgoal of the regularized version of $Q'$ into a distinct (in terms of relational template) subgoal of the regularized version of $Q$.\footnote{We can use {\em the regularized version} of $Q$ here, because the regularized version of a bag-set CQ query $Q$, by definition, has all subgoals of $Q$ albeit without duplication.} (Indeed, by definition, each SCVM from $Q'$ to $Q$ maps each element of the set ${\cal X}'$ $\bigcup$ $P'$ $\bigcup$ $M'_{noncopy}$ into a {\em distinct} element of the set ${\cal X}$ $\bigcup$ $P$ $\bigcup$ $M_{noncopy}$. Here, $\cal X$ is the set of all variables in the head of query $Q$, $P$ is the set of all constants used in $Q$, and $M_{noncopy}$ is the set of all multiset noncopy variables of $Q$; similarly for $Q'$.) As a result, for each every SCVM $\mu$ from $Q'$ to $Q$, $\mu$ maps each subgoal of the regularized  version of $Q'$ to a distinct subgoal of the regularized  version of  $Q$. 
	
	We make the same observation about all SCVMs from $Q$ to $Q'$, and thus obtain that the regularized versions of $Q$ and of $Q'$ must have the same number of subgoals. We conclude that  for each every SCVM $\mu$ from $Q'$ to $Q$, $\mu$ is an isomorphic mapping from the regularized  version of $Q'$ to the regularized  version of  $Q$. Hence, from (i) Q.E.D. 
\end{itemize} 
\reminder{Need to recheck the above proof of (ii)!} 

\begin{theorem} 
\label{old-set-magic-thm} 
Given set CQ queries $Q$ and $Q'$, such that $Q \equiv_{S} Q'$. Then  there exists a containment mapping from the query $Q'$ to the query $Q$. 
\end{theorem} 

Theorem~\ref{old-set-magic-thm} is the ``if'' part of the containment-mapping theorem \reminder{??? -- put in here the theorem ID} of \cite{ChandraM77}. Its proof is immediate from the facts that (i) each set CQ query is a CCQ query that belongs to the class ${\cal Q}_{cszs}$,  hence by Theorem~\ref{magic-mapping-prop} (using $Q \equiv_S Q'$) there must exist a SCVM from $Q'$ to $Q$, and that (ii) for every SCVM from set CQ query $Q'$ to set CQ query $Q$, the mapping is a containment mapping from $Q'$ to $Q$. (Item (ii) is immediate from Definition~\ref{magic-mapping-def}, because set CQ queries have no multiset variables.)

\begin{proposition} 
\label{old-explicit-wave-sufficient-prop}
Let $Q$ be a CCQ query, such that each copy-sensitive subgoal of $Q$ has zero set variables. Then $Q$ is an explicit-wave query. 
\end{proposition} 
	
\begin{proof} 
Let $Q$ be a CCQ query whose each copy-sensitive subgoal has zero set variables. In case where $Q$ does not have copy-sensitive subgoals, the query $Q_{H^d}$ does not have multiset-on subgoals by definition. Hence, it is immediate from Definition~\ref{old-expl-wave-def} that in this case, $Q$ is an explicit-wave query. Hence in the remainder of this proof we assume that $Q$ has at least one copy-sensitive subgoal. 

Let list $L^{(Q)}_{nc}$ be an arbitrary fixed ordering of all the elements of the set $M_{noncopy}$ of multiset noncopy variables of $Q$. Fix an arbitrary permutation $\pi^{(Q)}_{nc}$ of the list $L^{(Q)}_{nc}$. (In case $M_{noncopy} = \emptyset$, we have that the only possibility for $L^{(Q)}_{nc}$ is $L^{(Q)}_{nc}$ $=$ $[]$, that is the empty list, and that the only possible permutation of that empty list is $\pi^{(Q)}_{nc}$ $=$ $[]$. In this case $M_{noncopy} = \emptyset$, let the empty list be (the only possibility for) the  permutation $\pi^{(Q)}_{nc}$ that we fix for the remainder of this proof.) 

In case where for the fixed permutation $\pi^{(Q)}_{nc}$ there exists at most one  containment mapping, $\mu$, from the query $Q_{H^d}$ to itself,  such that $\mu$ agrees with $\pi^{(Q)}_{nc}$, we have that w.r.t. the permutation $\pi^{(Q)}_{nc}$, query $Q$ is an explicit-wave query by definition (by default). Hence for the remainder of this proof, for the chosen permutation $\pi^{(Q)}_{nc}$, assume that there exist at least two distinct containment mappings from $Q_{H^d}$ to itself,  such that all of the mappings agree with the permutation $\pi^{(Q)}_{nc}$. Let $\mu_1$ and $\mu_2$ be any two distinct such containment mappings. 


By our assumption that query $Q$ has copy-sensitive subgoals, the set of multiset-on subgoals of the  query $Q_{H^d}$ is not empty. Let $s$ be an arbitrary multiset-on subgoal of the query $Q_{H^d}$. We show that the atoms $\mu_1(s)$ and $\mu_2(s)$ are identical relational atoms. (This result holds for every CCQ query $Q$ whose each copy-sensitive subgoal of $Q$ has zero set variables, even when we do not assume that (as we assume in Definition~\ref{old-expl-wave-def}) each of $\mu_1(s)$ and $\mu_2(s)$  is a multiset-on subgoal of the query  $Q_{H^d}$.) 

Indeed, consider the set ${\cal S}^*$ of all those terms of the query $Q_{H^d}$ that are not set variables of the query $Q$. By definition of $Q_{H^d}$, the set ${\cal S}^*$ is the union of three sets: (i) the set $\cal X$ of all distinguished variables of the query $Q$; (ii) the set $M_{noncopy}$ of all multiset noncopy variables of the query $Q$; and (iii) the set $P$ of all constants used in the query $Q$. (Recall that none of the copy variables of the query $Q$ is a term of the query $Q_{H^d}$.) 

By definition of containment mapping, each containment mapping from the query $Q_{H^d}$ to itself maps each element of ${\cal X} \bigcup P$ to itself. Further, in case $M_{noncopy} \neq \emptyset$, for each containment mapping $\mu$ from $Q_{H^d}$ to itself such that $\mu$ agrees with the permutation  $\pi^{(Q)}_{nc}$, the restriction of $\mu$ to the list $L^{(Q)}_{nc}$ (which is a fixed ordering of the elements of the set $M_{noncopy}$) is the list $\pi^{(Q)}_{nc}$. 

Recall that all terms of the query $Q_{H^d}$ used in the subgoal $s$ of  $Q_{H^d}$ belong to the set ${\cal S}^*$. (Indeed, we have that for each copy-sensitive subgoal $s'$ of the query $Q$ and for each term $t$ of $Q$ used in $s'$ such that $t$ is not a copy variable, $t$ is not a set variable of $Q$ and hence must belong to the set ${\cal S}^*$.) Therefore, the mappings $\mu_1$ and $\mu_2$ yield the same atom when applied to the multiset-on subgoal $s$ of the query $Q_{H^d}$. Hence $Q$ is an explicit-wave query w.r.t. the permutation $\pi^{(Q)}_{nc}$. 

From the arbitrary choice of (a) the containment mappings $\mu_1$ and $\mu_2$ each agreeing with the fixed permutation $\pi^{(Q)}_{nc}$ of the list $L^{(Q)}_{nc}$, and of (b) the multiset-on subgoal $s$ of the query $Q_{H^d}$, we conclude that query $Q$ is an explicit-wave  query w.r.t. the permutation $\pi^{(Q)}_{nc}$ of the list $L^{(Q)}_{nc}$. From the arbitrary choice of the permutation $\pi^{(Q)}_{nc}$ of the list $L^{(Q)}_{nc}$, we conclude that query $Q$ is an explicit-wave CCQ query by definition. Q.E.D. 
\end{proof}

\reminder{Define the ``noncopy-permutation function'' of $Q''$ (for the $Q''$ and for the family of databases $\{ D_{\bar{N}^{(i)}}(Q) \}$): This is the sum of all the size-of-set-union terms for all the permutations of the (arbitrarily fixed) list of all multiset noncopy variables of the query. Also define how a ``noncopy-permutation function of $Q''$ has a grouping'', define ``the explicit wave of $Q$''.} 

\begin{proposition} 
\label{old-goldfish-prop} 
Let $Q$ be an explicit-wave CCQ query, and let $Q'$ be a CCQ query such that $Q' \equiv_C Q$. Then $Q'$ has a {\em noncopy-permutation monomial class} whose multiplicity monomial is the {\em explicit-wave monomial} of $Q$. 
\end{proposition} 

\subsubsection{Old Reasoning that Did Not Work}

Consider all the terms of ${\cal F}_{(Q)}(Q)$ that have the product of all of $N_1,\ldots,N_m$ (in case $m \geq 1$); in case $m = 0$ consider {\em all} the monomials of ${\cal F}_{(Q)}(Q)$. Call these monomials collectively ${\cal F}_{(Q)}^{(m)}(Q)$. By our assumption, $Q$ is an explicit-wave CCQ query.  \reminder{In case $m = 0$, the entire  ${\cal F}_{(Q)}(Q)$ must have just the exposed wave and {\em nothing else,} for some $\cal O$ etc. Hence in the remainder of this reasoning, we assume $m \geq 1$.}
Hence there does exist a permutation $\pi_0^{(m)}$ of $N_1,\ldots,N_m$, there does exist a monomial class ${\cal C}_0^{(Q)}$ with noncopy-signature $\pi_0^{(m)}$ and with a characteristic-wave copy-signature, and there does exist  a total-order vector $\cal O$ such that for all the elements $\bar{N}^{(i)}$ of $\cal N$ that satisfy all the inequalities of $Ord({\cal O})$, the multiplicity-monomial term \reminder{???} for $\pi_0^{(m)}$ is exactly the multiplicity monomial for the monomial class ${\cal C}_0^{(Q)}$. 

We show that it follows that for each CCQ query $Q'$ such that $Q \equiv_C Q'$, there exists a monomial class, call it ${\cal C}_0^{(Q')}$, whose noncopy-signature is a permutation of $N_1,\ldots,N_m$ and whose copy-signature is a characteristic wave. In ${\cal F}_{(Q)}^{(m)}(Q)$ {\em for the query $Q$,} group 

\reminder{Idea: Group all the elements of ${\cal F}_{(Q)}^{(m)}(Q)$ by permutation of $N_1,\ldots,N_m$. Within each such permutation, we have the size-of-Cartesian-product formula by the inclusion-exclusion principle. The exposed wave will be by itself. The group ${\cal F}_{(Q)}^{(m)}(Q')$ {\em for the query $Q'$} has exactly the same polynomial, in terms of $N_{m+1},\ldots,N_{m+w}$, as does ${\cal F}_{(Q)}^{(m)}(Q)$; hence we can group that polynomial (for $Q'$, in terms of $N_{m+1},\ldots,N_{m+w}$) exactly the same way as we did for $Q$. Then the exposed wave of $Q$ will be ``by itself'' in the result of that grouping and will be ``the odd one out''. Then, by assuming that $Q'$ does not have a characteristic-wave monomial class, we have to group that exposed wave with one of the grouping, but it will be ``the odd one out'' always. } 

\begin{itemize} 
	\item Take the polynomials ${\cal F}_{(Q)}^{(m)}(Q)$ and ${\cal F}_{(Q)}^{(m)}(Q')$; they must return the same value on all inputs in the domain for the total order $Ord({\cal O})$; $\cal O$ is the total-order vector for which $Q$ has the exposed wave. 
	\item Represent each of   ${\cal F}_{(Q)}^{(m)}(Q)$ and ${\cal F}_{(Q)}^{(m)}(Q')$ as sum of groups, where each group is the multivariate polynomial for the size of a union of {\em pairwise $\cal O$- nondominated sets}; make the {\em maximal number of groups} -- that is, if the size of a union of sets can be represented as a positive sum of two sizes of unions of (contributor) sets, then do represent as that sum. In other words, each group is for the size of union of non-disjoint sets. Clearly, the exposed wave for $Q$ is a separate group whose multivariate polynomial has exactly one monomial -- the characteristic-wave monomial for $Q$. 
	\item Prove that if, for two collections of pairwise $\cal O$-nondominated sets $A_1,\ldots,A_k$ and $B_1,\ldots,B_n$ it holds that the size of the union of the $A_j$'s is the same as the size of the union of the $B_l$'s then the two collections are identical to each other, up to permutation on the Cartesian products. (Proof: By contradiction. If the two collections are different then they have different top-level sum (that is, sum $|A_1| + \ldots + |A_k|$ is different from the sum $|B_1| + \ldots + |B_n|$), then the multivariate polynomials for the two top-level sums (which have the highest degrees in the original size-of-union polynomials) are not identical, then the original multivariate polynomials are not identical. Then contradiction.) 
	\item In the equation ${\cal F}_{(Q)}^{(m)}(Q)$ $=$ ${\cal F}_{(Q)}^{(m)}(Q')$ on the domain of $Ord({\cal O})$, cancel out all same-size unions. Then what remains is the sum on the left side (this sum has the characteristic-wave monomial group for $Q$) such that the sum on the right side does not have any element identical to the characteristic-wave monomial group for $Q$. Call the resulting equation $f_{\bar{N}}^{(m)}(Q)$ $=$ $f_{\bar{N}}^{(m)}(Q')$. 
	\item Prove that when $N_{m+1}$ $=$ $N_{m+2}$ $=$ $\ldots$ $=$ $N_{m+w}$, then for all groups of each of  $f_{\bar{N}}^{(m)}(Q)$ and $f_{\bar{N}}^{(m)}(Q')$, whenever a group's top-level sum has an elementary set with degree without ones ($1$-s), then the value of the size of the group on $N_{m+1}$ $=$ $N_{m+2}$ $=$ $\ldots$ $=$ $N_{m+w}$ $=$ $N^*$ is always $(N^*)^w$. 
	\item From the previous item, the number of such groups in $f_{\bar{N}}^{(m)}(Q)$ is the same as the number of such groups in $f_{\bar{N}}^{(m)}(Q')$. (To prove, just set $N_{m+1}$ $=$ $N_{m+2}$ $=$ $\ldots$ $=$ $N_{m+w}$ $=$ $N^*$) and count the left-hand-side coefficient and the right-hand-side coefficient for the term $(N^*)^w$. 
	\item Prove that for each group, setting each of $N_{m+1}$ $=$ $N_{m+2}$ $=$ $\ldots$ $=$ $N_{m+w}$ to the value {\em unity} results in the same size value ($= 1$) for each group in each of $f_{\bar{N}}^{(m)}(Q)$ and $f_{\bar{N}}^{(m)}(Q')$. Conclude that the total number of all groups in $f_{\bar{N}}^{(m)}(Q)$ equals the total number of all groups in $f_{\bar{N}}^{(m)}(Q')$. Conclude further (from previous item and from this item) that the total number of groups whose top-level sum has an elementary set with degree without ones, is the same in $f_{\bar{N}}^{(m)}(Q)$ and $f_{\bar{N}}^{(m)}(Q')$. 
	
	\item In $f_{\bar{N}}^{(m)}(Q)$ and $f_{\bar{N}}^{(m)}(Q')$, the {\em all-positive-coefficient} sum of all the top-degree (i.e., those not having {\em unity,} i.e., $= 1$, literals in the product) monomials is the same on both sides (otherwise $f_{\bar{N}}^{(m)}(Q)$ and $f_{\bar{N}}^{(m)}(Q')$ as totals would not be the same multivariate polynomials). Then, from the groups in $f_{\bar{N}}^{(m)}(Q)$ and $f_{\bar{N}}^{(m)}(Q')$ being all not-pairwise-the-same, we conclude that it is the {\em groupings} of these top-degree monomials that is the difference between $f_{\bar{N}}^{(m)}(Q)$ and $f_{\bar{N}}^{(m)}(Q')$. 
	
	\item By the number of ``plus'' and ``minus'' signs in the two polynomials $f_{\bar{N}}^{(m)}(Q)$ and $f_{\bar{N}}^{(m)}(Q')$ (after we ``open the parentheses'' of all the groups), we conclude that the number of sets being unioned within each group is pairwise the same between $f_{\bar{N}}^{(m)}(Q)$ and $f_{\bar{N}}^{(m)}(Q')$. (That is, there is a bijection from the set of groups in $f_{\bar{N}}^{(m)}(Q)$ to the set of groups in $f_{\bar{N}}^{(m)}(Q')$, such that the number of sets being unioned between each group and its image under the bijection is the same.) 
	
	\item By contrapositive, we assume that $f_{\bar{N}}^{(m)}(Q')$ does not have the characteristic-wave monomial of $Q$ (note that it has only $N_{m+1}$ through $N_{m+w}$ in the product, as the variables $N_1,\ldots,N_m$ are {\em not} used in either $f_{\bar{N}}^{(m)}(Q)$ or $f_{\bar{N}}^{(m)}(Q')$), call this monomial ${\cal P}^*$, either as a standalone grouping or as a top-level set in any of the size-of-nontrivial-union grouping. Then it must be that there are negative-coefficient terms in both of $f_{\bar{N}}^{(m)}(Q)$ and $f_{\bar{N}}^{(m)}(Q')$. 
	 In addition, by $f_{\bar{N}}^{(m)}(Q')$ being identical to $f_{\bar{N}}^{(m)}(Q)$ after all the groupings have been removed (i.e., after all the parentheses for the groupings have been removed in each polynomial), there are two cases: 

	\begin{itemize} 
		\item Suppose that 	the characteristic-wave monomial ${\cal P}^*$ of $Q$ is {\em absent} in both  $f_{\bar{N}}^{(m)}(Q)$ and  $f_{\bar{N}}^{(m)}(Q')$  after all the groupings have been removed. Then it must be that one of the size-of-union-of-sets groupings in $f_{\bar{N}}^{(m)}(Q')$, call this grouping $G^*$,  has ${\cal P}^*$ with the negative coefficient. This means that ${\cal P}^*$ is present in that grouping as a {\em below-the-top-level term.} From ${\cal P}^*$ being absent from  $f_{\bar{N}}^{(m)}(Q')$  after all the groupings have been removed, we conclude that ${\cal P}^*$ must also be present in one of the size-of-union-of-sets groupings in $f_{\bar{N}}^{(m)}(Q')$, where that grouping $G'$ is different from $G^*$ (!), as a {\em below-the-top-level term}  (from out assumption that ${\cal P}^*$ is not a top-level term in $f_{\bar{N}}^{(m)}(Q')$) with a positive coefficient. 
 	
		\item Suppose that the characteristic-wave monomial ${\cal P}^*$ of $Q$ is {\em present} in both  $f_{\bar{N}}^{(m)}(Q)$ and  $f_{\bar{N}}^{(m)}(Q')$  after all the groupings have been removed. Then it must be that one of the size-of-union-of-sets groupings in $f_{\bar{N}}^{(m)}(Q')$, call this grouping $G^*$, has ${\cal P}^*$ with the positive coefficient. By our assumption that ${\cal P}^*$ is not present in $f_{\bar{N}}^{(m)}(Q')$ either as a standalone grouping or as a top-level set in any of the size-of-nontrivial-union grouping,  this means that ${\cal P}^*$ is present in $G^*$ as a {\em below-the-top-level term.} 

	\end{itemize} 
	
	In both cases above, we discover ${\cal P}^*$ in $f_{\bar{N}}^{(m)}(Q')$, in some group, ``under'' some negative-coefficient level. In $f_{\bar{N}}^{(m)}(Q)$, this same (!!!) negative-coefficient level has under it a (positive-coefficient) term, call it ${\cal P}'$, such that ${\cal P}'$ is different from ${\cal P}^*$ (because ${\cal P}^*$ in $f_{\bar{N}}^{(m)}(Q)$ is not under any negative-coefficient level). 
	
	We now look at an arbitrary term, $T$, being directly under a negative-coefficient level, $L$ in some group. That term's (i.e., $T$'s) variables are directly dependent on the variables in $L$, because each $i$th (i.e., in position $i$) variable in $T$ is a $min$ (under the total order $Ord({\cal O})$) of all the (appropriate) $i$th variables in $L$. Hence we cannot replace the ${\cal P}'$ in $G^*$ by ${\cal P}^*$ because from ${\cal P}'$ and ${\cal P}^*$ not being identical to each other it follows that ${\cal P}'$ and ${\cal P}^*$ must be different in at least one position, under each ``vectorization/linearlization'' of each of ${\cal P}'$ and ${\cal P}^*$. Thus we have the desired contradiction, Q.E.D.


\end{itemize}

} 

{\small 
\bibliographystyle{abbrv}
\bibliography{icdt14arxiv}  
}

\newpage

\appendix 

For the convenience of the reviewers, this optional appendix provides additional information from the online paper \cite{fullversion} concerning the results presented in this current manuscript.

\section{Sufficient condition for bag \\ containment} 
\label{e-three-sec}

In this appendix we provide a detailed discussion of the relationship of Theorem~\ref{cvm-containm-thm} with the following result of \cite{ChaudhuriV93}. 

\begin{theorem}{\cite{ChaudhuriV93}} 
\label{chaudhuri-v-suffic-containm-thm}
Given two CCQ bag queries $Q$ and $Q'$ such that there exists a CVM from $Q'$ to $Q$. Then we have that $Q \sqsubseteq_B Q'$. 
\end{theorem} 

It is easy to see that Theorem~\ref{chaudhuri-v-suffic-containm-thm} is an immediate corollary of  Theorem~\ref{cvm-containm-thm}. 

Note that Theorem~\ref{chaudhuri-v-suffic-containm-thm} is formulated here using the syntax of \cite{Cohen06} that we adopt in this current paper. Recall that in \cite{ChaudhuriV93}, definitions of bag queries are written using the {\em implicit} bag syntax. That is, suppose that we know that a query $Q$ is a bag query. This means that we know that (i) for all the nondistinguished variables of $Q$ that are not copy variables, each such variable is a multiset (noncopy) variable of $Q$, and that (ii) all subgoals of $Q$ are copy-sensitive subgoals. In this case, we can drop altogether from the (explicit, i.e., in the style of \cite{Cohen06}) definition of $Q$ (a) the set $M$, and (b) all copy variables in all subgoals of $Q$ -- just because we know how to interpret all subgoals and all explicit variables of a bag query. 

The resulting implicit notation makes a CVM from $Q'$ to $Q$ ``look like'' a containment mapping. That is,  suppose there exists an ``onto-style containment mapping'' $\mu$ from bag query $Q'$ to bag query $Q$ when the definitions of both queries use the implicit bag syntax of \cite{ChaudhuriV93}. By the definition of $\mu$, we have that 
\begin{itemize} 
	\item[(1)] Each subgoal $l'$ of $Q'$ is associated by $\mu$ with a subgoal $l$ of $Q$, such that $\mu(l')$ and $l$ have identical relational templates; and that 
	\item[(2)] For each subgoal $l$ of $Q$, there exists at least one subgoal $l'$ of $Q'$ such that $\mu$ associates $l'$ with $l$. 
\end{itemize} 

When we change this definition of $\mu$ in such a way that $\mu$ still applies to $Q'$ and $Q$ using the (explicit) syntax of \cite{Cohen06}, it is easy to see that $\mu$ is exactly a CVM from the (explicitly defined) $Q'$ to the (explicitly defined) $Q$. Hence the formulation of Theorem~\ref{chaudhuri-v-suffic-containm-thm} reflects correctly the result of \cite{ChaudhuriV93}.

\section{The Nonsurjective Containment Example}
\label{chaudhuri-v-ex-sec} 

In this appendix we recall an example from \cite{ChaudhuriV93}. 

\begin{example}
\label{chaudhuri-v-ninety-three-example}
Let CCQ queries $Q$ and $Q'$ be as follows. 
\begin{tabbing}
tab me boo \= ju \kill
$Q(X,Z) \leftarrow p(X; i), s(U,X; j), s(V,Z; k), r(Z; l),$ \\
\> $\{ U,V,i,j,k,l \}.$ \\
$Q'(X,Z) \leftarrow p(X; i), s(U,Y; j), s(V,Y; k), r(Z; l),$ \\
\> $\{ U,V,Y,i,j,k,l \}.$
\end{tabbing}

For bag queries $Q$ and $Q'$, the authors of \cite{ChaudhuriV93} claim $Q \sqsubseteq_B Q'$, that is $Q \sqsubseteq_C Q'$ in the context of this present paper. 
\end{example} 

Observe that for the queries $Q$ and $Q'$ of Example~\ref{chaudhuri-v-ninety-three-example},  $(Q,$ $Q')$ is a containment-compatible CCQ pair. (That is, $Q$ and $Q'$ satisfy the necessary containment condition of Theorem~\ref{not-same-num-multiset-vars-thm}.) At the same time, it is easy to check  that no CVM exists from the query $Q'$ to the query $Q$. 

\section{CVMs vs multiset homomorphisms} 
\label{e-four-sec} 

In this appendix we provide Example~\ref{scvm-vs-mhom-ex} showing that general CVMs and multiset homomorphisms are incomparable when applied to pairs of CCQ queries. We also prove Proposition~\ref{sccm-vs-mhomom-prop}.  


\begin{example}
\label{scvm-vs-mhom-ex} 
Let CCQ queries $Q$ and $Q'$ 
be as follows. 
\begin{tabbing}
$Q(A) \leftarrow p(A,B), \ p(A,C), \ \{ B, C \} .$ \\
$Q'(D) \leftarrow p(D,E), \ p(D,F), \ \{ E \} .$ 
\end{tabbing}

Consider mapping $\mu$ from the terms of query $Q$ to the terms of query $Q'$, and mappings $\mu'$ and $\mu''$ from the terms of $Q'$ to the terms of $Q$: $\mu = \{ A \rightarrow D, B \rightarrow E, C \rightarrow E \}$; $\mu' = \{ D \rightarrow A, E \rightarrow B, F \rightarrow B \}$; and $\mu'' = \{ D \rightarrow A, E \rightarrow B, F \rightarrow C \}$. 
Mapping $\mu$ is a CVM but not a multiset-homomorphism (because $\mu$ maps $B$ and $C$ into the same multiset variable $E$ of $Q'$).  
Further,  each of $\mu'$ and $\mu''$ is a multiset-homomorphism but not a CVM. (For each of $\mu'$ and $\mu''$, the image of $\{ E \}$ under the mapping is not a superset of $\{ B,C \}$.) 
\nop{ 
Now consider a CCQ query $Q''$: 
\begin{tabbing}
$Q''(G) \leftarrow p(G,H; k), \ \{ k \} .$ 
\end{tabbing}
Observe that there does not exist a containment mapping from $Q''$ to $Q$ (or to $Q'$), nor does there exist a containment mapping from $Q$ (or $Q'$) to $Q''$, due to the presence of a copy variable in $Q''$ but not in $Q$ or $Q'$. 
} 
\end{example} 

\begin{proof}{(Proposition~\ref{sccm-vs-mhomom-prop})} 
The proof is immediate from Proposition~\ref{cvm-is-gcm-prop}. Indeed, suppose that for an equivalence-compatible CCQ pair $(Q,$ $Q')$, there exists a SCVM $\mu$ from $Q'$ to $Q$. By Proposition~\ref{cvm-is-gcm-prop}, $\mu$ is a generalized containment mapping from $Q'$ to the deregularized version of $Q$. Thus, using Definition~\ref{regulariz-def}, we obtain that condition (4) of Definition~\ref{magic-mapping-def}, when applied to $\mu$, to $Q'$, and to the deregularized version of $Q$, guarantees that condition (3) of Definition~\ref{multiset-homom-def} is satisfied by $\mu$. Observe that by $(Q,$ $Q')$ being an equivalence-compatible CCQ pair, we have that condition (3) of Definition~\ref{magic-mapping-def} for $\mu$ guarantees conditions (4) and (5) of Definition~\ref{multiset-homom-def} for $\mu$. Finally, the satisfaction by $\mu$ (when applied to $Q'$ and $Q$) of conditions (1) and (2)  of Definition~\ref{magic-mapping-def} guarantees the satisfaction by $\mu$ (when applied to $Q'$ and to the deregularized version of $Q$) of  conditions (1) and (2) of Definition~\ref{multiset-homom-def}. The opposite direction (that is, a multiset homomorphism $\varphi$ from $Q'$ to the deregularized version of $Q$ is always a SCVM from $Q'$ to $Q$) is proved using the above proof ``in the opposite direction.''  
\end{proof}  

\section{Proof of Sufficient Condition for a CCQ Query to Be an Explicit-Wave Query} 
\label{suffic-for-expl-wave-sec} 

This appendix provides a proof of Proposition~\ref{suffic-for-expl-wave-prop}.

\begin{proof} 
We prove that for all queries such as $Q$ in the statement of Proposition~\ref{suffic-for-expl-wave-prop}, each such query satisfies Definition~\ref{expl-wave-def}.  Indeed, consider a query $Q$ satisfying the condition that  each copy-sensitive subgoal of $Q$ has no set variables. In case $Q$ has at most one copy-sensitive subgoal, we obtain immediately that $Q$ satisfies Definition~\ref{expl-wave-def} (1). Thus, in the remainder of this proof we assume that $Q$ has at least two copy-sensitive subgoals. We will show that in this case, $Q$ always satisfies Definition~\ref{expl-wave-def} (2). 

Let $(\mu_1,$ $\mu_2)$ be an arbitrary pair of noncopy-permuting GCMs from $Q_{ce}$ to itself such that $\mu_1$ and $\mu_2$ agree on $M_{noncopy}$. 
Consider an arbitrary copy-sensitive subgoal of $Q$, call this subgoal $s$. By definition of the query $Q_{ce}$, $s$ must be a subgoal of $Q_{ce}$. The atom $s$ may have as arguments only constants, head variables of the query $Q$, and multiset variables. (Recall that no set variables of $Q$ may be arguments of $s$.) Now each of $\mu_1$ and $\mu_2$ map each constant to itself (by each of $\mu_1$ and $\mu_2$ being a GCM), and by each of   $\mu_1$ and $\mu_2$ being a mapping from $Q_{ce}$ to itself we obtain that each of $\mu_1$ and $\mu_2$ maps each head variable of $Q_{ce}$ (that is, each head variable of $Q$, by definition of $Q_{ce}$) to itself. Finally, consider each multiset {\em noncopy} variable, call it $Y$, such that $Y$ is an argument of the atom $s$. By $\mu_1$ and $\mu_2$ agreeing on $M_{noncopy}$, we have that $\mu_1(Y)$ and $\mu_2(Y)$ are the same term of the query $Q$. (Recall that, by definition of $Q_{ce}$, we have that for all terms that are present in $Q_{ce}$ but not in $Q$, each such term is a copy variable.)  We conclude that atoms $\mu_1(s)$ and $\mu_2(s)$  have the same relational template. Hence, $Q$ satisfies Definition~\ref{expl-wave-def} (2). Q.E.D. 
\end{proof}

\section{Query $Q$ of Example 4.1 is an\\ implicit-wave query}
\label{abc-app} 

\reminder{Into section title, put manually the ID of example: Example~\ref{intro-weird-ex}} 

In this appendix we show that query $Q$ of Example~\ref{intro-weird-ex} is an implicit-wave query.  We observe first that the query $Q$ has two copy-sensitive subgoals. Now the copy-enhanced version $Q_{ce}$ of $Q$ is exactly $Q$. (Recall that to construct the copy-enhanced version $Q_{ce}$ of query $Q$, all one needs to do is add distinct copy variables to all relational subgoals of $Q$. The query $Q$ of Example~\ref{intro-weird-ex} does not have any relational subgoals.) Consider two GCMs from $Q_{ce}$ to itself. (We use here the formulation, specifically the variable naming, for the query $Q$ as given in the beginning of Appendix~\ref{weird-ex-sec}.)

$\mu_1$ $=$ $\{$ $X_1 \rightarrow X_1$, $Y_1 \rightarrow Y_1$, $Y_2 \rightarrow Y_2$, $X_2 \rightarrow X_2$, $X_3 \rightarrow X_2$, $Y_3 \rightarrow Y_3$, $Y_4 \rightarrow Y_3$  $\}$; and 

$\mu_2$ $=$ $\{$ $X_1 \rightarrow X_1$, $Y_1 \rightarrow Y_1$, $Y_2 \rightarrow Y_2$, $X_2 \rightarrow X_3$, $X_3 \rightarrow X_3$, $Y_3 \rightarrow Y_4$, $Y_4 \rightarrow Y_4$  $\}$. 

The set $M_{noncopy}$ for the query $Q$, as well as for the query $Q_{ce}$, is $\{ Y_1, Y_2 \}$. Each of $\mu_1$ and $\mu_2$ is a noncopy-permuting GCM from $Q_{ce}$ to itself, because each of $\mu_1$ and $\mu_2$ maps each element of $M_{noncopy}$ to itself. For the same reason, the mappings $\mu_1$ and $\mu_2$ agree on $M_{noncopy}$. 

Mappings $\mu_1$ and $\mu_2$ map the first subgoal of $Q$ (which is an original copy-sensitive subgoal of $Q$) into atoms with different relational templates. Indeed,\linebreak  $\mu_1(r(X_1,Y_1,Y_2,X_2; Y_3))$ is atom $r(X_1,Y_1,Y_2,X_2; Y_3)$, and $\mu_2(r(X_1,Y_1,Y_2,X_2; Y_3))$ is atom $r(X_1,Y_1,Y_2,X_3; Y_4)$. We conclude that $Q$ is not an explicit-wave query.

\section{In Example 4.1, $Q$ $\equiv_C$ $Q'$ holds} 
\label{weird-ex-sec} 

\reminder{Put into the title of this section: ID of Example~\ref{intro-weird-ex}} 

In this appendix we show that, for the queries of Example~\ref{intro-weird-ex}, $Q \equiv_C Q'$ holds. 

\begin{proposition} 
\label{weird-equiv-prop} 
For the queries $Q$ and $Q'$ of Example~\ref{intro-weird-ex}, we have that $Q \equiv_C Q'$. 
\end{proposition} 

For the convenience of the exposition in the proof, we use a version of the query $Q'$ where all variables have been renamed into ``same-name'' {\em primed} variables. We also rename the copy variables in a consistent way. That is, we use 

\begin{tabbing} 
Hehetab b \= hehe \kill
$Q(X_1) \leftarrow r(X_1,Y_1,Y_2,X_2; Y_3), r(X_1,Y_1,Y_2,X_3; Y_4),$ \\
\> $\{ Y_1,Y_2,Y_3,Y_4 \} .$ \\
$Q'(X'_1) \leftarrow r(X'_1,Y'_1,Y'_2,X'_2; Y'_3), r(X'_1,Y'_1,Y'_2,X'_2; Y'_4),$ \\
\> $\{ Y'_1,Y'_2,Y'_3,Y'_4 \} .$ 
\end{tabbing}

\begin{proof} 
We will prove the claim of Proposition~\ref{weird-equiv-prop} if we show that for an arbitrary database $D$ and for an arbitrary constant $a$ $\in$ $adom(D)$, the sets $\Gamma^{(a)}_{\bar{S}}(Q,D)$ and $\Gamma^{(a)}_{\bar{S}}(Q',D)$ are of the same cardinality. (Recall the definition of query answer under combined semantics.) To prove this, it is sufficient to show that (for the fixed database $D$ and) for an arbitrary 3-tuple $t$ of  constants from $adom(D)$, the sets $\Gamma_{\bar{S}}(Q,D)[t]$  and $\Gamma_{\bar{S}}(Q',D)[t]$ are of the same cardinality. Here, by the set $\Gamma_{\bar{S}}(Q,D)[t]$ we denote the set of all tuples in $\Gamma_{\bar{S}}(Q,D)$ such that the projection of each tuple on the variables $X_1,Y_1,Y_2$, in this order, is exactly the fixed tuple $t$. Similarly,  by the set $\Gamma_{\bar{S}}(Q',D)[t]$ we denote the set of all tuples in $\Gamma_{\bar{S}}(Q',D)$ such that the projection of each tuple on the variables $X'_1,Y'_1,Y'_2$, in this order, is exactly the fixed tuple $t$. 

We now prove the latter claim. For the fixed database $D$, for the remainder of this proof fix a tuple $t = (a,b,c)$, for some $a,b,c$ $\in$ $adom(D)$, as described above. 

(1) We first show that whenever  the set $\Gamma_{\bar{S}}(Q,D)[t]$ is not empty, the sets $\Gamma_{\bar{S}}(Q,D)[t]$ and $\Gamma_{\bar{S}}(Q',D)[t]$ are of the same cardinality $k^2$, for some constant $k$ $\in$ ${\mathbb N}_+$ where $k$ is a copy number of some ground atom of the database $D$. 

Suppose that the set $\Gamma_{\bar{S}}(Q,D)[t]$ is not empty. Then there must exist in $D$ ground atoms (perhaps identical to each other) $g_1 = r(a,b,c,d; e)$ and $g_2 = r(a,b,c,f; h)$, for some $d,f$ $\in$ $adom(D)$ and for some $e,h$ $\in$ ${\mathbb N}_+$. These atoms $g_1$ and $g_2$ must, intuitively, be the images of the first and of the second subgoal of the query $Q$, respectively, under a valid assignment mapping from $Q$ to $D$. That is, formally, for the set $\Gamma_{\bar{S}}(Q,D)[t]$ to be a nonempty set, it must be that the mapping $\{$ $X_1 \rightarrow a,$ $Y_1 \rightarrow b,$ $Y_2 \rightarrow c,$ $X_2 \rightarrow d,$ $X_3 \rightarrow e$, $Y_3 \rightarrow 1$, $Y_4 \rightarrow 1$  $\}$ is a valid assignment mapping from all the terms of the query $Q$ to the elements of $adom(D)$ $\bigcup$ ${\mathbb N}_+$. The validity of this assignment mapping is justified by the presence of the ground atoms $g_1$ and $g_2$ in the database $D$. 

We now consider all those ground atoms in relation $R$ in the database $D$, such that each of the atoms has $a,b,c$, in this order, as the values of the first three attributes of the relation $R$, from left to right. We know that the set, call it $S[Q]$, of all such atoms is not empty, as $g_1$ and $g_2$ of the previous paragraph will be elements of this set. Now let the constant $k$ $\in$ ${\mathbb N}_+$ be the maximal value of the copy number among all the ground atoms in the set $S[Q]$.  From the set $S[Q]$, choose an arbitrary atom, call it $g$, whose copy number is $k$. Let $g$ be $r(a,b,c,l; k)$, for some $l$ $\in$ $adom(D)$. 

We now argue that for each $n_1,n_2$ $\in$ $\{ 1,\ldots,k \}$ and for the constant $l$ in the ground atom $g$, the mapping $\mu_{(n_1,n_2,l)}$ $=$ $\{$ $X_1 \rightarrow a,$ $Y_1 \rightarrow b,$ $Y_2 \rightarrow c,$ $X_2 \rightarrow l,$ $X_3 \rightarrow l$, $Y_3 \rightarrow n_1$, $Y_4 \rightarrow n_2$  $\}$ is a valid assignment mapping from all the terms of the query $Q$ to the elements of $adom(D)$ $\bigcup$ ${\mathbb N}_+$. Indeed, the required fact is immediate from the definition of the set $\Gamma(Q,D)$ and from the presence of the atom $g$ in the database $D$. 

Further, we argue that for each natural number $n_1$ that is strictly greater than the constant $k$, for each $n_2$ $\in$ ${\mathbb N}_+$, and for {\em each} constant $l$ $\in$ $adom(D)$, the  mapping $\mu_{(n_1,n_2,l)}$ as defined above is not a valid assignment mapping from all the terms of the query $Q$ to the elements of $adom(D)$ $\bigcup$ ${\mathbb N}_+$. Indeed, it is sufficient to observe that the set $S[Q]$ does not have atoms whose copy number is greater than $k$. (Recall that $\mu_{(n_1,n_2,l)}$ fixes the images of the variables $X_1$, $Y_1$, and $Y_2$ to the respective elements of the tuple $t$ $=$ $(a,b,c)$.) We show in a similar way that for each natural number $n_2$ that is strictly greater than the constant $k$, for each $n_1$ $\in$ ${\mathbb N}_+$, and for {\em each} constant $l$ $\in$ $adom(D)$,  the  mapping $\mu_{(n_1,n_2,l)}$ is not a valid assignment mapping from all the terms of the query $Q$ to the elements of $adom(D)$ $\bigcup$ ${\mathbb N}_+$. 

From the facts established about the mappings $\mu_{(n_1,n_2)}$ we conclude that the set $\Gamma_{\bar{S}}(Q,D)[t]$ has exactly $k^2$ elements. Now consider the set $\Gamma_{\bar{S}}(Q',D)[t]$. It is easy to show (in fact, easier than for $\Gamma_{\bar{S}}(Q,D)[t]$ as we did above) that the set $\Gamma_{\bar{S}}(Q',D)[t]$ also has exactly $k^2$ elements. (For each valid assignment mapping $\mu$ from all the terms of the query $Q'$ to the elements of $adom(D)$ $\bigcup$ ${\mathbb N}_+$, such that $\mu$ maps $X'_1$ to $a$, $Y'_1$ to $b$ and $Y'_2$ to $c$, $\mu$ induces a mapping from both subgoals of the query $Q'$ into {\em the same} ground atom of the database $D$. Specifically, for the ground atom $g$ $\in$ $S[Q]$ as defined above, there exists a valid assignment mapping of this form $\mu$,  such that the mapping induces a mapping from both subgoals of the query $Q'$ into the atom $g$.)

(2) Now suppose that for the above fixed $D$ and $t$, the set $\Gamma_{\bar{S}}(Q',D)[t]$ is not empty. We show that in this case, the sets $\Gamma_{\bar{S}}(Q,D)[t]$ and $\Gamma_{\bar{S}}(Q',D)[t]$ are of the same cardinality $p^2$, for some constant $p$ $\in$ ${\mathbb N}_+$ where $p$ is a copy number of some ground atom of the database $D$. The proof is symmetric to the proof of the claim (1) above. Q.E.D. 
\end{proof} 



\section{Necessary and Sufficient Conditions of [9] for Combined-Seman- tics Query Equivalence}
\label{cohen-expl-wave-app} 

Cohen in \cite{Cohen06} provides necessary and sufficient conditions for combined-semantics equivalence of CQ 
queries, possibly with negation, comparisons, and disjunction. For these necessary and sufficient conditions to be applicable, both queries to be tested for combined-semantics equivalence are to satisfy one of the following conditions: 

\begin{enumerate} 
	\item Neither of the two queries has set variables; or  
	\item Neither of the two queries has multiset variables; or  
	\item Neither of the two queries has same-name predicate twice or more in positive (i.e., nonnegated) subgoals; or 
	\item Each query is a join of a set (i.e., no multiset variables) subquery with a multiset (i.e., no set variables) subquery. The formal definition is that neither query may have a subgoal that would have both a multiset variable and a set variable; or   
	\item Neither query may have copy variables. 
\end{enumerate} 

Now consider a restriction of the query language studied in  \cite{Cohen06} to CCQ queries. In the remainder of this section, we consider the above conditions 1--5 only as applied to the queries that satisfy this restriction. (That is, in the remainder of this section we consider CQ combined-semantics queries only, without any extensions of this query language.) Under this query-language restriction, each of the above conditions 1--5 enforces that each {\em CCQ} query in question be an explicit-wave query, by Definition~\ref{expl-wave-def}  in this current paper. Specifically:  

\begin{enumerate} 
	\item Whenever neither of the two queries has set variables, then both queries are explicit-wave queries because in each query, each copy-sensitive subgoal has no set variables. (See Proposition~\ref{suffic-for-expl-wave-prop} in this current paper for this syntactic sufficient condition for a CCQ query to be an explicit-wave query.) 
	\item Whenever neither of the two queries has multiset variables, then neither query has copy-sensitive subgoals. Hence, both queries in question are explicit-wave queries by Definition~\ref{expl-wave-def} (1).   
	\item Whenever neither of the two queries has same-name predicate twice or more in positive (i.e., nonnegated) subgoals, then both queries are explicit-wave queries because neither (CCQ) query has self-joins. (The fact that a CCQ query without self-joins is an explicit-wave query is an easy inference from Definition~\ref{expl-wave-def} (2).)  
	\item Whenever neither query may have a subgoal that would have both a multiset variable and a set variable, then both queries are explicit-wave queries because in each query, each copy-sensitive subgoal has no set variables.   (See Proposition~\ref{suffic-for-expl-wave-prop} in this current paper for this syntactic sufficient condition for a CCQ query to be an explicit-wave query.) 
	\item Whenever neither query may have copy variables, then both queries are explicit-wave queries  by Definition~\ref{expl-wave-def} (1). 
\end{enumerate} 

We conclude that if we apply to {\em only} CCQ queries  the necessary and sufficient conditions of \cite{Cohen06}  for query combined-semantics equivalence, then each of these conditions would be applicable exclusively to pairs of explicit-wave queries.  Thus, {\em when all the queries in question are required to be CCQ queries,} we have that all the necessary and sufficient conditions of \cite{Cohen06} for combined-semantics equivalence of queries are subsumed by Theorem~\ref{necess-suff-equiv-cond-thm} of this current paper. 

Observe that the CCQ query $Q$ of Example~\ref{weird-two-ex} does not satisfy (individually) any of the conditions 1--5 of this section. Thus, none of the necessary and sufficient query-equivalence conditions of \cite{Cohen06} would apply to a pairing of this query with an {\em arbitrary}  query in the query language considered in \cite{Cohen06}. (By definition, see Definition~\ref{ccq-def}, CCQ queries do belong to the query language considered in \cite{Cohen06}.) We make the same observation about the CCQ query $Q'$  of Example~\ref{weird-two-ex}, as well as about the query CCQ $Q$ of Example~\ref{mh-vs-sscm-ex}. Still, by Theorem~\ref{necess-suff-equiv-cond-thm} 
of this current paper we obtain that (i) $Q \equiv_C \hspace{-0.55cm} / \hspace{0.5cm} Q'$ for the queries of Example~\ref{weird-two-ex}, and that (ii)  $Q \equiv_C Q'$ for the queries of Example~\ref{mh-vs-sscm-ex}. 

\end{document}